\def\isarxiv{1}
\documentclass[11pt]{article}
\ifdefined\isarxiv
\usepackage{authblk}
\fi

\usepackage[utf8]{inputenc}
\usepackage[T1]{fontenc}
\usepackage[margin=1in]{geometry}
\usepackage{braket}
\usepackage{amsmath, amsthm, amssymb, mathtools, thmtools}
\usepackage{thm-restate}
\usepackage{amsfonts,mathrsfs}
\usepackage{acronym}
\usepackage{mleftright}
\usepackage{xcolor}
\usepackage{charter}
\DeclareMathAlphabet{\mathcal}{OMS}{zplm}{m}{n}

\usepackage{letltxmacro}
\LetLtxMacro{\originaleqref}{\eqref}
\renewcommand{\eqref}{Eq.~\originaleqref}

\usepackage{tikz}
\usetikzlibrary{arrows}
\definecolor{b2}{rgb}{0.0, 0.53, 0.74}
\definecolor{royalpurple}{rgb}{0.47, 0.32, 0.66}
\definecolor{sinopia}{rgb}{0.8, 0.25, 0.04}

\makeatletter
\def\th@plain{%
  \thm@notefont{}%
  \itshape %
}
\def\th@definition{%
  \thm@notefont{}%
  \normalfont %
}
\makeatother

\usepackage[pagebackref]{hyperref}
\usepackage[capitalise]{cleveref}
\hypersetup{colorlinks=true,citecolor=blue,linkcolor=red}

\crefname{equation}{Eq.}{Eqs.}
\crefname{claim}{Claim}{Claims}
\newcommand*{\email}[1]{\href{mailto:#1}{\nolinkurl{#1}} } 

\usepackage{enumerate,enumitem}
\usepackage{subcaption}
\usepackage{color,soul}
\usepackage{algorithm, algpseudocode, algorithmicx}
\usepackage{comment}
\usepackage{xspace}
\usepackage{stmaryrd}

\numberwithin{equation}{section}

\theoremstyle{plain}
\newtheorem{theorem}{Theorem}[section]
\newtheorem{lemma}[theorem]{Lemma}

\newtheorem{proposition}[theorem]{Proposition}
\newtheorem{corollary}[theorem]{Corollary}

\newtheorem{fact}[theorem]{Fact}
\newtheorem{claim}[theorem]{Claim}

\theoremstyle{definition}
\newtheorem{definition}[theorem]{Definition}

\theoremstyle{remark}
\newtheorem{remark}[theorem]{Remark}

\newcommand{\real}{\mathbb{R}}
\newcommand{\complex}{\mathbb{C}}

\newcommand{\zdom}{\vec{z}_{\mathrm{dom}}}
\newcommand{\ztail}{\vec{z}_{\mathrm{tail}}}
\newcommand{\mudom}{\vec{\mu}_{\mathrm{dom}}}
\newcommand{\mutail}{\vec{\mu}_{\mathrm{tail}}}

\DeclareMathOperator{\diag}{diag}

\newcommand{\mat}[1]{\boldsymbol{#1}}
\renewcommand{\vec}[1]{\boldsymbol{#1}}

\newcommand{\norm}[1]{\mleft\| #1 \mright\|}

\newcommand{\uinorm}[1]{{\left\vert\kern-0.25ex\left\vert\kern-0.25ex\left\vert #1 
                \right\vert\kern-0.25ex\right\vert\kern-0.25ex\right\vert}}

\newcommand{\Id}{\mathbf{I}}

\DeclareMathOperator{\expect}{\mathbb{E}}
\DeclareMathOperator{\prob}{\mathbb{P}}

\newcommand{\order}{\mathcal{O}}

\definecolor{mygreen}{RGB}{80,180,0}

\renewcommand{\d}{\mathrm{d}}
\newcommand{\iu}{\mathrm{i}}
\renewcommand{\Re}{\mathrm{Re}}
\renewcommand{\Im}{\mathrm{Im}}
\renewcommand{\hat}[1]{\widehat{#1}}
\renewcommand{\tilde}[1]{\widetilde{#1}}

\newcommand{\C}{\mathbb{C}}
\newcommand{\mds}{\mathrm{md}}
\newcommand{\unitcircle}{\mathbb{T}}
\newcommand{\for}{\quad \text{for }}
\usepackage[most]{tcolorbox}
\newcommand{\zgap}{\Delta_{\vec{z}}}

\DeclareMathOperator{\Toep}{\mathbf{Toep}}
\DeclareMathOperator{\Vand}{\mat{V}}
\newcommand{\pinv}{+}
\newcommand{\adj}{\dagger}
\newcommand{\vvec}{\vec{v}}

\newcommand{\outprod}[1]{#1^{\vphantom{\adj}}#1^\adj}

\newcommand{\pinvk}[2]{\left({#1}^{\pinv}\right)^{#2}}

\newcommand{\abs}[1]{\left\lvert#1\right\rvert}
\newcommand{\poly}{\operatorname{poly}}

\ifdefined\isarxiv
\author[1]{Zhiyan Ding\thanks{\texttt{zding.m@berkeley.edu}.}}
\author[2]{Ethan N. Epperly\thanks{\texttt{eepperly@caltech.edu}.}}
\author[1,3,4]{Lin Lin\thanks{\texttt{linlin@math.berkeley.edu}.}}
\author[5]{Ruizhe Zhang\thanks{\texttt{rzzhang@berkeley.edu}.}}
\affil[1]{\emph{Department of Mathematics, University of California, Berkeley}}
\affil[2]{\emph{Department of Computing and Mathematical Sciences, California Institute of Technology, Pasadena, CA, USA}}
\affil[3]{\emph{Applied Mathematics and Computational Research Division, Lawrence Berkeley National Laboratory}}
\affil[4]{\emph{Challenge Institute for Quantum Computation, University of California, Berkeley}}
\affil[5]{\emph{Simons Institute for the Theory of Computing}}
\else
\author{}
\fi

\begin{document}
\date{} 
\title{The ESPRIT algorithm under high noise: \\ Optimal error scaling and noisy super-resolution}

\begin{titlepage}
  \maketitle

\begin{abstract}
Subspace-based signal processing techniques, such as the Estimation of Signal Parameters via Rotational Invariant Techniques (ESPRIT) algorithm, are popular methods for spectral estimation.
These algorithms can achieve the so-called super-resolution scaling under low noise conditions, surpassing the well-known Nyquist limit. However, the performance of these algorithms under high-noise conditions is not as well understood.
Existing state-of-the-art analysis indicates that ESPRIT and related algorithms can be resilient even for signals where each observation is corrupted by statistically independent, mean-zero noise of size $\order(1)$, but these analyses only show that the error $\epsilon$ decays at a slow rate $\epsilon=\mathcal{\tilde{O}}(n^{-1/2})$ with respect to the cutoff frequency $n$ (i.e., the maximum frequency of the measurements).
In this work, we prove that under certain assumptions, the ESPRIT algorithm can attain a significantly improved error scaling $\epsilon = \mathcal{\tilde{O}}(n^{-3/2})$, exhibiting noisy super-resolution scaling beyond the Nyquist limit $\epsilon = \mathcal{O}(n^{-1})$ given by the Nyquist-–Shannon sampling theorem. We further establish a theoretical lower bound and show that this scaling is optimal.
Our analysis introduces novel matrix perturbation results, which could be of independent interest.
\end{abstract}
  \thispagestyle{empty}
\end{titlepage}

{\hypersetup{linkcolor=black}
\thispagestyle{empty}
\tableofcontents
\thispagestyle{empty}
}
\newpage
\setcounter{page}{1} 

\section{Introduction}
Spectral estimation is a fundamental problem in statistical signal processing.
The goal of spectral estimation is to reconstruct fine details of a signal from noisy measurements. 
Algorithms for spectral estimation have broad applications spanning image/audio processing and compression \cite{gnh08}, acoustics and electromagnetics~\cite{Krim_1996,Schmidt_1986}, geophysics \cite{ssf90}, inverse problems~\cite{Fannjiang_2010,Fannjiang_remote}, time series analysis~\cite{Cryer2008}, spectrometry \cite{scg96}, machine learning \cite{hbh18,qgr22}, and quantum computing~\cite{Somma2019,LinTong2022}. 

In this paper, we will consider the following formulation of the spectral estimation problem: Consider a positive spectral measure %
\begin{equation*}
    \mu = \sum_{i=1}^{d} \mu_i\delta_{f_i}, \quad \mu_i\ge 0
\end{equation*}
representing the superposition of $d$ point sources.
The \emph{locations} of the sources are denoted by $f_i \in [0, 1)$ and the \emph{intensities} are denoted by $\mu_i$.

We assume access to noisy measurements of the Fourier transform of the spectral measure:
\begin{equation} \label{eq:noisy-signal}
    g_j =\int^1_{0} \exp(2\pi \iu  j \omega) \mathrm{d}\mu(\omega) + E_j \coloneqq \sum_{i=1}^d \mu_i z_i^j + E_j \for j = 0,\ldots,n-1.
\end{equation}
Here, $j$ is the frequency for each measurement and $z_i \coloneqq \exp(2\pi \iu f_i)$ are complex numbers on the unit circle $\unitcircle \coloneqq \{ z \in \complex : |z| = 1\}$. The parameter $n$ is called the \emph{cutoff frequency}\footnote{For some applications, the indices $j$ have units of time, meaning the $f_i$ have units of $\mathrm{time}^{-1}$.
For this reason, the $f_i$ are sometimes referred to as ``frequencies'' in some applications. 
This clash of terminology is unfortunate.
For this work, the term ``frequency'' refers to a frequency $j$ of a complex exponential $\omega \mapsto \exp(2\pi \iu j \omega)$ defining the measurements \cref{eq:noisy-signal}, and the quantities $f_i$ are referred to as ``locations''.}.

The goal of spectral estimation is to recover the locations $f_i$ (or equivalently, the numbers $z_i = \exp(2\pi \iu f_i)$) and the intensities $\mu_i$ from the measurements \cref{eq:noisy-signal}.
In applications, it is typically sufficient to recover only the first $r \ll d$ locations and intensities, which we refer to as the \emph{dominant} part of the spectrum. Suppose $\mu_i$ are in a non-increasing order.
In this setting, the measurements $g_j$ can be decomposed as follows:
\begin{equation*}
    g_j = \underbrace{\sum_{i=1}^r \mu_i z_i^j}_{\text{signal}\, \equiv \,\text{dominant spectrum}} + \underbrace{\sum_{i=r+1}^d \mu_i z_i^j}_{\text{bias} \, \equiv \, \text{tail spectrum}} + \underbrace{E_j}_{\text{measurement noise}}.
\end{equation*}
We see that $g_j$ decomposes as a sum of the signal $\sum_{i=1}^r \mu_i z_i^j$ from the dominant part of the spectrum, a deterministic \emph{bias} $\sum_{i=r+1}^d \mu_i z_i^j$ from the tail of the spectrum, and \emph{measurement noise} $E_j$. 
In this work, we assume the following conditions on the locations, intensities, and measurement noise:

\begin{enumerate}[label=\Alph*.]
\item \textbf{Separation of locations.} All dominant locations are separated from each other and from non-dominant locations:
\begin{equation}\label{eqn:condition_zgap}
\zgap \coloneqq  \min_{1 \leq i \leq r, 1 \leq i'\leq d, i \neq i'} |z_i - z_{i'}| > 0.
\end{equation}
It is important to note that separation between non-dominant locations is not required.

\item \textbf{Separation of intensities.} 
We assume the intensities are positive and ordered
\begin{equation*}
    1\geq \mu_1 \ge \mu_2 \ge \cdots \ge \mu_r > \mu_{r+1} \ge \mu_{r+2} \ge \cdots \mu_d>0\,.
\end{equation*}
We assume the cumulative intensity of non-dominant locations is bounded\footnote{The $1/8$ factor is an arbitrary constant that we chose. We believe any constant strictly less than one (say, 0.999) will work by slightly changing the proof, but we have not checked the proof in detail with this modification.}:
\begin{equation}\label{eqn:condition_mutal}
\mu_{\rm tail} \coloneqq \sum_{i=r+1}^{d} \mu_i\le \frac{1}{8} \cdot \mu_r.
\end{equation}

\item \textbf{Independent, mean-zero, sub-Gaussian measurement noise.} We assume that $\{E_j\}_{j=1}^n$ are independent complex random variables with zero mean $\expect[E_j] = 0$ and sub-Gaussian tail decay. More specifically, we assume there exists $\alpha > 0$ such that
\begin{equation}\label{eqn:alpha_E_j}
\mathbb{P}(|E_j|\geq t)\leq 2\exp\left(-\frac{t^2}{2\alpha^2}\right) \text{ for every } t\geq 0 \text{ and every } j = 0,\ldots,n-1.
\end{equation}

\end{enumerate}

\begin{tcolorbox}[colback=white,colframe=black,width=\columnwidth,boxsep=5pt,arc=4pt]
  \textbf{Spectral estimation problem (with bias and noise):} Under assumptions A--C, estimate the dominant locations $\{f_i\}_{i=1}^r$ (or, equivalently, $\{z_i\}_{i=1}^r$) and intensities $\{\mu_i\}_{i=1}^r$ to precision $\epsilon$.
\end{tcolorbox}

This setup is directly motivated by the eigenvalue estimation problem on quantum computers, which has received significant attention recently~\cite{Somma2019,Stroeks_2022,Dutkiewicz2022heisenberglimited,LinTong2022,wfz22,WanBertaCampbell2022,PhysRevA.108.062408,ShenCampsSzaszEtAl2023,Ding2023simultaneous}.
We also expect that this setup could be applicable in many fields, such as communication, sensing, and audio processing.

When $\mu_{\rm tail}=0$ (no bias) and $\sup_j \abs{E_j}=\poly(\epsilon)$ (small noise), it is possible to estimate the locations (all of which are dominant, since we have assumed $\mu_{\rm tail}=0$) to arbitrary precision $\epsilon$ as long as $n=\Omega(\zgap^{-1})$~\cite{moitra_super-resolution_2015}.  This is the \emph{super-resolution} scaling in signal processing (see e.g.~\cite{don92,crt06,cf13,moitra_super-resolution_2015,9000636,li2022stability}), which is superior to the \emph{Nyquist error scaling} $\epsilon = \mathcal{O}(n^{-1})$ given by the  Nyquist--Shannon sampling theorem. 
This scaling is clearly not feasible in the presence of large bias \emph{or} measurement noise. 
Even when $\mu_{\mathrm{tail}} = 0$, the ability to solve the spectral estimation problem with \emph{arbitrary} precision $\epsilon$ as $n \rightarrow \infty$ is not immediately obvious when the measurement noise is much larger than $\epsilon$, though we may still hope to resolve dominant locations with the Nyquist error scaling.
In support of this, Price and Song introduced a sparse Fourier transform algorithm designed for a continuous version of the spectral estimation problem \cite{ps15}.
The Price--Song algorithm requires access to noisy estimates of the continuous Fourier transform $g(t) = \int_{0}^{1} \exp(2\pi \iu t\omega) \mathrm{d}\mu(\omega)$ for $t \in \mathbb{R}$ and, under the conditions $\mu_{\mathrm{tail}},\ 
\sup_j |E_j| = \mathcal{O}(\mu_r)$, produces estimates of the dominant locations with error $\epsilon = \mathcal{O}(n^{-1})$, where $n\in\mathbb{R}$ represents the maximum $t$ utilized in the algorithm.
This work inspires us to adopt a generalized definition of super-resolution:

\begin{definition}[Noisy super-resolution scaling] \label{def:noisy-super}
    An algorithm satisfies \emph{noisy super-resolution scaling}\footnote{Our use of the term ``noisy super-resolution scaling'' differs from that in  \cite{moitra_super-resolution_2015}, where the term refers to classical super-resolution scaling in the presence of small noise. We note that, in this definition, the error $\epsilon$ may also depend on other parameters, which we suppress in the $o$-notation. Throughout this paper, the estimation error $\epsilon$ is measured in the \emph{matching distance} (see \eqref{eqn:optimal_matching}).} if it recovers the locations up to error strictly superior to the Nyquist error scaling, i.e., $\epsilon = o(n^{-1})$.
\end{definition}

In this paper, we focus on the Estimation of Signal Parameters via Rotational Invariant Techniques (ESPRIT) algorithm~\cite{paulraj_1986_subspace}. %
Since its introduction in 1986, the ESPRIT algorithm has proven one of the most effective methods for spectral estimation in practice.
The main question of this paper is as follows:

\begin{tcolorbox}[colback=white,colframe=black,width=\columnwidth,boxsep=5pt,arc=4pt]
\textbf{Question:} Can ESPRIT achieve a noisy super-resolution scaling for solving the spectral estimation problem with bias and measurement noise?
\end{tcolorbox}

Why might we expect noisy super-resolution scaling for the ESPRIT algorithm?
First, due to noise cancellation, the cumulative impact of the noise might asymptotically decrease relative to the signal as $n$ increases. Specifically, given a fixed, normalized vector $\vec{v}\in\mathbb{C}^n$, the magnitude of the inner product between the signal and $\vec{v}$ can be as large as $\order(\sqrt{n})$. From the central limit theorem, the typical magnitude of the inner product between the noise $\{E_j\}$ and $\vec{v}$ is $\order(1)$, much smaller. 
By adapting the matrix perturbation bounds derived by Moitra in \cite{moitra_super-resolution_2015} and employing matrix concentration inequalities, we prove that ESPRIT can estimate dominant locations with a \emph{central limit error scaling}  $\epsilon=\tilde{\order}(n^{-1/2})$ (\cref{thm:esprit-0.5}).  In particular, when $\mu_{\text{tail}}=0$ and $\sup_{j} \abs{E_j}=\poly(\epsilon)$, this calculation reduce to the standard \emph{super-resolution} scaling in \cite{moitra_super-resolution_2015}. %

Second, a Cram\'er--Rao bound for signal processing (see, e.g., \cite{sn89})  hints that noisy super-resolution scaling may be reachable: the accuracy in estimating dominant locations could potentially be as low as \(\epsilon = \Omega(n^{-3/2-o(1)})\). This superior scaling can be attributed to the substantial increase in information contained in high-frequency data, which extends beyond mere noise cancellation as \(n\) increases. This paper includes a self-contained proof of the lower bound for the spectral estimation problem (\cref{thm:lower_bound}), confirming that no algorithm can achieve an error rate better than \(\epsilon = \Omega(n^{-3/2})\) in estimating dominant locations (\cref{thm:lower_bound}). We refer to \(\epsilon = \Omega(n^{-3/2})\)as the \emph{optimal error scaling}.

The main contribution of this paper is a significant improvement over the central limit error scaling outlined in \cref{thm:esprit-0.5}. We establish that ESPRIT achieves the optimal error scaling for estimating dominant locations \emph{and} dominant intensities (\cref{thm:esprit-1.5}). To our knowledge, this is the first demonstration of a noisy super-resolution scaling (in the sense of \cref{def:noisy-super}) for spectral estimation under comparable assumptions.

\paragraph{Notation.}
A complete summary of our notation is provided in \cref{sec:notation}. 
We make special comments about our use of asymptotic notation.
Throughout this paper, most of our results are non-asymptotic, and we use the big-${\cal O}$ notation $f(n)=\order(g(n))$ to indicate that there exists a universal constant $C>0$ such that $f(n)\leq C\cdot g(n)$. Similarly, $f(n)=\Omega(g(n))$ or $f(n)=\Theta(g(n))$ imply the existence of a uniform constant $C>0$ such that $f(n)\geq C\cdot g(n)$ or $C^{-1}\cdot g(n)\leq f(n)\leq C\cdot g(n)$.
We augment $\order(B)$, $\Omega(B)$, and $\Theta(B)$ with a tilde to suppress polylogarithmic factors in $B$.%

\subsection{ESPRIT algorithm and central limit error scaling}\label{sec:esprit}

The ESPRIT algorithm~\cite{roy_esprit-estimation_1989} starts by rearranging the noisy measurements $\vec{g} = (g_0,\ldots,g_{n-1})$ into a Hankel or Toeplitz matrix. In this paper, we consider the Toeplitz version of the algorithm,
\begin{equation}\label{eqn:toeplitz}
    \hat{\mat{T}} \coloneqq \Toep(\vec{g}) \coloneqq \begin{bmatrix}
        g_0 & \overline{g_1} & \overline{g_2} & \cdots & \overline{g_{n-1}} \\
        g_1 & g_0 & \overline{g_1} & \cdots & \overline{g_{n-2}} \\
        g_2 & g_1 & g_0 & \cdots & \overline{g_{n-3}} \\ 
        \vdots & \vdots & \vdots & \ddots & \vdots \\
        g_{n-1} & g_{n-2} & g_{n-3} & \cdots & g_0
    \end{bmatrix}.
\end{equation}
The ESPRIT algorithm (see \cref{alg:ESPRIT}) is a subspace-type algorithm, which uses properties of the subspace spanned by the $r$ dominant eigenvectors of the Hermitian Toeplitz matrix $\hat{\mat{T}}$. 
Note that the version of ESPRIT discussed here assumes that the intensities $\mu_i$ are real and positive.\footnote{
ESPRIT can also be utilized even when the coefficients $\{\mu_i\}^d_{i=1}$ are complex numbers. In cases where $\{\mu_i\}^d_{i=1}$ are not real and positive, we adapt the first step of ESPRIT by substituting the eigendecomposition with the singular value decomposition. The output of the algorithm consists of dominant locations and intensities with large $|\mu_i|$.
}

\begin{algorithm}[t]
\caption{ESPRIT algorithm for spectral estimation}\label{alg:ESPRIT}
\begin{algorithmic}[1]
    \Require Measurements $\vec{g} = (g_0,\ldots,g_{n-1})\in\complex^n$; number of dominant locations $r$
    \Ensure Estimates $\vec{\hat{z}}$ and $\vec{\hat{\mu}}$ for dominant locations and intensities
    \State $(\hat{\mat{Q}},\hat{\mat{\Sigma}}) \gets \Call{Eig}{\Toep(\vec{g})}$ \Comment{Eigendecomposition with eigenvalues in decreasing order}
    \State $\mat{\hat{Q}}_\uparrow\gets\mat{\hat{Q}}(1:n-1,1:r)$\Comment{First $n-1$ rows and first $r$ columns of $\mat{\hat{Q}}$}
    \State $\mat{\hat{Q}}_\downarrow\gets\mat{\hat{Q}}(2:n,1:r)$\Comment{Last $n-1$ rows and first $r$ columns of $\mat{\hat{Q}}$}
    \State $\mat{\hat{W}} \gets \mat{\hat{Q}}_\uparrow^\pinv \mat{\hat{Q}}_\downarrow$ \Comment{Apply pseudoinverse} \label{line:W_matrix}
    \State $\vec{\hat{\lambda}} \gets \Call{Eig}{\hat{\mat{W}}}$ \Comment{Compute eigendecomposition}
    \State Reorder $\vec{\hat{\lambda}}$ so that $0\leq \mathrm{arg}(\hat{\lambda}_1)\leq \mathrm{arg}(\hat{\lambda}_2)\leq \cdots\leq \mathrm{arg}(\hat{\lambda}_r)<2\pi$ 
    \State $\vec{\hat{z}}_r \gets \exp(\iu\mathrm{arg}(\vec{\hat{\lambda}}))$
    \State $\vec{\hat{\mu}}_r\gets \Vand_n(\hat{\vec{z}}_r)^\pinv \vec{g}$\label{ln:mu_r} \Comment{Optimal intensities by least-squares}
\end{algorithmic}    
\end{algorithm}

Define the location and intensity vectors
\[
\zdom=\left(z_1,z_2,\cdots,z_r\right)\in\mathbb{T}^r,\quad \ztail=\left(z_{r+1},\cdots,z_d\right)\in\mathbb{T}^{d-r}\,,
\]
\[
\mudom=\left(\mu_1,\mu_2,\cdots,\mu_r\right)\in\mathbb{R}_+^r,\quad \vec{\mu}_{\rm tail}=\left(\mu_{r+1},\cdots,\mu_d\right)\in\mathbb{R}_+^{d-r}\,,
\]
In the limit of small bias and noise, $\hat{\mat{T}}$ is an approximation for the Toeplitz matrix of the \emph{exact} signal:
\begin{equation}\label{eqn:T}
\mat{T}=\mat{V}_n(\zdom) \cdot \diag(\mudom) \cdot \mat{V}_n(\zdom)^\dagger\,.
\end{equation}
where $\dagger$ indicates the conjugate transpose and the Vandermonde matrix $\mat{V}_n(\vec{z})\in \C^{n\times k}$ with vector $\vec{z}\in \mathbb{T}^k$ is defined as 
\begin{equation} \label{eq:vander-def}
    \mat{V}_n(\vec{z}) \coloneqq \begin{bmatrix}
        \vvec_n(z_1) & \cdots & \vvec_n(z_k)
    \end{bmatrix} \quad \text{where }~ \vvec_n(z) \coloneqq \begin{bmatrix}
        1 & z & z^2 & \cdots & z^{n-1} 
    \end{bmatrix}^\top.
\end{equation}
The Toeplitz matrix $\mat{\hat{T}}$ can then be written
\begin{equation}\label{T_hat}
\hat{\mat{T}}=\mat{T}+\mat{E}, \quad \mat{E}\coloneqq \mat{E}_{\rm tail}+\mat{E}_{\mathrm{random}}, \quad \mat{E}_{\rm tail} \coloneqq \mat{V}_n(\ztail) \cdot \diag(\mutail) \cdot \mat{V}_n(\ztail)^\dagger\,.
\end{equation}
Here, the measurement noise is also rearranged into a Hermitian Toeplitz matrix $\mat{E}_{\mathrm{random}}=\Toep((E_0,\dots,E_{n-1}))\in \C^{n\times n}$ according to \eqref{eqn:toeplitz}.

When there is no bias or noise, the ESPRIT algorithm can reconstruct the signal exactly \cite{paulraj_1986_subspace}, which is briefly reviewed in \cref{sec:vand_mat_eig} for completeness. More generally, to quantify the error of our recovery, we employ the optimal matching distance: Given $\vec{a},\hat{\vec{a}}\in\complex^r$, the optimal matching distance is
\begin{equation}\label{eqn:optimal_matching}
\mds\left(\hat{\vec{a}},\vec{a}\right)=\min_{\pi}\max_{1\leq i\le r}|\vec{a}_{i}-\hat{\vec{a}}_{\pi(i)}|\,,
\end{equation}
The minimum is taken over all permutations $\pi$ on $\{1,\ldots,r\}$.

We first present the central limit error scaling of the ESPRIT algorithm. 
The proof is given in \cref{sec:esprit-0.5}, which uses matrix concentration inequalities (\cref{lem:Mecks}) to control $\norm{\mat{E}_{\mathrm{random}}}_2$, matrix perturbation theory, and the bound on the singular values of Vandermonde matrices
derived in \cite{moitra_super-resolution_2015}.

\begin{theorem}[Central limit error scaling of ESPRIT] \label{thm:esprit-0.5}
    Consider the spectral estimation problem under assumptions \eqref{eqn:condition_mutal} and \eqref{eqn:alpha_E_j}. Assume %
    $n=\Omega\left(\alpha^2/(\zgap^2\mu_r^2)+1/\zgap\right)$.
    With probability $1 - 1/n^2$,
    \begin{equation}\label{eqn:previous_z_est}
        \mds(\hat{\vec{z}}_r,\zdom)= \order \left( \frac{\mu_{\rm tail}}{\mu_r\zgap n} + \frac{\alpha\sqrt{\log n}}{\mu_r \sqrt{n}} \right).
    \end{equation}
\end{theorem}

\begin{remark}[Classical super-resolution scaling]\label{rem:super_resolution} 
 We observe that the aforementioned theorem can capture the classical super-resolution scenario. Specifically, if $\mu_{\rm tail}=\mathcal{O}(\epsilon)$ and $n=\Omega(1/\zgap)$, setting the noise level $\alpha=\widetilde{\Theta}(\epsilon \sqrt{n}\mu_r)$ is sufficient to guarantee $\mds(\hat{\vec{z}}_r,\zdom)=\mathcal{O}(\epsilon)$. This condition corresponds to the scenario where the tail error and random noise are small, namely $\|\mat{E}\|=\mathcal{O}(\epsilon)$ with high probability, suppressing dependence on $n,\mu_r$ in the $\mathcal{O}$-notation.
\end{remark}

\subsection{Contribution}\label{sec:main_contribution}

The main contribution of this paper is a novel analysis of the ESPRIT algorithm, which demonstrates the optimal error scaling of ESPRIT. These findings are summarized in the following theorem:

\begin{theorem}[Optimal error scaling of ESPRIT, {\normalfont informal version of \cref{thm:esprit-1.5_formal}}] \label{thm:esprit-1.5}
Consider the spectral estimation problem under assumptions \eqref{eqn:condition_mutal} and \eqref{eqn:alpha_E_j}. Assume       $\alpha>1$ and
$n=\Omega\left(\frac{r\alpha^2}{\mu^2_r\zgap^{5/3}}\right)$. %
    With probability $1-1/n^2$, the location estimation of ESPRIT satisfies:%
\begin{align}\label{eqn:z_estimation_error}
       \mds(\hat{\vec{z}}_r,\zdom)&= \tilde{\order} \left( \frac{r^{1.5}\alpha^3}{\mu_r^{3} \zgap^{1.5} n^{1.5}}\right)\,.
    \end{align}
The intensity estimation of ESPRIT satisfies:
\begin{align}\label{eqn:mu_estimation_error}
        \mds(\hat{\vec{\mu}}_r,\mudom)= \tilde{\order}\left(\frac{r^{2.5}\alpha^3}{\mu_r^3\Delta_{\vec{z}}^{1.5}\sqrt{n}}\right)\,.
    \end{align}
Here, $\widetilde{\mathcal{O}}$ suppresses polylogarithmic factors in $n$, $\mu^{-1}_r$, and $\zgap^{-1}$. 
Note that the permutation in the definition of the matching distance is the same for both \cref{eqn:z_estimation_error,eqn:mu_estimation_error}.
\end{theorem}

\begin{remark}[The assumption $\alpha > 1$]
In the above theorem, we assume $\alpha>1$ to simplify the form of our bounds.
Note that \cref{thm:esprit-1.5} still yields results for the ESPRIT algorithm under small noise and zero noise $E_j \equiv 0$ since $\alpha$ is only required to be an \emph{upper bound} on the sub-Gaussian rate of tail decay.
Our proof can be extended to provide sharper estimates for small values of $\alpha$.

\begin{remark}[Intensity normalization]
We assume $\|\vec{\mu}\|_1 =1$ throughout this paper.
Without assuming this normalization, the intensity error \eqref{eqn:mu_estimation_error} will have an extra factor of $\|\vec{\mu}\|_1$.  
\end{remark}

\end{remark}

We provide the proof of Theorem \ref{thm:esprit-1.5} in \cref{sec:pf_main_result}. %
In our analysis, we use the result that $\|\mat{E}_{\mathrm{random}}\|=\widetilde{\mathcal{O}}(\alpha\sqrt{n})$ and $|E_j|=\order(\alpha \sqrt{\log n})$ with a probability at least $1-1/n^2$ (see the detailed discussion in \cref{sec:error-matrix-bounds}, \cref{lem:Mecks}). Moreover, it is worth noting that these same properties remain true (albeit with larger implied constants in the $\tilde{\order}$ notation) for any success probability of $1-1/n^C$, where $C > 0$ is a constant. Thus, our results hold for any desired inverse-polynomial failure probability.

To our knowledge, the above theorem is the \textbf{first} demonstration of $\tilde{\order}(n^{-1.5})$ error scaling---or indeed any noisy super-resolution scaling (\cref{def:noisy-super})---for any frequency estimation algorithm under assumptions A--C.
Moreover, as a byproduct of our $\tilde{\order}(n^{-1.5})$ scaling for the locations $\{z_i\}_{i=1}^r$, we also show that the intensities $\{\mu_i\}_{i=1}^r$ can be recovered with an error of $\mathcal{\tilde{O}}(n^{-0.5})$.

To complement our $\tilde{\order}(n^{-1.5})$ scaling result for the ESPRIT algorithm (\cref{thm:esprit-1.5}), we provide a lower bound showing that this scaling is \textbf{optimal} for any algorithm. The lower bound result is summarized in the following theorem.

\begin{theorem}[Spectral estimation lower bound, {\normalfont informal version of \cref{thm:lower_bound_formal}}]\label{thm:lower_bound} 
Let $\{g_j\}_{j=1}^n$ be defined in \eqref{eq:noisy-signal} and let $\{E_j\}_{j=1}^n$ be i.i.d.\ standard Gaussians ${\cal N}(0,1)$. Then, for any $\eta > 0$, no algorithm can estimate $z_i$ with error $\order(n^{-1.5-\eta})$ or $\mu_i$ with error $\order(n^{-0.5-\eta})$ error for all $i\in \{1,\ldots,r\}$ with constant success probability.

\end{theorem}

We present the proof of \cref{thm:lower_bound} in \cref{sec:lower_bound}.
This result is similar to the Cram\'er--Rao bound in signal processing (cf.\ \cite{sn89}).
Due to the different settings, we provide a self-contained proof using the total variation distance between Gaussian distributions.

\subsection{Related work}\label{sec:related_work}

There is a long list of algorithms and theoretical results have been developed for spectral estimation problems under various assumptions: Approaches include Prony’s method~\cite{Prony_first,sauer_pronys_2018}; subspace-based methods such as ESPRIT \cite{roy_esprit-estimation_1989,9000636,li2022stability,PhysRevA.108.062408}, MUSIC~\cite{Schmidt_1986,lf16,li2022stability}, and the matrix pencil method~\cite{hua_matrix_1990,moitra_super-resolution_2015}; optimization-based methods~\cite{crt06,cf13,cf14,Ding2023simultaneous,6576276}; sparse Fourier transform methods~\cite{Boufou_2012,ps15,jls23}; Fourier interpolation~\cite{ckps16,sswz23,ms24}; and tensor decomposition-based methods~\cite{hk15,cm21}, to name a few. 
Among these works, a substantial body of research~\cite{li2022stability,9000636,Li_2015,Yihong_2020,Fan2022EfficientAF,moitra_super-resolution_2015,cf13,cf14,6576276,ps15,jls23,hk15} is dedicated to the classical ``super-resolution" regime. This regime prioritizes identifying the smallest $n$ required to achieve high precision, operating under the assumption of small noise. In particular, when $\mu_{\rm tail}=0$ and $E_j=\mathcal{O}(\mathrm{poly}(\epsilon))$, setting $n=\Theta(1/\zgap)$ is enough for achieving $\epsilon$-accuracy.

Although the super-resolution regime only requires a small $n$ that is independent of the precision $\epsilon$, the constraints of this regime may be unrealistic in certain real-world applications. For example, in quantum computing, energy estimation problems~\cite{PhysRevA.108.062408,Ding2023simultaneous,LinTong2022,wfz22,WanBertaCampbell2022}, where the goal is to estimate $f_j$ accurately, often involve signals affected by non-dominant signals and random noise of order $\order(1)$ with zero expectation. In this case, the classical super-resolution type assumption $n=\Theta(1/\zgap)$ is no longer adequate.
To tackle this challenge, significant research has been done using approaches based on sparse Fourier transforms (e.g., in~\cite{ps15,jls23} for the continuous setting and in~\cite{Boufou_2012,hikp12a,hikp12_soda,ghikps13,ik14,ikp14,kap16,kap17,kvz19,nsw19} for the discrete setting). For example, in~\cite{ps15}, the authors propose a sparse Fourier transform algorithm that assumes access to the continuous Fourier transform $g(t) = \int_0^1 \exp(2\pi \iu \omega t) \mathrm{d}\mu(\omega)$. To achieve $\epsilon$-accuracy for the location estimation, the algorithm requires the cutoff frequency $n = \max |t| =\Theta(1/\epsilon)$ when $\epsilon$ is smaller than the spectral gap $\zgap$. 
The spectral estimation problem with bias and noise also appears naturally in quantum algorithms for eigenvalue estimation~\cite{Somma2019,Stroeks_2022,Dutkiewicz2022heisenberglimited,PhysRevA.108.062408,Ding2023simultaneous,LinTong2022,wfz22,ShenCampsSzaszEtAl2023,WanBertaCampbell2022}, but the cost model is different from the current setting\footnote{In quantum algorithms for eigenvalue estimation, the observation cost $g_j$ typically increases linearly with $j$. Hence, to gather a dataset for $j=0,\ldots,n-1$, the cumulative cost is $\propto \sum_{j=0}^{n-1} j=\mathcal{O}(n^2)$. Applying the ESPRIT algorithm with the current analysis $n = \tilde{\mathcal{O}}(\epsilon^{-2/3})$ leads to a total cost of $\tilde{\mathcal{O}}(\epsilon^{-4/3})$. Consequently, the total cost of ESPRIT is higher than that of several referenced quantum algorithms, which use a higher cutoff frequency $n = \tilde{\mathcal{O}}(\epsilon^{-1})$ and non-uniform samples to achieve the so-called Heisenberg-limited scaling of the total cost, $\tilde{\mathcal{O}}(\epsilon^{-1})$. This study may inspire further examination of the performance of quantum eigenvalue estimation algorithms when the cutoff frequency is set to $n = \tilde{\mathcal{O}}(\epsilon^{-2/3})$, a topic not yet explored in the existing literature to our knowledge.
}. 
We also note that the access model in this work is different from many previous works.
In this work, we assume we have access to a single noisy measurement at the equispaced frequencies $j = 0,\ldots,n-1$; many other works assume either the ability to query a frequency multiple times, to reduce variance or the ability to query arbitrary real- or integer-valued frequencies $j$.

\subsection{Technical overview}\label{sec:technical_overview}

In this section, we present an overview of the new technical advancements that are essential in our improved analysis.
While these findings are specifically designed for the ESPRIT algorithm, we believe the underlying concepts may also apply to other subspace-based signal processing techniques.
Therefore, we anticipate that these technical results may be of independent interest.

\subsubsection{Vandermonde matrix and eigenbasis}\label{sec:vand_mat_eig}

Recall that the ESPRIT algorithm reduces spectral estimation to matrix operations on the Toeplitz matrix $\mat{\hat{T}}$ of measurements, defined in \cref{eqn:toeplitz}.
In turn, we treat $\mat{\hat{T}}$ as a perturbation of a ``clean'' Toeplitz matrix $\mat{T}$ comprising only the dominant part of the spectrum; see \cref{eqn:T}.
When ESPRIT is run using the clean Toeplitz matrix $\mat{T}$, the algorithm is exact.
Therefore, understanding the ESPRIT algorithm requires matrix perturbation theory to relate the matrices $\mat{\hat{T}}$ and $\mat{T}$.

The beating heart of the ESPRIT algorithm is the relation between two factorizations of the Toeplitz matrix $\mat{T}$, the Vandermonde decomposition (also defined above in \cref{eqn:T})
\begin{equation*}
    \mat{T} = \Vand_n(\zdom) \cdot \diag(\mudom) \cdot \Vand_n(\zdom)^\adj
\end{equation*}
and the eigenvalue decomposition
\begin{equation*}
    \mat{T} = \mat{Q} \mat{\Sigma}\mat{Q}^\adj, \quad \mat{\Sigma} = \diag(\lambda_1(\mat{T}),\ldots,\lambda_r(\mat{T}),\ldots,0,\ldots,0).
\end{equation*}
See \cref{eq:vander-def} for the definition of $\Vand_n(\zdom)$.
Since $\mat{T}$ is Hermitian, its eigenvalues are real, and its eigenvectors are orthonormal. As $\mat{T}$ has rank $r$, it possesses $r$ nonzero eigenvalues, which we arrange in decreasing order. Since $\mat{T}$ has many zero eigenvalues, the eigenvalue decomposition can be expressed more compactly as:
\begin{equation*}
    \mat{T} = \mat{Q}_r^{\vphantom{\adj}} \mat{\Sigma}_r^{\vphantom{\adj}}\mat{Q}_r^\adj, \quad \mat{\Sigma}_r = \diag(\lambda_1(\mat{T}),\ldots,\lambda_r(\mat{T})),
\end{equation*}
where $\mat{Q}_r$ denotes the first $r$ columns of $\mat{Q}$.
We are ultimately interested in computing the Vandermonde decomposition of $\mat{T}$, as it contains the spectral information $\zdom$ and $\mudom$.
Lacking direct access to this decomposition, we instead compute the eigenvalue decomposition and use the ESPRIT algorithm to discover the spectral information.

The first of our main technical results is the following lemma relating the Vandermonde matrix $\Vand_n(\zdom)$ to the basis $\mat{Q}_r$ of eigenvectors. The proof is deferred to \cref{sec:vand-eig}.

\begin{restatable}[Vandermonde matrix and eigenbasis]{lemma}{vandereigen} \label{lem:vander-eigen}
    Assume that $n = \Omega(1/\zgap)$.
    There is a unique invertible matrix $\mat{P}_v$ for which
    \begin{equation}\label{eqn:reshape_Q}
        \mat{Q}_r = \frac{1}{\sqrt{n}} \Vand_n(\zdom) \mat{P}_v.
    \end{equation}
    This matrix is nearly unitary in the sense that
    \begin{equation}\label{eqn:P_orthogonal}
        \norm{\mat{P}_v^\adj \mat{P}_v - \Id_r}_2 = \norm{\mat{P}_v \mat{P}_v^\adj - \Id_r}_2 = \order \left( \frac{1}{\zgap n} \right).
    \end{equation}
    In particular, 
    \begin{equation*}
        \norm{\mat{P}_v}_2, \norm{\smash{\mat{P}_v^{-1}}}_2 = \order(1).
    \end{equation*}
\end{restatable}

Here, and going forward, $\norm{\cdot}_2$ denotes the $\ell_2\to\ell_2$ operator norm.
This result allows us to transfer between the eigenbasis $\mat{Q}_r$ and the Vandermonde matrix $\Vand_n(\zdom)$ by multiplying by a nearly unitary matrix $\mat{P}_v$.
In particular, the matrix $\mat{P}_v$ is a near-isometry, so we do not lose much by interconverting between $\mat{Q}_r$ and $\Vand_n(\zdom)$.
We reference \cref{lem:vander-eigen} throughout our analysis, and the presence of this result is a crucial factor in determining the effectiveness of the ESPRIT algorithm.

\cref{lem:vander-eigen} immediately implies the correctness of the ESPRIT algorithm when it runs with the clean Toeplitz matrix $\mat{T}$. More specifically, let $\mat{P}:=\frac{1}{\sqrt{n}}\mat{P}_v$. By \cref{lem:vander-eigen}, we have $(\mat{Q}_r)_\uparrow = \mat{V}_n(\zdom)_\uparrow \cdot \mat{P}$, and 
\begin{align*}
    (\mat{Q}_r)_\downarrow = \mat{V}_n(\zdom)_\downarrow \cdot \mat{P} = \mat{V}_n(\zdom)_\uparrow \cdot \diag(\zdom) \cdot \mat{P} = (\mat{Q}_r)_\uparrow \cdot \mat{P}^{-1} \diag(\zdom)\mat{P}\,,
\end{align*}
where the second step follows from the structure of the Vandermonde matrix. Then, we get that
\begin{align*}
    \mat{W} = (\mat{Q}_r)_\uparrow^+ (\mat{Q}_r)_\downarrow = (\mat{Q}_r)_\uparrow^+ \cdot (\mat{Q}_r)_\uparrow \cdot \mat{P}^{-1} \diag(\zdom)\mat{P} = \mat{P}^{-1} \diag(\zdom)\mat{P}\,,
\end{align*}
whose eigenvalues are exactly the location vector $\zdom$. Therefore, the ESPRIT algorithm can recover $\zdom$ exactly when the input signal has no biase and noise.

\subsubsection{Second-order perturbation for dominant eigenspace}

The next main technical result is a perturbation expansion for the dominant eigenvectors of the noisy matrix $\mat{\hat{T}}$ in terms of the error matrices $\mat{E} = \mat{E}_{\rm random} + \mat{E}_{\rm tail}$, defined in \cref{T_hat}.
Let 
\begin{equation*}
    \mat{\hat{T}} = \mat{\hat{Q}} \mat{\hat\Sigma}\mat{\hat{Q}}^\adj
\end{equation*}
be the eigenvalue decomposition of $\mat{\hat{T}}$, and let $\mat{\hat{Q}}_r$ denote the first $r$ columns of the eigenmatrix $\mat{\hat{Q}}$.
In the existing analysis of the ESPRIT algorithm~\cite{9000636}, a crucial step involves quantifying the ``distance" between $\hat{\mat{Q}}_r$ and $\mat{Q}_r$.
In our setting, application of standard matrix perturbation theory (i.e., the Davis--Kahan $\sin \theta$ theorem, \cref{thm:sin}) can be used to show that there exists a unitary matrix $\mat{P}_r$ such that
\begin{equation} \label{eq:eigen-elementary}
    \norm{\mat{\hat{Q}}_r \mat{P}_r - \mat{Q}_r}_2 \le \tilde{\order}\left( \frac{1}{\sqrt{n}} \right).
\end{equation}
Here, we have suppressed the dependence of the right-hand side on all parameters except for $n$; see \cref{lem:eigenvectors-weak} for a full statement.
In order to observe the optimal $\tilde{\order}(n^{-1.5})$ scaling of the ESPRIT algorithm, a more precise estimate is required.
To this end, we developed the following new lemma:

\begin{lemma}[Second-order perturbation for dominant eigenspace, {\normalfont informal version of \cref{lem:second-order-formal}}]\label{lem:second_order_expansion_informal}
Assume that the conditions of \cref{thm:esprit-1.5} are satisfied.
Then there exists a unitary %
matrix $\mat{U}_r\in\mathbb{C}^{r\times r}$ such that%
\begin{align*}
\hat{\mat{Q}}_r\mat{U}_r=&~\mat{Q}_r+
\underbrace{\sum^\infty_{k=0}\mat{\Pi}_{\mat{Q}^\perp_r}\left(\mat{E}_{\rm tail}\mat{\Pi}_{\mat{Q}^\perp_r}\right)^{k}\mat{E}_{\rm random}\mat{Q}_r\left(\mat{\Sigma}_r^{-1}\right)^{k+1}+\frac{\alpha\sqrt{\log n}}{\mu_r \sqrt{\Delta_{\vec{z}}n}}\mat{\Pi}_{\mat{Q}_r}\tilde{\mat{Q}}_1}_{\text{first order term}}+\underbrace{\frac{\alpha^2\log n}{\mu_r^2\zgap n}\tilde{\mat{Q}}_{2}}_{\text{second order term}}\,.
\end{align*}
Here $\mat{\Pi}_{\mat{Q}_r}=\mat{Q}_r\mat{Q}^\dagger_r$ is the projector\footnote{Here, and going forward, all projectors are \emph{orthogonal projectors}.}  onto the column space of $\mat{Q}_r$, $\mat{\Pi}_{\mat{Q}^\perp_r} = \Id_r - \mat{\Pi}_{\mat{Q}_r}$ is the projector onto the space that is orthogonal to $\mat{Q}_r$, and $\tilde{\mat{Q}}_1$ and $\tilde{\mat{Q}}_2$ are matrices of $\norm{\cdot}_2$ norm $\mathcal{O}(1)$.
\end{lemma}

Although the result of this lemma is complicated, the main conclusion is this: Up to rotation by a unitary, the perturbed matrix of eigenvectors $\hat{\mat{Q}}_r$ can be decomposed as a sum of four terms:
\begin{enumerate}
    \item The eigenvectors $\mat{Q}_r$ of the noiseless Toeplitz matrix $\mat{T}$, 
    \item A term of size $\tilde{\order}(1/\sqrt{n})$ that is orthogonal to the matrix $\mat{Q}_r$ and depends linearly on the matrix $\mat{E}_{\rm random}$ (whose norm is $\sim \tilde{\order}(1/\sqrt{n})$ with high probability),
    \item A term of size $\tilde{\order}(1/\sqrt{n})$ that is contained in the range of the matrix $\mat{Q}_r$, and
    \item Second-order terms of size $\tilde{\order}(1/n)$.
\end{enumerate}
Our new result \cref{lem:second_order_expansion_informal} has the same $\tilde{\order}(1/\sqrt{n})$ scaling as the standard result \cref{eq:eigen-elementary}, but we have identified the structure of the leading order terms, which is crucial to our sharp analysis of the ESPRIT algorithm.

To prove \cref{lem:second_order_expansion_informal}, we represent the spectral projector $\outprod{\mat{\hat{Q}}_r}$ as a contour integral over the resolvent $\left(\zeta \Id_n - \mat{\hat{T}}\right)^{-1}$ (see \cref{app:resolvents}), expand the resolvent in a Neumann series, and explicitly compute all the terms in the result perturbation series.
This approach results in an explicit convergent series representation of $\outprod{\mat{\hat{Q}}_r}$ (\cref{thm:expansion-explicit}) in terms of \emph{Schur polynomials} (\cref{def:schur_poly}).
From this expansion, it is relatively straightforward to prove \cref{lem:second_order_expansion_informal}.
See \cref{sec:second-order} for details.

We believe that a result of the form \cref{lem:second_order_expansion_informal} should also exist for non-Hermitian Hankel and Toeplitz matrices, allowing our methods to be extended to study the error scaling of other spectral estimation algorithms.

\subsubsection{Perturbation for \texorpdfstring{$\mat{Q}_\uparrow^\pinv \mat{Q}_\downarrow$}{Q\_up Q\_down}}

Up to a final post-processing step, the location estimates $\hat{z}_i$ produced by the ESPRIT algorithm applied to the noisy data \cref{eq:noisy-signal} are the eigenvalues of the matrix $\mat{\hat{Q}}_\uparrow^\pinv \mat{\hat{Q}}_\downarrow$.
Here, $\pinv$ denotes the Moore--Penrose pseudoinverse (\cref{sec:pseudoinverse}) and $\mat{\hat{Q}}_\uparrow,\mat{\hat{Q}}_\downarrow$ denote the first $n-1$ and last $n-1$ rows of $\mat{\hat{Q}}_r$, i.e., $\mat{\hat{Q}}_\uparrow = \mat{\hat{Q}}_r(1:n-1,1:r)$ and $\mat{\hat{Q}}_\downarrow = \mat{\hat{Q}}_r(2:n,1:r)$.
To show ESPRIT is effective, we need to show the eigenvalues of $\mat{\hat{Q}}_\uparrow^\pinv \mat{\hat{Q}}_\downarrow$ are close to the eigenvalues of the matrix $\mat{Q}_\uparrow^\pinv \mat{Q}_\downarrow$ associated with the clean Toeplitz matrix $\mat{T}$.
In general, the matrices $\mat{\hat{Q}}_\uparrow^\pinv \mat{\hat{Q}}_\downarrow$ and $\mat{Q}_\uparrow^\pinv \mat{Q}_\downarrow$ will be far apart. 
However, the eigenvalues of a matrix are invariant under a similarity transformation $\mat{A} \mapsto \mat{P}\mat{A}\mat{P}^{-1}$.
Thus, $\mat{P}\mat{\hat{Q}}_\uparrow^\pinv \mat{\hat{Q}}_\downarrow\mat{P}^{-1}$ has the same eigenvalues as $\mat{\hat{Q}}_\uparrow^\pinv \mat{\hat{Q}}_\downarrow$ for any invertible matrix $\mat{P}$.
Thus, to prove ESPRIT achieves error $\tilde{\order}(n^{-1.5})$ it is sufficient to construct an invertible matrix $\mat{P}$ for which
\begin{equation*}
    \norm{\mat{P}\mat{\hat{Q}}_\uparrow^\pinv \mat{\hat{Q}}_\downarrow\mat{P}^{-1} - \mat{Q}_\uparrow^\pinv \mat{Q}_\downarrow}_2 \le \tilde{\order}\left(\frac{1}{n^{1.5}} \right)
\end{equation*}
and $\mat{P}$ is an approximate isometry $\norm{\mat{P}}_2, \norm{\mat{P}^{-1}}_2 = \order(1)$.
This is exactly the result of the following theorem:

\begin{theorem}[Eigenvectors comparison, strong estimate, {\normalfont informal version of \cref{thm:key_comparision}}]\label{thm:improve_similarity}
Assume that the conditions of \cref{thm:esprit-1.5} are satisfied. Then, with probability $1-1/n^2$, there exists an invertible matrix $\mat{P}\in\mathbb{C}^{r\times r}$ such that%
\begin{equation}\label{eqn:eigen_comparison_strong}
\left\|\mat{P}\mat{\hat{Q}}_\uparrow^\pinv \mat{\hat{Q}}_\downarrow\mat{P}^{-1}-\mat{Q}_\uparrow^\pinv \mat{Q}_\downarrow\right\|_2=\tilde{{\cal O}}\left(\frac{r^{1.5}\alpha^3}{\mu_r^3\Delta_{\vec{z}}^{1.5}n^{1.5}}\right)\,.%
\end{equation}
Moreover, $\mat{P}$ is an approximate isometry:
\begin{equation} \label{eq:p-approx-isometry}
\norm{\mat{P}}_2, \norm{\smash{\mat{P}^{-1}}}_2 = \order(1)\,.
\end{equation}
\end{theorem}

\begin{remark}[The matrix $\mat{P}$] We emphasize that the matrix $\mat{P}$ in the above theorem is not unitary.
In fact, we believe that if $\mat{P}$ were restricted to be unitary, then the best possible scaling for the right-hand side of \cref{eqn:eigen_comparison_strong} would be $\tilde{\order}(1/\sqrt{n})$. In our proof, $\mat{P}$ is explicitly constructed to ensure the column space of $\hat{\mat{Q}}_r\mat{P}^{-1}-\mat{Q}_r$ is almost orthogonal to that of $\mat{Q}_r$. This helps us to ensure 
\begin{equation}\label{eqn:Q_r_orthogonal}
\hat{\mat{Q}}_r\mat{P}^{-1}=\mat{Q}_r+\underbrace{\mat{E}_{\rm first}}_{=\mathcal{O}(n^{-0.5})}+\underbrace{\mat{E}_{\rm second}}_{=\mathcal{O}(n^{-1})}+\mathcal{O}(n^{-1.5})\,,
\end{equation}
where the column spaces of $\mat{E}_{\rm first}$ and $\mat{E}_{\rm second}$ are orthogonal to that of $\mat{Q}_r$. The detailed statement is presented in \cref{lem:good_P_construction} and \cref{sec:align_step1_defer}. 
\end{remark}

Using this estimate, the optimal ESPRIT scaling \cref{thm:esprit-1.5} follows immediately from an application of the Bauer--Fike-type result (\cref{lem:W-to-frequency}), a standard perturbation result for eigenvalues.

This result cannot be proven using standard matrix perturbation theory alone, which only provides a weak estimate of $\tilde{\mathcal{O}}(1/\sqrt{n})$ (\cref{lem:eigenvectors-weak}). To prove Theorem \ref{thm:improve_similarity}, we rely on the novel eigenspace perturbation result \cref{lem:second_order_expansion_informal}, a careful series expansion of $\mat{P}\mat{\hat{Q}}_\uparrow^\pinv \mat{\hat{Q}}_\downarrow\mat{P}^{-1}$, and the Toeplitz structure. Following the structure result in \cref{lem:second_order_expansion_informal}, we construct an invertible matrix $\mat{P}$ satisfying \eqref{eqn:Q_r_orthogonal} (for details see \cref{sec:align_step1_defer} and \cref{lem:good_P_construction}). We then calculate the series expansion of $\mat{P}\mat{\hat{Q}}_\uparrow^\pinv \mat{\hat{Q}}_\downarrow\mat{P}^{-1}$ and plug in the formulas of $\mat{E}_{\rm first}$ and $\mat{E}_{\rm second}$.
Ignoring higher-order terms, proving \cref{thm:improve_similarity} hinges on establishing:
\[
\left\|\mat{Q}_\downarrow^\dagger\left(\mat{E}_\uparrow-\mat{E}_\downarrow(\mat{Q}_\downarrow)^\dagger\mat{Q}_\uparrow\right)\right\|_2\sim n^{-1.5}\quad \text{and}\quad \left\|(\mat{E}_{\rm first})^Q_\uparrow-(\mat{E}_{\rm first})_\downarrow(\mat{Q}_\downarrow)^\dagger\mat{Q}_\uparrow\right\|_2\sim n^{-1}\,,
\]
where $\mat{E}=\mat{E}_{\rm first}+\mat{E}_{\rm second}$. The first bound can be shown by the orthogonality between $\mat{Q}_r$ and $\mat{E}$. Although the left-hand side involves the vectors with $\uparrow$ and $\downarrow$, the difference caused by the missing entry is small, which ensures the bound of $\mathcal{O}(n^{-1})$. The rigorous result is summarized in \cref{cor:second_component}. To establish the second bound, we incorporate the observation that $(\mat{Q}_\downarrow)^\dagger\mat{Q}_\uparrow \approx \mathrm{diag}(\zdom^{-1})$ and the structure of $\mat{E}_{\rm first}$ into the expansion. Because of the Toeplitz structure of the error matrix, the main terms in the expansion are of the form
\[
\frac{1}{\sqrt{n}}\mat{A}_\uparrow\mat{v}_n(z_i) - \frac{1}{\sqrt{n}}\mat{A}_\downarrow\mat{v}_n(z_i)z^{-1}_i\,,
\]
where $\mat{A}$ is a Toeplitz matrix with $\mat{A}=\mathcal{O}(n^{-0.5})$ and $\sup_{i,j}|\mat{A}_{i,j}|=\mathcal{O}(n^{-1})$. The Toeplitz structure allows us to show that these terms are of order $\mathcal{O}(n^{-1})$, rather than the trivial bound of $\mathcal{O}(n^{-0.5})$. The detailed calculation is shown in the proof of \cref{lem:decompose_1,lem:fine_grained_1}.

\subsection{Organization}

The rest of the paper is structured as follows. In Section \ref{sec:esprit-0.5}, we present the direct extension of the previous theoretical findings to prove \cref{thm:esprit-0.5} for completeness. The formal statement and the proof of Theorem \ref{thm:esprit-1.5} are provided in Section \ref{sec:pf_main_result}, under the assumption of the correctness of Lemma \ref{lem:second-order-formal} and Theorem \ref{thm:improve_similarity}. Section \ref{sec:second-order} is dedicated to discussing the proof of Lemma \ref{lem:second-order-formal}, while in Section \ref{sec:strong_sketch}, we provide the proof of Theorem \ref{thm:improve_similarity}. %

In \cref{sec:prelim}, we define the notations and provide some preliminaries. In \cref{sec:van_mat}, we establish some properties of Vandermonde matrix. \cref{app:esprit-0.5,app:second-order,sec:align_proof_defer} contain the technical details and proofs of \cref{sec:esprit-0.5,sec:second-order,sec:strong_sketch}, respectively. In \cref{sec:lower_bound}, we prove the lower bound for the spectral estimation problem.

\section{Proof of the central limit error scaling}\label{sec:esprit-0.5}

We begin by proving \cref{thm:esprit-0.5}, which establishes a $\tilde{\order}\left(n^{-0.5}\right)$ sample-accuracy tradeoff result for the ESPRIT algorithm. To our knowledge, previous analyses of Toeplitz or Hankel-based signal processing algorithms (e.g., \cite{moitra_super-resolution_2015,9000636}) always assume the absence of non-dominant components to the spectrum (i.e., $\mu_{\rm tail}=0$), making these results incomparable to \cref{thm:esprit-0.5}. 
The extension of previous analyses to the scenario of $\mu_{\rm tail}>0$ is a relatively straightforward application of the Davis--Kahan $\sin \theta$ theorem (see~\cref{thm:sin}).
Furthermore, it is worth mentioning that even when $\mu_{\rm tail}=0$, previous analyses only established the central limit-type scaling $\epsilon = \tilde{\order}(n^{-0.5})$.

The main proof roadmap is the same as the previous work~\cite{9000636} with small modifications. 
All technical proofs and calculations are deferred to \cref{app:esprit-0.5}.

The first step of this argument is to develop a quantitative estimate that relates the eigenmatrices $\mat{Q}_r$ of the noise-free data $\mat{T}$ to eigenmatrices $\mat{\hat{Q}}_r$ of the noisy data $\mat{\hat{T}}$:

\begin{lemma}[Eigenvectors comparison, weak estimate] \label{lem:eigenvectors-weak}
    Consider the spectral estimation problem under assumptions \eqref{eqn:condition_mutal} and \eqref{eqn:alpha_E_j}. 
    Assume $n = \Omega(1/\zgap + 1/\mu_r^2)$. With probability at least $1 - 1/n^2$, there exists a unitary matrix $\mat{U}_r\in \C^{r\times r}$ such that
    \begin{align}
        \norm{\mat{\hat{Q}}_r - \mat{Q}_r\mat{U}_r}_2 &\le \order \left(  \frac{\mu_{\rm tail}}{\mu_r\zgap n} + \frac{\alpha\sqrt{\log n}}{\mu_r \sqrt{n}} \right),  \label{eq:Q-compare-weak}\\
        \norm{\mat{\hat{Q}}_\uparrow^\pinv \mat{\hat{Q}}_\downarrow - \mat{U}_r^\adj \mat{Q}_\uparrow^\pinv \mat{Q}_\downarrow \mat{U}_r}_2 &\le \order \left(  \frac{\mu_{\rm tail}}{\mu_r\zgap n} + \frac{\alpha\sqrt{\log n}}{\mu_r \sqrt{n}} \right) \label{eq:W-compare-weak}\,,
    \end{align}
where $\mat{Q}_\uparrow=\mat{Q}(1:n-1,1:r)$ and $\mat{Q}_\downarrow=\mat{Q}(2:n,1:r)$.
\end{lemma}

This result is similar to existing results such as \cite[Lem.~2.7]{moitra_super-resolution_2015} and \cite[Lems.~2, 5, \& 6]{9000636} and our proof (in~\cref{sec:eigenvectors-weak})  is based on similar ideas.
The main ingredients in our proof of this result are Moitra's bounds on the singular values of Vandermonde matrices (\cref{sec:moitra}), standard results in matrix perturbation theorem (\cref{app:eigen-perturbation}), and a comparison lemma between the Vandermonde and eigenbases (\cref{lem:vander-eigen}). 

To convert the bound \cref{eq:W-compare-weak} into actionable information for the ESPRIT algorithm, we use the following result:

\begin{lemma}[From $\mat{Q}_\uparrow^+\mat{Q}_\downarrow$ bounds to location estimation] \label{lem:W-to-frequency}
    Assume $n = \Omega(1/\zgap)$ and let $\mat{P}$ be an invertible matrix satisfying $\norm{\mat{P}}_2,\norm{\smash{\mat{P}^{-1}}}_2 = \order(1)$. Then the output of the ESPRIT algorithm on the noisy signal $\{g_j\}_{j=0}^{n-1}$ satisfies
    \begin{equation}\label{eq:Lemma_2_2_weak}
        \mds(\hat{\vec{z}}_r,\zdom)= \order\left( r \norm{\mat{\hat{Q}}_\uparrow^\pinv \mat{\hat{Q}}_\downarrow - \mat{P}^{-1} \mat{Q}_\uparrow^\pinv \mat{Q}_\downarrow \mat{P}}_2\right).
    \end{equation}
    Furthermore, if
    \begin{equation}\label{eqn:Lemma_2_2_improve}
         \norm{\mat{\hat{Q}}_\uparrow^\pinv \mat{\hat{Q}}_\downarrow - \mat{P}^{-1} \mat{Q}_\uparrow^\pinv \mat{Q}_\downarrow \mat{P}}_2 < c\zgap/2,
    \end{equation}
    for a sufficiently small constant $c\in (0,1)$ that only depends on $\mat{P}$, then
    \begin{equation}\label{eq:Lemma_2_2_strong}
        \mds(\hat{\vec{z}}_r,\zdom)= \order\left( \norm{\mat{\hat{Q}}_\uparrow^\pinv \mat{\hat{Q}}_\downarrow - \mat{P}^{-1} \mat{Q}_\uparrow^\pinv \mat{Q}_\downarrow \mat{P}}_2\right).
     \end{equation}
\end{lemma}

This result slightly improves Lemma 2 in \cite{9000636} and is proven by the combination of the Bauer--Fike theorem (\cref{thm:bauer-fike}) and the resolvent approach for eigenvalue perturbation (\cref{lem:eig_location}). The proof is deferred to \cref{sec:W-to-frequency}. %

Combining \cref{lem:eigenvectors-weak,lem:W-to-frequency} with $\mat{P} = \mat{U}$ immediately yields the $\tilde{\order}(n^{-0.5})$ error bound in \cref{thm:esprit-0.5}.

Ultimately, the loss in this argument is that the perturbation bounds in \cref{lem:eigenvectors-weak} are too weak.
To prove the $\tilde{\order}(n^{-1.5})$ error scaling (i.e., \cref{thm:esprit-1.5}), we will employ the stronger estimate \cref{thm:improve_similarity}.
Combining this strong estimate \cref{thm:improve_similarity} with \cref{lem:W-to-frequency} will prove \cref{thm:esprit-1.5}.

\section{Proof of the optimal error scaling}\label{sec:pf_main_result}

In this section, we present the formal statement of our main result (\cref{thm:esprit-1.5}) and the proof. %

\begin{theorem}[Optimal error scaling of ESPRIT, {\normalfont formal version of \cref{thm:esprit-1.5}}] \label{thm:esprit-1.5_formal}
Consider the spectral estimation problem under assumptions \eqref{eqn:condition_mutal} and \eqref{eqn:alpha_E_j}. Assume       $\alpha>1$ and 
\begin{align}\label{eq:proof_main_n_choice_small}
    n=\Omega\left(\left(\frac{r\alpha^{4/3}}{\mu_r^{4/3}\Delta_{\vec{z}}^{2/3}} + \frac{\alpha^{2}}{\mu_r^2\Delta_{\vec{z}}}\right)\cdot \log(r,\alpha,1/\mu_r,1/\Delta_{\vec{z}})\right)\,.
\end{align}
With probability $1-1/n^2$, the location estimation of ESPRIT satisfies:
\begin{align}\label{eqn:z_estimation_error_formal}
       \mds(\hat{\vec{z}}_r,\zdom)&= {\cal O}\left(\left(\frac{\alpha\sqrt{\log n}}{\mu_r\sqrt{\Delta_{\vec{z}}} }+r^{3/2}\right)\frac{r\alpha^2\log n}{\mu_r^2\Delta_{\vec{z}}n^{3/2}}\right)\,.
\end{align}
    
If we further assume 
\begin{align}\label{eq:proof_main_n_choice}
    n=\Omega\left(\left(\frac{r\alpha^{4/3}}{\mu_r^{4/3}\Delta_{\vec{z}}^{4/3}} + \frac{\alpha^{2}}{\mu_r^2\Delta_{\vec{z}}^{5/3}}\right)\cdot \log(r,\alpha,1/\mu_r,1/\Delta_{\vec{z}})\right)\,,
\end{align}
then the location estimation satisfies:
\begin{align}\label{eq:z_est_error_strong_formal}
    {\rm md}(\hat{\vec{z}}_r,\zdom)={\cal O}\left(\left(\frac{\alpha\sqrt{\log n}}{\mu_r\sqrt{\Delta_{\vec{z}}} }+r^{3/2}\right)\frac{\alpha^2\log n}{\mu_r^2\Delta_{\vec{z}}n^{3/2}}\right)\,,
\end{align}
and the intensity estimation of ESPRIT satisfies:
\begin{align}\label{eqn:mu_estimation_error_formal}
        \mds(\hat{\vec{\mu}}_r,\mudom)= {\cal O}\left(\left(\frac{\alpha\sqrt{\log n}}{\mu_r\sqrt{\Delta_{\vec{z}}} }+r^{3/2}\right) \frac{r\alpha^2\log n}{\mu_r^2\Delta_{\vec{z}}\sqrt{n}}\right)\,. 
\end{align}
Note that the permutation in the definition of the matching distance is the same for both \cref{eq:z_est_error_strong_formal,eqn:mu_estimation_error_formal}.
\end{theorem}

The high-level idea of the proof is as follows. First, the location estimation result, established in \cref{thm:esprit-1.5}, is a direct consequence of the strong eigenvector comparison result introduced in \cref{thm:improve_similarity}. Specifically, after demonstrating that $\mat{\hat{Q}}_\uparrow^\pinv \mat{\hat{Q}}_\downarrow$ is approximately similar to $\mat{Q}_\uparrow^\pinv \mat{Q}_\downarrow$, it becomes apparent that the eigenvalues of $\mat{\hat{Q}}_\uparrow^\pinv \mat{\hat{Q}}_\downarrow$ and $\mat{Q}_\uparrow^\pinv \mat{Q}_\downarrow$ are within $\mathcal{O}(n^{-1.5})$ proximity to each other (see \cref{lem:W-to-frequency}). Subsequently, establishing the proximity between $\hat{\vec{z}}_r$ and $\zdom$ allows us to conclude that $\Vand_n(\hat{\vec{z}}_r)$ closely approximates the exact Vandermonde matrix $\Vand_n(\zdom)$. The intensity estimation result can then be derived through perturbation theory applied to the least-squares problem.

\begin{proof}[Proof of \cref{thm:esprit-1.5_formal}]
We prove the error scaling of the location estimation and intensity estimation of the ESPRIT algorithm below.
\paragraph{Location estimation.}
According to \cref{thm:key_comparision} (formal version of \cref{thm:improve_similarity}) and the assumption on $n$ (\cref{eq:proof_main_n_choice_small}), with probability at least $1-1/n^2$, the following holds:
\begin{equation}\label{eq:proof_main_loc_cond}
\begin{aligned}
    \left\|\mat{P}\right\|_2\left\|\smash{\mat{P}^{-1}}\right\|_2\cdot \norm{\mat{\hat{Q}}_\uparrow^\pinv \mat{\hat{Q}}_\downarrow - \mat{P}^{-1} \mat{Q}_\uparrow^\pinv \mat{Q}_\downarrow \mat{P}}_2 = &~ {\cal O}(1)\cdot \norm{\mat{\hat{Q}}_\uparrow^\pinv \mat{\hat{Q}}_\downarrow - \mat{P}^{-1} \mat{Q}_\uparrow^\pinv \mat{Q}_\downarrow \mat{P}}_2 \\
    = &~ {\cal O}\left(\left(\frac{\alpha\sqrt{\log n}}{\mu_r\sqrt{\Delta_{\vec{z}}} }+r^{3/2}\right)\frac{\alpha^2\log n}{\mu_r^2\Delta_{\vec{z}}n^{3/2}}\right)\,,
\end{aligned}
\end{equation}
where we use $\left\|\mat{P}^{-1}\right\|_2,\left\|\mat{P}\right\|_2=\mathcal{O}(1)$ by \cref{eq:p-approx-isometry} in the first step. 

Then, by \cref{eq:Lemma_2_2_weak} in \cref{lem:W-to-frequency}, we obtain that
\begin{align*}
    {\rm md}(\hat{\vec{z}}_r,\zdom)={\cal O}\left(\left(\frac{\alpha\sqrt{\log n}}{\mu_r\sqrt{\Delta_{\vec{z}}} }+r^{3/2}\right)\frac{r\alpha^2\log n}{\mu_r^2\Delta_{\vec{z}}n^{1.5}}\right)\,.
\end{align*}
The first location estimation guarantee (\eqref{eqn:z_estimation_error_formal}) is proved.

If we further assume that a stronger condition on $n$ (\cref{eq:proof_main_n_choice}) holds,
then by \cref{eq:proof_main_loc_cond}, we have
\begin{align*}
    \norm{\mat{\hat{Q}}_\uparrow^\pinv \mat{\hat{Q}}_\downarrow - \mat{P}^{-1} \mat{Q}_\uparrow^\pinv \mat{Q}_\downarrow \mat{P}}_2 < c\frac{\Delta_{\vec{z}}}{2}\,,
\end{align*}
where $c\in (0,1)$ that satisfies the condition (\cref{eqn:Lemma_2_2_improve}) in \cref{lem:W-to-frequency}. By \cref{eq:Lemma_2_2_strong}, we obtain that
\begin{align*}
    {\rm md}(\hat{\vec{z}}_r,\zdom)={\cal O}\left(\left(\frac{\alpha\sqrt{\log n}}{\mu_r\sqrt{\Delta_{\vec{z}}} }+r^{3/2}\right)\frac{\alpha^2\log n}{\mu_r^2\Delta_{\vec{z}}n^{1.5}}\right)\,.
\end{align*}
The second location estimation guarantee (\cref{eq:z_est_error_strong_formal}) is proved.

\paragraph{Intensity estimation.}
To show \eqref{eqn:mu_estimation_error}, we assume that the optimal permutation $\pi$ is identity without loss of generality. Let $\mat{V}_r=\mat{V}_n(\zdom)$ and $\hat{\mat{V}}_r=\mat{V}_n(\hat{\mat{z}}_r)$. By Line~\ref{ln:mu_r} of \cref{alg:ESPRIT}, we have
\[
\begin{aligned}
\vec{\hat{\mu}}_r=\hat{\mat{V}}_r^+\vec{g}=\left(\hat{\mat{V}}_r^\dagger\hat{\mat{V}}_r\right)^{-1}\hat{\mat{V}}_r^\dagger\big(\mat{V}_r\mudom+\mat{V}_n(\ztail)\vec{\mu}_{\rm tail}+\mat{E}_{\rm random}(:,1)\big)\,.
\end{aligned}
\]
Here, $\mat{E}_{\rm random}(:,1)$ represents the first column of $\mat{E}_{\rm random}$.
This implies
\begin{equation}\label{eqn:difference}
\begin{aligned}
\left\|\vec{\hat{\mu}}_r-\mudom\right\|_2
\leq \left\|\left(\hat{\mat{V}}_r^+\mat{V}_r-\mat{I}_r\right)\mudom\right\|_2+\left\|\hat{\mat{V}}_r^+\mat{V}_n(\ztail)\mutail\right\|_2+\left\|\hat{\mat{V}}_r^+\mat{E}_{\rm random}(:,1)\right\|_2\,.
\end{aligned}
\end{equation}

To bound the three terms, we first observe that when $n$ is sufficiently large, i.e., \cref{eq:proof_main_n_choice} holds, the location estimation (\cref{eqn:z_estimation_error}) guarantees that $\|\vec{z}_{\rm dom}-\hat{\vec{z}}_r\|_\infty <\frac{1}{4}\Delta_{\vec{z}}$, which implies that
\begin{align}\label{eq:z_hat_z_sep}
    \min_{1 \leq i,j\leq r, i \neq j} |\hat{z}_i - \hat{z}_{j}|\geq \Delta_{\vec{z}}/2,\quad \min_{1 \leq i\leq r, 1\leq j\leq d, i \neq j} |\hat{z}_i - z_{j}|\geq \Delta_{\vec{z}}/2\,.
\end{align}

For the first term in \cref{eqn:difference}, we have
\begin{equation}\label{eqn:first_term_1}
\begin{aligned}
\norm{(\hat{\mat{V}}_r^+ \mat{V}_r-\mat{I}_r)\vec{\mu}_{\rm dom}}_2=&~\left\|\left(\left(\hat{\mat{V}}_r^\dagger\hat{\mat{V}}_r\right)^{-1}\hat{\mat{V}}_r^\dagger\mat{V}_r-\mat{I}_r\right)\mudom\right\|_2\\
=&~\mathcal{O}\left(r^{1/2}\right)\cdot \norm{\hat{\mat{V}}_r^+ \hat{\mat{V}}_r-\mat{I}_r + \hat{\mat{V}}_r^+ (\mat{V}_r-\hat{\mat{V}}_r)}_2\\
= &~ \mathcal{O}\left(r^{1/2}\right)\cdot \norm{\hat{\mat{V}}_r^+}_2\cdot \norm{\mat{V}_r-\hat{\mat{V}}_r}_2\,,
\end{aligned}
\end{equation}
where the second equality follows from $\norm{\vec{\mu}_{\rm dom}}_2={\cal O}(r^{1/2})$, and the last equality follows from $\hat{\mat{V}}_r^+ \hat{\mat{V}}_r=\mat{I}_r$ by the definition of pseudoinverse.
For $\big\|\hat{\mat{V}}_r^+\big\|_2$, since $\hat{\mat{V}}_r$ is a Vandermonde matrix, by \cref{eq:z_hat_z_sep} and \cref{thm:moitra}, 
\begin{align}\label{eq:V_r_hat_inv}
    \left\|\hat{\mat{V}}_r^+\right\|_2=\sigma_{\min}\left(\hat{\mat{V}}_r\right)^{-1}={\cal O}\left(\frac{1}{\sqrt{n}}\right)\,.
\end{align}
Next, for $\norm{\mat{V}_r-\hat{\mat{V}}_r}_2$,  we have
\begin{align*}
   \norm{\mat{V}_r-\hat{\mat{V}}_r}_2^2 \leq \norm{\mat{V}_r-\hat{\mat{V}}_r}_F^2 = &~ \sum_{1\leq i\leq r,0\leq j\leq n-1}|z^j_i-\hat{z}^j_i|^2\\
   = &~ \sum_{1\leq i\leq r,0\leq j\leq n-1}{\cal O}\left(j^2\right)\cdot |z_i-\hat{z}_i|^2\\
   = &~ \sum_{0\leq j\leq n-1}{\cal O}(j^2)\cdot r\cdot {\cal O}\left(\left(\frac{\alpha^2\log n}{\mu_r^2\Delta_{\vec{z}} }+r^3\right)\frac{\alpha^4\log^2 n}{\mu_r^4\Delta_{\vec{z}}^2n^{3}}\right)\\
   = &~ {\cal O}\left(\left(\frac{\alpha^2\log n}{\mu_r^2\Delta_{\vec{z}} }+r^3\right)\frac{r\alpha^4\log^2 n}{\mu_r^4\Delta_{\vec{z}}^2}\right)\,,
\end{align*}
where the second equality follows from $|z_i|=|\hat{z}_i|=1$ for any $i\in [r]$, and the third equality follows from the location estimation guarantee (\cref{eq:z_est_error_strong_formal}). Hence,
\begin{align}\label{eq:V_r_V_r_hat_diff}
    \norm{\mat{V}_r-\hat{\mat{V}}_r}_2 = {\cal O}\left(\left(\frac{\alpha\sqrt{\log n}}{\mu_r\sqrt{\Delta_{\vec{z}}} }+r^{3/2}\right)\frac{\sqrt{r}\alpha^2\log n}{\mu_r^2\Delta_{\vec{z}}}\right)\,.
\end{align}
Plugging \cref{eq:V_r_hat_inv,eq:V_r_V_r_hat_diff} back into \eqref{eqn:first_term_1}, we obtain
\[
\left\|\left(\hat{\mat{V}}_r^+\mat{V}_r-\mat{I}_r\right)\mudom\right\|_2={\cal O}\left(\left(\frac{\alpha\sqrt{\log n}}{\mu_r\sqrt{\Delta_{\vec{z}}} }+r^{3/2}\right)\frac{r\alpha^2\log n}{\mu_r^2\Delta_{\vec{z}}\sqrt{n}}\right)\,.%
\]

For the second term in \cref{eqn:difference}, using \cref{eq:z_hat_z_sep} and a similar proof as \cref{prop:tail-error}, we have
\[
\left\|\hat{\mat{V}}_r^\dagger\mat{V}_n(\ztail)\mutail\right\|_2\leq \sum^d_{j=r+1}\mu_j\left\|\hat{\mat{V}}_r^\dagger\vec{v}_n(z_j)\right\|_2=\mathcal{O}\left(\frac{\mu_{\rm tail}}{\Delta_{\vec{z}}}\right)\,.
\]
Similar to \cref{eq:V_r_hat_inv}, by \cref{eq:z_hat_z_sep} and \cref{thm:moitra}, we also have $\|(\hat{\mat{V}}_r^\dagger\hat{\mat{V}}_r)^{-1}\|_2={\cal O}(\frac{1}{n})$.
Combining them together, we obtain
\begin{equation}\label{eqn:last_term}
\left\|\hat{\mat{V}}_r^+\mat{V}_n(\ztail)\mutail\right\|_2\leq \norm{(\hat{\mat{V}}_r^\dagger\hat{\mat{V}}_r)^{-1}}_2\cdot \left\|\hat{\mat{V}}_r^\dagger\mat{V}_n(\ztail)\mutail\right\|_2=\mathcal{O}\left(\frac{\mu_r}{\Delta_{\vec{z}}n}\right)\,.
\end{equation}

For the last term in \cref{eqn:difference}, by the concentration inequality of sub-Gaussian vector (see e.g., \cite{jng19}), $\|\mat{E}_{\rm random}(:,1)\|_2={\cal O}(\alpha \sqrt{n\log n})$ with probability $1-1/(2n^2)$. Then, we have
\[
\begin{aligned}
&\left\|\hat{\mat{V}}_r^+\mat{E}_{\rm random}(:,1)\right\|_2\\
\leq&~\left\|\left(\hat{\mat{V}}_r^\dagger\hat{\mat{V}}_r\right)^{-1}\mat{V}^\dagger_r\mat{E}_{\rm random}(:,1)\right\|_2+\left\|\left(\hat{\mat{V}}_r^\dagger\hat{\mat{V}}_r\right)^{-1}\right\|_2\cdot \left\|\hat{\mat{V}}_r-\mat{V}_r\right\|_2\cdot \left\|\mat{E}_{\rm random}(:,1)\right\|_2\\
=&~\left\|\left(\hat{\mat{V}}_r^\dagger\hat{\mat{V}}_r\right)^{-1}\mat{V}^\dagger_r\mat{E}_{\rm random}(:,1)\right\|_2+{\cal O}\left(\left(\frac{\alpha\sqrt{\log n}}{\mu_r\sqrt{\Delta_{\vec{z}}} }+r^{3/2}\right)\frac{\sqrt{r}\alpha^3\log^{3/2} n}{\mu_r^2\Delta_{\vec{z}}\sqrt{n}}\right)\\
\leq&~{\cal O}\left(\frac{1}{n}\right)\cdot \left\|\mat{V}^\dagger_r\mat{E}_{\rm random}(:,1)\right\|_2+{\cal O}\left(\left(\frac{\alpha\sqrt{\log n}}{\mu_r\sqrt{\Delta_{\vec{z}}} }+r^{3/2}\right)\frac{\sqrt{r}\alpha^3\log^{3/2} n}{\mu_r^2\Delta_{\vec{z}}\sqrt{n}}\right)\\
\leq &~ {\cal O}\left(\frac{\sqrt{r}}{n}\right)\cdot \max_{1\leq i\leq r}~\left|\mat{v}^\dagger_n(z_i)\mat{E}_{\rm random}(:,1)\right|+{\cal O}\left(\left(\frac{\alpha\sqrt{\log n}}{\mu_r\sqrt{\Delta_{\vec{z}}} }+r^{3/2}\right)\frac{\sqrt{r}\alpha^3\log^{3/2} n}{\mu_r^2\Delta_{\vec{z}}\sqrt{n}}\right)\,%
\end{aligned}
\]
Here, we use \cref{eq:V_r_V_r_hat_diff} in the third step and \cref{eq:V_r_hat_inv} in the third and fourth steps.
Notice that $\mat{v}^\dagger_n(z_i)\mat{E}_{\rm random}(:,1)$ is sub-Gaussian with parameter $\norm{\mat{v}_n(z_i)}_2\cdot \alpha=\alpha\sqrt{n}$. By sub-Gaussian concentration and union bound, we obtain that with probability $1-1/(2n^2)$,
\begin{align*}
    \max_{1\leq i\leq r}~\left|\mat{v}^\dagger_n(z_i)\mat{E}_{\rm random}(:,1)\right| = {\cal O}(\alpha \sqrt{n\log (nr)})\,.
\end{align*}
Hence, we obtain
\begin{align*}
    \left\|\hat{\mat{V}}_r^+\mat{E}_{\rm random}(:,1)\right\|_2 = &~ {\cal O}\left(\frac{\sqrt{r}\alpha \sqrt{\log (nr)}}{\sqrt{n}}+\left(\frac{\alpha\sqrt{\log n}}{\mu_r\sqrt{\Delta_{\vec{z}}} }+r^{3/2}\right)\frac{\sqrt{r}\alpha^3\log^{3/2} n}{\mu_r^2\Delta_{\vec{z}}\sqrt{n}}\right)\\
    = &~ {\cal O}\left(\left(\frac{\alpha\sqrt{\log n}}{\mu_r\sqrt{\Delta_{\vec{z}}} }+r^{3/2}\right)\frac{\sqrt{r}\alpha^3\log^{3/2} n}{\mu_r^2\Delta_{\vec{z}}\sqrt{n}}\right)\,.
\end{align*}

Plugging these three terms back into \cref{eqn:difference}, we have
\begin{align*}
    \left\|\vec{\hat{\mu}}_r-\mudom\right\|_2 = &~ {\cal O}\left(\left(\frac{\alpha\sqrt{\log n}}{\mu_r\sqrt{\Delta_{\vec{z}}} }+r^{3/2}\right)\cdot \left(\frac{r\alpha^2\log n}{\mu_r^2\Delta_{\vec{z}}\sqrt{n}} + \frac{\sqrt{r}\alpha^3\log^{3/2} n}{\mu_r^2\Delta_{\vec{z}}\sqrt{n}}\right) + \frac{\mu_r}{\Delta_{\vec{z}}n}\right)\\
    = &~ {\cal O}\left(\left(\frac{\alpha\sqrt{\log n}}{\mu_r\sqrt{\Delta_{\vec{z}}} }+r^{3/2}\right) \frac{r\alpha^2\log n}{\mu_r^2\Delta_{\vec{z}}\sqrt{n}}\right)\,.
\end{align*}
where the last step follows from the assumption of $n$ (\cref{eq:proof_main_n_choice}).

The proof of the theorem is completed.
\end{proof}

\section{Second-order eigenvector perturbation theory} \label{sec:second-order}

In this section, we provide a formal statement of the second-order eigenvector perturbation result in \cref{lem:second_order_expansion_informal}. We outline the main steps involved in proving this lemma and defer technical details to \cref{app:second-order}. In the following part of this section, we always assume \eqref{eqn:condition_mutal}, \eqref{eqn:alpha_E_j}, and \cref{thm:esprit-1.5_formal} \eqref{eq:proof_main_n_choice_small} and omit these conditions later.

\begin{lemma}[Second-order perturbation for dominant eigenspace, {\normalfont formal version of \cref{lem:second_order_expansion_informal}}] \label{lem:second-order-formal}
    There exists a unitary matrix $\mat{U}_r\in\complex^{r\times r}$ such that
    \begin{align}\label{eq:second_order_expansion_formal}
        \hat{\mat{Q}}_r \mat{U}_r = &~ \mat{Q}_r  + \sum_{k=1}^\infty  (\mat{\Pi}_{\mat{Q}_r^\perp} \mat{E})^{k-1}\mat{\Pi}_{\mat{Q}_r^\perp}\mat{E}_{\rm random}\mat{Q}_r\mat{\Lambda}_r^{-k} \notag\\ 
        ~&+  \order \left( \frac{1}{\zgap n} + \frac{\alpha\sqrt{\log n}}{\mu_r \sqrt{n}} \right) \mat{\Pi}_{\mat{Q}_r} \mat{\tilde{Q}}_1  + \order \left( \frac{1}{n} \left(\frac{1}{\zgap} + \frac{\alpha^2\log n}{\mu_r^2}\right) \right)\mat{\tilde{Q}}_2,
    \end{align}
    where $\mat{\tilde{Q}}_1,\mat{\tilde{Q}}_2\in \C^{r\times r}$ are matrices of unit norm $\norm{\tilde{\mat{Q}}_1}_2 = \norm{\tilde{\mat{Q}}_2}_2 = 1$.

    Furthermore, it holds that
    \begin{align}\label{eq:second_order_expansion_formal_2}
    \hat{\mat{Q}}_r \mat{U}_r = &~ \mat{Q}_r  + \sum_{k=1}^\infty  (\mat{\Pi}_{\mat{Q}_r^\perp} \mat{E}_{\rm tail})^{k-1}\mat{\Pi}_{\mat{Q}_r^\perp}\mat{E}_{\rm random}\mat{Q}_r\mat{\Lambda}_r^{-k} \notag\\ 
        ~&+  \order \left( \frac{\alpha\sqrt{\log n}}{\mu_r \sqrt{\Delta_{\vec{z}}n}} \right) \mat{\Pi}_{\mat{Q}_r} \mat{\tilde{Q}}_1  + \order \left( \frac{\alpha^2\log n}{\mu_r^2\zgap n}\right) \mat{\tilde{Q}}_2',    
    \end{align}
    where $\mat{\tilde{Q}}_2'\in \C^{r\times r}$ with $\norm{\mat{\tilde{Q}}_2'}_2=1$.
\end{lemma}

Recall that for a matrix $\mat{U}$ with orthonormal columns, we notate the projectors $\mat{\Pi}_{\mat{U}} = \mat{U}\mat{U}^\adj$ and $\mat{\Pi}_{\mat{U}^\perp} = \Id_n - \mat{\Pi}_{\mat{U}}$.
We will extensively use the contour integral representation of the spectral projectors $\outprod{\hat{\mat{Q}}_r}$ and $\outprod{\mat{Q}_r}$ of $\mat{\hat{T}}$ and $\mat{T}$ (\cref{app:resolvents}):

\begin{proposition}[Expansion of spectral projector] \label{prop:projector-expansion}
    It holds that
    \begin{equation*}
        \mat{\Pi}_{\hat{\mat{Q}}_r} = \mat{\Pi}_{\mat{Q}_r}  + \sum_{k=1}^\infty \frac{1}{2\pi \iu} \oint_{\mathcal{C}_r } (\zeta - \mat{T})^{-1} \left( \mat{E} (\zeta - \mat{T})^{-1} \right)^k \, \d\zeta\,,
    \end{equation*}
    where $\mathcal{C}_r $ is a circle in the complex plane defined below in \cref{eq:def_circle_C}.%
\end{proposition}

\begin{proof}[Proof of \cref{prop:projector-expansion}]
    From \cref{lem:perturbation-resolvents}, it is sufficient to  prove $\|\mat{E} (\zeta - \mat{T})^{-1}\|<1$ on an appropriate circle $\mathcal{C}_r$.
    For convenience, we break the proof up into steps.
    
    \paragraph{Step 1: Where are the eigenvalues?}
    By \cref{cor:moitra_T} and $n>\frac{4\pi}{\Delta_z}+2$, we know that 
    \begin{align}\label{eq:eigen_pos_1}
       \lambda_1(\mat{T}) \leq  \mu_1 \left(n-1+\frac{2\pi}{\Delta_{\vec{z}}}\right) \leq \frac{3}{2}n\mu_1\,,\quad \lambda_r(\mat{T}) \geq \mu_r \left(n-1-\frac{2\pi}{\Delta_{\vec{z}}}\right) \geq \frac{1}{2}n\mu_{r}\,.
    \end{align}
    Then, by Weyl's perturbation theorem (\cref{thm:weyl}) and \cref{prop:error-matrix}, the eigenvalues of $\mat{\hat{T}}$ satisfy
    \begin{equation}\label{eq:eigen_pos_2}
    \begin{aligned}
        \lambda_1(\mat{\hat{T}}) \le&~ \lambda_1(\mat{T}) + \|\mat{E}\|_2 \leq \frac{3}{2} n\mu_1 + n\mu_{\rm tail} + {\cal O}(\alpha \sqrt{n\log n})\leq \frac{27}{16} n\mu_1\,,\\
        \lambda_r(\mat{\hat{T}}) \ge &~ \lambda_r(\mat{T})-\|\mat{E}\|_2 \geq \frac{1}{2} n\mu_r - n\mu_{\rm tail} - {\cal O}(\alpha \sqrt{n \log n}) \geq \frac{5}{16} n\mu_r\,,\\
        |\lambda_{r+1}(\mat{\hat{T}})| \le &~ \|\mat{E}\|_2 = n\mu_{\rm tail} + {\cal O}(\alpha \sqrt{n \log n}) \leq \frac{3}{16}n\mu_r\,,
    \end{aligned}
    \end{equation}
    assuming $\mu_{\rm tail}\leq \frac{1}{8}\mu_r$ and sufficiently large $n$ (\eqref{eq:proof_main_n_choice_small})
    
    Let $\mathcal{C}_r $ be a circle in the complex plane centered at $n(\mu_1 + \frac{1}{8}\mu_r)$ with radius $n(\mu_1 - \frac{1}{8}\mu_r)$, i.e.,
    \begin{align}\label{eq:def_circle_C}
        \mathcal{C}_r  = \left\{z\in \C~\Big|~\left|z-n\left(\mu_1 + \frac{1}{8}\mu_r\right)\right|=n\left(\mu_1 - \frac{1}{8}\mu_r\right)\right\}\,.
    \end{align}
    Then, by \cref{eq:eigen_pos_1,eq:eigen_pos_2}, %
    the following properties hold (see \cref{fig:B_r} for an illustration):
    \begin{enumerate}
        \item[a)] The first $r$ eigenvalues of $\mat{T}$ and $\hat{\mat{T}}$ are included in $\mathcal{C}_r $. 
        \item[b)] All remaining eigenvalues of $\mat{T}$ (which are zero) and $\hat{\mat{T}}$ are outside of $\mathcal{C}_r $.
        \item[c)] For any $z\in \mathcal{C}_r $ and $1\leq i\leq d$,
        \begin{align}\label{eq:circle_boundary_distance}
            |z-\lambda_i(\mat{T})|\geq \frac{1}{4}n\mu_r\,.
        \end{align}
    \end{enumerate}
    Thus, $\mathcal{C}_r $ satisfies the eigenvalue conditions in \cref{lem:perturbation-resolvents}.
    
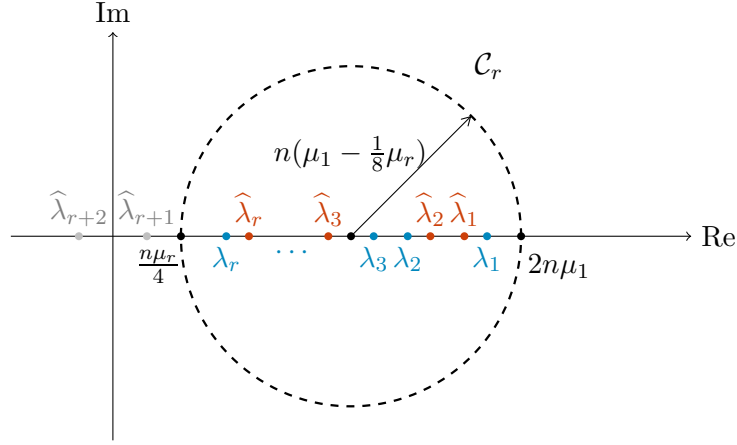
\begin{figure}[ht]
\centering
\begin{tikzpicture}[scale=1.5]
    \draw[->] (-1,0) -- (5,0) node[right] {$\Re$};
    \draw[->] (-0.1,-1.8) -- (-0.1,1.8) node[above] {$\Im$};
    
    \def\mx{3} %
    \def\dmu{2} %
    \def\r{\mx/2} %
    
    \draw[dashed, line width=0.3mm] (\mx/2+\dmu/4,0) circle (\mx/2); %

    \fill (\mx/2+\dmu/4,0) circle[radius=1pt]; %
    \draw[arrows={-{angle 60}}] (\mx/2+\dmu/4,0) -- (\mx/2+\dmu/4 + \r*0.70710678118,\r*0.70710678118) node[xshift=-1.6 cm, yshift = -0.5 cm]{$n(\mu_1-\frac{1}{8}\mu_r)$};
    \fill (\dmu/4,0) circle[radius=1pt] node[xshift=-0.3 cm, yshift = -0.4 cm]{$\frac{n\mu_r}{4}$};
    \fill (\mx+\dmu/4,0) circle[radius=1pt] node[xshift=0.5 cm, yshift = -0.4 cm]{$2n\mu_1$};
    \fill[fill=lightgray] (0.2, 0) circle[radius=1pt] node[above, text=gray]{$\hat{\lambda}_{r+1}$};
    \fill[fill=lightgray] (-0.4, 0) circle[radius=1pt] node[above, text=gray]{$\hat{\lambda}_{r+2}$};
    \fill[fill=b2] (3.2, 0) circle[radius=1pt] node[below, text=b2]{$\lambda_1$};
    \fill[fill=b2] (2.5, 0) circle[radius=1pt] node[below, text=b2]{$\lambda_2$};
    \fill[fill=b2] (2.2, 0) circle[radius=1pt] node[below, text=b2]{$\lambda_3$};
    \node at (1.5, 0) [below, text=b2] {$\cdots$};
    \fill[fill=b2] (0.9, 0) circle[radius=1pt] node[below, text=b2]{$\lambda_r$};
    \fill[fill=sinopia] (3, 0) circle[radius=1pt] node[above, text=sinopia]{$\hat{\lambda}_1$};
    \fill[fill=sinopia] (2.7, 0) circle[radius=1pt] node[above, text=sinopia]{$\hat{\lambda}_2$};
    \fill[fill=sinopia] (1.8, 0) circle[radius=1pt] node[above, text=sinopia]{$\hat{\lambda}_3$};
    \fill[fill=sinopia] (1.1, 0) circle[radius=1pt] node[above, text=sinopia]{$\hat{\lambda}_r$};
    \node at (\mx/2+1.5,\r-0.2) [above right] {$\mathcal{C}_r$};
\end{tikzpicture}
\caption{All the non-zero eigenvalues of $\mat{T}$ and the first $r$ eigenvalues of $\hat{\mat{T}}$ are contained in $\mathcal{C}_r$.}\label{fig:B_r}
\end{figure}
    
    \paragraph{Step 2: Bound the error on the circle.}
    Next, we establish a uniform bound on $\norm{(\zeta - \mat{T})^{-1}\mat{E}}_2$ so that we may apply \cref{lem:perturbation-resolvents}. Fix any $\zeta \in \mathcal{C}_r $.
    We begin with the triangle inequality
    \begin{equation} \label{eq:resolvent-error-bound}
        \norm{\mat{E}(\zeta - \mat{T})^{-1}}_2 \le \norm{\mat{E}_{\rm random}(\zeta - \mat{T})^{-1}}_2 + \norm{\mat{E}_{\rm tail}(\zeta - \mat{T})^{-1}}_2.
    \end{equation}
    
    For the first term, we use \cref{lem:toeplitz} to obtain the bound
    \begin{equation}\label{eq:zeta_T_inv_bound}
        \norm{\mat{E}_{\rm random}(\zeta - \mat{T})^{-1}}_2 \le \norm{\mat{E}_{\rm random}}_2 \cdot \norm{(\zeta - \mat{T})^{-1}}_2 \le {\cal O}(\alpha \sqrt{n\log n})\cdot \frac{4}{n\mu_r}=\order\left( \frac{\alpha\sqrt{\log n}}{\mu_r\sqrt{n}} \right)\,,
    \end{equation}
    where the second step follows from \cref{eq:circle_boundary_distance} that $$\norm{(\zeta - \mat{T})^{-1}}_2=\max_{i\in [d]}~|\zeta-\lambda_i(\mat{T})|^{-1}\leq \frac{4}{n\mu_r}\,.$$ %
    
    For the second term in \cref{eq:resolvent-error-bound}, consider the spectral decomposition 
    \begin{align*}
        (\zeta - \mat{T})^{-1} = \mat{Q}(\zeta - \mat{\Lambda})^{-1} \mat{Q}^\adj = \mat{Q}_r(\zeta - \mat{\Lambda}_r)^{-1} \mat{Q}^\adj_r + \zeta^{-1}\mat{\Pi}_{\mat{Q}_r^\perp}\,.
    \end{align*}
    Then, we have
    \begin{align*}
        \norm{\mat{E}_{\rm tail}(\zeta - \mat{T})^{-1}}_2 \leq &~ \norm{\mat{E}_{\rm tail} \mat{Q}_r(\zeta - \mat{\Lambda}_r)^{-1} \mat{Q}_r^\adj}_2 + |\zeta^{-1}|\cdot \norm{\mat{E}_{\rm tail} \mat{\Pi}_{\mat{Q}_r^\perp}}_2\\
        \leq &~ \norm{\mat{E}_{\rm tail} \mat{Q}_r}_2\cdot \norm{(\zeta - \mat{\Lambda}_r)^{-1}}_2 + |\zeta^{-1}|\cdot \norm{\mat{E}_{\rm tail}}_2\\
        \leq &~ \frac{4}{n\mu_r}\cdot \left(\norm{\mat{E}_{\rm tail} \mat{Q}}_2+\norm{\mat{E}_{\rm tail}}_2\right)\\
        \leq &~ \order\left(\frac{\mu_{\rm tail}}{\zgap\mu_r n}\right) + \frac{4 \mu_{\rm tail}}{\mu_r}\,,
    \end{align*}
    where the first step follows from the triangle inequality, the second step follows from $\|\mat{Q}_r\|_2=\|\mat{\Pi}_{\mat{Q}_r^\perp}\|_2=1$, the third step follows from $\norm{(\zeta - \mat{T})^{-1}}_2\leq \frac{4}{n\mu_r}$, the fourth step follows from \cref{prop:error-matrix} and \cref{prop:tail-error}.
    
    Plugging the above two displays into \cref{eq:resolvent-error-bound} and using the assumptions on $n$ and $\mu_{\rm tail}$, we obtain that 
    \begin{equation*}
        \norm{\mat{E}(\zeta - \mat{T})^{-1}}_2\leq {\cal O}\left(\frac{\alpha\sqrt{\log n}}{\mu_r\sqrt{n}} + \frac{\mu_{\rm tail}}{\zgap\mu_r n}\right) + \frac{4 \mu_{\rm tail}}{\mu_r}< 1\,.
    \end{equation*}
    This bound holds for any $\zeta\in \mathcal{C}_r $, so we obtain the desired expansion by \cref{lem:perturbation-resolvents}.
\end{proof}

The series expansion in \cref{prop:projector-expansion} is rather nontransparent, as the higher-order terms are expressed in terms of contour integrals.
Fortunately, these integrals can be evaluated exactly, leading to explicit formulas for all terms in the expansion in terms of Schur polynomials.
As the results are rather intricate, so we defer the statement and proof to \cref{app:explicit-expansion}.
We can then bound the higher-order terms in this expression, yielding the following representation:

\begin{restatable}[Expansion of spectral projector, simplified]{lemma}{expansionsimplified} \label{lem:expansion-simplified}
    It holds that
    \begin{equation} \label{eq:spectral-projector-simplified}
        \begin{split}
        \mat{\Pi}_{\hat{\mat{Q}}_r} =&~ \mat{\Pi}_{\mat{Q}_r} + \sum_{k=1}^\infty \left( \pinvk{\mat{T}}{k} \mat{E}_{\rm random} \mat{\Pi}_{\mat{Q}_r^\perp} (\mat{E}\mat{\Pi}_{\mat{Q}_r^\perp})^{k-1} + (\mat{\Pi}_{\mat{Q}_r^\perp} \mat{E})^{k-1}\mat{\Pi}_{\mat{Q}_r^\perp}\mat{E}_{\rm random}\pinvk{\mat{T}}{k} \right) \\ 
        ~&+ \order \left( \frac{1}{n} \left( \frac{\alpha^2 \log n}{\mu_r^2} + \frac{1}{\zgap} \right)  \right) \mat{G},
        \end{split}
    \end{equation}
    for some matrix $\mat{G}$ of norm $\norm{\mat{G}}_2 = 1$.
\end{restatable}

To complete the proof of \cref{lem:second-order-formal}, we just need to transfer this result about the projectors $\mat{\Pi}_{\hat{\mat{Q}}_r}$ and $\mat{\Pi}_{\mat{Q}_r}$ to the rectangular matrices $\hat{\mat{Q}}_r$ and $\mat{Q}_r$.
To do so, we use \cref{prop:distance-from-angles}.

\begin{proof}[Proof of \cref{lem:second-order-formal}]
    From the proof of \cref{lem:eigenvectors-weak} (specifically, \cref{eq:W-bound-2}), we know that there exists a unitary matrix $\mat{U}_r$ such that
    \begin{equation*}
        \norm{\hat{\mat{Q}}_r \mat{U}_r - \mat{Q}_r }_2 = \order \left( \frac{1}{\zgap n} + \frac{\alpha\sqrt{\log n}}{\mu_r \sqrt{n}} \right)\,,
    \end{equation*}
    where we omit $\mu_{\rm tail}/\mu_r<1$.
    
    On the other hand, notice that
    \begin{align*}
        \norm{\sum_{k=1}^\infty  (\mat{\Pi}_{\hat{\mat{Q}}_r^\perp} \mat{E})^{k-1}\mat{\Pi}_{\hat{\mat{Q}}_r^\perp}\mat{E}_{\rm random}\mat{Q}_r\mat{\Lambda}_r^{-k} }_2 \leq &~ \sum_{k=1}^\infty \norm{(\mat{\Pi}_{\hat{\mat{Q}}_r^\perp} \mat{E})^{k-1}\mat{\Pi}_{\hat{\mat{Q}}_r^\perp}\mat{E}_{\rm random}\mat{Q}_r\mat{\Lambda}_r^{-k}}_2\\
        = &~ \sum_{k=1}^\infty \norm{(\mat{\Pi}_{\hat{\mat{Q}}_r^\perp} \mat{E})^{k-1}\mat{\Pi}_{\hat{\mat{Q}}_r^\perp}\mat{E}_{\rm random}(\mat{T}^+)^k}_2\cdot \|\mat{Q}_r\|_2\\
        \leq &~ \sum_{k=1}^\infty \order\left(\frac{\alpha\sqrt{\log n}}{\sqrt{n} \mu_r}\right) \cdot C^{k-1}\\
        = &~ \order\left(\frac{\alpha\sqrt{\log n}}{\mu_r \sqrt{n}} \right)\,,
    \end{align*}
    where the first step follows from the triangle inequality, the second step follows from the definition of the pseudo-inverse of $\mat{T}$, the third step follows from \cref{eq:leading-order-bounds}, and the last step follows from $0<C<\frac{1}{2}$ by \cref{lem:terms}.
    
    Thus, we can write 
    \begin{equation} \label{eq:expansion-plus-remainder}
        \hat{\mat{Q}}_r \mat{U}_r = \mat{Q}_r  + \sum_{k=1}^\infty  (\mat{\Pi}_{\mat{Q}_r^\perp} \mat{E})^{k-1}\mat{\Pi}_{\mat{Q}_r^\perp}\mat{E}_{\rm random}\mat{Q}_r\mat{\Lambda}_r^{-k} + \left(\mat{\Pi}_{\mat{Q}}  + \mat{\Pi}_{\mat{Q}_r^\perp}\right) \mat{R}
    \end{equation}
    where $\mat{R}$ is the remainder matrix of norm
    \begin{align*}
        \norm{\mat{R}}_2 = &~ \norm{\hat{\mat{Q}}_r \mat{U}_r - \mat{Q}_r  - \sum_{k=1}^\infty  (\mat{\Pi}_{\mat{Q}_r^\perp} \mat{E})^{k-1}\mat{\Pi}_{\mat{Q}_r^\perp}\mat{E}_{\rm random}\mat{Q}_r\mat{\Lambda}_r^{-k}}_2\\
        \leq &~ \norm{\hat{\mat{Q}}_r \mat{U}_r - \mat{Q}_r }_2 + \norm{\sum_{k=1}^\infty  (\mat{\Pi}_{\mat{Q}_r^\perp} \mat{E})^{k-1}\mat{\Pi}_{\mat{Q}_r^\perp}\mat{E}_{\rm random}\mat{Q}_r\mat{\Lambda}_r^{-k} }_2\\
        = &~ \order \left( \frac{1}{\zgap n} + \frac{\alpha\sqrt{\log n}}{\mu_r \sqrt{n}} \right).
    \end{align*}

    To prove the result, we just need to bound $\norm{\smash{\mat{\Pi}_{\mat{Q}_r^\perp} \mat{R}}}_2$.
    To do so, we multiply each side of \cref{eq:expansion-plus-remainder} with its adjoint. Since $\mat{U}_r$ is a unitary matrix, the LHS becomes
    \begin{align*}
        \hat{\mat{Q}}_r \mat{U}_r \mat{U}_r^\adj \hat{\mat{Q}}_r^\adj = \hat{\mat{Q}}_r\hat{\mat{Q}}_r^\adj=\mat{\Pi}_{\hat{\mat{Q}}_r}\,.
    \end{align*}
    The RHS becomes
    \begin{align*}
        &\mat{Q}_r\mat{Q}_r^\dagger + \mat{Q}_r\left(\sum_{k=1}^\infty  (\mat{\Pi}_{\mat{Q}_r^\perp} \mat{E})^{k-1}\mat{\Pi}_{\mat{Q}_r^\perp}\mat{E}_{\rm random}\mat{Q}_r\mat{\Lambda}_r^{-k}\right)^\adj + \left(\sum_{k=1}^\infty  (\mat{\Pi}_{\mat{Q}_r^\perp} \mat{E})^{k-1}\mat{\Pi}_{\mat{Q}_r^\perp}\mat{E}_{\rm random}\mat{Q}_r\mat{\Lambda}_r^{-k}\right) \mat{Q}_r^\adj\\
        ~&+ \mat{Q}_r\mat{R}^\adj + \mat{R}\mat{Q}_r^\adj\\
        ~&+\left(\sum_{k=1}^\infty  (\mat{\Pi}_{\mat{Q}_r^\perp} \mat{E})^{k-1}\mat{\Pi}_{\mat{Q}_r^\perp}\mat{E}_{\rm random}\mat{Q}_r\mat{\Lambda}_r^{-k}+\mat{R}\right)\left(\sum_{k=1}^\infty  (\mat{\Pi}_{\mat{Q}_r^\perp} \mat{E})^{k-1}\mat{\Pi}_{\mat{Q}_r^\perp}\mat{E}_{\rm random}\mat{Q}_r\mat{\Lambda}_r^{-k}+\mat{R}\right)^\adj\\
        = &~ \mat{\Pi}_{\mat{Q}_r} + \sum_{k=1}^\infty (\mat{T}^+)^{k} \mat{E}_{\rm random} \mat{\Pi}_{\mat{Q}_r^\perp}(\mat{E} \mat{\Pi}_{\mat{Q}_r^\perp} )^{k-1} + \sum_{k=1}^\infty  (\mat{\Pi}_{\mat{Q}_r^\perp} \mat{E})^{k-1}\mat{\Pi}_{\mat{Q}_r^\perp}\mat{E}_{\rm random}(\mat{T}^+)^k\\
        ~&+ \mat{Q}_r\mat{R}^\adj + \mat{R}\mat{Q}_r^\adj +\order \left( \frac{\alpha^2 \log n}{\mu_r^2n} \right) \mat{G}'\,,
    \end{align*}
    where $\mat{G'}$ is some matrix of norm $\|\mat{G}'\|_2=1$, and the last step follows from
    \begin{align*}
        &\norm{\left(\sum_{k=1}^\infty  (\mat{\Pi}_{\mat{Q}_r^\perp} \mat{E})^{k-1}\mat{\Pi}_{\mat{Q}_r^\perp}\mat{E}_{\rm random}\mat{Q}_r\mat{\Lambda}_r^{-k}+\mat{R}\right)\left(\sum_{k=1}^\infty  (\mat{\Pi}_{\mat{Q}_r^\perp} \mat{E})^{k-1}\mat{\Pi}_{\mat{Q}_r^\perp}\mat{E}_{\rm random}\mat{Q}_r\mat{\Lambda}_r^{-k}+\mat{R}\right)^\adj}_2\\
        = &~ {\cal O}\left(\Big\|\sum_{k=1}^\infty  (\mat{\Pi}_{\mat{Q}_r^\perp} \mat{E})^{k-1}\mat{\Pi}_{\mat{Q}_r^\perp}\mat{E}_{\rm random}\mat{Q}_r\mat{\Lambda}_r^{-k} \Big\|_2^2 + \norm{\mat{R}}_2^2\right)\\
        = &~ {\cal O}\left(\frac{\alpha^2 \log n}{\mu_r^2 n} + \frac{1}{\Delta_{\vec{z}}^2 n^2}\right)={\cal O}\left(\frac{\alpha^2 \log n}{\mu_r^2 n}\right)
    \end{align*}
        
    Comparing LHS and RHS with \cref{eq:spectral-projector-simplified}, we see that
    \begin{align*}
        \norm{\mat{Q}_r\mat{R}^\dagger + \mat{R}^\adj \mat{Q}_r}_2 = &~  \norm{\order \left( \frac{1}{n} \left( \frac{\alpha^2 \log n}{\mu_r^2} + \frac{1}{\zgap} \right)  \right) \mat{G} - \order \left( \frac{\alpha^2 \log n}{\mu_r^2 n} \right) \mat{G}'}_2\\
        = &~\order \left( \frac{1}{n} \left( \frac{\alpha^2 \log n}{\mu_r^2} + \frac{1}{\zgap} \right)  \right)\,.
    \end{align*}
    Now, we estimate
    \begin{align*}
        \norm{\mat{\Pi}_{\mat{Q}_r^\perp} \mat{R}}_2 \leq &~  \norm{\mat{\Pi}_{\mat{Q}_r^\perp} \mat{R}\mat{Q}_r^\adj}_2 = \norm{\mat{\Pi}_{\mat{Q}_r^\perp} \left(\mat{Q}_r\mat{R}^\adj + \mat{R}\mat{Q}_r^\adj\right)^\dagger}_2 \\
        \le &~ \norm{\mat{Q}_r\mat{R}^\adj + \mat{R}\mat{Q}_r^\adj}_2= \order \left( \frac{1}{n} \left( \frac{\alpha^2\log n}{\mu_r^2} + \frac{1}{\zgap} \right)  \right)\,,
    \end{align*}
    where the second step follows from $\mat{\Pi}_{\mat{Q}_r^\perp}\mat{Q}_r=\mat{0}$.

    Plugging this estimate into \cref{eq:expansion-plus-remainder} gives \cref{eq:second_order_expansion_formal}. 
    
    For the furthermore part, fix any $k\geq 2$ and consider the difference of replacing $\mat{E}$ with $\mat{E}_{\rm tail}$ in \cref{eq:second_order_expansion_formal}:

    \begin{align*}
    &\left\|(\mat{\Pi}_{\mat{Q}_r^\perp} \mat{E})^{k-1}\mat{\Pi}_{\mat{Q}_r^\perp}\mat{E}_{\rm random}\mat{Q}_r\mat{\Lambda}_r^{-k} - (\mat{\Pi}_{\mat{Q}_r^\perp} \mat{E}_{\rm tail})^{k-1}\mat{\Pi}_{\mat{Q}_r^\perp}\mat{E}_{\rm random}\mat{Q}_r\mat{\Lambda}_r^{-k}\right\|_2\\
    \leq &~ \left\|(\mat{\Pi}_{\mat{Q}_r^\perp} (\mat{E}_{\rm tail} + \mat{E}_{\rm random}))^{k-1}-(\mat{\Pi}_{\mat{Q}_r^\perp} \mat{E}_{\rm tail})^{k-1}\right\|_2\cdot \norm{\mat{\Pi}_{\mat{Q}_r^\perp}\mat{E}_{\rm random}\mat{Q}_r\mat{\Lambda}_r^{-k}}_2\\
    \leq &~ \sum_{t=1}^{k-1} \binom{k-1}{t}\cdot \norm{\mat{E}_{\rm random}}_2^t \cdot \norm{\mat{E}_{\rm tail}}_2^{k-1-t}\cdot \norm{\mat{E}_{\rm random}}_2\cdot \norm{\mat{\Lambda}_r^{-k}}_2\\
    = &~ \sum_{t=1}^{k-1} \binom{k-1}{t}\cdot {\cal O}\left(\alpha \sqrt{n\log n}\right)^{t+1} \cdot (n\mu_{\rm tail})^{k-1-t} \cdot {\cal O}\left(\frac{1}{n\mu_r}\right)^k\\
    \leq &~ {\cal O}\left(\frac{\mu_{\rm tail}}{\mu_r}\right)^{k-2}\cdot {\cal O}\left(\frac{\alpha^2 \log n}{\mu_r^2 n}\right)\,,
\end{align*}
where the third step follows from \cref{lem:Mecks,prop:error-matrix,cor:moitra_T}, and the last step follows from the identity $\binom{m}{p}\leq 2^m$.

Summing over $k$, we get that
\begin{align*}
    &\norm{\sum_{k=1}^\infty(\mat{\Pi}_{\mat{Q}_r^\perp} \mat{E})^{k-1}\mat{\Pi}_{\mat{Q}_r^\perp}\mat{E}_{\rm random}\mat{Q}_r\mat{\Lambda}_r^{-k} - (\mat{\Pi}_{\mat{Q}_r^\perp} \mat{E}_{\rm tail})^{k-1}\mat{\Pi}_{\mat{Q}_r^\perp}\mat{E}_{\rm random}\mat{Q}_r\mat{\Lambda}_r^{-k}}_2\\
    \leq &~ \sum_{k=2}^\infty {\cal O}\left(\frac{\mu_{\rm tail}}{\mu_r}\right)^{k-2}\cdot {\cal O}\left(\frac{\alpha^2 \log n}{\mu_r^2 n}\right)\\
    = &~ {\cal O}\left(\frac{\alpha^2 \log n}{\mu_r^2 n}\right)\,,
\end{align*}
where the last step follows from $\mu_{\rm tail}<\mu_r / 8$. Therefore, we obtain that
\begin{align*}
    \hat{\mat{Q}}_r \mat{U}_r = &~ \mat{Q}_r  + \sum_{k=1}^\infty  (\mat{\Pi}_{\mat{Q}_r^\perp} \mat{E}_{\rm tail})^{k-1}\mat{\Pi}_{\mat{Q}_r^\perp}\mat{E}_{\rm random}\mat{Q}_r\mat{\Lambda}_r^{-k}\\ 
        ~&+  \order \left( \frac{1}{\zgap n} + \frac{\alpha\sqrt{\log n}}{\mu_r \sqrt{n}} \right) \mat{\Pi}_{\mat{Q}_r} \mat{\tilde{Q}}_1  + \order \left( \frac{1}{n} \left(\frac{1}{\zgap} + \frac{\alpha^2\log n}{\mu_r^2}\right) \right)\mat{\tilde{Q}}_2'\\
      = &~   \mat{Q}_r  + \sum_{k=1}^\infty  (\mat{\Pi}_{\mat{Q}_r^\perp} \mat{E}_{\rm tail})^{k-1}\mat{\Pi}_{\mat{Q}_r^\perp}\mat{E}_{\rm random}\mat{Q}_r\mat{\Lambda}_r^{-k} +\order \left( \frac{\alpha\sqrt{\log n}}{\mu_r \sqrt{\Delta_{\vec{z}}n}} \right) \mat{\Pi}_{\mat{Q}_r} \mat{\tilde{Q}}_1\\
      ~&+ \order \left( \frac{\alpha^2\log n}{\mu_r^2\zgap n}\right) \mat{\tilde{Q}}_2'\,,
\end{align*}
where $\norm{\mat{\tilde{Q}}_2'}_2=1$, and the last step follows from $n=\Omega(\alpha^{-2}\Delta_{\vec{z}}^{-1})$.
    
    The proof of the lemma is then completed.
\end{proof}

\section{Strong eigenvector comparison}\label{sec:strong_sketch}
The goal of this section is to prove the strong version of the eigenvectors comparison theorem stated below, which is the key component of obtaining the optimal convergence rate of the ESPRIT algorithm.  
\begin{theorem}[Eigenvectors comparison: strong estimation, {\normalfont formal version of Theroem~\ref{thm:improve_similarity}}]\label{thm:key_comparision} 
Assume \eqref{eqn:condition_mutal} and \eqref{eqn:alpha_E_j}. 
Suppose 
\[
n=\Omega\left(\left(\frac{r\alpha^{4/3}}{\mu_r^{4/3}\Delta_{\vec{z}}^{2/3}} + \frac{\alpha^{2}}{\mu_r^2\Delta_{\vec{z}}}\right)\cdot \log(r,\alpha,1/\mu_r,1/\Delta_{\vec{z}})\right)\,.
\]
Then, there exists a unitary matrix $\mat{U}_r\in \C^{r\times r}$ and an invertible (not necessarily unitary) matrix $\mat{P}\in \C^{r\times r}$ such that
\begin{equation}\label{eqn:matrix_perb_improve}
    \left\|\mat{P}\mat{U}_r^\dagger\mat{\hat{Q}}^{+}_\downarrow\mat{\hat{Q}}_\uparrow\mat{U}_r\mat{P}^{-1}-\mat{Q}^{+}_\downarrow\mat{Q}_\uparrow\right\|_2={\cal O}\left(\left(\frac{\alpha\sqrt{\log n}}{\mu_r\sqrt{\Delta_{\vec{z}}} }+r^{1.5}\right)\frac{\alpha^2\log n}{\mu_r^2\Delta_{\vec{z}}n^{1.5}}\right)\,,
\end{equation}
where $[\cdot]^{+}$ is the pseudo-inverse. Here $\mat{Q}_\uparrow=\mat{Q}(:n-1,:r)$ is the first $n-1$ rows and first $r$ columns of $\mat{Q}$ and $\mat{Q}_\uparrow=\mat{Q}(1:,:r)$ is the last $n-1$ rows and first $r$ columns of $\mat{Q}$.
\end{theorem}

For simplicity, let $\mat{E}_1$ and $\mat{E}_2$ be the following two parts of the error matrix $\mat{E}$:
\[
\mat{E}_1:=\mat{E}_{\rm tail}=\mat{V}_n(\ztail) \cdot \diag(\mutail) \cdot \mat{V}_n(\ztail)^\dagger,\quad \mat{E}_2:=\mat{E}_{\mathrm{random}}\,.
\]

Our proof of \cref{thm:key_comparision} consists of three steps:

\begin{itemize}
    \item {\bf Step 1:} Construction of the ``good'' aligning matrix $\mat{P}$ (\cref{sec:align_step1})
    \item {\bf Step 2:} Taylor expansion with respect to the error terms (\cref{sec:align_step2})
    \item {\bf Step 3:} Error cancelation in the Taylor expansion (\cref{sec:align_step3})
\end{itemize}

We will introduce these three steps in \cref{sec:align_step1} - \cref{sec:align_step3} and prove \cref{thm:key_comparision} in \cref{sec:align_all}. Most detialed calculations are deferred to \cref{sec:align_proof_defer}.

\subsection{Construction of the ``good'' \texorpdfstring{$\mat{P}$}{P}}\label{sec:align_step1}
In this step, we employ the second-order perturbation expansion of $\hat{\mat{Q}}_r$ in terms of the error matrices $\mat{E}_{\rm tail}$ and $\mat{E}_{\rm random}$ (\cref{lem:second_order_expansion_informal}) to construct an invertible matrix $\mat{P}$ designed to eliminate all error terms whose column space aligns with that of $\mat{Q}_r$. More specifically, as summarized in the following lemma, the matrix $\mat{P}$ defined as \cref{eq:def_good_P} is almost-identity (Property I) and $\hat{\mat{Q}}_r \mat{U}_r \mat{P}^{-1} - \mat{Q}_r$ has certain structure (Property II), which is crucial for the subsequent two steps.

\begin{restatable}{lemma}{goodPconstruction}\label{lem:good_P_construction}
There exists two matrices $\tilde{\mat{Q}}_1,\tilde{\mat{Q}}_2\in \C^{n\times r}$ satisfying the following properties: 
\begin{itemize}
\item $\norm{\tilde{\mat{Q}}_1}_2={\cal O}(1)$, $\norm{\tilde{\mat{Q}}_2}_2={\cal O}(1)$.

\item Define an matrix $\mat{P}\in \C^{r\times r}$:%
\begin{align}\label{eq:def_good_P}
\mat{P}:=\mat{I}_r+{\cal O}\left(\frac{\alpha\sqrt{\log n}}{\mu_r\sqrt{\Delta_z n}}\right)\mat{Q}^\dagger_r\mat{\Pi}_{\mat{Q}_r}\tilde{\mat{Q}}_1+{\cal O}\left(\frac{\alpha^2\log n}{\mu^{2}_r\Delta_{\vec{z}} n}\right)\mat{Q}^\dagger_r\mat{\Pi}_{\mat{Q}_r}\tilde{\mat{Q}}_{2}\,,    
\end{align}
where $\mat{\Pi}_{\mat{Q}_r}$ is the projection onto the column space of $\mat{Q}_r$.  Then, $\mat{P}$ is an invertible matrix and satisfies the following two properties:
\begin{itemize}
    \item {\it Property I: } 
    \begin{equation}\label{eqn:P_bound}
\left\|\mat{P}^{-1}-\mat{I}_r\right\|_2={\cal O}\left(\frac{\alpha\sqrt{\log n}}{\mu_r\sqrt{\Delta_z n}}\right)\,,\quad \left\|\mat{P}-\mat{I}_r\right\|_2={\cal O}\left(\frac{\alpha\sqrt{\log n}}{\mu_r\sqrt{\Delta_z n}}\right)\,.
\end{equation}
    \item {\it Property II:} There exists a unitary matrix $\mat{U}_r\in\mathbb{C}^{r\times r}$ such that
\begin{equation}\label{eqn:second_Q_perturb}
\begin{aligned}
\left\|\mat{\hat{Q}}_r\mat{U}_r\mat{P}^{-1}-\mat{Q}_r\right\|_2\leq {\cal O}\left(\frac{\alpha\sqrt{\log n}}{\mu_r\sqrt{\Delta_z n}}\right)\,,
\end{aligned}
\end{equation}
and
\begin{equation}\label{eqn:informal_expansion}
\begin{aligned}
\hat{\mat{Q}}_r\mat{U}_r\mat{P}^{-1}=&~ \mat{Q}_r+\sum^\infty_{k=0}\left(\mat{I}_n-\frac{1}{n}\mat{V}_n(\zdom)\mat{V}_n(\zdom)^\dagger\right)\mat{E}_{\rm tail}^k\mat{E}_{\rm random}\mat{Q}_r\left(\mat{\Sigma}_r^{-1}\right)^{k+1}\mat{P}^{-1}\\
~&+ {\cal O}\left(\frac{\alpha^2\log n}{\mu_r^2\Delta_z n}\right)\mat{\Pi}_{\mat{Q}^\perp_r}\tilde{\mat{Q}}_{2}\mat{P}^{-1}+ {\cal O}\left(\frac{\alpha}{\mu_r\Delta_{\vec{z}}n^{1.5}}\right)\,.  
\end{aligned}
\end{equation}
\end{itemize}
\end{itemize}
\end{restatable}
The proof is deferred to \cref{sec:align_step1_defer}.

\subsection{Taylor expansion with respect to the error terms}\label{sec:align_step2}
In this step, we use the Neumann series to expand $\left(\hat{\mat{Q}}_r\mat{U}_r\mat{P}^{-1}\right)^\dagger_{\uparrow}\left(\hat{\mat{Q}}_r\mat{U}_r\mat{P}^{-1}\right)_{\downarrow}$ and insert the formula in \eqref{eqn:informal_expansion}. Upon organizing various terms and controlling all higher-order terms (as detailed in \cref{lem:step_2_first}), we conclude in \cref{lem:coarse_bound} that proving \cref{thm:improve_similarity} hinges on establishing: 
    \begin{equation}\label{eqn:key_main_text}
    \left\|(\mat{E}^Q_\downarrow)^\dagger\left(\mat{E}^Q_\uparrow-\mat{E}^Q_\downarrow(\mat{Q}_\downarrow)^\dagger\mat{Q}_\uparrow\right)\right\|_2\sim n^{-1.5}\quad \text{and}\quad\left\|\mat{Q}_\downarrow^\dagger\left(\mat{E}^Q_\uparrow-\mat{E}^Q_\downarrow(\mat{Q}_\downarrow)^\dagger\mat{Q}_\uparrow\right)\right\|_2\sim n^{-1.5}\,.
    \end{equation}

Without loss of generality, we assume $\mat{U}_r=\mat{I}_r$. For convenience, we define some notions for the error terms appearing in the expansion:%
\[
\mat{E}^Q=\mat{\hat{Q}}_r\mat{P}^{-1}-\mat{Q}_r,\quad \mat{E}^Q_\uparrow=\mat{E}^Q(:-1,:r),\quad \mat{E}^Q_\downarrow=\mat{E}^Q(1:,:r)\,.
\]
Recall
\[
\mat{Q}_{\uparrow}=\mat{Q}(:n-1,:r),\quad \hat{\mat{Q}}_{\uparrow}=\hat{\mat{Q}}(:n-1,:r),\quad \mat{Q}_{\downarrow}=\mat{Q}(1:,:r),\quad \hat{\mat{Q}}_{\downarrow}=\hat{\mat{Q}}(1:,:r)
\]
It's easy to check that 
\begin{align}\label{eq:E_Q_up_down}
    \mat{E}^Q_\uparrow=\hat{\mat{Q}}_{\uparrow}\mat{P}^{-1}-\mat{Q}_\uparrow=\left(\widetilde{\mat{E}}^Q_1+\widetilde{\mat{E}}^Q_2\right)_\uparrow,\quad \text{and}\quad \mat{E}^Q_\downarrow=\hat{\mat{Q}}_{\downarrow}\mat{P}^{-1}-\mat{Q}_\downarrow=\left(\widetilde{\mat{E}}^Q_1+\widetilde{\mat{E}}^Q_2\right)_\downarrow\,.
\end{align}

The following lemma provides a second-order truncation for the Neumann series expansion (\cref{def:neumann}) of $\mat{P}\mat{\hat{Q}}^{+}_\downarrow\mat{\hat{Q}}_\uparrow\mat{P}^{-1}$:
\begin{restatable}{lemma}{stepIIfirst}\label{lem:step_2_first}
Let $\mat{P}$ be defined as \eqref{eq:def_good_P}. Then, we have
\begin{align}
    & \mat{P}\mat{\hat{Q}}^{+}_\downarrow\mat{\hat{Q}}_\uparrow\mat{P}^{-1}-\mat{Q}^{+}_\downarrow\mat{Q}_\uparrow\notag\\
    =&~\left((\mat{E}^Q_\downarrow)^\dagger\mat{Q}_\uparrow+\mat{Q}_\downarrow^\dagger\mat{E}^Q_\uparrow\right)-\left((\mat{E}^Q_\downarrow)^\dagger\mat{Q}_\downarrow+\mat{Q}_\downarrow^\dagger\mat{E}^Q_\downarrow\right)(\mat{Q}_\downarrow)^\dagger\mat{Q}_\uparrow\tag{first order terms}\\
~&+(\mat{E}^Q_\downarrow)^\dagger(\mat{E}^Q_\uparrow)-\left((\mat{E}^Q_\downarrow)^\dagger\mat{Q}_\downarrow+\mat{Q}_\downarrow^\dagger\mat{E}^Q_\downarrow\right)\left((\mat{E}^Q_\downarrow)^\dagger\mat{Q}_\uparrow+\mat{Q}_\downarrow^\dagger\mat{E}^Q_\uparrow\right)\tag{second order terms}\\
~&+\left(-(\mat{E}^Q_\downarrow)^\dagger(\mat{E}^Q_\downarrow)+\left((\mat{E}^Q_\downarrow)^\dagger\mat{Q}_\downarrow+\mat{Q}_\downarrow^\dagger\mat{E}^Q_\downarrow\right)^2\right)(\mat{Q}_\downarrow)^\dagger\mat{Q}_\uparrow\tag{second order terms}\\
~&+\mathcal{O}\left(\frac{\alpha^3\log^3 n}{\mu^{3}_r(\Delta_{\vec{z}} n)^{1.5}}\right)\,,\label{eqn:error_expansion}
\end{align}
\end{restatable}

By re-grouping and upper-bounding some high-order terms, we obtain the key result of this sub-section:

\begin{restatable}{lemma}{coarsebound}\label{lem:coarse_bound}
It holds that
\begin{equation}\label{eq:PQQP_QQ_key_ineq}
\begin{aligned}
&\left\|\mat{P}\mat{\hat{Q}}^{+}_\downarrow\mat{\hat{Q}}_\uparrow\mat{P}^{-1}-\mat{Q}^{+}_\downarrow\mat{Q}_\uparrow\right\|_2\\
\leq &~ \left\|(\mat{E}^Q_\downarrow)^\dagger\left(\mat{E}^Q_\uparrow-\mat{E}^Q_\downarrow(\mat{Q}_\downarrow)^\dagger\mat{Q}_\uparrow\right)\right\|_2+\mathcal{O}\left(1+\frac{\alpha\sqrt{\log n}}{\mu_r\sqrt{\Delta_{\vec{z}} n}}\right)\left\|\mat{Q}_\downarrow^\dagger\left(\mat{E}^Q_\uparrow-\mat{E}^Q_\downarrow(\mat{Q}_\downarrow)^\dagger\mat{Q}_\uparrow\right)\right\|_2\\
~&+ \mathcal{O}\left(\frac{\alpha^3\log^3 n}{\mu^{3}_r(\Delta_{\vec{z}} n)^{1.5}} \right)\,.
\end{aligned}
\end{equation}
\end{restatable}

The proofs \cref{lem:step_2_first,lem:coarse_bound} of are deferred to \cref{sec:align_step2_defer}.

\subsection{Error cancellation in the Taylor expansion}\label{sec:align_step3}
In this step, we will prove \cref{eqn:key_main_text}. Because $\norm{\mat{E}^Q_\downarrow}_2\leq \norm{\hat{\mat{Q}}_r\mat{P}^{-1}-\mat{Q}_r}_2\sim n^{-0.5}$ (\cref{eqn:second_Q_perturb}), it suffices to establish:
\begin{equation}\label{eqn:key_equality}
\begin{aligned}
\norm{\mat{E}^Q_\uparrow-\mat{E}^Q_\downarrow(\mat{Q}_\downarrow)^\dagger\mat{Q}_\uparrow}_2=&~\mathcal{O}\left(\frac{r^{1.5}\alpha}{\mu_r n}+\frac{\alpha^2\log n}{\mu^{2}_r\Delta_{\vec{z}} n}\right),\\
\norm{\mat{Q}_\downarrow^\dagger\left(\mat{E}^Q_\uparrow-\mat{E}^Q_\downarrow(\mat{Q}_\downarrow)^\dagger\mat{Q}_\uparrow\right)}_2=&~ \mathcal{O}\left(\frac{r^2\alpha^2\log n}{\mu^{2}_r\Delta_{\vec{z}} n^{1.5}} + \frac{\alpha\sqrt{\log n}}{\mu_r(\Delta_{\vec{z}} n)^{1.5}}\right)\,.
\end{aligned}
\end{equation}

We first define the errors in the column space of $\mat{Q}_r$ and orthogonal to this subspace:
\begin{subequations}\label{eqn:E_tilde_Q}
\begin{equation}%
\widetilde{\mat{E}}^Q_1:=\sum^\infty_{k=0}\left(\mat{I}_n-\frac{1}{n}\mat{V}_n(\zdom)\mat{V}_n(\zdom)^\dagger\right)\mat{E}_1^k\mat{E}_2\mat{Q}_r\left(\mat{\Sigma}_r^{-1}\right)^{k+1}\mat{P}^{-1}\,.
\end{equation}
and
\begin{equation}
\widetilde{\mat{E}}^Q_2:={\cal O}\left(\frac{\alpha^2\log n}{\mu^{2}_r\Delta_{\vec{z}} n}\right)\mat{\Pi}_{\mat{Q}_r^\perp}\hat{\mat{Q}}_{2}\mat{P}^{-1}\,.
\end{equation}
\end{subequations}
Then, we immediately have 
\begin{align}\label{eq:E_Q_up_down_2}
\mat{E}^Q_\uparrow=\left(\widetilde{\mat{E}}^Q_1+\widetilde{\mat{E}}^Q_2\right)_\uparrow,\quad \text{and}\quad \mat{E}^Q_\downarrow=\left(\widetilde{\mat{E}}^Q_1+\widetilde{\mat{E}}^Q_2\right)_\downarrow\,.
\end{align}

\paragraph{Proof of the first equation in \cref{eqn:key_equality}.}
Since $\norm{\widetilde{\mat{E}}^Q_2}_2\sim n^{-1}$ by its definition, we only need to consider $(\tilde{\mat{E}}_1^Q)_\uparrow - (\tilde{\mat{E}}_1^Q)_\downarrow(\mat{Q}_\downarrow)^\dagger \mat{Q}_\uparrow$. A key observation is that $\mat{Q}_r$ approximates $\frac{1}{\sqrt{n}}\mat{V}_n(\zdom)$ closely (see \cref{lem:vander-eigen}), implying $(\mat{Q}_\downarrow)^\dagger\mat{Q}_\uparrow\approx \mathrm{diag}(\zdom^{-1})$. After expanding $(\tilde{\mat{E}}_1^Q)_\uparrow - (\tilde{\mat{E}}_1^Q)_\downarrow(\mat{Q}_\downarrow)^\dagger \mat{Q}_\uparrow$, we obtain the following lemma:

\begin{restatable}{lemma}{decomposeI}\label{lem:decompose_1} 
Let $\mat{P}_v$ be defined in \cref{lem:vander-eigen}. We have
\begin{subequations}
\begin{equation}\label{eq:EPPV_lem}
\begin{aligned}
    \left\|\mat{E}^Q_\uparrow-\mat{E}^Q_\downarrow(\mat{Q}_\downarrow)^\dagger\mat{Q}_\uparrow\right\|_2=&~\mathcal{O}\left(\Big\|\left(\widetilde{\mat{E}}^Q_1\mat{P}\mat{P}^\dagger_v\right)_\uparrow-\left(\widetilde{\mat{E}}^Q_1\mat{P}\right)_\downarrow(\mat{Q}_\downarrow)^\dagger\mat{Q}_\uparrow\mat{P}^\dagger_v\Big\|_2+\frac{\alpha^2\log n}{\mu^{2}_r\Delta_{\vec{z}} n}\right)\,,
\end{aligned}
\end{equation}
where
\begin{equation}\label{eqn:EPPV_1_all}
\begin{aligned}
\widetilde{\mat{E}}^Q_1\mat{P}\mat{P}^\dagger_v=
&~\sum^\infty_{k=0}\mat{E}_1^k\mat{E}_2\left(\frac{1}{\sqrt{n}}\mat{V}_n(\zdom)\right)\left(n\diag(\mudom)\right)^{-(k+1)}\\
~&- \frac{1}{n}\mat{V}_n(\zdom)\mat{V}_n(\zdom)^\dagger \mat{E}_2\left(\frac{1}{\sqrt{n}}\mat{V}_n(\zdom)\right) \left(n\diag(\mudom)\right)^{-1}\\
~&+\mathcal{O}\left(\frac{\alpha\sqrt{\log n}}{\mu^{2}_r\Delta_{\vec{z}} n^{1.5}}\right)\,,    
\end{aligned}
\end{equation}
and
\begin{equation}\label{eqn:EPPV_2_all}
\begin{aligned}
&\left(\widetilde{\mat{E}}^Q_1\mat{P}\right)_\downarrow(\mat{Q}_\downarrow)^\dagger\mat{Q}_\uparrow\mat{P}^\dagger_v\\
=&~\left(\sum^\infty_{k=0}\mat{E}_1^k\mat{E}_2\left(\frac{1}{\sqrt{n}}\mat{V}_n(\zdom)\right)\left(n\diag(\mudom)\right)^{-(k+1)}\right)_\downarrow\diag(\zdom^{-1})\\
~&-\left(\frac{1}{n}\mat{V}_n(\zdom)\mat{V}_n(\zdom)^\dagger\mat{E}_2\left(\frac{1}{\sqrt{n}}\mat{V}_n(\zdom)\right)\left(n\diag(\mudom)\right)^{-1}\right)_\downarrow\diag(\zdom^{-1})\\
~&+\mathcal{O}\left(\frac{\alpha\sqrt{\log n}}{\mu^{2}_r\Delta_{\vec{z}} n^{1.5}}\right)\,.
\end{aligned}    
\end{equation}
\end{subequations}
\end{restatable}

Then, we find that the main components take the form:
    \[
    \frac{1}{\sqrt{n}}\mat{A}_\uparrow\mat{v}_n(z_i)-\frac{1}{\sqrt{n}}\mat{A}_\downarrow\mat{v}_n(z_i)z^{-1}_i\,,
    \]
    where $\mat{A}$ is a Toeplitz matrix satisfying $\mat{A}=\mathcal{O}(n^{-0.5})$ and $\sup_{i,j}\left|\mat{A}_{i,j}\right|=\mathcal{O}(n^{-1})$ and $\vvec_n(z)\in\mathbb{C}^n$ is defined as $(\vvec_n(z))_i=z^i$. By leveraging the structure of Toeplitz matrices, we obtain:
    \[
    \left\|\frac{1}{\sqrt{n}}\left(\mat{A}_\uparrow\mat{v}_n(z_i)-\mat{A}_\downarrow\mat{v}_n(z_i)z^{-1}_i\right)_k\right\|_\infty=\left\|\frac{1}{\sqrt{n}}\left(-\mat{A}_{k+1,1}z^{-1}_i+\overline{\mat{A}_{k,n}}z^{n-1}_i\right)\right\|_\infty=\mathcal{O}(n^{-1.5})
    \]
    for $1\leq k\leq n-1$. This implies:
    \[
    \left\|\frac{1}{\sqrt{n}}\mat{A}_\uparrow\mat{v}_n(z_i)-\frac{1}{\sqrt{n}}\mat{A}_\downarrow\mat{v}_n(z_i)z^{-1}_i\right\|_2=\mathcal{O}(n^{-1})
    \]
    instead of the trivial bound $\mathcal{O}(n^{-0.5})$. 
    Based on this intuition, we obtain the following lemma:

\begin{restatable}{lemma}{finegrainedI}\label{lem:fine_grained_1}
\begin{align}\label{eq:EPPV_finer_bound}
\Big\|\left(\widetilde{\mat{E}}^Q_1\mat{P}\mat{P}^\dagger_v\right)_\uparrow-\left(\widetilde{\mat{E}}^Q_1\mat{P}\right)_\downarrow(\mat{Q}_\downarrow)^\dagger\mat{Q}_\uparrow\mat{P}^\dagger_v\Big\|_2 \leq \mathcal{O}\left(\frac{r^{1.5}\alpha}{\mu_r n}+\frac{\alpha\sqrt{\log n}}{\mu^{2}_r\Delta_{\vec{z}} n^{1.5}}\right)\,.
\end{align}    
\end{restatable}

Proofs \cref{lem:decompose_1,lem:fine_grained_1} of are deferred to \cref{sec:defer_key_eq_1st}. They immediately imply the first equation of \eqref{eqn:key_equality}.

\begin{corollary}\label{cor:first_component}
It holds that
\begin{align}\label{eq:key_ineq_1_cor}
\left\|\mat{E}^Q_\uparrow-\mat{E}^Q_\downarrow(\mat{Q}_\downarrow)^\dagger\mat{Q}_\uparrow\right\|_2=&~\mathcal{O}\left(\frac{r^{1.5}\alpha}{\mu_r n}+\frac{\alpha^2\log n}{\mu^{2}_r\Delta_{\vec{z}} n}\right)
\end{align}
\end{corollary}
\begin{proof}
We have
\begin{align*}
&\left\|\mat{E}^Q_\uparrow-\mat{E}^Q_\downarrow(\mat{Q}_\downarrow)^\dagger\mat{Q}_\uparrow\right\|_2\\
\leq &~ \mathcal{O}\left(\Big\|\left(\widetilde{\mat{E}}^Q_1\mat{P}\mat{P}^\dagger_v\right)_\uparrow-\left(\widetilde{\mat{E}}^Q_1\mat{P}\right)_\downarrow(\mat{Q}_\downarrow)^\dagger\mat{Q}_\uparrow\mat{P}^\dagger_v\Big\|_2+\frac{\alpha^2\log n}{\mu^{2}_r\Delta_{\vec{z}} n}\right)\\
\leq &~ \mathcal{O}\left(\frac{r^{1.5}\alpha}{\mu_r n}+\frac{\alpha\sqrt{\log n}}{\mu^{2}_r\Delta_{\vec{z}} n^{1.5}} + \frac{\alpha^2\log n}{\mu^{2}_r\Delta_{\vec{z}} n}\right)\\
= &~ \mathcal{O}\left(\frac{r^{1.5}\alpha}{\mu_r n}+\frac{\alpha^2\log n}{\mu^{2}_r\Delta_{\vec{z}} n}\right)\,,
\end{align*}
where the first step follows from \eqref{eq:EPPV_lem}, the second step follows from \eqref{eq:EPPV_finer_bound}.
\end{proof}

\paragraph{Proof of the second equation in \cref{eqn:key_equality}.}
The idea is similar to the proof of the first equation. First, we (approximately) decompose $\mat{Q}_\downarrow^\dagger\left(\mat{E}^Q_\uparrow-\mat{E}^Q_\downarrow(\mat{Q}_\downarrow)^\dagger\mat{Q}_\uparrow\right)$ into four terms: 
\begin{restatable}{lemma}{decomposeII}\label{lem:decompose_2} 
Let $\mat{V}_r:=\frac{1}{\sqrt{n}}\mat{V}_n(\zdom)$. We have
\begin{equation}\label{eq:deompose_2}
\begin{aligned}
\left\|\mat{Q}_\downarrow^\dagger\left(\mat{E}^Q_\uparrow-\mat{E}^Q_\downarrow(\mat{Q}_\downarrow)^\dagger\mat{Q}_\uparrow\right)\right\|_2\leq &~\mathcal{O}\left(\left\|\left((\mat{V}_r)_\uparrow\right)^\dagger\left(\widetilde{\mat{E}}^Q_1\right)_\uparrow\mat{P}\right\|_2+\left\|\left((\mat{V}_r)_\uparrow\right)^\dagger\left(\widetilde{\mat{E}}^Q_2\mat{P}\right)_\uparrow\right\|_2\right.\\
~&+ \left\|\left((\mat{V}_r)_\downarrow\right)^\dagger\left(\widetilde{\mat{E}}^Q_1\right)_\downarrow\mat{P}\right\|_2+\left\|\left((\mat{V}_r)_\downarrow\right)^\dagger\left(\widetilde{\mat{E}}^Q_2\mat{P}\right)_\downarrow\right\|_2\\
~&+\left.\frac{\alpha\sqrt{\log n}}{\mu_r(\Delta_{\vec{z}} n)^{1.5}}\right)\,.    
\end{aligned}
\end{equation}
\end{restatable}

Then, we upper-bound the four terms. 
\begin{itemize}
    \item For $\left((\mat{V}_r)_\uparrow\right)^\dagger\left(\widetilde{\mat{E}}^Q_2\mat{P}\right)_\uparrow$ and $\left((\mat{V}_r)_\downarrow\right)^\dagger\left(\widetilde{\mat{E}}^Q_2\mat{P}\right)_\downarrow$, first notice that $(\mat{V}_r)^\adj (\tilde{\mat{E}}_2^Q)=\mat{0}$, since $\widetilde{\mat{E}}^Q_2\mat{P}\propto \mat{\Pi}_{\mat{Q}^\perp_r}\tilde{\mat{Q}}_{2}$ and the column spaces of $\mat{V}_r$ and $\mat{Q}_r$ are the same. Although the matrices in these two terms involve $\uparrow$ and $\downarrow$, we are still able to show that their norms are small.
    \item For $\left((\mat{V}_r)_\uparrow\right)^\dagger\left(\widetilde{\mat{E}}^Q_1\mat{P}\right)_\uparrow$ and $\left((\mat{V}_r)_\downarrow\right)^\dagger\left(\widetilde{\mat{E}}^Q_1\mat{P}\right)_\downarrow$, we use the structural properties of $\mat{V}_r$, $\mat{E}_1$, and $\mat{E}_2$ to bound each entry of these terms, which will lead to the spectral norm bounds.
\end{itemize}
More specifically, we establish the following lemma:
\begin{restatable}{lemma}{finegrainedII}\label{lem:fine_grained_2}
\begin{subequations}
\begin{equation}\label{eq:V_EQ_P_up}
\left\|\left((\mat{V}_r)_\uparrow\right)^\dagger\left(\widetilde{\mat{E}}^Q_2\mat{P}\right)_\uparrow\right\|_2 \leq {\cal O}\left(\frac{r\alpha^2\log n}{\mu^{2}_r\Delta_{\vec{z}} n^{1.5}}\right),~~~
\left\|\left((\mat{V}_r)_\uparrow\right)^\dagger\left(\widetilde{\mat{E}}^Q_1\right)_\uparrow\mat{P}\right\|_2\leq \mathcal{O}\left(\left(r+\frac{\sqrt{r}}{\Delta_{\vec{z}}}\right)\frac{\alpha\sqrt{\log n}}{\mu_rn^{1.5}}\right)\,,    
\end{equation}
and
\begin{equation}\label{eq:V_EQ_P_down}
\left\|\left((\mat{V}_r)_\downarrow\right)^\dagger\left(\widetilde{\mat{E}}^Q_2\mat{P}\right)_\downarrow\right\|_2 \leq {\cal O}\left(\frac{r\alpha^2\log n}{\mu^{2}_r\Delta_{\vec{z}} n^{1.5}}\right),~~~
\left\|\left((\mat{V}_r)_\downarrow\right)^\dagger\left(\widetilde{\mat{E}}^Q_1\right)_\downarrow\mat{P}\right\|_2\leq \mathcal{O}\left(\left(r+\frac{\sqrt{r}}{\Delta_{\vec{z}}}\right)\frac{\alpha\sqrt{\log n}}{\mu_rn^{1.5}}\right)\,,     
\end{equation}
\end{subequations}
\end{restatable}

Proofs of \cref{lem:decompose_2,lem:fine_grained_2} are deferred to \cref{sec:defer_key_eq_2nd}. They imply the second inequality of \eqref{eqn:key_equality}.

\begin{corollary}\label{cor:second_component}
It holds that
\begin{align}\label{eq:key_ineq_2_cor}
    \left\|\mat{Q}_\downarrow^\dagger\left(\mat{E}^Q_\uparrow-\mat{E}^Q_\downarrow(\mat{Q}_\downarrow)^\dagger\mat{Q}_\uparrow\right)\right\|_2=&~ \mathcal{O}\left(\frac{r\alpha^2\log n}{\mu^{2}_r\Delta_{\vec{z}} n^{1.5}} + \frac{\alpha\sqrt{\log n}}{\mu_r(\Delta_{\vec{z}} n)^{1.5}}\right)\,.
\end{align}
\end{corollary}
\begin{proof}
We have
\begin{align*}
\left\|\mat{Q}_\downarrow^\dagger\left(\mat{E}^Q_\uparrow-\mat{E}^Q_\downarrow(\mat{Q}_\downarrow)^\dagger\mat{Q}_\uparrow\right)\right\|_2
\leq&~ \mathcal{O}\left(\left\|\left((\mat{V}_r)_\uparrow\right)^\dagger\left(\widetilde{\mat{E}}^Q_1\right)_\uparrow\mat{P}\right\|_2+\left\|\left((\mat{V}_r)_\uparrow\right)^\dagger\left(\widetilde{\mat{E}}^Q_2\mat{P}\right)_\uparrow\right\|_2\right.\\
~&+\left.
\left\|\left((\mat{V}_r)_\downarrow\right)^\dagger\left(\widetilde{\mat{E}}^Q_1\right)_\downarrow\mat{P}\right\|_2+\left\|\left((\mat{V}_r)_\downarrow\right)^\dagger\left(\widetilde{\mat{E}}^Q_2\mat{P}\right)_\downarrow\right\|_2\right.\\
~&+\left.\frac{\alpha\sqrt{\log n}}{\mu_r(\Delta_{\vec{z}} n)^{1.5}}\right)\\
\leq &~ \mathcal{O}\left(\frac{r\alpha^2\log n}{\mu^{2}_r\Delta_{\vec{z}} n^{1.5}} + \frac{r\alpha\sqrt{\log n}}{\mu_rn^{1.5}}+\frac{\sqrt{r}\alpha\sqrt{\log n}}{\mu_r\Delta_{\vec{z}}n^{1.5}}+\frac{\alpha\sqrt{\log n}}{\mu_r(\Delta_{\vec{z}} n)^{1.5}}\right)\\
= &~ \mathcal{O}\left(\frac{r\alpha^2\log n}{\mu^{2}_r\Delta_{\vec{z}} n^{1.5}} + \frac{\alpha\sqrt{\log n}}{\mu_r(\Delta_{\vec{z}} n)^{1.5}}\right)\,,
\end{align*}
where the first step follows from \eqref{eq:deompose_2}, and the second step follows from \eqref{eq:V_EQ_P_up} and \eqref{eq:V_EQ_P_down}.
\end{proof}

\subsection{Proof of \texorpdfstring{\cref{thm:key_comparision}}{Theorem}}\label{sec:align_all}
\begin{proof}%
First, the constructions of the matrices $\mat{U}_r$ and $\mat{P}$ are shown in \cref{lem:good_P_construction}. Since $\mat{U}_r$ is an orthogonal matrix that changes the basis of $\mat{Q}$, without loss of generality, we may consider $\mat{U}_r = \mat{I}_r$. It remains to prove that \eqref{eqn:matrix_perb_improve} holds for this $\mat{P}$.  

By \cref{lem:coarse_bound}, we have
\begin{align*}
&\left\|\mat{P}\mat{\hat{Q}}^{+}_\downarrow\mat{\hat{Q}}_\uparrow\mat{P}^{-1}-\mat{Q}^{+}_\downarrow\mat{Q}_\uparrow\right\|_2\\
\leq &~ \left\|(\mat{E}^Q_\downarrow)^\dagger\left(\mat{E}^Q_\uparrow-\mat{E}^Q_\downarrow(\mat{Q}_\downarrow)^\dagger\mat{Q}_\uparrow\right)\right\|_2+\mathcal{O}\left(1+\frac{\alpha\sqrt{\log n}}{\mu_r\sqrt{\Delta_{\vec{z}} n}}\right)\left\|\mat{Q}_\downarrow^\dagger\left(\mat{E}^Q_\uparrow-\mat{E}^Q_\downarrow(\mat{Q}_\downarrow)^\dagger\mat{Q}_\uparrow\right)\right\|_2\\
~&+ \mathcal{O}\left(\frac{\alpha^3\log^3 n}{\mu^{3}_r(\Delta_{\vec{z}} n)^{1.5}} \right)\\
\leq &~ \left\|\mat{E}^Q\right\|_2\cdot \left\|\left(\mat{E}^Q_\uparrow-\mat{E}^Q_\downarrow(\mat{Q}_\downarrow)^\dagger\mat{Q}_\uparrow\right)\right\|_2+\mathcal{O}(1)\cdot \left\|\mat{Q}_\downarrow^\dagger\left(\mat{E}^Q_\uparrow-\mat{E}^Q_\downarrow(\mat{Q}_\downarrow)^\dagger\mat{Q}_\uparrow\right)\right\|_2\\
&+\mathcal{O}\left(\frac{\alpha^3\log^3 n}{\mu^{3}_r(\Delta_{\vec{z}} n)^{1.5}} \right)\,,
\end{align*}
where we use $\left\|(\mat{E}^Q_\downarrow)^\dagger\right\|_2\leq \left\|\mat{E}^Q\right\|_2$ and $n=\Omega\left(\frac{\alpha^2}{\mu_r^2\Delta_{\vec{z}}}\right)$ in the second step.

By \eqref{eqn:second_Q_perturb} and \eqref{eq:key_ineq_1_cor},
\begin{align*}
    \left\|\mat{E}^Q\right\|_2\cdot \left\|\left(\mat{E}^Q_\uparrow-\mat{E}^Q_\downarrow(\mat{Q}_\downarrow)^\dagger\mat{Q}_\uparrow\right)\right\|_2 \leq &~ \mathcal{O}\left(\frac{\alpha\sqrt{\log n}}{\mu_r\sqrt{\Delta_{\vec{z}}n}}\cdot \left(\frac{r^{1.5}\alpha}{\mu_r n}+\frac{\alpha^2\log n}{\mu^{2}_r\Delta_{\vec{z}} n}\right)\right)\\
    =&~ {\cal O}\left(\frac{r^{1.5}\alpha^2\sqrt{\log n}}{\mu_r^{2}\Delta_{\vec{z}}^{0.5} n^{1.5}}+\frac{\alpha^3\log^{1.5} n}{\mu^{3}_r\Delta_{\vec{z}}^{1.5} n^{1.5}}\right)\,.
\end{align*}
And by \eqref{eq:key_ineq_2_cor},
\begin{align*}
    {\cal O}(1)\cdot \left\|\mat{Q}_\downarrow^\dagger\left(\mat{E}^Q_\uparrow-\mat{E}^Q_\downarrow(\mat{Q}_\downarrow)^\dagger\mat{Q}_\uparrow\right)\right\|_2=&~ \mathcal{O}\left(\frac{r\alpha^2\log n}{\mu^{2}_r\Delta_{\vec{z}} n^{1.5}} + \frac{\alpha\sqrt{\log n}}{\mu_r(\Delta_{\vec{z}} n)^{1.5}}\right)\,.
\end{align*}

Thus, we obtain that
\begin{align*}
    \left\|\mat{P}\mat{\hat{Q}}^{+}_\downarrow\mat{\hat{Q}}_\uparrow\mat{P}^{-1}-\mat{Q}^{+}_\downarrow\mat{Q}_\uparrow\right\|_2
    \leq &~ {\cal O}\left(\frac{\alpha^3\log^{1.5} n}{\mu^{3}_r\Delta_{\vec{z}}^{1.5} n^{1.5}}+\frac{r^{1.5}\alpha^2\log n}{\mu^{2}_r\Delta_{\vec{z}} n^{1.5}} \right)\,,
\end{align*}
which completes the proof of the theorem.
\end{proof}

\ifdefined\isarxiv
\vspace{1em}
\noindent\textbf{Acknowledgments.}
This material is partially supported by the U.S. Department of Energy, Office of Science, National Quantum Information Science Research Centers, Quantum Systems Accelerator (Z.D.), by the U.S.\ Department of Energy, Office of Science, Office of Advanced Scientific Computing Research, Department of Energy Computational Science Graduate Fellowship under Award Number DE-SC0021110 (E.N.E.), by the Applied Mathematics Program of the US Department of Energy (DOE) Office of Advanced Scientific Computing Research under contract number DE-AC02-05CH1123 and the Simons Investigator program (L.L.), and by DOE Grant No. DE-SC0024124 (R.Z.). 
We thank the Institute for Pure and Applied Mathematics (IPAM) for its hospitality throughout the semester-long program   ``Mathematical and Computational Challenges in Quantum Computing'' in Fall 2023 during which this work was initiated. We thank Haoya Li, Hongkang Ni, Joel Tropp, Rob Webber, and Lexing Ying for insightful discussions.

\vspace{1em}
\noindent\textbf{Disclaimer.}
This report was prepared as an account of work sponsored by an agency of the
United States Government. Neither the United States Government nor any agency thereof, nor
any of their employees, makes any warranty, express or implied, or assumes any legal liability
or responsibility for the accuracy, completeness, or usefulness of any information, apparatus,
product, or process disclosed, or represents that its use would not infringe privately owned
rights. Reference herein to any specific commercial product, process, or service by trade name,
trademark, manufacturer, or otherwise does not necessarily constitute or imply its
endorsement, recommendation, or favoring by the United States Government or any agency
thereof. The views and opinions of authors expressed herein do not necessarily state or reflect
those of the United States Government or any agency thereof.
\fi

\newpage

\appendix

\section{Preliminaries}\label{sec:prelim}
\subsection{Notation} \label{sec:notation}
We use $\iu\coloneqq\sqrt{-1}$ to denote the imaginary unit. For a complex number $z=a + b\iu\in \C$, we use $\overline{z}=a-b\iu$ to denote the complex conjugate of $z$.  

We use $\mathbb{T}$ to denote the unit circle in the complex plane. We use $\|\cdot\|_2$ to denote the matrix spectral norm and $\|\cdot\|_{\rm F}$ to denote the Frobenius norm. 
$\norm{\cdot}_p$ denotes the vector $p$-norm.

For a matrix $\mat{A}\in \C^{n\times m}$, we use $\mat{A}(:a,:b)$ to denote the submatrix with rows from 1 to $a$ and columns from 1 to $b$ and define $\mat{A}_r$ as the matrix that contains first $r$ columns of $\mat{A}$. We use $\mat{A}^\adj$ to represent the Hermitian transpose (or conjugate transpose) of $\mat{A}$, i.e., $(\mat{A}^\adj)_{ij}=\overline{\mat{A}_{ji}}$. In addition, $\lambda_i(\mat{A})$ means the $i$-th eigenvalue of $\mat{A}$, $\sigma_i(\mat{A})$ means the $i$-th singular value of $\mat{A}$, $\lambda_{\min}(\mat{A})$ means the smallest nonzero eigenvalue of $\mat{A}$ if $\mat{A}$ is nonnegative definite, and $\sigma_{\min}(\mat{A})$ means the smallest nonzero singular value of $\mat{A}$. For a matrix $\mat{A}$ with orthonormal columns, we use $\mat{\Pi}_{\mat{A}}$ and $\mat{\Pi}_{\mat{A}^\perp}$ to denote the projectors to the column space and the orthogonal subspace:
\begin{equation*}
    \mat{\Pi}_{\mat{A}} \coloneqq \mat{A}\mat{A}^\adj, \quad \mat{\Pi}_{\mat{A}^\perp} \coloneqq \Id - \mat{\Pi}_{\mat{A}}.
\end{equation*}
For two matrices $\mat{A},\mat{B}\in \C^{m\times n}$ and some quantity $f\in \mathbb{R}_+$, we use $\mat{A}\leq \mat{B}+f$ to denote $\norm{\mat{A}-\mat{B}}_2\leq f$.

\subsection{Pseudoinverses} \label{sec:pseudoinverse}

We will make frequent use of the (Moore--Penrose) pseudoinverse:

\begin{definition}[Moore--Penrose pseudoinverse]\label{def:pseudoinverse}
Let $\mat{A}\in\mathbb{C}^{m\times n}$ be a matrix with SVD $\mat{A}=\mat{U}\mat{\Sigma}\mat{V}^\dagger$. The \emph{Moore--Penrose pseudoinverse} (or just pseudoinverse) of $\mat{A}$ is defined as
\[
\mat{A}^+=\mat{V}\mat{\Sigma}^{+}\mat{U}^\dagger\,,
\]
where $\mat{\Sigma}^\pinv$ is an $n\times m$ matrix with $(\mat{\Sigma}^\pinv)_{ij}=(\mat{\Sigma}_{ii})^{-1}$ if $i = j$ and $\mat{\Sigma}_{ii}\ne 0$ and  $(\mat{\Sigma}^\pinv)_{ij} = 0$ otherwise.
\end{definition}
In this paper, we only use the pseudoinverse of full-rank (rectangular) matrices. In particular, if $\mat{A}\in\mathbb{C}^{m\times n}$ is a matrix of full column rank, its pseudoinverse is
\begin{equation}\label{eq:pseudo_inv_col}
\mat{A}^+=\left(\mat{A}^\dagger\mat{A}\right)^{-1}\mat{A}^\dagger\,.
\end{equation}

We will make use of the reverse order for the pseudoinverse \cite{gre66}:

\begin{fact}[Reverse order law of pseudoinverse]\label{fac:rev_order_pinv}
For a full column-rank matrix $\mat{A}\in \C^{m\times n}$ and a full row-rank matrix $\mat{B}\in \C^{n\times p}$ , $(\mat{A}\mat{B})^{\pinv}=\mat{B}^{\pinv}\mat{A}^{\pinv}$.    
\end{fact}

Perturbation theory for pseudoinverses is provided by the following theorem \cite[eq.~(21)]{Han87}:

\begin{theorem}[Perturbation for pseudoinverses] \label{thm:pseudoinverse-perturbation}
    Let $\mat{A}$ be a matrix of full column rank and suppose that $\norm{\mat{E}}_2 < \sigma_{\rm min}(\mat{A})$.
    Then
    \begin{equation*}
        \norm{(\mat{A}+\mat{E})^\pinv - \mat{A}^\pinv}_2 \le \frac{3\norm{\smash{\mat{A}^\pinv}}_2^2\norm{\mat{E}}_2}{1 - \norm{\smash{\mat{A}^\pinv}}_2\norm{\mat{E}}_2}.
    \end{equation*}
\end{theorem}

\subsection{Matrix norms}
In this paper, we mainly focus on the spectral norm of a matrix: $\|\mat{A}\|_2\coloneqq \sup_{\|\vec{x}\|_2=1}\|\mat{A}\vec{x}\|_2$ for any $\mat{M}\in \C^{n\times m}$. We will frequently use the following two basic facts about the spectral norm:
\begin{fact}[Sub-multiplicativity of spectral norm]
For any $\mat{A}\in \C^{n\times m}$ and $\mat{B}\in \C^{m\times p}$, it holds that
\begin{align*}
    \|\mat{A}\mat{B}\|_2\leq \|\mat{A}\|_2\cdot \|\mat{B}\|_2\,.
\end{align*}
\end{fact}
\begin{fact}[Spectral norm of submatrix]\label{fac:submatrix_norm}
Let $\mat{A}\in \C^{n\times m}$ and $\mat{B}$ be any sub-matrix of $\mat{A}$. Then, we have
\begin{align*}
    \|\mat{B}\|_2\leq \|\mat{A}\|_2\,.
\end{align*}
\end{fact}
\begin{proof}
Without loss of generality, we may assume that $\mat{A}=\begin{bmatrix}\mat{B} & *\\ * & *\end{bmatrix}$, since permuting the columns and rows does not change the spectral norm. Then, we have
\begin{align*}
    \|\mat{B}\|_2=\sup_{\norm{\vec{x}}_2=1}~\norm{\mat{B}\vec{x}}_2 \leq \sup_{|\vec{x}|_2=1} ~ \norm{\begin{bmatrix}\mat{B} & *\\ * & *\end{bmatrix} \begin{bmatrix}\vec{x} & \vec{0}\end{bmatrix}}_2\leq \sup_{\|\vec{y}\|_2=1}~\norm{\mat{A}\vec{y}}_2=\|\mat{A}\|_2\,.
\end{align*}
The result is proven.
\end{proof}

\subsection{Neumann series}

The Neumann series is one of the most powerful tools in matrix analysis.
It is the matrix generalization of a geometric series.

\begin{definition}[Neumann series]\label{def:neumann}
Given a matrix $\mat{A}\in \C^{n\times n}$ with $\|\mat{A}\|_2< 1$, the \emph{Neumann series} of $\mat{A}\in \C^{n\times n}$ is the convergent series representation
\begin{equation*}
    (\Id_n - \mat{A})^{-1} = \sum_{k=0}^\infty \mat{A}^k.
\end{equation*}
\end{definition}

\subsection{Eigenvalue and eigenvector perturbation theory} \label{app:eigen-perturbation}

Our argument will require several standard results from matrix perturbation theory.
First, we will make use of additive and multiplicative perturbation theorems for eigenvalues of Hermitian matrices.
For additive perturbation, we will use Weyl's theorem:

\begin{theorem}[Weyl's theorem] \label{thm:weyl}
    Let $\mat{A},\mat{E}\in\complex^{n\times n}$ be Hermitian matrices.
    Then, 
    \begin{equation*}
        \lambda_i(\mat{A}) - \norm{\mat{E}}_2 \le \lambda_i(\mat{A}+\mat{E}) \le \lambda_i(\mat{A}) + \norm{\mat{E}}_2 \for i=1,2,\ldots,n.
    \end{equation*}
\end{theorem}

For multiplicative perturbation theory, we will use Ostrowski's theorem.

\begin{theorem}[Ostrowski's theorem]\label{thm:ostrowski}
    Let $\mat{A} \in \complex^{r\times r}$ be a nonsingular Hermitian matrix and $\mat{B} \in \complex^{n\times r}$ have full column-rank.
    Then, $\mat{B}\mat{A}\mat{B}^\adj$ has $r$ nonzero eigenvalues satisfying the bounds
    \begin{equation*}
        \sigma_{\rm min}^2(\mat{B}) \lambda_i(\mat{A}) \le \lambda_i\left(\mat{B}\mat{A}\mat{B}^\adj\right) \le \sigma_{\rm max}^2(\mat{B}) \lambda_i(\mat{A}) \for i=1,2,\ldots,r.
    \end{equation*}
\end{theorem}

Ostrowski's theorem for the case in which $\mat{A}$ and $\mat{B}$ are both square appears in \cite[Thm.~4.5.9]{horn_matrix_2012}.
The generalization to $\mat{B}$ rectangular presented here is straightforward.

Next, we will need perturbation bounds of eigenspaces of Hermitian matrices.
The classical measure of the distance between subspaces is the principal angle, which admits multiple equivalent definitions \cite[Sec.~I.5]{SS90}.
We use the following definition: for matrices $\mat{U}$ and $\mat{V}$ with orthonormal columns, the sin of the principal angle between the ranges of $\mat{U}$ and $\mat{V}$ is
\begin{equation*}
    \sin \theta(\mat{U},\mat{V}) \coloneqq \norm{\outprod{\mat{U}} - \outprod{\mat{V}}}_2.
\end{equation*}
With this definition in hand, we now state the classical result for the perturbation of eigenspaces of Hermitian matrices, due to a Davis and Kahan
\cite[$\sin \theta$ theorem]{davis_rotation_1970}: 

\begin{theorem}[Davis--Kahan $\sin \theta$ theorem] \label{thm:sin}
Let $\mat{T}$, $\hat{\mat{T}}$, $\mat{Q}$, $\hat{\mat{Q}}$ be defined as \cref{sec:notation}, 
    \begin{equation*}
        \sin \theta\left(\mat{Q}_r,\mat{\hat{Q}}_r\right) \le \frac{\norm{\mat{E}\mat{Q}_r}_2}{\lambda_r(\mat{T}) - \lambda_{r+1}(\mat{\hat{T}})}.
    \end{equation*}
\end{theorem}

The following result (see, e.g., the proof of \cite[Lem.~2.7]{moitra_super-resolution_2015}) shows that a bound on the angle between subspaces yields a bound between $\mat{Q}$ and $\mat{\hat{Q}}$, after multiplication by a unitary matrix.

\begin{proposition}[Distance from angles] \label{prop:distance-from-angles}
    Let $\mat{Q}_r$ and $\mat{\hat{Q}}_r$ be matrices with orthonormal columns.
    There exists a unitary matrix such that
    \begin{equation*}
        \norm{\mat{\hat{Q}}_r - \mat{Q}_r\mat{U}}_2 \le 2 \sin \theta\left(\mat{Q}_r,\mat{\hat{Q}}_r\right) = 2\norm{\outprod{\mat{Q}_r} - \outprod{\mat{\hat{Q}}_r}}_2.
    \end{equation*}
\end{proposition}

Last, we will need one result for the perturbation of diagonalizable matrices, a version of the famous Bauer--Fike theorem augmented by Gerschgorin and Ostrowski's continuity argument.
We take this result from \cite{SS90}:

\begin{theorem}[Bauer--Fike theorem, {\normalfont {\cite[Thm.~IV.3.3]{SS90}}}] \label{thm:bauer-fike}
    Let $\mat{A} = \mat{V} \diag(\vec{\lambda})\mat{V}^{-1} \in \complex^{r\times r}$ be a diagonalizable matrix and let $\vec{\hat{\lambda}}$ be the eigenvalues of $\mat{A}+\mat{E}$.
    Then
    \begin{equation*}
        \max_{1\le i\le r} \min_{1\le j\le r} | \hat{\lambda}_i - \lambda_j| \le \norm{\mat{V}}_2 \norm{\smash{\mat{V}^{-1}}}_2 \norm{\mat{E}}_2.
    \end{equation*}
    Moreover, we have the following bound in the optimal matching distance, defined in \cref{eqn:optimal_matching}:
    \begin{equation*}
        \operatorname{md} (\vec{\lambda},\vec{\hat{\lambda}}) \le (2r-1)\norm{\mat{V}}_2 \norm{\smash{\mat{V}^{-1}}}_2 \norm{\mat{E}}_2.
    \end{equation*}
\end{theorem}

\subsection{Spectral perturbation via resolvents} \label{app:resolvents}

The last results that we will need from classical matrix perturbation theory are formulas for spectral projectors using contour integrals and resolvents.
For this section, let $\mat{A} \in \complex^{n\times n}$ be a Hermitian matrix with eigenvalues $\lambda_1 \ge \cdots \ge \lambda_n$ associated to orthonormal eigenvectors $\vec{q}_1,\ldots,\vec{q}_n$.

The core object is the resolvent:

\begin{definition}[Resolvent]
    The resolvent function $\mat{R} : \complex \setminus \{\lambda_1,\ldots,\lambda_n\} \to \complex^{n\times n}$ of $\mat{A}$ is 
    \begin{equation*}
        \mat{R}(\zeta) \coloneqq (\zeta \Id - \mat{A})^{-1}.
    \end{equation*}
    For simplcity, we shall often omit ``$\Id$'' and write the resolvent merely as $\mat{R}(\zeta) = (\zeta - \mat{A})^{-1}$.
\end{definition}

We first state a lemma that estimates the locations of the perturbed eigenvalues via the resolvent method. 
\begin{lemma}[Eigenvalue perturbation via resolvents, {\normalfont {\cite[Lemma 2.1]{Micheal_2022}}, see also {\cite[Theorem IV.3.18]{Tos_1980}}}]\label{lem:eig_location}
Let $\mat{A}\in \C^{r\times r}$ and $\tilde{\mat{A}}$ be its perturbation. For any $i\in [r]$, suppose the $i$-th eigenvalue of $\mat{A}$ is $2\delta$-separated from other eigenvalues, i.e., $\min_{j\ne i}|\lambda_i - \lambda_j|\geq 2\delta>0$. Let ${\cal C}$ be a circle in the complex plane centered at $\lambda_i$ with radius $\delta$. If the following condition holds:
\begin{align*}
    \norm{\tilde{\mat{A}}-\mat{A}}_2 \cdot \sup_{z\in {\cal C}} \norm{(z\Id_r - \mat{A})^{-1}}_2<1\,,
\end{align*}
then $\tilde{\mat{A}}$ also has exactly one eigenvalue contained in ${\cal C}$.
\end{lemma}

Next, we present the following results that express the spectral projector as the contour 
 integral of the resolvent:

\begin{theorem}[Spectral theory via resolvents, {\normalfont {\cite[Sec. I.5]{Tos_1980}}}] \label{thm:spectral-resolvent}
    Let $\mathcal{C}$ be a simple closed curve in $\complex$ such that $\lambda_{i_1},\ldots,\lambda_{i_\ell}$ lie within the interior of $\mathcal{C}$ and all other eigenvalues lie outside $\mathcal{C}$.
    Then
    \begin{equation*}
        \sum_{j=1}^\ell \vec{q}_{i_j}^{\vphantom{\adj}}\vec{q}_{i_j}^\adj = \frac{1}{2\pi \iu} \oint_{\mathcal{C}} (\zeta - \mat{A})^{-1} \, \d \zeta.
    \end{equation*}
    Here, as always, $\oint_{\mathcal{C}}$ denotes a contour integral around the curve $\mathcal{C}$, traversed counterclockwise.
\end{theorem}

This result follows immediately from the Cauchy integral formula and the following spectral representation of the resolvent $(\zeta - \mat{A})^{-1} = \sum_{i=1}^n (\zeta - \lambda_i)^{-1} \vec{q}_i^{\vphantom{\adj}} \vec{q}_i^{\adj}$.

\Cref{thm:spectral-resolvent} can be used to derive a convergent perturbation series for spectral projectors of a perturbation $\mat{A}+\mat{E}$ of $\mat{A}$ by expanding the resolvent $(\zeta - (\mat{A}+\mat{E}))^{-1}$ in a Neumann series in $\mat{E}$.
Here is one such result:

\begin{lemma}[Perturbation theory via resolvents] \label{lem:perturbation-resolvents}
    Assume $\mat{A}$ is rank-$r$, $\lambda_{r+1} = \cdots = \lambda_n = 0$ and let $\mat{A}+\mat{E}\in\complex^{n\times n}$ be a Hermitian perturbation of $\mat{A}$ with eigenvalues $\hat{\lambda}_1 \ge \cdots \hat{\lambda}_n$ with associated eigenvectors $\vec{\hat{q}}_1,\ldots,\vec{\hat{q}}_n$.
    Let $\mathcal{C}$ be a simple closed curve in $\complex$ such that $\mathcal{C}$ $\lambda_1,\ldots,\lambda_r,\hat{\lambda}_1,\ldots,\hat{\lambda}_r$ lie inside $\mathcal{C}$ and all other eigenvalues of both $\mat{A}$ and $\mat{A}+\mat{E}$ lie outside $\mathcal{C}$.
    Assume that
    \begin{equation*}
        \norm{\mat{E}(\zeta - \mat{A})^{-1}}_2 < 1 \for \zeta \in \mathcal{C}.
    \end{equation*}
    Then, denoting $\mat{\Pi} = \sum_{i=1}^r \outprod{\vec{q}_i}$ and $\mat{\hat{\Pi}} = \sum_{i=1}^r \outprod{\vec{\hat{q}}_i}$, we have
    \begin{equation*}
        \mat{\hat{\Pi}} = \mat{\Pi} + \sum_{k=1}^\infty \frac{1}{2\pi \iu}\oint_{\mathcal{C}} (\zeta - \mat{A})^{-1} \left( \mat{E} (\zeta - \mat{A})^{-1} \right)^k \, \d\zeta.
    \end{equation*}
\end{lemma}

\begin{proof}
    Since $\mathcal{C}$ is a simple closed curve, it is compact.
    Therefore, there exists $0 < \alpha < 1$ such that 
    \begin{equation*}
        \norm{\mat{E}(\zeta - \mat{A})^{-1}}_2 \le \alpha \for \zeta \in \mathcal{C}.
    \end{equation*}
    Therefore, the resolvent $(\zeta - (\mat{A}+\mat{E}))^{-1}$ has a Neumann expansion
    \begin{equation*}
        (\zeta - (\mat{A}+\mat{E}))^{-1} = \sum_{k=0}^\infty (\zeta - \left(  (\zeta - \mat{A})^{-1} \mat{E}\right)^k\mat{A})^{-1}.
    \end{equation*}
    that is uniformly convergent on $\mathcal{C}$.
    Therefore, integrating both sides we obtain
    \begin{equation*}
        \frac{1}{2\pi \iu} \oint_{\mathcal{C}} (\zeta - (\mat{A}+\mat{E}))^{-1}\, \d \zeta = \frac{1}{2\pi \iu}\sum_{k=0}^\infty \oint_{\mathcal{C}} (\zeta - \mat{A})^{-1} \left( \mat{E} (\zeta - \mat{A})^{-1} \right)^k \, \d\zeta.
    \end{equation*}
    By \cref{thm:spectral-resolvent}, the left-hand side is $\mat{\hat{\Pi}}$ and the first term in the right-hand side is $\mat{\Pi}$:
    \begin{equation*}
        \mat{\hat{\Pi}} = \mat{\Pi} + \sum_{k=1}^\infty \frac{1}{2\pi \iu}\oint_{\mathcal{C}} (\zeta - \mat{A})^{-1} \left( \mat{E} (\zeta - \mat{A})^{-1} \right)^k \, \d\zeta. 
    \end{equation*}
    This is the stated conclusion.
\end{proof}

\section{Vandermonde matrices}\label{sec:van_mat}

In this section, we discuss results for Vandermonde matrices that will be heavily used throughout our analysis.
We begin in \cref{sec:moitra} by discussing Moitra's bounds \cite{moitra_super-resolution_2015} for the singular values of a Vandermonde matrix. 
Then, in \cref{sec:vand-eig}, we present a comparison lemma for the Vandermonde basis and eigenbasis of a Toeplitz matrix.

\subsection{Moitra's singular value bounds} \label{sec:moitra}

For us, one of the core technical tools for analyzing ESPRIT is bound on the singular values of Vandermonde matrices, proven in the seminal work of Moitra \cite{moitra_super-resolution_2015}.
The main result is as follows: 

\begin{theorem}[Vandermonde matrix: singular values, {\normalfont adapted from {\cite[Thm. 1.1]{moitra_super-resolution_2015}}}] \label{thm:moitra}
    Let $\vec{z} \in \unitcircle^s$ be a set of frequencies and define a gap
    \begin{equation*}
        \delta_{\vec{z}} \coloneqq \min_{i\neq j}|z_i-z_j|.
    \end{equation*}
    Then if $n > 1 + 2\pi/\delta_{\vec{z}}$, the largest and smallest singular values of $\Vand_n(\vec{z})$ satisfy
    \begin{equation*}
        \sigma_{\rm max}(\Vand_n(\vec{z})) \le \sqrt{n - 1 + \frac{2\pi}{\delta_{\vec{z}}}}, \quad \sigma_{\rm min}(\Vand_n(\vec{z})) \ge \sqrt{n - 1 - \frac{2\pi}{\delta_{\vec{z}}}}.
    \end{equation*}
\end{theorem}

Moitra's original results were stated in terms of the minimum separation of the $z_i$'s in the wrap-around metric; here, we have transferred these results to use the standard distance on $\complex$.

This result immediately yields the following useful corollary:

\begin{corollary}[Vandermonde matrix: Gram matrix and outer product] \label{cor:moitra}\label{cor:V_dagger_V_dom}
    In the setting of \cref{thm:moitra},
    \begin{align*}
        \norm{ \frac{1}{n} \Vand_n(\vec{z})^\adj \Vand_n(\vec{z}) - \Id_r }_2 = \order\left( \frac{1}{\delta_{\vec{z}} n}\right),\quad\norm{ \frac{1}{n} \Vand_n(\vec{z}) \Vand_n(\vec{z})^\adj - \Id_n}_2 \le 1\,.
    \end{align*}
\end{corollary}

\subsection{Vandermonde basis and eigenbasis} \label{sec:vand-eig}
The exact Toeplitz matrix $\mat{T}$ \cref{eqn:T} has two useful factorizations, each of which packages useful information about the matrix.
The first factorization is the eigendecomposition
\begin{equation}\label{eqn:eigendecomposition}
    \mat{T} = \mat{Q} \mat{\Sigma} \mat{Q}^\adj \for \mat{Q} \in \complex^{n\times n}, \: \mat{\Sigma} \in \real^{n\times n}.
\end{equation}
The eigendecomposition is useful because it can be easily computed using standard methods from numerical linear algebra.
The second useful factorization is the \emph{Vandermonde decomposition}
\begin{equation*}
    \mat{T} = \Vand_n(\zdom)\diag(\mudom) \Vand_n(\zdom)^\adj.
\end{equation*}
The Vandermonde decomposition contains information about the frequency spectrum.
However, it can only be computed indirectly, such as by procedures such as ESPRIT which use the eigendecomposition as a subroutine.

Recall that $$\zgap:=\min_{1\leq i\leq r,1\leq j\leq d,i\neq j}|z_i-z_j|\leq \min_{1\leq i,j\leq r,i\neq j}|z_i-z_j|.$$ A direct consequence of \cref{thm:moitra} is the following estimates for the eigenvalues of $\mat{T}$:

\begin{corollary}[Toeplitz matrix: eigenvalues]\label{cor:moitra_T}
Assume $n > 1 + 2\pi/\Delta_{\vec{z}}$.
It holds that
\begin{equation}\label{eqn:eigenvalue_prop}
\lambda_1(\mat{T}) = \mu_1\left(n-1+\frac{2\pi}{\Delta_{\vec{z}}}\right), \quad \lambda_r(\mat{T}) = \mu_r \left(n-1-\frac{2\pi}{\Delta_{\vec{z}}}\right)\,,
\end{equation}
and $\lambda_i(\mat{T}) = 0$ for $i=r+1,\ldots,n$. 
\end{corollary}
\begin{proof}
Applying Ostrowski's theorem (\cref{thm:ostrowski}) to the Vandermonde decomposition $\mat{T} = \Vand_n(\zdom) \cdot \diag(\mudom) \Vand_n(\zdom)^\adj$ gives that for any $1\leq i \leq r$,
\begin{align*}
    \sigma_{\min}^2(\mat{V}_n(\zdom))\cdot \mu_i \leq \lambda_i(\mat{T})\leq \sigma_{\max}^2(\mat{V}_n(\zdom))\cdot \mu_i\,.
\end{align*}
By \cref{thm:moitra}, we have
\begin{align*}
    \lambda_1(\mat{T}) \leq \mu_1\left(n-1+\frac{2\pi}{\Delta_{\vec{z}}}\right),\quad \text{and}\quad \lambda_r(\mat{T})\geq \mu_r \left(n-1-\frac{2\pi}{\Delta_{\vec{z}}}\right)\,.
\end{align*}
The corollary then follows.
\end{proof}

The essential fact that we will make frequent use of throughout our analysis is that when $n\gg 1/\zgap$, the eigenbasis $\mat{Q}$ and the scaled Vandermonde basis $\Vand_n(\zdom)/\sqrt{n}$ are equivalent up to multiplication by a nearly unitary matrix $\mat{P}_v$.
The result is as follows:
\vandereigen*

\begin{proof}

Specifically, because the image of $\mat{Q}_r$ %
equals to the image of $\mat{V}_n(\zdom)$ by the definition of $\mat{Q}_r$, there exists an invertible matrix $\mat{P}_v\in\mathbb{C}^{r\times r}$ such that
\begin{align*}%
\mat{Q}_r=\frac{1}{\sqrt{n}}\mat{V}_n(\zdom)\mat{P}_v\,.
\end{align*}
By multiplying both sides of this equation by the pseudoinverse of $\mat{V}_n(\zdom)/\sqrt{n}$, we obtain
\begin{align}\label{eq:def_P_v}
\mat{P}_v:=\left(\frac{1}{n}\mat{V}^\dagger_n(\zdom)\mat{V}_n(\zdom)\right)^{-1}\frac{1}{\sqrt{n}}\mat{V}^\dagger_n(\zdom)\mat{Q}_r\,.
\end{align}

It remains to prove \cref{eqn:P_orthogonal}.
Notice that the matrices $\mat{P}_v \mat{P}_v^\dagger-\Id_r$ and $\mat{P}_v^\dagger \mat{P}_v-\Id_r$ are adjoints of each other. Thus, we only need to prove one of the equality in \cref{eqn:P_orthogonal}.
Consider the second equality. We have
\begin{equation} \label{eq:P_v_manipulations}
    \begin{split}
    \left\|\mat{P}_v^\dagger \mat{P}_v-\Id_r\right\|_2=&~\left\|\mat{P}_v^\dagger \mat{P}_v-\mat{Q}_r^\dagger\mat{Q}_r\right\|_2\\
    =&~ \left\|\mat{P}_v^\dagger\left(\Id_r-\frac{\mat{V}^\dagger_n(\zdom)\mat{V}_n(\zdom)}{n}\right)\mat{P}_v\right\|_2\\
    =&~ \mathcal{O}\left(\frac{1}{\Delta_{\vec{z}} n}\cdot \|\mat{P}_v^\dagger\mat{P}_v\|_2\right)\\
    = &~ {\cal O}\left(\frac{1}{\Delta_{\vec{z}}n}\cdot \left(\|\mat{P}^\dagger_v \mat{P}_v-\Id_r\|_2 + \|\Id_r\|_2\right)\right)\\
    = &~ {\cal O}\left(\frac{1}{\Delta_{\vec{z}}n}\right)\cdot \left\|\mat{P}^\dagger_v \mat{P}_v-\Id_r\right\|_2 + {\cal O}\left(\frac{1}{\Delta_{\vec{z}}n}\right)\,.
    \end{split}
\end{equation}
The first step follows from the fact that $\mat{Q}_r$ has orthonormal columns, $\mat{Q}_r^\dagger \mat{Q}_r=\Id_r$.
The second step follows from the identity $\mat{Q}_r=\mat{V}_n(\zdom)\mat{P}_v/\sqrt{n}$, the submultipicative property of the norm $\norm{\cdot}_2$, and the fact that $\norm{\mat{P}_v^\adj\mat{P}_v}_2^2 = \norm{\mat{P}_v}_2^2 = \norm{\mat{P}_v^\adj}_2^2$.
The third step follows from \cref{cor:V_dagger_V_dom}.
The fourth step follows from triangle inequality. 
Now, rearrange \cref{eq:P_v_manipulations} to obtain the desired conclusion
\begin{align*}%
    \left\|\mat{P}_v^\dagger \mat{P}_v-\Id_r\right\|_2 = {\cal O}\left(\frac{1/(\Delta_{\vec{z}}n)}{1-1/(\Delta_{\vec{z}}n)}\right)=\mathcal{O}\left(\frac{1}{\Delta_{\vec{z}} n}\right)\,.
\end{align*}    
The lemma is then proved.
\end{proof}

As a corollary of this result, we obtain a result for $\mat{Q}_\downarrow$, defined to be $\mat{Q}_r$ with its first row deleted.

\begin{corollary}[Approximate unitarity of $\mat{Q}_\downarrow$]\label{cor:approximate-orthonormality}
    Assuming $n > 1 + 2\pi/\zgap$, it holds that
    \begin{equation} \label{eq:approximate-orthonormality}
        \norm{\mat{Q}_\downarrow^\adj \mat{Q}_\downarrow - \Id_r}_2 = \order \left( \frac{1}{\zgap n} \right).
    \end{equation}
    In addition,
    \begin{equation} \label{eq:W-bound-1}
        \norm{\mat{Q}_\uparrow}_2, \norm{\mat{Q}_\uparrow^\pinv}_2 = \order(1).
    \end{equation}
\end{corollary}

\begin{proof}
    Utilizing the matrix from \cref{lem:vander-eigen}, we have
    \begin{equation*}
        \mat{Q}_\downarrow = \frac{1}{\sqrt{n}} \Vand_n(\zdom)_\downarrow \mat{P}_v.
    \end{equation*}
    By the definition of the Vandermonde matrix, 
    \begin{equation*}
        \Vand_n(\zdom)_\downarrow = \Vand_{n-1}(\zdom) \cdot \mat{D} \quad \text{for } \mat{D} = \diag(\zdom).
    \end{equation*}
    Thus, we obtain
    \begin{equation} \label{eq:near-orth-proof}
        \begin{split}
        \norm{\mat{Q}_\downarrow^\adj \mat{Q}_\downarrow - \Id_r}_2 
        &= \norm{\frac{1}{n} \mat{P}_v^\adj \mat{D}^\adj \Vand_{n-1}(\zdom)^\adj \Vand_{n-1}(\zdom) \mat{D} \mat{P}_v - \Id_r}_2 \\
        &\le \norm{\frac{1}{n} \mat{P}_v^{-1} \mat{D}^\adj \Vand_{n-1}(\zdom)^\adj \Vand_{n-1}(\zdom) \mat{D} \mat{P}_v - \Id_r}_2 \\
        &\qquad + \frac{1}{n} \norm{\mat{P}_v^\adj - \mat{P}_v^{-1}}_2 \norm{\Vand_{n-1}(\zdom)}_2^2\norm{\mat{P}_v}_2.
        \end{split}
    \end{equation}
    For the inequality, we use the triangle inequality and the submultiplicative property of the norm $\norm{\cdot}_2$.
    To bound $\norm{\mat{P}_v^\adj - \mat{P}_v^{-1}}_2$, employ the submultiplicative property again and use \cref{lem:vander-eigen}:
    \begin{equation*}
        \norm{\mat{P}_v^\adj - \mat{P}_v^{-1}}_2 \le \norm{\mat{P}_v^\adj\mat{P}_v - \Id_r}_2\norm{\mat{P}_v^{-1}}_2 = \order \left( \frac{1}{\zgap n} \right).
    \end{equation*}
    Using this bound, \cref{thm:moitra} to bound $\norm{\Vand_{n-1}(\zdom)}_2$, and \cref{lem:vander-eigen} to bound $\norm{\mat{P}_v}_2$, \cref{eq:near-orth-proof} simplifies to
    \begin{align*}
        \norm{\mat{Q}_\downarrow^\adj \mat{Q}_\downarrow - \Id_r}_2 &= \norm{\mat{P}_v^{-1}\mat{D}^\adj \left( \frac{1}{n} \Vand_{n-1}(\zdom)^\adj \Vand_{n-1}(\zdom) - \Id_r \right)\mat{D}\mat{P}_v }_2 + \order \left( \frac{1}{\zgap n} \right) \\
        &= \order\left( \norm{\frac{1}{n} \Vand_{n-1}(\zdom)^\adj \Vand_{n-1}(\zdom) - \Id_r}_2 + \frac{1}{\zgap n} \right)\,.
    \end{align*}
    By \cref{thm:moitra}, the eigenvalues of $\Vand_{n-1}(\zdom)^\adj \Vand_{n-1}(\zdom)$ within a range $n \pm \order(1/\zgap)$.
    Thus, the eigenvalues of $\Vand_{n-1}(\zdom)^\adj \Vand_{n-1}(\zdom)/n - \Id_r$ are within a range $\pm \order(\tfrac{1}{\zgap n})$, from which it follows that
    \begin{equation*}
        \norm{\mat{Q}_\downarrow^\adj \mat{Q}_\downarrow - \Id_r}_2 =\order\left( \norm{\frac{1}{n} \Vand_{n-1}(\zdom)^\adj \Vand_{n-1}(\zdom) - \Id_r}_2 + \frac{1}{\zgap n} \right) = \order \left( \frac{1}{\zgap n}\right).
    \end{equation*}
    This is the first conclusion \cref{eq:approximate-orthonormality}.

    The conclusion \cref{eq:approximate-orthonormality} entails that the singular values of $\mat{Q}_\downarrow$ satisfy the following bound
    \begin{equation*}
        \left| \sigma_j^2(\mat{Q}_\downarrow) - 1 \right| = \order \left( \frac{1}{\zgap n} \right) \for j = 1,\ldots,r.
    \end{equation*}
    Therefore, it follows that $\sigma_j(\mat{Q}_\downarrow) = \Theta(1)$ for each $j$, from which \cref{eq:W-bound-1} follows.
\end{proof}

\section{Deferred proofs for \texorpdfstring{\cref{sec:esprit-0.5}}{Section 2}} \label{app:esprit-0.5}

In this section, we prove \cref{lem:eigenvectors-weak,lem:W-to-frequency} from \cref{sec:esprit-0.5}.
We begin in \cref{sec:error-matrix-bounds} by proving some supplementary error bounds that we will use.
Then, we prove \cref{lem:eigenvectors-weak} in  \cref{sec:eigenvectors-weak} and \cref{lem:W-to-frequency} in \cref{sec:W-to-frequency}.

\subsection{Error matrix bounds} \label{sec:error-matrix-bounds}

The error matrix $\mat{E} = \mat{E}_{\rm tail} + \mat{E}_{\rm random}$ neatly decomposes into a sum of two terms.
In this section, we develop appropriate bounds for each of these terms separately.

For $\mat{E}_{\rm random}$, we employ the following estimate for the norm of the Toeplitz matrix of random errors:

\begin{lemma}[Norm of a random Toeplitz matrix] \label{lem:toeplitz}\label{lem:Mecks}
    Assume that the errors $\{E_j\}_{j=0}^{n-1}$ satisfy condition~\eqref{eqn:alpha_E_j}.
    Then, for a universal constant $C > 0$, 
    \begin{equation*}
        \prob \left\{ \norm{\mat{E}_{\rm random}}_2 \ge C\alpha\sqrt{n \log n} \right\} \le \frac{1}{n^2}.
    \end{equation*}
    If the random variables $E_j$ are Gaussian, we have the lower bound $\expect [\norm{\mat{E}_{\rm random}}_2] = \Omega(\alpha\sqrt{n \log n})$, so this result is sharp.
\end{lemma}

This result appears in the proof of \cite[Thm.~2 (Page 320)]{Meckes2007OnTS} for $E_j$ real.
Estimates of this form with explicit constants in the real or complex case can easily be obtained from the theory of matrix concentration (see, e.g., \cite[sec.~4.4]{Tro15} for a specific calculation).

Using this result, we can bound the norm of the error matrix in total:

\begin{proposition}[Error matrix, norm bound] \label{prop:error-matrix}
It holds that
\begin{align}\label{eqn:bound_tail}
    \|\mat{E}_{\rm tail}\|_2 \le n\mu_{\rm tail}.
\end{align}
    Assume the hypotheses of \cref{lem:toeplitz}.
    Then, with probability at least $1-1/n$, 
    \begin{equation*}
        \norm{\mat{E}}_2 = n \mu_{\rm tail} + \order(\alpha \sqrt{n \log n}).
    \end{equation*}
\end{proposition}

\begin{proof}
    Invoke the triangle inequality and the fact that $\norm{\vvec_n(z)\vvec_n(z)^\adj}_2 = n$ for any $z \in \unitcircle$:
    \begin{align*}
        \|\mat{E}_{\rm tail}\|_2\leq \sum_{i=r+1}^d \norm{\mu_i \vvec_n(z_i)\vvec_n(z_i)^\adj}_2 = \sum_{i=r+1}^d \mu_i \cdot n = n\mu_{\rm tail},
    \end{align*}
    where the last step follows from the definition of $\mu_{\rm tail}$ in \cref{eqn:condition_mutal}.

    Combining this with \cref{lem:toeplitz}, we get that
    \begin{equation*}
        \norm{\mat{E}}_2 \le \norm{\mat{E}_{\rm tail}} + \norm{\mat{E}_{\rm random}}_2 \le n \mu_{\rm tail} + \order(\alpha\sqrt{n\log n})\,. 
    \end{equation*}
    This is the desired conclusion.
\end{proof}

Finally, we need the following specialized bound on $\norm{\mat{E}_{\rm tail}\mat{Q}_r}_2$:

\begin{proposition}[Tail error in the dominant eigenspace] \label{prop:tail-error}
    We have the bound
    \begin{equation*}
        \norm{\mat{E}_{\rm tail}\mat{Q}_r}_2 = \order\left( \frac{\mu_{\rm tail}}{\zgap}\right).
    \end{equation*}
\end{proposition}

\begin{proof}
Write the Vandermonde decomposition of $\mat{E}_{\rm tail}$ in outer product form and use the triangle inequality, obtaining
\begin{align*}
    \norm{\mat{E}_{\rm tail}\mat{Q}_r}_2 =&~ \norm{\sum_{i=r+1}^n \vvec_n(z_i) \mu_i \vvec_n(z_i)^\adj \mat{Q}_r}_2\\
    \le &~\sum_{i=r+1}^d \mu_i \norm{\vvec_n(z_i)}_2 \norm{\vvec_n(z_i)^\adj \mat{Q}_r}_2 \\
    =&~ \sqrt{n}\mu_{\rm tail} \cdot  \max_{r+1\le i\le d}~ \norm{\vvec_n(z_i)^\adj \mat{Q}_r}_2\,,
\end{align*}
where the second step follows from $\|\vec{v}_n(z)\|_2=\sqrt{n}$ for any $z\in \unitcircle$.

By \cref{lem:vander-eigen}, $\mat{Q}_r = \Vand_n(\zdom) \mat{P}_v/\sqrt{n}$, where $\mat{P}_v$ satisfies $\norm{\mat{P}_v}_2 = \order(1)$.
Thus,
\begin{equation*}
    \max_{r+1\le i\le d}~ \norm{\vvec_n(z_i)^\adj \mat{Q}_r}_2 = \order\left( \frac{1}{\sqrt{n}}\max_{r+1\le i\le d}~ \norm{\vvec_n(z_i)^\adj \Vand_n(\zdom)}_2\right).
\end{equation*}
Consider any $i\in \{r+1,\dots,d\}$. Notice that 
\begin{align*}
    \vvec_n(z_i)^\adj \Vand_n(\zdom) = \left(\Vand_n(\vec{w})^\adj \Vand_n(\vec{w})-n\Id_{r+1}\right)(r+1,:r)\,,
\end{align*}
where $\vec{w} =  (\zdom, z_i)\in \C^{r+1}$. Moreover, by the definition of $\Delta_{\vec{z}}$, we know that $\min_{1\leq j\leq r}|z_i-z_j|\geq \Delta_{\vec{z}}$. Thus, invoking \cref{cor:moitra} and \cref{fac:submatrix_norm} gives
\begin{align*}
    \norm{\vvec_n(z_i)^\adj \Vand_n(\zdom)}_2 \leq n\cdot \norm{\frac{1}{n}\Vand_n(\vec{w})^\adj \Vand_n(\vec{w}) - \Id_{r+1} }_2 = {\cal O}\left(\frac{1}{\Delta_{\vec{z}}}\right)\,.
\end{align*}
Therefore,
\begin{equation*}
    \norm{\mat{E}_{\rm tail}\mat{Q}_r}_2\leq \sqrt{n}\mu_{\rm tail} \cdot {\cal O}\left(\frac{1}{\sqrt{n}\Delta_{\vec{z}}}\right) =   \order\left( \frac{\mu_{\rm tail}}{\zgap}\right)\,. 
\end{equation*}
The desired result has been established.
\end{proof}

\subsection{Proof of \texorpdfstring{\cref{lem:eigenvectors-weak}}{Lg}} \label{sec:eigenvectors-weak}

In this section, we provide a proof of \cref{lem:eigenvectors-weak}, broken into two steps.%

\paragraph{Step 1: Proof of \eqref{eq:Q-compare-weak}.}
    The result of \eqref{eq:Q-compare-weak} follows directly by combining matrix perturbation theorems (\cref{app:eigen-perturbation}), singular value bounds for Vandermonde matrices (\cref{sec:moitra}), and the error matrix bounds developed in \cref{sec:error-matrix-bounds}.
    
    By Weyl's theorem (\cref{thm:weyl}) and \cref{prop:error-matrix},
    \begin{equation*}
        \lambda_{r+1}\left(\mat{\hat{T}}\right) \le \lambda_{r+1}(\mat{T}) + \norm{\mat{E}}_2 \le n\mu_{\rm tail} + \order(\alpha\sqrt{n\log n})~~.
    \end{equation*}
    By Ostrowski's theorem (\cref{thm:ostrowski}) and \cref{thm:moitra},
    \begin{equation*}
        \lambda_r(\mat{T}) \ge \lambda_r(\diag(\mudom)) \sigma_{\rm min}(\Vand(\zdom))^2 \ge \mu_r \left(n - 1 - \frac{1}{\zgap}\right) = \Omega(\mu_r n).
    \end{equation*}
    In the final equality, we use the condition $n = \Omega(1/\zgap)$.
    Invoking the previous two displays and, the generalized $\sin \theta$ theorem (\cref{thm:sin}), \cref{prop:tail-error}, and the hypothesis $\mu_r = \Omega(\mutail)$, we obtain
    \begin{equation*}
        \sin \theta(\mat{Q}_r,\mat{\hat{Q}}_r) = \frac{\order(\mu_{\rm tail}/ \zgap +\alpha\sqrt{n\log n})}{\Omega(\mu_r n)} = \order \left( \frac{\mu_{\rm tail}}{\mu_r\zgap n}+ \frac{\alpha\sqrt{\log n}}{\mu_r \sqrt{n}} \right). 
    \end{equation*}
    By \cref{prop:distance-from-angles}, there exists a unitary matrix $\mat{U}_r$ such that
\begin{equation} \label{eq:W-bound-2}
    \norm{\mat{\hat{Q}_r} - \mat{Q}_r\mat{U}_r}_2 \le 2 \sin \theta(\mat{Q}_r,\mat{\hat{Q}}_r) = \order \left( \frac{\mu_{\rm tail}}{\mu_r\zgap n}+ \frac{\alpha\sqrt{\log n}}{\mu_r \sqrt{n}} \right).
\end{equation}

\paragraph{Step 2: Proof of \cref{eq:W-compare-weak}.}
Begin by invoking the triangle inequality and the submultiplicativity of $\norm{\cdot}_2$:
\begin{equation} \label{eq:W-bound-3}
    \begin{split}
    \norm{\mat{\hat{Q}}^\pinv \mat{\hat{Q}}_\downarrow - \mat{U}_r^\adj \mat{Q}^\pinv \mat{Q}_\downarrow\mat{U}_r}_2 &\le \norm{(\mat{\smash{\hat{Q}}}_\uparrow^\pinv - \mat{U}_r^\adj\mat{Q}_\uparrow^\pinv) \mat{\smash{\hat{Q}}}_\downarrow}_2 + \norm{\mat{U}_r^\adj\mat{Q}_\uparrow^\pinv (\mat{\hat{Q}}_\downarrow - \mat{Q}_\downarrow\mat{U}_r)}_2 \\
    &\le \norm{\mat{\smash{\hat{Q}}}_\uparrow^\pinv - \mat{U}_r^\adj\mat{Q}_\uparrow^\pinv}_2 + \norm{\mat{Q}_\uparrow^\pinv}_2 \norm{\mat{\hat{Q}}_r - \mat{Q}_r\mat{U}_r}_2
    \end{split}
\end{equation}
In the second inequality, we used the fact that adding a row to a matrix can only increase its norm, together with the fact that $\norm{\smash{\mat{\hat{Q}}}_r}_2 = 1$.

By \cref{thm:pseudoinverse-perturbation}, we have
\begin{equation*}
    \norm{\mat{\smash{\hat{Q}}}_\uparrow^\pinv - \mat{U}_r^\adj\mat{Q}_\uparrow^\pinv}_2 \le \frac{3\norm{\smash{\mat{Q}_\uparrow^\pinv}}_2^2\norm{\smash{\mat{\hat{Q}}_r - \mat{Q}_r\mat{U}_r}}_2}{1 - \norm{\smash{\mat{Q}_\uparrow^\pinv}}_2 \|\mat{\hat{Q}}_r - \mat{Q}_r\mat{U}_r\|_2} = \order \left( \frac{\mu_{\rm tail}}{\mu_r\zgap n}+ \frac{\alpha\sqrt{\log n}}{\mu_r \sqrt{n}} \right).
\end{equation*}
In the last bound, we use \cref{eq:W-bound-2,eq:W-bound-1}.
Substituting this display into \cref{eq:W-bound-3} and using the bounds \cref{eq:W-bound-1,eq:W-bound-2} again proves \cref{eq:W-compare-weak}.

\subsection{Proof of \texorpdfstring{\cref{lem:W-to-frequency}}{Lemma 2.2}} \label{sec:W-to-frequency}
Utilizing the matrix $\mat{P}_v$ defined in \cref{lem:vander-eigen} and the relations
\begin{equation*}
    \Vand_n(\zdom)_\uparrow = \Vand_{n-1}(\zdom), \quad \Vand_n(\zdom)_\downarrow = \Vand_{n-1}(\zdom) \diag(\zdom)\,,
\end{equation*}
we obtain the following diagonalization of $\mat{Q}_\uparrow^\pinv \mat{Q}_\downarrow$:
\begin{equation*}
    \mat{Q}_\uparrow^\pinv \mat{Q}_\downarrow = \mat{P}_v^{-1} \Vand_{n-1}(\zdom)^\pinv \Vand_{n-1}(\zdom) \diag(\zdom) \mat{P}_v = \mat{P}_v^{-1} \diag(\zdom)\mat{P}_v\,.
\end{equation*}
Observe that the eigenvalues of $\mat{Q}_\uparrow^\pinv \mat{Q}_\downarrow$ are $\zdom$, which satisfy $\zdom = \exp(\iu \arg(\zdom))$. We also note that $\zdom$ are also the eigenvalues of 
\begin{align*}
    \mat{P}^{-1}\mat{Q}^+_\uparrow\mat{Q}_\downarrow\mat{P}=(\mat{P}_v\mat{P})^{-1}\diag(\zdom)(\mat{P}_v\mat{P})\,.
\end{align*}

Now, we apply the Bauer--Fike theorem (\cref{thm:bauer-fike}).
Let $\vec{\hat{\lambda}}$ be the eigenvalues of $\mat{\hat{Q}}_\uparrow^\pinv \mat{\hat{Q}}_\downarrow$.
Then, by the optimal matching distance condition of the Bauer--Fike theorem (\cref{thm:bauer-fike}),
\begin{equation}\label{eq:matching_dis_proof}
\begin{aligned}
     \operatorname{md}(\vec{\hat{z}}_r, \zdom) &= \order\left( \operatorname{md}(\vec{\hat{\lambda}}, \zdom)\right) \\
    &\le \order \left( \norm{\mat{P}_v\mat{P}}_2\norm{\smash{\mat{P}^{-1}\mat{P}_v^{-1}}}_2 \norm{\mat{\hat{Q}}_\uparrow^\pinv \mat{\hat{Q}}_\downarrow - \mat{P}^{-1} \mat{Q}_\uparrow^\pinv \mat{Q}_\downarrow \mat{P}}_2 \right) \\
    &= \order \left(r\norm{\mat{\hat{Q}}_\uparrow^\pinv \mat{\hat{Q}}_\downarrow - \mat{P}^{-1} \mat{Q}_\uparrow^\pinv \mat{Q}_\downarrow \mat{P}}_2 \right)\,.
\end{aligned}
\end{equation}
In the second line, the estimates $\norm{\mat{P}_v}_2, \norm{\smash{\mat{P}_v^{-1}}}_2 = \order(1)$ from \cref{lem:vander-eigen}, the assumption $\norm{\mat{P}}_2,\norm{\smash{\mat{P}^{-1}}}_2 = \order(1)$, and the submultiplicative property of $\norm{\cdot}_2$.

Next, we prove the ``Furthermore'' part of \cref{lem:W-to-frequency}. According to \cref{lem:vander-eigen}, we set $c\in (0,1)$ such that $\norm{\mat{P}_v\mat{P}}_2\cdot\norm{(\mat{P}_v\mat{P})^{-1}}_2 \leq \frac{1}{c}$. Suppose 
\begin{align*}
    \norm{\mat{\hat{Q}}_\uparrow^\pinv \mat{\hat{Q}}_\downarrow - \mat{P}^{-1} \mat{Q}_\uparrow^\pinv \mat{Q}_\downarrow \mat{P}}_2<c\cdot \frac{\Delta_{\vec{z}}}{2}\,.
\end{align*}
For any $i\in [r]$, define ${\cal C}:=\{\xi~|~|\xi-z_i|=\Delta_{\vec{z}}/2\}$. Then, for any $\xi\in {\cal C}$, we have
\begin{align*}
    \norm{\left(\xi \Id_r - \mat{P}^{-1}\mat{Q}_\uparrow^+\mat{Q}_\downarrow \mat{P}\right)^{-1}}_2 = &~ \norm{(\mat{P}_v\mat{P})^{-1} \cdot (\diag(\zdom) - \xi\Id)^{-1}\cdot (\mat{P}_v\mat{P})}_2\\
    \leq &~ \norm{\mat{P}_v\mat{P}}_2\cdot \norm{(\mat{P}_v\mat{P})^{-1}}_2\cdot \norm {(\diag(\zdom) - \xi\Id)^{-1}}_2\\
    \leq &~ \frac{1}{c}\cdot \max_{j\in [r]} |z_j - \xi|^{-1}\leq \frac{2}{c\Delta_{\vec{z}}}\,,
\end{align*}
where the last step follows from $\zdom$ is $\Delta_{\vec{z}}$-separated and $\xi\in {\cal C}$. Thus, we know that the condition of \cref{lem:eig_location} is satisfied:
\begin{align*}
    \norm{\mat{\hat{Q}}_\uparrow^\pinv \mat{\hat{Q}}_\downarrow - \mat{P}^{-1} \mat{Q}_\uparrow^\pinv \mat{Q}_\downarrow \mat{P}}_2 \cdot \sup_{\xi\in {\cal C}}\norm{\left(\xi \Id_r - \mat{P}^{-1}\mat{Q}_\uparrow^+\mat{Q}_\downarrow \mat{P}\right)^{-1}}_2 < 1.
\end{align*}
By \cref{lem:eig_location}, there exists exactly one $j\in [r]$ such that $|\hat{\lambda}_j-z_i|< \Delta_{\vec{z}}/2$. Since the above argument works for any $i\in [r]$, we obtain that
\[
\max_{i\in[r]}\min_{j\in[r]}~\left|\hat{\lambda}_j-z_{i}\right|<\Delta_{\vec{z}}/2\,.
\]
Because $\zdom$ are $\Delta_{\vec{z}}$-separated, there exists a permutation $\pi$ such that 
\[
\max_{i\in[r]}~\left|\hat{\lambda}_i-z_{\pi(i)}\right|<\Delta_{\vec{z}}/2\,.
\]
Also, for any $i\in [r]$ and $j\ne i$, $\left|\hat{\lambda}_i-z_{\pi(j)}\right|>\Delta_{\vec{z}}/2$. 

Thus, by \cref{thm:bauer-fike}, we have
\begin{align*}
    \max_{1\leq i\leq r} |\hat{\lambda}_i - z_{\pi(i)}|=\max_{1\leq i\leq r} \min_{1\leq j\leq r} |\hat{\lambda}_i-z_j|\leq &~ {\cal O}\left(\norm{\mat{P}}_2 \norm{\smash{\mat{P}^{-1}}}_2\norm{\mat{\hat{Q}}_\uparrow^\pinv \mat{\hat{Q}}_\downarrow - \mat{P}^{-1} \mat{Q}_\uparrow^\pinv \mat{Q}_\downarrow \mat{P}}_2\right)\\
    = &~  {\cal O}\left(\norm{\mat{\hat{Q}}_\uparrow^\pinv \mat{\hat{Q}}_\downarrow - \mat{P}^{-1} \mat{Q}_\uparrow^\pinv \mat{Q}_\downarrow \mat{P}}_2\right)\,,
\end{align*}
which implies that
\begin{align*}
    {\rm md}(\hat{\vec{z}}_r,\zdom)={\cal O}\left(\norm{\mat{\hat{Q}}_\uparrow^\pinv \mat{\hat{Q}}_\downarrow - \mat{P}^{-1} \mat{Q}_\uparrow^\pinv \mat{Q}_\downarrow \mat{P}}_2\right)\,.
\end{align*}

The proof of the lemma is then completed.

\section{Deferred proofs for \texorpdfstring{\cref{sec:second-order}}{Section 5}} \label{app:second-order}

This section provides supplementary proofs for \cref{sec:second-order}.
In \cref{app:explicit-expansion}, we develop explicit formulas for the higher-order terms in the perturbation expansion result, \cref{prop:projector-expansion}.
In \cref{app:explicit-expansion-bounds}, we show how to bound the terms appearing in this expression.
We conclude by proving \cref{lem:expansion-simplified} in \cref{app:expansion-simplified}.

\subsection{Expansion of spectral projector: explicit formulas} \label{app:explicit-expansion}

In this section, we explicitly expand the perturbed projector $\mat{\Pi}_{\hat{\mat{Q}}_r}$ in terms of $\mat{\Pi}_{\mat{Q}_r}$, $\mat{T}^+$, $\mat{\Pi}_{\mat{Q}^\perp_r}$, and $\mat{E}$, where $\mat{T}^+$ is the pseudoinverse of $\mat{T}$. In particular, our goal is to identify the polynomial $\textbf{Poly}_{\hat{\mat{Q}}_r}$ so that 
\begin{equation}\label{eqn:polynomial_form}
\mat{\Pi}_{\hat{\mat{Q}}_r}=\mat{\Pi}_{\mat{Q}_r}+\textbf{Poly}_{\hat{\mat{Q}}_r}\left(\mat{T}^+,\mat{\Pi}_{\mat{Q}^\perp_r},\mat{E}\right)
\end{equation}
A straightforward approach to computing $\textbf{Poly}_{\hat{\mat{Q}}_r}$ involves exhaustively expanding each eigenvector projection within the contour integral (\cref{prop:projector-expansion}), then attempting to combine them to form $\mat{T}^+$. Interestingly, exploiting the inherent symmetry of the contour integral unexpectedly simplifies many complex expressions into a Schur polynomial, facilitating the combination of various projections.

To start, we first introduce the Schur polynomials:

\begin{definition}[Vandermonde determinant and Schur polynomials]\label{def:schur_poly}
    For numbers $a_1,\ldots,a_\ell\in\complex$, the \emph{Vandermonde determinant} is
    \begin{equation*}
        W(a_1,\ldots,a_\ell) \coloneqq \det \begin{bmatrix}
            a_1^{\ell-1} & a_2^{\ell-1} & \cdots & a_{\ell-1}^{\ell-1} & a_\ell^{\ell-1} \\
            a_1^{\ell-2} & a_2^{\ell-2} & \cdots & a_{\ell-1}^{\ell-2} & a_\ell^{\ell-2} \\
            \vdots & \vdots & \ddots & \vdots & \vdots \\
            a_1^1 & a_2^1 & \cdots & a_{\ell-1}^1 & a_\ell^1 \\
            1 & 1 & \cdots & 1 & 1
        \end{bmatrix} = \prod_{i < j} (a_i - a_j).
    \end{equation*}
    For non-increasingly ordered nonnegative integers $\vec{\gamma} = (\gamma_1,\ldots,\gamma_\ell)\in\mathbb{Z}_+$, the \emph{Schur polynomial} is
    \begin{equation*}
        s_{\vec{\gamma}}(a_1,\ldots,a_\ell) = \frac{1}{W(a_1,\ldots,a_\ell)} \det \begin{bmatrix}
            a_1^{\gamma_1+\ell-1} & a_2^{\gamma_1+\ell-1} & \cdots & a_{\ell-1}^{\gamma_1+\ell-1} & a_\ell^{\gamma_1+\ell-1} \\
            a_1^{\gamma_2+\ell-2} & a_2^{\gamma_2+\ell-2} & \cdots & a_{\gamma_2+\ell-2}^{\ell-1} & a_\ell^{\gamma_2+\ell-2} \\
            \vdots & \vdots & \ddots & \vdots & \vdots \\
            a_1^{\gamma_{\ell-1}+1} & a_2^{\gamma_{\ell-1}+1} & \cdots & a_{\ell-1}^{\gamma_{\ell-1}+1} & a_\ell^{\gamma_{\ell-1}+1} \\
            a_1^{\gamma_\ell} & a_2^{\gamma_\ell} & \cdots & a_{\ell-1}^{\gamma_\ell} & a_\ell^{\gamma_\ell}
        \end{bmatrix}.
    \end{equation*}
\end{definition}

After deriving the explicit expansion, we need to employ the following two properties of the Schur polynomial~\cite{MR3443860} to constrain the higher-order terms in \cref{app:explicit-expansion-bounds}:

\begin{proposition}[Properties of Schur polynomial] \label{prop:schur-polynomial}
    Let $\vec{\gamma} = (\gamma_1,\ldots,\gamma_\ell)\in\mathbb{Z}_+$ be a collection of non-increasingly ordered nonnegative integers.
    The Schur polynomial $s_{\vec{\gamma}}(x_1,\dots,x_\ell)$ satisfies the following properties:
    \begin{enumerate}
        \item \textbf{Polynomial.} $s_{\vec{\gamma}}$ is a symmetric, homogeneous polynomial of degree $\|\vec{\gamma}\|_1$. And the maximum degree of any variable $x_i$ is $\|\vec{\gamma}\|_\infty$.
        \item \textbf{Coefficients.} The coefficients of $s_{\vec{\gamma}}$ are nonnegative integers and sum to
        \begin{equation*}
            s_{\vec{\gamma}}(1,1,\ldots,1) = \prod_{i < j} \frac{\gamma_i - \gamma_j + j - i}{j-i}.
        \end{equation*}
        In particular, for any $m\in \mathbb{Z}_+$, 
        \begin{equation*}
            s_{(m,m,\ldots,m,0)}(x_1,\dots,x_\ell) = \sum_{\substack{\vec{\beta}\in \{0,1,\dots,m\}^\ell:\\\|\vec{\beta}\|_1=(\ell-1)m}}~ \prod_{i\in [\ell]} x_i^{\beta_i}\,,
        \end{equation*}
        and
        \begin{equation*}
            s_{(m,m,\ldots,m,0)}(1,1,\ldots,1) = \prod_{i=1}^{\ell-1} \frac{m + \ell - i}{\ell-i} = \binom{m + \ell - 1}{\ell - 1} \for m \in \mathbb{Z}_+.
        \end{equation*}
    \end{enumerate}
\end{proposition}

We also need to use the standard cofactor expansion for the determinant calculation.

\begin{fact}[Cofactor expansion of determinant]\label{fac:cofactor_expansion}
For any $n$-by-$n$ matrix $A$, the cofactor expansion of the determinant along the first row is
\begin{align*}
    \det(\mat{A}) = \sum_{i=1}^n (-1)^{i-1} \det(\mat{A}_{\backslash i})\cdot \mat{A}_{1,i}\,,
\end{align*}
where $\mat{A}_{\backslash i}$ is a sub-matrix of $\mat{A}$ by removing the first row and the $i$th column.
\end{fact}

To streamline notation, we define:
\begin{equation*}
    \begin{alignedat}{2}
        &\mat{\Pi}_0 \coloneqq \mat{\Pi}_{\mat{Q}_r^\perp}; \quad &&\mat{\Pi}_j \coloneqq  \mat{Q}(:,j)\mat{Q}(:,j)^\adj\\
        &\lambda_0\coloneqq 0; \quad &&\lambda_j\coloneqq \lambda_j(\mat{T})
    \end{alignedat}
    \quad \quad \for j=1,\dots,r\,.
\end{equation*}

Now, we proceed to derive the polynomial $\textbf{Poly}_{\hat{\mat{Q}}_r}$. Initially, we naively decompose the matrix contour integrals outlined in \cref{prop:projector-expansion} into summations, where each term is a multiplication of eigenvector projections. In particular, we have
\begin{equation}\label{eqn:naive_expansion}
\mat{\Pi}_{\hat{\mat{Q}}_r} = \mat{\Pi}_{\mat{Q}_r}  + \sum_{k=1}^\infty A_{i_1,i_2,\dots,i_{k+1}}\sum_{\vec{i} \in \{0,1\ldots,r\}^{k+1}} \mat{\Pi}_{i_1} \mat{E}\mat{\Pi}_{i_2}\cdots \mat{E}\mat{\Pi}_{i_{k+1}}\,.
\end{equation}

We observe that $A_{i_1,i_2,\dots,i_{k+1}}$ in each term can be expressed as scalar contour integrals, which can be explicitly written using Schur polynomials. The result is summarized as follows:
\begin{proposition}[Contour integral] \label{prop:contour-integral}
    For any $\ell\in \mathbb{N}$, $m\geq -1$, and $\vec{i} \in \{1,\ldots,r\}^\ell$,
    \begin{equation} \label{eq:A-func}
        A_m(\lambda_{i_1},\ldots,\lambda_{i_\ell}) \coloneqq\frac{1}{2\pi \iu} \oint_{\mathcal{C}_r} \frac{1}{(\zeta - \lambda_{i_1}) \cdots (\zeta - \lambda_{i_\ell}) \zeta^{m+1}}\, \d\zeta = (-1)^{\ell-1} \frac{s_{(m,m,\ldots,m,0)}(\lambda_{i_1},\ldots,\lambda_{i_{\ell}})}{\prod_{j=1}^{\ell} \lambda_{i_j}^{m+1}},
    \end{equation}
    where $\mathcal{C}_r$ denotes the contour in \cref{prop:projector-expansion}.
\end{proposition}

\begin{proof}
    The function $A_m$ is continuous on the interior of $\mathcal{C}_r$. 
    Consequently, we are free to assume without loss of generality that $\lambda_{i_1},\ldots,\lambda_{i_\ell}$ are distinct.
    The formula we derive will then extend to all $\lambda_{i_1},\ldots,\lambda_{i_\ell}$ in the interior of $\mathcal{C}_r$ by continuity.

    Using Cauchy's residue theorem (see e.g. \cite{stein2010complex}) to evaluate the contour integral in \cref{eq:A-func}, we obtain
    \begin{equation*}
        A_m(\lambda_{i_1},\ldots,\lambda_{i_\ell}) = \sum_{j=1}^\ell \frac{1}{\lambda_{i_j}^{m+1}}\prod_{p\ne j} \frac{1}{(\lambda_{i_j} - \lambda_{i_p})}
    \end{equation*}
    To consolidate this sum, we multiply by the numerator and denominator of the $j$-th summand by $(-1)^{j-1}W(\lambda_{i_1},\ldots,\lambda_{i_{j-1}},\lambda_{i_{j+1}},\ldots,\lambda_{i_\ell})\prod_{p\ne j} \lambda_{i_p}^{m+1}$, resulting in
    \begin{align*}
        A_m(\lambda_{i_1},\ldots,\lambda_{i_\ell}) = &~ 
        \sum_{j=1}^\ell\frac{ (-1)^{j-1} W(\lambda_{i_1},\ldots,\lambda_{i_{j-1}},\lambda_{i_{j+1}},\ldots,\lambda_{i_\ell})\prod_{p\ne j} \lambda_{i_p}^{m+1}}{(-1)^{j-1} W(\lambda_{i_1},\ldots,\lambda_{i_{j-1}},\lambda_{i_{j+1}},\ldots,\lambda_{i_\ell})\prod_{p\ne j} \lambda_{i_p}^{m+1}\cdot \lambda_{i_j}^{m+1} \prod_{p\ne j} (\lambda_{i_j} - \lambda_{i_p})}\\
        = &~ \frac{\sum_{j=1}^\ell (-1)^{j-1} W(\lambda_{i_1},\ldots,\lambda_{i_{j-1}},\lambda_{i_{j+1}},\ldots,\lambda_{i_\ell})\prod_{p\ne j} \lambda_{i_p}^{m+1}}{W(\lambda_{i_1},\ldots,\lambda_{i_\ell})\prod_{p=1}^\ell \lambda_{i_p}^{m+1}}\,,
    \end{align*}
    where the second step follows from
    \begin{align*}
        W(\lambda_{i_1},\ldots,\lambda_{i_{j-1}},\lambda_{i_{j+1}},\ldots,\lambda_{i_\ell}) \cdot \prod_{p\ne j} (\lambda_{i_j} - \lambda_{i_p}) = &~ \prod_{\substack{\alpha < \beta\\\alpha,\beta\ne j}}(\lambda_{i_\alpha} - \lambda_{i_\beta})\cdot \prod_{p\ne j} (\lambda_{i_j} - \lambda_{i_p})\\
        = &~ (-1)^{j-1} \prod_{\alpha < \beta}(\lambda_{i_\alpha} - \lambda_{i_\beta})\\
        = &~ (-1)^{j-1} W(\lambda_{i_1},\ldots,\lambda_{i_\ell})\,.
    \end{align*}
    We recognize the numerator of this expression of a cofactor expansion of a determinant (\cref{fac:cofactor_expansion}), obtaining
    \begin{equation} \label{eq:A-expression}
        A_m(\lambda_{i_1},\ldots,\lambda_{i_\ell}) = \frac{D(\lambda_{i_1},\ldots,\lambda_{i_\ell})}{W(\lambda_{i_1},\ldots,\lambda_{i_\ell})\prod_{p=1}^\ell \lambda_{i_p}^{m+1}}
    \end{equation}
    where
    \begin{equation*}
        D(\lambda_{i_1},\ldots,\lambda_{i_\ell}) \coloneqq  \det \begin{bmatrix}
            \prod_{p\ne 1} \lambda_{i_p}^{m+1} & \prod_{p\ne 2} \lambda_{i_p}^{m+1} & \cdots & \prod_{p\ne \ell-1} \lambda_{i_p}^{m+1} & \prod_{p\ne \ell} \lambda_{i_p}^{m+1} \\
            \lambda_{i_1}^{\ell-2} & \lambda_{i_2}^{\ell-2} & \cdots & \lambda_{i_{\ell-1}}^{\ell-2} & \lambda_{i_{\ell}}^{\ell-2} \\
            \vdots & \vdots & \ddots & \vdots & \vdots \\
            \lambda_{i_1}^1 & \lambda_{i_2}^1 & \cdots & \lambda_{i_{\ell-1}}^1 & \lambda_{i_\ell}^1 \\
            1 & 1 & \cdots & 1 & 1
        \end{bmatrix}\,.
    \end{equation*}
    By pulling out a factor of $\prod_{p=1}^\ell \lambda_{i_p}^{m+1}$ from the first row, we obtain
    \begin{equation*}
        D(\lambda_{i_1},\ldots,\lambda_{i_\ell}) = \prod_{p=1}^\ell \lambda_{i_p}^{m+1}\cdot \det \begin{bmatrix}
            \lambda_{i_1}^{-(m+1)} & \lambda_{i_2}^{-(m+1)} & \cdots & \lambda_{i_{\ell-1}}^{-(m+1)} & \lambda_{i_\ell}^{-(m+1)} \\
            \lambda_{i_1}^{\ell-2} & \lambda_{i_2}^{\ell-2} & \cdots & \lambda_{i_{\ell-1}}^{\ell-2} & \lambda_{i_{\ell}}^{\ell-2} \\
            \vdots & \vdots & \ddots & \vdots & \vdots \\
            \lambda_{i_1}^1 & \lambda_{i_2}^1 & \cdots & \lambda_{i_{\ell-1}}^1 & \lambda_{i_\ell}^1 \\
            1 & 1 & \cdots & 1 & 1
        \end{bmatrix}.
    \end{equation*}
    Now shuffle the first row to the bottom of the matrix and absorb a factor of $\lambda_{i_j}^{m+1}$ into each column $j$, obtaining
    \begin{align*}
        D(\lambda_{i_1},\ldots,\lambda_{i_\ell}) &= (-1)^{\ell-1} \det \begin{bmatrix}
            \lambda_{i_1}^{m+\ell-1} & \lambda_{i_2}^{m+\ell-1} & \cdots & \lambda_{i_{\ell-1}}^{m+\ell-1} & \lambda_{i_{\ell}}^{m+\ell-1} \\
            \lambda_{i_1}^{m+\ell-2} & \lambda_{i_2}^{m+\ell-2} & \cdots & \lambda_{i_{\ell-1}}^{m+\ell-2} & \lambda_{i_{\ell}}^{m+\ell-2} \\
            \vdots & \vdots & \ddots & \vdots & \vdots \\
            \lambda_{i_1}^{m+1} & \lambda_{i_2}^{m+1} & \cdots & \lambda_{i_{\ell-1}}^{m+1} & \lambda_{i_\ell}^{m+1} \\
            1 & 1 & \cdots & 1 & 1
        \end{bmatrix} \\
        &= (-1)^{\ell-1} s_{(m,m,\ldots,m,0)}(\lambda_{i_1},\ldots,\lambda_{i_\ell}) W(\lambda_{i_1},\ldots,\lambda_{i_\ell})\,,
    \end{align*}
    where the second step follows from the definition of the Schur polynomial (\cref{def:schur_poly}).
    Substituting this expression back in \cref{eq:A-expression} yields the stated result.
\end{proof}

An important observation stemming from \cref{prop:contour-integral} is that $A_m$ is a polynomial of $\lambda_i$ for $i\neq 0$. Thus, when inserting $A_m$ back into \eqref{eqn:naive_expansion}, we can see each $\lambda_i$ can be a variable $x$ that varies within the set $\{\lambda_j\}_{j=1}^r$ rather than as a fixed eigenvalue. Based on this intuition, we introduce the following definition.
\begin{definition}[Interleaving evaluation]\label{def:interleave_eval}
    For a function $f$ which is a polynomial in the inverses $x_1^{-1},\ldots,x_{k+1}^{-1}$ of $k+1$ variables, define the $\mat{T},\mat{E},\mat{\Pi}_{0}$-\emph{interleaving evaluation} for any monomial $x_1^{-\beta_1}x_2^{-\beta_2} \cdots x_{k+1}^{-\beta_{k+1}}$ as:
        \begin{align*}
        &\left( x_1^{-\beta_1} x_2^{-\beta_2} \cdots x_{k+1}^{-\beta_{k+1}} \right)[\mat{T},\mat{\Pi}_{0},\mat{E}]\coloneqq \\
        &\quad \left(\left(\mat{\Pi}_{0}\right)^{[\beta_1=0]}\pinvk{\mat{T}}{\beta_1}\right)\cdot \mat{E} \left(\left(\mat{\Pi}_{0}\right)^{[\beta_2=0]}\pinvk{\mat{T}}{\beta_2}\right)\cdots \mat{E}  \left(\left(\mat{\Pi}_{0}\right)^{[\beta_{k+1}=0]}\pinvk{\mat{T}}{\beta_{k+1}}\right)\,,
    \end{align*}
    where $[\beta_i=0]=1$ if $\beta_i=0$ and $[\beta_i=0]=0$ otherwise.
    And the $\mat{T},\mat{E},\mat{\Pi}_{0}$-\emph{interleaving evaluation} for the general polynomial $f$ is defined by linearity. %
\end{definition}

\begin{remark}[Example of interleaving evaluation]
For expository purposes, we give one example for above definition:
\begin{equation*}
    (x_1^{-1}x_2^0x_3^{-1} + x_1^0 x_2^0x_3^{-2})[\mat{T},\mat{\Pi}_{0},\mat{E}] = \mat{T}^{\pinv} \mat{E}\mat{\Pi}_{0}\mat{E}\mat{T}^\pinv + \mat{\Pi}_{0} \mat{E} \mat{\Pi}_{0} \mat{E} \pinvk{\mat{T}}{2}.
\end{equation*}
We see that each variable $x_i^{-\beta}$ raised to a negative power leads to a power of the pseudoinverse $\pinvk{\mat{T}}{\beta}$ in the expression.
Similarly, each variable $x_i^0$ raised to the zero power yields a copy of $\mat{\Pi}_{0}$.
Finally, these powers of $\mat{T}^\pinv$ and $\mat{\Pi}_{0p}$ are interleaved by copies of the matrix $\mat{E}$.    
\end{remark}

Finally, using the definition of the interleaving evaluation, \cref{prop:projector-expansion} and \eqref{eqn:naive_expansion} which expand $\mat{\Pi}_{\mat{\hat{Q}}_r}$ in terms of contour integrals, and \cref{prop:contour-integral} to evaluating those contour integrals, we obtain an explicit form for $\textbf{Poly}_{\hat{\mat{Q}}_r}$ in \eqref{eqn:polynomial_form}:

\begin{theorem}[Expansion of spectral projector, explicit form] \label{thm:expansion-explicit}
    It holds that
    \begin{equation*}
        \mat{\Pi}_{\mat{\hat{Q}}_r} = \mat{\Pi}_{\mat{Q}_r} + \sum_{k=1}^\infty \sum_{\mathcal{J} \subseteq \{1,\ldots,k+1\}} \mat{F}_{k,\mathcal{J}}
    \end{equation*}
    where
    \begin{equation} \label{eq:FkT}
        \mat{F}_{k,\mathcal{J}} \coloneqq \left( A_{k - |\mathcal{J}|}( x_i : i \in \mathcal{J}) \right)[\mat{T},\mat{\Pi}_{0},\mat{E}]\,,
    \end{equation}
    where $A_m$ is the function defined in \cref{eq:A-func}.
\end{theorem}

The above theorem can be intuitively derived by exploiting the polynomial structure of the coefficients in Equation \eqref{eqn:naive_expansion}. For example, consider any term $A_{i_1,i_2,0\dots,0}\mat{\Pi}_{i_1} \mat{E} \mat{\Pi}_{i_2}\mat{E}\mat{\Pi}_0\cdots \mat{\Pi}_0$ with $i_1\neq 0$ and $i_2\neq 0$. Since $A_{i_1,i_2,0\dots,0}$ is a polynomial of $\lambda^{-1}_{i_1}$ and $\lambda^{-1}_{i_2}$, as demonstrated in \cref{prop:contour-integral}, adding up $1\leq i_1,i_2\leq r$ yields $\mat{F}_{k,\{1,2\}}$.

\begin{proof}
    We begin by defining some notation.
    Fix an index $k\ge 1$ and let
    \begin{equation*}
        \mat{C}_k \coloneqq  \frac{1}{2\pi \iu} \oint_{\mathcal{C}_r} (\zeta - \mat{T})^{-1} \left( \mat{E} (\zeta - \mat{T})^{-1} \right)^k \, \d\zeta\,.
    \end{equation*}
    Our goal in this proof will be to show that $\mat{C}_k = \sum_{\mathcal{J} \subseteq \{1,\ldots,k+1\}} \mat{F}_{k,\mathcal{J}}$.
    By the spectral theorem, the resolvent $(\zeta - \mat{T})^{-1}$ has the spectral decomposition
    \begin{equation} \label{eq:spectral-expansion-appendix}
        (\zeta - \mat{T})^{-1} = \sum_{j=0}^r (\zeta - \lambda_j)^{-1} \mat{\Pi}_j.
    \end{equation}
    
    Using the spectral expansion \cref{eq:spectral-expansion-appendix} in the definition of $\mat{C}_k$, we obtain
    \begin{equation*}
         \mat{C}_k = \sum_{\vec{i} \in \{0,1\ldots,r\}^{k+1}} \mat{\Pi}_{i_1} \cdot \mat{E}\mat{\Pi}_{i_2}\cdots \mat{E}\mat{\Pi}_{i_{k+1}} \cdot \frac{1}{2\pi \iu} \oint_{\mathcal{C}_r} \frac{1}{(\zeta - \lambda_{i_1}) \cdots (\zeta - \lambda_{i_{k+1}})}\, \d\zeta.
    \end{equation*}
    To evaluate this sum, re-index the sum by the support $\mathcal{J}$ of the index vector $\vec{i}$, obtaining
    \begin{equation*}
         \mat{C}_k = \sum_{\mathcal{J} \subseteq \{1,\ldots,k+1\}} \sum_{\substack{i_j = 0 \text{ for } j \notin \mathcal{J} \\ i_j \in \{1,\ldots,r\} \text{ for } j \in \mathcal{J}}}  \mat{\Pi}_{i_1} \cdot \mat{E}\mat{\Pi}_{i_2}\cdots \mat{E}\mat{\Pi}_{i_{k+1}} \cdot \frac{1}{2\pi \iu} \oint_{\mathcal{C}_r} \frac{1}{\zeta^{(k-|\mathcal{J}|)+1}\prod_{j \in \mathcal{J}}(\zeta - \lambda_{i_j})}\, \d\zeta.
    \end{equation*}
    By \cref{prop:contour-integral}, the contour integral in this expression can be evaluated as:
    \begin{align*}
        \frac{1}{2\pi \iu} \oint_{\mathcal{C}_r} \frac{1}{\zeta^{(k-|\mathcal{J}|)+1}\prod_{j \in \mathcal{J}}(\zeta - \lambda_{i_j})}\, \d\zeta = A_{k-|\mathcal{J}|}(\lambda_{i_j}:j\in \mathcal{J})\,,
    \end{align*}
    where $A_{k-|\mathcal{J}|}$ is defined by \cref{eq:A-func}. We may recognize $A_{k-|\mathcal{J}|}(x_j:j\in {\cal J})$ as a $(k+1)$-variate polynomial of $x_1^{-1},\dots,x_{k+1}^{-1}$ (by multiplying $\prod_{j\notin {\cal J}}x_j^0$). Then, we have 
    \begin{align*}
        A_{k-|\mathcal{J}|}(\lambda_{i_j}:j\in \mathcal{J}) = \sum_{\vec{\beta}\in \mathbb{N}^{k+1}} A_{\vec{\beta}} \cdot \prod_{j=1}^{k+1}\lambda_{i_j}^{-\beta_j}\,,
    \end{align*}
    where $A_{\vec{\beta}}\in \mathbb{R}$ is the coefficient. 
    Notice that 
    \begin{align*}
        &\sum_{\substack{i_j = 0 \text{ for } j \notin \mathcal{J} \\ i_j \in \{1,\ldots,r\} \text{ for } j \in \mathcal{J}}}  \mat{\Pi}_{i_1} \cdot \mat{E}\mat{\Pi}_{i_2}\cdots \mat{E}\mat{\Pi}_{i_{k+1}} \cdot \prod_{j=1}^{k+1}\lambda_{i_j}^{-\beta_j}\\
        = &~ \left(\sum_{i_1} \frac{\mat{\Pi}_{i_1}}{\lambda_{i_1}^{\beta_1}}\right)\cdot \mat{E}\left(\sum_{i_2} \frac{\mat{\Pi}_{i_2}}{\lambda_{i_2}^{\beta_2}}\right) \cdots \mat{E}\left(\sum_{i_{k+1}} \frac{\mat{\Pi}_{i_{k+1}}}{\lambda_{i_{k+1}}^{\beta_{k+1}}}\right)\\
        = &~ \left(\left(\mat{\Pi}_{0}\right)^{[\beta_1=0]}\pinvk{\mat{T}}{\beta_1}\right)\cdot \mat{E} \left(\left(\mat{\Pi}_{0}\right)^{[\beta_2=0]}\pinvk{\mat{T}}{\beta_2}\right) \cdots \mat{E} \left(\left(\mat{\Pi}_{0}\right)^{[\beta_{k+1}=0]}\pinvk{\mat{T}}{\beta_{k+1}}\right)\\
        = &~ \left(\prod_{j=1}^{k+1}x_j^{-\beta_j}\right)[\mat{T},\mat{\Pi}_{0},\mat{E}]\,,
    \end{align*}
    where the summation of $i_j$ in the second step depends on whether $j\in {\cal J}$, the third step follows from $\sum_{j=1}^r \lambda_j^{-t}\mat{\Pi}_j = (\mat{T}^+)^t$ for any $t>0$, and the last step follows from \cref{def:interleave_eval}.
    By linearity, it implies that
    \begin{align*}
        \sum_{\substack{i_j = 0 \text{ for } j \notin \mathcal{J} \\ i_j \in \{1,\ldots,r\} \text{ for } j \in \mathcal{J}}}  \mat{\Pi}_{i_1} \cdot \mat{E}\mat{\Pi}_{i_2}\cdots \mat{E}\mat{\Pi}_{i_{k+1}} \cdot A_{k-|\mathcal{J}|}(\lambda_{i_j} : j \in \mathcal{J})=\left( A_{k - |\mathcal{J}|}( x_i : i \in \mathcal{J}) \right)[\mat{T},\mat{\Pi}_{0},\mat{E}]\,.
    \end{align*}

    Thus, we obtain that
    \begin{equation*}
         \mat{C}_k = \sum_{\mathcal{J} \subseteq \{1,\ldots,k+1\}} \left( A_{k - |\mathcal{J}|}( x_i : i \in \mathcal{J}) \right)[\mat{T},\mat{\Pi}_{0},\mat{E}] = \sum_{\mathcal{J} \subseteq \{1,\ldots,k+1\}} \mat{F}_{k,\mathcal{J}}\,.
    \end{equation*}
    The theorem then follows by \cref{prop:projector-expansion}.
\end{proof}

\subsection{Expansion of spectral projector: bounds on terms} \label{app:explicit-expansion-bounds}

In this section, we bound the terms appearing in the expansion of the spectral projector, \cref{thm:expansion-explicit}.
Our result is as follows:

\begin{lemma}[Terms in expansion of spectral projector] \label{lem:terms}
    Fix a number $k\ge 2$ and a subset $\mathcal{J} \subseteq \{ 1,\ldots, k+1 \}$.
    Then there exists a constant $0 < C < 3/8$ such that
    \begin{enumerate}[label=(\alph*)]
        \item If $\mathcal{J} = \emptyset$, then $\mat{F}_{k,\mathcal{J}} = \mat{0}$.
        \item If $|\mathcal{J}| \ge 2$, then %
        \begin{equation*}
            \norm{\mat{F}_{k,\mathcal{J}}}_2= \order \left( \frac{1}{\zgap n} + \frac{\alpha\sqrt{\log n}}{\mu_r \sqrt{n}} \right)^{|\mathcal{J}|} C^{k-|\mathcal{J}|}.
        \end{equation*}
        \item If $\mathcal{J} = \{j\}$ for $2 \le i \le k$, then 
        \begin{equation*}
            \norm{\mat{F}_{k,\mathcal{J}}}_2= \order \left( \frac{1}{\zgap^2 n^2} + \frac{\alpha^2\log n}{\mu_r^2 n} \right) C^{k-2}.
        \end{equation*}
        \item For $\mathcal{J} = \{1\}$ and $\mathcal{J} = \{k+1\}$, we have the representations
        \begin{align*}
            \mat{F}_{k,\{1\}} = \pinvk{\mat{T}}{k} \mat{E}_{\rm random} \mat{\Pi}_{\mat{Q}_r^\perp} (\mat{E}\mat{\Pi}_{\mat{Q}_r^\perp})^{k-1} + \order \left( \frac{1}{\zgap n} \right)C^{k-1} \norm{\mat{G}_1}_2 \\
            \mat{F}_{k,\{k+1\}} =   (\mat{\Pi}_{\mat{Q}_r^\perp} \mat{E})^{k-1}\mat{\Pi}_{\mat{Q}_r^\perp}\mat{E}_{\rm random}\pinvk{\mat{T}}{k}  + \order \left( \frac{1}{\zgap n} \right)C^{k-1} \norm{\mat{G}_2}_2
        \end{align*}
        where $\norm{\mat{G}_1}_2,\norm{\mat{G}_2}_2= 1$.
        For the leading-order terms, we have the bounds
        \begin{equation} \label{eq:leading-order-bounds}
            \begin{split}
            \norm{\pinvk{\mat{T}}{k} \mat{E}_{\rm random} \mat{\Pi}_{\mat{Q}_r^\perp} (\mat{E}\mat{\Pi}_{\mat{Q}_r^\perp})^{k-1}}_2 &= \order\left(\frac{\alpha\sqrt{\log n}}{\sqrt{n} \mu_r}\right) C^{k-1} \\
            \norm{(\mat{\Pi}_{\mat{Q}_r^\perp} \mat{E})^{k-1}\mat{\Pi}_{\mat{Q}_r^\perp}\mat{E}_{\rm random}\pinvk{\mat{T}}{k}}_2 &= \order\left(\frac{\alpha\sqrt{\log n}}{\sqrt{n} \mu_r}\right) C^{k-1}.
            \end{split}
        \end{equation}
    \end{enumerate}
\end{lemma}

\begin{proof}
    For ease of reading, we break the proof into steps.

    \paragraph{Proof of (a).} 
    The conclusion $\mat{F}_{k,\emptyset} = \mat{0}$ follows directly from the formula \cref{eq:FkT} and $A_{k-\emptyset}=0$.

    \paragraph{Expression of $\mat{F}_{k,\mathcal{J}}$ as a sum of monomials.} 
    For the remaining parts of the result, we shall decompose $\mat{F}_{k,\mathcal{J}}$ as a sum of monomials.
    By \cref{prop:schur-polynomial} and the definition of $\mat{F}_{k,\mathcal{J}}$, $\mat{F}_{k,\mathcal{J}}$ is equal to $(-1)^{|\mathcal{J}|-1}$ times a polynomial with degree $-k$. More specifically,
    \begin{equation}\label{F_{k,j}}
        \mat{F}_{k,\mathcal{J}}=(-1)^{|\mathcal{J}|-1}\sum_{\substack{\vec{\beta}\in \mathbb{N}^{k+1},\sum^{k+1}_{i=1}\beta_i=k\\\left\{j|\beta_j\neq 0\right\}\subseteq\mathcal{J}}} A_{\vec{\beta}} \cdot \left( x_1^{-\beta_1} \cdots x_{k+1}^{-\beta_{k+1}} \right)[\mat{T},\mat{\Pi}_{\mat{Q}_r^\perp},\mat{E}],
    \end{equation}
    where the coefficients $A_{\vec{\beta}}$ are nonnegative integers. In addition, according to the second point of \cref{prop:schur-polynomial}, we have
    \begin{equation}\label{eqn:A_beta}
        \sum_{\substack{\vec{\beta}\in \mathbb{N}^{k+1},\sum^{k+1}_{i=1}\beta_i=k\\\left\{j|\beta_j\neq 0\right\}\subseteq\mathcal{J}}} A_{\vec{\beta}}=\binom{k-1}{|\mathcal{J}|-1}
    \end{equation}
    
    In view of the identities $\pinvk{\mat{T}}{\alpha} = \mat{\Pi}_{\mat{Q}_r}\pinvk{\mat{T}}{\alpha}\mat{\Pi}_{\mat{Q}_r}$ and $\mat{\Pi}_{\mat{Q}_r^\perp} = \mat{\Pi}_{\mat{Q}_r^\perp} \cdot \mat{\Pi}_{\mat{Q}_r^\perp}$, we can write
    \begin{equation*}
        \left( x_1^{-\beta_1} \cdots x_{k+1}^{-\beta_{k+1}} \right)[\mat{T},\mat{\Pi}_{\mat{Q}_r^\perp},\mat{E}] = \pinvk{\mat{T}}{\beta_1} \prod_{j=1}^k \left[(\mat{\Pi}({\cal J})_j \mat{E} \mat{\Pi}({\cal J})_{j+1}) \pinvk{\mat{T}}{\beta_{j+1}}\right]
    \end{equation*}
    where $\mat{\Pi}({\cal J})_j = \mat{\Pi}_{\mat{Q}_r}$ if $j \in \mathcal{J}$ and $\mat{\Pi}({\cal J})_j = \mat{\Pi}_{\mat{Q}_r^\perp}$ otherwise.
    Here, we use the convention $\pinvk{\mat{T}}{0} = \Id_n$.%

    \paragraph{Bounding each monomial.}
    By the submultiplicative property of $\norm{\cdot}_2$, we have a norm bound on each of these monomials.
    \begin{equation} \label{eq:monomial-bound-1}
        \begin{split}
        \norm{\left( x_1^{-\beta_1} \cdots x_{k+1}^{-\beta_{k+1}} \right)[\mat{T},\mat{\Pi}_{\mat{Q}_r^\perp},\mat{E}]}_2 \le&~ \norm{\mat{T}^\pinv}^k \prod_{j=1}^k \norm{\mat{\Pi}({\cal J})_j \mat{E}\mat{\Pi}({\cal J})_{j+1}}_2 \\
        =&~ \order\left(\frac{1}{\mu_r n}\right)^k \prod_{j=1}^k \norm{\mat{\Pi}({\cal J})_j \mat{E}\mat{\Pi}({\cal J})_{j+1}}_2\,,
        \end{split}
    \end{equation}
    where the second step follows from $\norm{\mat{T}^\pinv}_2 = \order(\frac{1}{\mu_r n})$ by \cref{cor:moitra_T}.
    
    To bound the factors $\norm{\mat{\Pi}({\cal J})_j \mat{E}\mat{\Pi}({\cal J})_{j+1}}_2$, we consider two cases. When $j\in {\cal J}$ or $j+1\in {\cal J}$, this term is either $\|\mat{\Pi}_{\mat{Q}^\perp} \mat{E}\mat{\Pi}_{\mat{Q}_r}\|_2$ or $\|\mat{\Pi}_{\mat{Q}_r} \mat{E}\mat{\Pi}_{\mat{Q}_r}\|_2$, which can be upper-bounded by $\|\mat{E}\mat{\Pi}_{\mat{Q}_r}\|_2$. Then, we have
    \begin{align*}
        \|\mat{E}\mat{\Pi}_{\mat{Q}_r}\|_2 \leq \|\mat{E}_{\rm tail}\mat{\Pi}_{\mat{Q}_r}\|_2 + \|\mat{E}_{\rm random}\mat{\Pi}_{\mat{Q}_r}\|_2 \leq {\cal O}\left(\frac{\mu_{\rm tail}}{\Delta_{\vec{z}}}\right) + \|\mat{E}_{\rm tail}\|_2={\cal O}\left(\frac{\mu_{\rm tail}}{\Delta_{\vec{z}}}+\alpha \sqrt{n\log n}\right)\,,
    \end{align*}
    where the second step follows from \cref{prop:tail-error} and the third step follows from \cref{lem:toeplitz}. Otherwise, this term becomes $\|\mat{\Pi}_{\mat{Q}^\perp}\mat{E}\mat{\Pi}_{\mat{Q}^\perp}\|_2$. By \cref{prop:error-matrix},
    \begin{align*}
        \|\mat{\Pi}_{\mat{Q}_r^\perp}\mat{E}\mat{\Pi}_{\mat{Q}_r^\perp}\|_2 \leq \|\mat{E}\|_2=n\mu_{\rm tail} +{\cal O}\left( \alpha \sqrt{n\log n}\right)\,.
    \end{align*}
    Therefore, we obtain that
    \begin{equation} \label{eq:monomial-bound-2}
        \norm{\mat{\Pi}({\cal J})_j \mat{E}\mat{\Pi}({\cal J})_{j+1}}_2 \le \begin{cases}
            \order \left( \frac{\mu_{\rm tail}}{\zgap} + \alpha\sqrt{n \log n} \right) & \text{if } j \in\mathcal{J} \text{ or } j+1\in\mathcal{J}, \\ 
            n \mu_{\rm tail}+\order \left(\alpha\sqrt{n \log n} \right) & \text{for any value of $j$}.
        \end{cases}.
    \end{equation}

    Given any $\mathcal{J}$, we notice that at least $|\mathcal{J}|$ of the terms $\mat{\Pi}({\cal J})_j \mat{E}\mat{\Pi}({\cal J})_{j+1}$ satisfy $j\in\mathcal{J}$ or $j+1\in\mathcal{J}$.
    Thus, substituting the bounds \cref{eq:monomial-bound-2} into \cref{eq:monomial-bound-1}, we obtain
    \begin{equation}\label{eqn:each_bound}
        \norm{\left( x_1^{-\beta_1} \cdots x_{k+1}^{-\beta_{k+1}} \right)[\mat{T},\mat{\Pi}_{\mat{Q}_r^\perp},\mat{E}]}_2 \le \order \left( \frac{1}{\zgap n} + \frac{\alpha\sqrt{\log n}}{\mu_r \sqrt{n}} \right)^{|\mathcal{J}|} a^{k-|\mathcal{J}|},
    \end{equation}
    where $$0 < a \leq \frac{\mu_{\rm tail}}{\mu_r} + {\cal O}\left(\frac{\alpha\sqrt{\log n}}{\mu_r\sqrt{n}}\right)<\frac{3}{16}$$ is an absolute constant.
    Define $C\coloneqq 2a$. Now, we are ready to show (b)--(d):
    
    \paragraph{Proof of (b).}

    Using \eqref{eqn:each_bound} and \eqref{eqn:A_beta} to bound every term in \eqref{F_{k,j}}, we conclude that
    \begin{equation*}
        \norm{\mat{F}_{k,\mathcal{J}}}_2 \le \binom{k -1}{|\mathcal{J}| - 1} \order \left( \frac{1}{\zgap n} + \frac{\alpha\sqrt{\log n}}{\mu_r \sqrt{n}} \right)^{|\mathcal{J}|} a^{k-|\mathcal{J}|}\leq \order \left( \frac{1}{\zgap n} + \frac{\alpha\sqrt{\log n}}{\mu_r \sqrt{n}} \right)^{|\mathcal{J}|} C^{k-|\mathcal{J}|}\,,
    \end{equation*}
    where we use the trivial bound $\binom{m}{p} \le 2^m$ in the second inequality. This proves (b).

    \paragraph{Proof of (c).} If $\mathcal{J} = \{j\}$ for some $2\le j \le k$, we know that $A_{k-1}(x_j)=x_j^{-k}$. Then, \cref{eq:monomial-bound-2} implies
    \begin{equation*}
        \norm{\mat{\Pi}({\cal J})_{j-1} \mat{E}\mat{\Pi}({\cal J})_j}_2, ~~\norm{\mat{\Pi}({\cal J})_{j} \mat{E}\mat{\Pi}({\cal J})_{j+1}}_2 \le \order \left( \frac{\mu_{\rm tail}}{\zgap} + \alpha\sqrt{n \log n} \right).
    \end{equation*}
    Similar to \eqref{eqn:each_bound}, using the second bound in \cref{eq:monomial-bound-2} for the other $k-2$ factors in \cref{eq:monomial-bound-1}, we obtain
    \begin{align*}
        \norm{\left( x_1^{0} \cdots x_{j-1}^0x_j^{-k}x_{j+1}^0\cdots x_{k+1}^{0} \right)[\mat{T},\mat{\Pi}_{\mat{Q}_r^\perp},\mat{E}]}_2 \le&~  \order \left( \frac{1}{\zgap n} + \frac{\alpha\sqrt{\log n}}{\mu_r \sqrt{n}} \right)^2 a^{k-2}\\
        \le&~ \order \left( \frac{1}{\zgap n} + \frac{\alpha\sqrt{\log n}}{\mu_r \sqrt{n}} \right)^2 C^{k-2}
    \end{align*}
    Since $\mat{F}_{k,\mathcal{J}}$ is a single monomial of the form we just bounded, so part (c) is proven.

    \paragraph{Proof of (d).}
    If $\mathcal{J} = \{1\}$ or $\mathcal{J} = \{k+1\}$, then the projector $\mat{\Pi}({\cal J})_j$ is on the end, so only the first or the last factor $\|\mat{T}^+\|\norm{\mat{\Pi}({\cal J})_{j} \mat{E}\mat{\Pi}({\cal J})_{j+1}}_2$ is $\tilde{\order}(1/\sqrt{n})$ and the other factors are $\order(1)$. Thus, this term gives us $\mathcal{O}(1/\sqrt{n})$. We need to obtain its specific formula.

    We treat the case $\mathcal{J} = \{1\}$, with the case $\mathcal{J} = \{k+1\}$ following similarly.
    Using our formula for $\mat{F}_{k,\{1\}}$ and the decomposition $\mat{E} = \mat{E}_{\rm random} + \mat{E}_{\rm tail}$, we have that
    \begin{equation*}
        \mat{F}_{k,\{1\}} = \pinvk{\mat{T}}{k} \mat{E} \mat{\Pi}_{\mat{Q}_r^\perp} (\mat{E}\mat{\Pi}_{\mat{Q}_r^\perp})^{k-1} = \pinvk{\mat{T}}{k} \mat{E}_{\rm random} \mat{\Pi}_{\mat{Q}_r^\perp} (\mat{E}\mat{\Pi}_{\mat{Q}_r^\perp})^{k-1} + \pinvk{\mat{T}}{k} \mat{E}_{\rm tail} \mat{\Pi}_{\mat{Q}_r^\perp} (\mat{E}\mat{\Pi}_{\mat{Q}_r^\perp})^{k-1}\,.
    \end{equation*}

    For the first term, we have
    \begin{align*}
    \norm{\pinvk{\mat{T}}{k} \mat{E}_{\rm random} \mat{\Pi}_{\mat{Q}_r^\perp} (\mat{E}\mat{\Pi}_{\mat{Q}_r^\perp})^{k-1}}_2 \leq &~ \norm{\mat{T}^+}^k_2 \cdot \norm{\mat{E}_{\rm random}}_2 \cdot \norm{\mat{\Pi}_{\mat{Q}_r^\perp} (\mat{E}\mat{\Pi}_{\mat{Q}_r^\perp})^{k-1}}_2\\
    \leq &~ {\cal O}\left(\frac{1}{\mu_r n}\right)^k\cdot {\cal O}\left(\alpha \sqrt{n\log n}\right)\cdot \left(n\mu_{\rm tail} + {\cal O}\left(\alpha \sqrt{n\log n}\right)\right)^{k-1}\\
    \leq &~ \order\left(\frac{\alpha\sqrt{\log n}}{\sqrt{n} \mu_r}\right) C^{k-1}, 
    \end{align*}
    where the second step follows from \cref{cor:moitra_T,lem:Mecks,eq:monomial-bound-2}.

    For the second term, we have
    \begin{align*}
    \norm{\pinvk{\mat{T}}{k} \mat{E}_{\rm tail} \mat{\Pi}_{\mat{Q}_r^\perp} (\mat{E}\mat{\Pi}_{\mat{Q}_r^\perp})^{k-1}}_2=&~\norm{\pinvk{\mat{T}}{k} \mat{\Pi}_{\mat{Q}_r}\mat{E}_{\rm tail} \mat{\Pi}_{\mat{Q}_r^\perp} (\mat{E}\mat{\Pi}_{\mat{Q}_r^\perp})^{k-1}}_2\\
    \le &~  \norm{\pinvk{\mat{T}}{k}}_2 \cdot \norm{\mat{\Pi}_{\mat{Q}_r}\mat{E}_{\rm tail}}_2\cdot \norm{\mat{\Pi}_{\mat{Q}_r^\perp} (\mat{E}\mat{\Pi}_{\mat{Q}_r^\perp})^{k-1}}_2\\
    \leq &~ {\cal O}\left(\frac{1}{\mu_r n}\right)^k\cdot {\cal O}\left(\frac{\mu_{\rm tail}}{\Delta_{\vec{z}}}\right)\cdot \left(n\mu_{\rm tail} + {\cal O}\left(\alpha \sqrt{n\log n}\right)\right)^{k-1}\\
    \le &~\order \left( \frac{1}{\zgap n} \right)C^{k-1}\,,    
    \end{align*}
    where we use \eqref{prop:tail-error} for $\mat{\Pi}_{\mat{Q}_r}\mat{E}_{\rm tail}$ in the third step.
    This concludes the proof of (d).
\end{proof}

\subsection{Proof of \texorpdfstring{\cref{lem:expansion-simplified}}{Lemma 4.3}} \label{app:expansion-simplified}
\expansionsimplified*

\begin{proof}
    Begin with the explicit expansion of the spectral projector (\cref{thm:expansion-explicit}):
    \begin{equation*}
        \mat{\Pi}_{\mat{\hat{Q}}_r} = \mat{\Pi}_{\mat{Q}_r} + \sum_{k=1}^\infty \sum_{\mathcal{J} \subseteq \{1,\ldots,k+1\}} \mat{F}_{k,\mathcal{J}}.
    \end{equation*}
    The summands $\mat{F}_{k,\mathcal{J}}$ are defined in \cref{eq:FkT}.
    Now use \cref{lem:terms} to obtain
    \begin{equation} \label{eq:expansion}
    \begin{aligned}
        \mat{\Pi}_{\mat{\hat{Q}}_r} = &~\mat{\Pi}_{\mat{Q}_r} + \sum_{k=1}^\infty \left( \pinvk{\mat{T}}{k} \mat{E}_{\rm random} \mat{\Pi}_{\mat{Q}_r^\perp} (\mat{E}\mat{\Pi}_{\mat{Q}_r^\perp})^{k-1} + (\mat{\Pi}_{\mat{Q}_r^\perp} \mat{E})^{k-1}\mat{\Pi}_{\mat{Q}_r^\perp}\mat{E}_{\rm random}\pinvk{\mat{T}}{k} \right) \\
        ~&+ \sum_{k=1}^\infty \Bigg[ \order \left( \frac{1}{\zgap n} \right)C^{k-1} + k\cdot\order \left( \frac{1}{\zgap^2 n^2} + \frac{\alpha^2\log n}{\mu_r^2 n} \right) C^{k-2} + D_k \Bigg] \mat{G}\,.
    \end{aligned}
    \end{equation}
    with $\norm{\mat{G}}_2 \le 1$ and 
    \begin{equation*}
        D_k \coloneqq  \sum_{\substack{\mathcal{J} \subseteq \{1,\ldots,k+1\} \\ |\mathcal{J}|\ge 2}}  \order \left( \frac{1}{\zgap n} + \frac{\alpha\sqrt{\log n}}{\mu_r \sqrt{n}} \right)^{|\mathcal{J}|} C^{k-|\mathcal{J}|}.
    \end{equation*}

    To bound $D_k$, begin by rewriting this as a sum over the size $t$ of the subset $\mathcal{J}$, of which there are $\binom{k+1}{t}$ such subsets.
    Therefore,
    \begin{equation*}
       D_k = \sum_{t=2}^{k+1} \binom{k+1}{t} \cdot \order \left( \frac{1}{\zgap n} + \frac{\alpha\sqrt{\log n}}{\mu_r \sqrt{n}} \right)^{t} C^{k-t}\,.
    \end{equation*}
    Now invoke the trivial bound $\binom{m}{p} \le 2^m$ to obtain
    \begin{equation*}
        D_k = \sum_{t=2}^{k+1} \order \left( \frac{1}{\zgap n} + \frac{\alpha \sqrt{\log n}}{\mu_r \sqrt{n}} \right)^{t} (2C)^{k-t} = \order \left( \frac{1}{\zgap^2 n^2} + \frac{\alpha^2\log n}{\mu_r^2 n} \right) \cdot (2C)^k\,.
    \end{equation*}
    
    Substituting in \cref{eq:expansion} and summing over $k$, the last term becomes
    \begin{align*}
        &\sum_{k=1}^\infty \Bigg[ \order \left( \frac{1}{\zgap n} \right)C^{k-1} + k\cdot\order \left( \frac{1}{\zgap^2 n^2} + \frac{\alpha^2\log n}{\mu_r^2 n} \right) C^{k-2} + D_k \Bigg] \mat{G}\\
        = &~ \sum_{k=1}^\infty \Bigg[ \order \left( \frac{1}{\zgap n} \right)C^{k-1} + k\cdot\order \left( \frac{1}{\zgap^2 n^2} + \frac{\alpha^2\log n}{\mu_r^2 n} \right) C^{k-2} + \order \left( \frac{1}{\zgap^2 n^2} + \frac{\alpha^2\log n}{\mu_r^2 n} \right) \cdot (2C)^k \Bigg] \mat{G}\\
        = &~ {\cal O}\left(\frac{1}{\Delta_{\vec{z}}n}+\frac{\alpha^2 \log n}{\mu_r^2 n}\right)\mat{G}\,,
    \end{align*}
    where the last step follows from $C\in (0, 3/8)$.

    The lemma is then proved.
\end{proof}

\section{Deferred proofs for \texorpdfstring{\cref{sec:strong_sketch}}{Section 5}}\label{sec:align_proof_defer}

\subsection{Proof for Step 1: Construction of the ``good'' \texorpdfstring{$\mat{P}$}{P}}\label{sec:align_step1_defer}

\goodPconstruction*

\begin{proof}

By the eigenvector perturbation lemma (\cref{lem:second-order-formal}), we obtain that there exists a unitary matrix $\mat{U}_r\in\mathbb{C}^{r\times r}$ such that%
\begin{equation}\label{eqn:eigenvector_expansion}
\begin{aligned}
\hat{\mat{Q}}_r\mat{U}_r=&~\mat{Q}_r+
\underbrace{\sum^\infty_{k=0}\mat{\Pi}_{\mat{Q}^\perp_r}\left(\mat{E}_1\mat{\Pi}_{\mat{Q}^\perp_r}\right)^{k}\mat{E}_2\mat{Q}_r\left(\mat{\Sigma}^{-1}_r\right)^{k+1}+ \order \left(\frac{\alpha\sqrt{\log n}}{\mu_r \sqrt{\zgap n}} \right)\mat{\Pi}_{\mat{Q}_r}\tilde{\mat{Q}}_1}_{\text{first order term}}+\underbrace{\order \left( \frac{\alpha^2\log n}{\mu_r^2 \zgap n}\right) \tilde{\mat{Q}}_{2}}_{\text{second order term}}\\
=&\underbrace{\left[\mat{Q}_r+{\cal O}\left(\frac{\alpha\sqrt{\log n}}{\mu_r\sqrt{\Delta_z n}}\right)\mat{\Pi}_{\mat{Q}_r}\tilde{\mat{Q}}_1+{\cal O}\left(\frac{\alpha^2\log n}{\mu^{2}_r\Delta_{\vec{z}} n}\right)\mat{\Pi}_{\mat{Q}_r}\tilde{\mat{Q}}_{2}\right]}_{\text{in the column space of $\mat{Q}_r$}}\\
&+\left[\sum^\infty_{k=0}\mat{\Pi}_{\mat{Q}^\perp_r}\left(\mat{E}_1\mat{\Pi}_{\mat{Q}^\perp_r}\right)^{k}\mat{E}_2\mat{Q}_r\left(\mat{\Sigma}_r^{-1}\right)^{k+1}+{\cal O}\left(\frac{\alpha^2\log n}{\mu^{2}_r\Delta_{\vec{z}} n}\right)\mat{\Pi}_{\mat{Q}^\perp_r}\tilde{\mat{Q}}_{2}\right]
\end{aligned}
\end{equation}
where $\mat{\Pi}_{\mat{Q}_r}$ is the projection onto the column space of $\mat{Q}_r$ and $\mat{\Pi}_{\mat{Q}^\perp_r}$ is the projection onto the column space that is orthogonal to the column space of $\mat{Q}_r$, and $\tilde{\mat{Q}}_{1}, \tilde{\mat{Q}}_{2}\in\mathbb{C}^{n\times r}$ satisfying $\left\|\tilde{\mat{Q}}_{1}\right\|_2=\mathcal{O}(1)$ and $\left\|\tilde{\mat{Q}}_{2}\right\|_2=\mathcal{O}(1)$.

By \cref{clm:first_order_simplify}, the first-order error term that is orthogonal to the column space of $\mat{Q}_r$ can be simplified as follows:
\begin{align}\label{eqn:first_order_simplification}
\sum^\infty_{k=0}\mat{\Pi}_{\mat{Q}^\perp_r}\left(\mat{E}_1\mat{\Pi}_{\mat{Q}^\perp_r}\right)^{k}\mat{E}_2\mat{Q}_r\left(\mat{\Sigma}_r^{-1}\right)^{k+1}=&~\sum^\infty_{k=0}\left(\mat{I}_n-\frac{1}{n}\mat{V}_n(\zdom)\mat{V}_n(\zdom)^\dagger\right)\mat{E}_1^k\mat{E}_2\mat{Q}_r\left(\mat{\Sigma}_r^{-1}\right)^{k+1}\notag\\
+&~{\cal O}\left(\frac{\alpha\sqrt{\log n}}{\mu_r\Delta_{\vec{z}}n^{1.5}}\right)\,.
\end{align}
It gives that%
\begin{equation}
\begin{aligned}
    \hat{\mat{Q}}_r\mat{U}_r=&~ \mat{Q}_r+{\cal O}\left(\frac{\alpha\sqrt{\log n}}{\mu_r\sqrt{\Delta_z n}}\right)\mat{\Pi}_{\mat{Q}_r}\tilde{\mat{Q}}_1+{\cal O}\left(\frac{\alpha^2\log n}{\mu^{2}_r\Delta_{\vec{z}} n}\right)\mat{\Pi}_{\mat{Q}_r}\tilde{\mat{Q}}_{2}\\
    +&~ \sum^\infty_{k=0}\left(\mat{I}_n-\frac{1}{n}\mat{V}_n(\zdom)\mat{V}_n(\zdom)^\dagger\right)\mat{E}_1^k\mat{E}_2\mat{Q}_r\left(\mat{\Sigma}_r^{-1}\right)^{k+1}\\
    +&~ {\cal O}\left(\frac{\alpha^2\log n}{\mu^{2}_r\Delta_{\vec{z}} n}\right)\mat{\Pi}_{\mat{Q}^\perp_r}\tilde{\mat{Q}}_{2}+{\cal O}\left(\frac{\alpha\sqrt{\log n}}{\mu_r\Delta_{\vec{z}}n^{1.5}}\right)\,.
\end{aligned}
\end{equation}

Because the image of $\mat{Q}_r+{\cal O}\left(\frac{\alpha\sqrt{\log n}}{\mu_r\sqrt{\Delta_z n}}\right)\mat{\Pi}_{\mat{Q}_r}\tilde{\mat{Q}}_1+{\cal O}\left(\frac{\alpha^2\log n}{\mu^{2}_r\Delta_{\vec{z}} n}\right)\mat{\Pi}_{\mat{Q}_r}\tilde{\mat{Q}}_{2}$ equals to the image of $\mat{Q}_r$, there exists an invertible matrix $\mat{P}\in\mathbb{C}^{r\times r}$ such that
\begin{align*}%
\mat{Q}_r+{\cal O}\left(\frac{\alpha\sqrt{\log n}}{\mu_r\sqrt{\Delta_z n}}\right)\mat{\Pi}_{\mat{Q}_r}\tilde{\mat{Q}}_1+{\cal O}\left(\frac{\alpha^2\log n}{\mu^{2}_r\Delta_{\vec{z}} n}\right)\mat{\Pi}_{\mat{Q}_r}\tilde{\mat{Q}}_{2}=\mat{Q}_r\mat{P}\,.
\end{align*}
Multiplying $\mat{Q}^\dagger_r$ on both sides, we have
\begin{align*}
\mat{P}:=&~\mat{Q}^\dagger_r\left(\mat{Q}_r+{\cal O}\left(\frac{\alpha\sqrt{\log n}}{\mu_r\sqrt{\Delta_z n}}\right)\mat{\Pi}_{\mat{Q}_r}\tilde{\mat{Q}}_1+{\cal O}\left(\frac{\alpha^2\log n}{\mu^{2}_r\Delta_{\vec{z}} n}\right)\mat{\Pi}_{\mat{Q}_r}\tilde{\mat{Q}}_{2}\right)\\
=&~\mat{I}_r+{\cal O}\left(\frac{\alpha\sqrt{\log n}}{\mu_r\sqrt{\Delta_z n}}\right)\mat{Q}^\dagger_r\mat{\Pi}_{\mat{Q}_r}\tilde{\mat{Q}}_1+{\cal O}\left(\frac{\alpha^2\log n}{\mu^{2}_r\Delta_{\vec{z}} n}\right)\mat{Q}^\dagger_r\mat{\Pi}_{\mat{Q}_r}\tilde{\mat{Q}}_{2}\,,
\end{align*}
which gives the expression of $\mat{P}$ \eqref{eq:def_good_P} in the lemma.

Property I and II follow from \cref{clm:good_P_property}. 
The proof of the lemma is then completed.
\end{proof}

\subsubsection{Technical claims}

\begin{claim}[Simplifying the first-order error term]\label{clm:first_order_simplify}
The first-order error term in the orthogonal space of $\mat{Q}_r$ in \eqref{eqn:eigenvector_expansion} can be expressed as:
\begin{align*}
\sum^\infty_{k=0}\mat{\Pi}_{\mat{Q}^\perp_r}\left(\mat{E}_1\mat{\Pi}_{\mat{Q}^\perp_r}\right)^{k}\mat{E}_2\mat{Q}_r\left(\mat{\Sigma}_r^{-1}\right)^{k+1}=&~\sum^\infty_{k=0}\left(\mat{I}_n-\frac{1}{n}\mat{V}_n(\zdom)\mat{V}_n(\zdom)^\dagger\right)\mat{E}_1^k\mat{E}_2\mat{Q}_r\left(\mat{\Sigma}_r^{-1}\right)^{k+1}\\
&~ +{\cal O}\left(\frac{\alpha\sqrt{\log n}}{\mu_r\Delta_{\vec{z}}n^{1.5}}\right)\,.
\end{align*}    
\end{claim}
\begin{proof}
Note that
\begin{equation}\label{eq:proj_Q_perp_approx}
\begin{aligned}
    \mat{\Pi}_{\mat{Q}^\perp_r} = &~ \mat{I}_n - \mat{Q}_r\mat{Q}_r^\dagger \\
    = &~ \mat{I}_n-\frac{1}{n}\mat{V}_n(\zdom)\mat{P}_v\mat{P}_v^\dagger\mat{V}_n(\zdom)^\dagger\\
    = &~ \mat{I}_n-\frac{1}{n}\mat{V}_n(\zdom)\mat{V}_n(\zdom)^\dagger  + \frac{1}{n}\mat{V}_n(\zdom)\left(\mat{P}_v\mat{P}_v^\dagger - I_r\right)\mat{V}_n(\zdom)^\dagger\\
    = &~ \mat{I}_n-\frac{1}{n}\mat{V}_n(\zdom)\mat{V}_n(\zdom)^\dagger  + \left\|\mat{P}_v\mat{P}_v^\dagger - I_r\right\|_2 \cdot \left\|\frac{1}{n}\mat{V}_n(\zdom)\mat{V}_n(\zdom)^\dagger\right\|_2\\
    = &~ \mat{I}_n-\frac{1}{n}\mat{V}_n(\zdom)\mat{V}_n(\zdom)^\dagger  + \mathcal{O}\left(\frac{1}{\Delta_{\vec{z}} n}\right) \cdot \left\|\frac{1}{n}\mat{V}_n(\zdom)\mat{V}_n(\zdom)^\dagger\right\|_2\\
    = &~ \mat{I}_n-\frac{1}{n}\mat{V}_n(\zdom)\mat{V}_n(\zdom)^\dagger  + \mathcal{O}\left(\frac{1}{\Delta_{\vec{z}} n}\right)\,,
\end{aligned}
\end{equation}
where the first step follows from $\mat{Q}_r\mat{Q}_r^\dagger = \mat{\Pi}_{\mat{Q}_r}$, the second step follows from \eqref{eqn:reshape_Q}, the third and the fourth steps are straightforward, the fifth step follows from \cref{lem:vander-eigen}, and the last step follows from \cref{cor:V_dagger_V_dom} that $\left\|\frac{1}{n}\mat{V}_n(\zdom)\mat{V}_n(\zdom)^\dagger\right\|_2={\cal O}\left(1+\frac{1}{\Delta_{\vec{z}} n}\right)= {\cal O}(1)$. 
Then, by triangle inequality, we have
\begin{align*}
    &\left\|\mat{\Pi}_{\mat{Q}^\perp_r}\left(\mat{E}_1\mat{\Pi}_{\mat{Q}^\perp_r}\right)^{k}\mat{E}_2\mat{Q}_r\left(\mat{\Sigma}_r^{-1}\right)^{k+1} - \left(\mat{I}_n-\frac{1}{n}\mat{V}_n(\zdom)\mat{V}_n(\zdom)^\dagger\right)\mat{E}_1^k\mat{E}_2\mat{Q}_r\left(\mat{\Sigma}_r^{-1}\right)^{k+1}\right\|_2\\
    = &~ \left\|\mat{\Pi}_{\mat{Q}^\perp_r}\left(\mat{E}_1\mat{\Pi}_{\mat{Q}^\perp_r}\right)^{k}\mat{E}_2\mat{Q}_r\left(\mat{\Sigma}_r^{-1}\right)^{k+1} - \mat{\Pi}_{\mat{Q}^\perp_r}\mat{E}_1^k\mat{E}_2\mat{Q}_r\left(\mat{\Sigma}_r^{-1}\right)^{k+1}\right\|_2\\
    &~+ \left\|\left(\mat{\Pi}_{\mat{Q}^\perp_r}-\Big(\mat{I}_n-\frac{1}{n}\mat{V}_n(\zdom)\mat{V}_n(\zdom)^\dagger\Big)\right)\cdot \mat{E}_1^k\mat{E}_2\mat{Q}_r\left(\mat{\Sigma}_r^{-1}\right)^{k+1}\right\|_2\,.
\end{align*}

For the first term, we have
\begin{align*}
&\left\|\mat{\Pi}_{\mat{Q}^\perp_r}\left(\mat{E}_1\mat{\Pi}_{\mat{Q}^\perp_r}\right)^{k}\mat{E}_2\mat{Q}_r\left(\mat{\Sigma}_r^{-1}\right)^{k+1}-\mat{\Pi}_{\mat{Q}^\perp_r}\mat{E}_1^k\mat{E}_2\mat{Q}_r\left(\mat{\Sigma}_r^{-1}\right)^{k+1}\right\|_2\\
\leq&~\left\|\left(\left(\mat{E}_1\mat{\Pi}_{\mat{Q}^\perp_r}\right)^{k}-\mat{E}_1^k\right)\cdot \mat{E}_2\mat{Q}_r\left(\mat{\Sigma}_r^{-1}\right)^{k+1}\right\|_2\\
\leq &~ \left\|\left(\mat{E}_1\mat{\Pi}_{\mat{Q}^\perp_r}\right)^{k}-\mat{E}_1^k\right\|_2\cdot {\cal O}\left(\frac{\alpha\sqrt{n\log n}}{(n\mu_r)^{k+1}}\right)\\
=&~\left\|\left(\mat{E}_1\mat{\Pi}_{\mat{Q}^\perp_r}\right)^{k}-\left(\mat{E}_1\left(\mat{\Pi}_{\mat{Q}^\perp_r}+\mat{\Pi}_{\mat{Q}_r}\right)\right)^k\right\|_2\cdot {\cal O}\left(\frac{\alpha\sqrt{n\log n}}{(n\mu_r)^{k+1}}\right)\\
\leq &~ \sum_{l=1}^k \binom{k}{l}\cdot \left\|\mat{E}_1\mat{\Pi}_{\mat{Q}_r}\right\|_2^l\cdot \left\|\mat{E}_1\mat{\Pi}_{\mat{Q}^\perp_r}\right\|_2^{k-l}\cdot {\cal O}\left(\frac{\alpha\sqrt{n\log n}}{(n\mu_r)^{k+1}}\right)\\
\leq &~ \sum_{l=1}^k \binom{k}{l}\cdot {\cal O}\left(\frac{\mu_{\rm tail}}{\Delta_{\vec{z}}}\right)^l\cdot {\cal O}(n\mu_{\rm tail})^{k-l}\cdot {\cal O}\left(\frac{\alpha\sqrt{n\log n}}{(n\mu_r)^{k+1}}\right)\\
= &~ \sum_{l=0}^{k-1}\binom{k-1}{l}\cdot {\cal O}\left(\frac{\mu_{\rm tail}}{\Delta_{\vec{z}}}\right)^l\cdot {\cal O}(n\mu_{\rm tail})^{k-1-l}\cdot k
\cdot \frac{\mu_{\rm tail}\alpha\sqrt{n\log n}}{\Delta_{\vec{z}}(n\mu_r)^{k+1}}\\
= &~ {\cal O}\left(\frac{\mu_{\rm tail}}{\Delta_{\vec{z}}} + n\mu_{\rm tail}\right)^{k-1}\cdot k
\cdot \frac{\mu_{\rm tail}\alpha\sqrt{n\log n}}{\Delta_{\vec{z}}(n\mu_r)^{k+1}}\\
= &~ {\cal O}\left(\frac{\mu_{\rm tail}}{\mu_r\Delta_{\vec{z}}n} + \frac{\mu_{\rm tail}}{\mu_r}\right)^{k-1}k \cdot \frac{\mu_{\rm tail}\alpha\sqrt{\log n}}{\mu_r^2\Delta_{\vec{z}}n^{1.5}}\,,
\end{align*}
where the first step follows from $\|\mat{\Pi}_{\mat{Q}_r^\perp}\|_2=1$, the second step follows from $\|\mat{E}_2\|_2={\cal O}(\alpha\sqrt{n\log n})$ and $\|\mat{\Sigma}_r^{-1}\|_2=(n\mu_r)^{-1}$ by \eqref{eqn:eigenvalue_prop} in \cref{cor:moitra_T}, the third step follows from $\mat{\Pi}_{\mat{Q}_r}+\mat{\Pi}_{\mat{Q}^\perp_r}=\mat{I}_r$, the fourth step follows from the binomial expansion and triangle inequality, the fifth step follows from $\|\mat{E}_1\mat{\Pi}_{\mat{Q}_r^\perp}\|_2\leq \|\mat{E}_1\|_2={\cal O}(n\mu_{\rm tail})$ by \eqref{eqn:bound_tail}
  and $\|\mat{E}_1\mat{\Pi}_{\mat{Q}_r}\|_2\leq \|\mat{E}_1\mat{Q}_r\|_2={\cal O}\left(\frac{\mu_{\rm tail}}{\Delta_{\vec{z}}}\right)$ by \cref{prop:tail-error},
  the sixth step follows from changing the summation index, and the last two steps are direct calculations.

For the second term, we have
\begin{align*}
&\left\|\left(\mat{\Pi}_{\mat{Q}^\perp_r}-\Big(\mat{I}_n-\frac{1}{n}\mat{V}_n(\zdom)\mat{V}_n(\zdom)^\dagger\Big)\right)\cdot \mat{E}_1^k\mat{E}_2\mat{Q}_r\left(\mat{\Sigma}_r^{-1}\right)^{k+1}\right\|_2\\
= &~
 {\cal O}\left(\frac{1}{\Delta_{\vec{z}}n}\right)\cdot \left\| \mat{E}_1^k\mat{E}_2\mat{Q}_r\left(\mat{\Sigma}_r^{-1}\right)^{k+1}\right\|_2\\
= &~  {\cal O}\left(\frac{1}{\Delta_{\vec{z}}n}\right) \cdot {\cal O}(n\mu_{\rm tail})^k \cdot {\cal O}(\alpha\sqrt{n\log n})\cdot (n\mu_r)^{-k-1}\\
= &~ {\cal O}\left(\frac{\mu_{\rm tail}}{\mu_r}\right)^k\cdot \frac{\alpha\sqrt{\log n}}{\mu_r\Delta_{\vec{z}}n^{1.5}}\,,
\end{align*}
where the first step follows from \eqref{eq:proj_Q_perp_approx}, the second step follows from $\left\|\mat{E}_{1}\right\|_2=\mathcal{O}\left(n\mu_{\rm tail}\right)$, $\left\|\mat{E}_{2}\right\|_2=\mathcal{O}\left(\alpha\sqrt{n\log(n)}\right)$, and $\|\mat{\Sigma}_r^{-1}\|_2=1/(n\mu_r)$ by  \eqref{eqn:eigenvalue_prop} in \cref{cor:moitra_T}. %

Combining them together, we have
\begin{align*}
    &\left\|\mat{\Pi}_{\mat{Q}^\perp_r}\left(\mat{E}_1\mat{\Pi}_{\mat{Q}^\perp_r}\right)^{k}\mat{E}_2\mat{Q}_r\left(\mat{\Sigma}_r^{-1}\right)^{k+1} - \left(\mat{I}_n-\frac{1}{n}\mat{V}_n(\zdom)\mat{V}_n(\zdom)^\dagger\right)\mat{E}_1^k\mat{E}_2\mat{Q}_r\left(\mat{\Sigma}_r^{-1}\right)^{k+1}\right\|_2\\
    \leq &~ {\cal O}\left(\frac{\mu_{\rm tail}}{\mu_r\Delta_{\vec{z}}n} + \frac{\mu_{\rm tail}}{\mu_r}\right)^{k-1}k \cdot \frac{\mu_{\rm tail}\alpha\sqrt{\log n}}{\mu_r^2\Delta_{\vec{z}}n^{1.5}} + {\cal O}\left(\frac{\mu_{\rm tail}}{\mu_r}\right)^k\cdot \frac{\alpha\sqrt{\log n}}{\mu_r\Delta_{\vec{z}}n^{1.5}}\\
    = &~ {\cal O}\left(\frac{\mu_{\rm tail}}{\mu_r}\right)^{k}k \cdot \frac{\alpha\sqrt{\log n}}{\mu_r\Delta_{\vec{z}}n^{1.5}}\,,
\end{align*}
where the second step follows from $n\geq 1/(\mu_r \Delta_{\vec{z}})$.
Thus, summing over $k$ gives that
\begin{align*}
&\left\|\sum^\infty_{k=0}\mat{\Pi}_{\mat{Q}^\perp_r}\left(\mat{E}_1\mat{\Pi}_{\mat{Q}^\perp_r}\right)^{k}\mat{E}_2\mat{Q}_r\left(\mat{\Sigma}_r^{-1}\right)^{k+1}-\sum^\infty_{k=0}\left(\mat{I}_n-\frac{1}{n}\mat{V}_n(\zdom)\mat{V}_n(\zdom)^\dagger\right)\mat{E}_1^k\mat{E}_2\mat{Q}_r\left(\mat{\Sigma}_r^{-1}\right)^{k+1}\right\|_2\notag\\
\leq &~\sum_{k=1}^\infty {\cal O}\left(\frac{\mu_{\rm tail}}{\mu_r}\right)^{k}k \cdot \frac{\alpha\sqrt{\log n}}{\mu_r\Delta_{\vec{z}}n^{1.5}}\notag\\
= &~ {\cal O}\left(\frac{\alpha\sqrt{\log n}}{\mu_r\Delta_{\vec{z}}n^{1.5}}\right)\,.
\end{align*}
where the last step follows from ${\cal O}(\mu_{\rm tail}/\mu_r)<1$.

The claim is then proved.
\end{proof}

\begin{claim}[Properties of the ``good'' $\mat{P}$]\label{clm:good_P_property}
The invertible matrix $\mat{P}$ defined by \eqref{eq:def_good_P} satisfies the following properties:
\begin{itemize}
    \item {\it Property I: } 
    \begin{equation}%
\left\|\mat{P}^{-1}-\mat{I}_r\right\|_2={\cal O}\left(\frac{\alpha\sqrt{\log n}}{\mu_r\sqrt{\Delta_z n}}\right)\,,\quad \left\|\mat{P}-\mat{I}_r\right\|_2={\cal O}\left(\frac{\alpha\sqrt{\log n}}{\mu_r\sqrt{\Delta_z n}}\right)\,.
\end{equation}
    \item {\it Property II:}
    \begin{equation}%
\begin{aligned}
\left\|\mat{\hat{Q}}_r\mat{U}_r\mat{P}^{-1}-\mat{Q}_r\right\|_2\leq \mathcal{O}\left(\frac{\alpha\sqrt{\log n}}{\mu_r\sqrt{\zgap n}}\right)\,.
\end{aligned}
\end{equation}%
\end{itemize}
\end{claim}
\begin{proof}
We prove each of the properties below.

\textbf{Proof of Property I:}

By the definition of $\mat{P}$ (\eqref{eq:def_good_P}), we have
\begin{align*}
    \left\|\mat{P}-\mat{I}_r\right\|_2 = &~ {\cal O}\left(\frac{\alpha\sqrt{\log n}}{\mu_r\sqrt{\Delta_z n}}\right)\cdot \left\|\mat{Q}^\dagger_r\mat{\Pi}_{\mat{Q}_r}\tilde{\mat{Q}}_1\right\|_2 + {\cal O}\left(\frac{\alpha^2\log n}{\mu_r^2\Delta_z n}\right)\cdot \left\|\mat{Q}^\dagger_r\mat{\Pi}_{\mat{Q}_r}\tilde{\mat{Q}}_{2}\right\|_2={\cal O}\left(\frac{\alpha\sqrt{\log n}}{\mu_r\sqrt{\Delta_z n}}\right)\,,
\end{align*}
where the second step follows from $\|\mat{Q}_r\|_2=1$, $\|\tilde{\mat{Q}}_1\|_2={\cal O}(1)$, and $\|\tilde{\mat{Q}}_2\|_2={\cal O}(1)$ by their definitions. It also implies that
\begin{align*}
    \left\|\mat{P}^{-1}-\mat{I}_r\right\|_2={\cal O}\left(\frac{\alpha\sqrt{\log n}}{\mu_r\sqrt{\Delta_z n}}\right)\,.
\end{align*}

\textbf{Proof of Property II:}

Consider
\[
\hat{\mat{Q}}_r\mat{U}_r\mat{P}^{-1}=\mat{Q}_r+\widetilde{\mat{E}}^Q_1+\widetilde{\mat{E}}^Q_2+\mathcal{O}\left(\frac{\alpha\sqrt{\log n}}{\mu_r\Delta_{\vec{z}}n^{1.5}}\right)\,,
\]
where
\begin{subequations}%
\begin{equation}%
\widetilde{\mat{E}}^Q_1:=\sum^\infty_{k=0}\left(\mat{I}_n-\frac{1}{n}\mat{V}_n(\zdom)\mat{V}_n(\zdom)^\dagger\right)\mat{E}_1^k\mat{E}_2\mat{Q}_r\left(\mat{\Sigma}_r^{-1}\right)^{k+1}\mat{P}^{-1}\,.
\end{equation}
and
\begin{equation}
\widetilde{\mat{E}}^Q_2:={\cal O}\left(\frac{\alpha^2\log n}{\mu^{2}_r\Delta_{\vec{z}} n}\right)\mat{\Pi}_{\mat{Q}_r^\perp}\hat{\mat{Q}}_{2}\mat{P}^{-1}\,.
\end{equation}
\end{subequations}

We first bound the norm of $\widetilde{\mat{E}}^Q_1$:%
\begin{equation}\label{eq:bound_E1_Q_wt}
\begin{aligned}
    \big\|\widetilde{\mat{E}}^Q_1\big\|_2 \leq &~ \left\|\mat{I}_n-\frac{1}{n}\mat{V}_n(\zdom)\mat{V}_n(\zdom)^\dagger\right\|_2\cdot \sum^\infty_{k=0}\left\|\mat{E}_1^k\mat{E}_2\mat{Q}_r\left(\mat{\Sigma}_r^{-1}\right)^{k+1}\right\|_2\cdot \|\mat{P}^{-1}\|_2\\
    \leq &~ \mathcal{O}(1) \cdot \sum^\infty_{k=0}\left\|\mat{E}_1^k\mat{E}_2\mat{Q}_r\left(\mat{\Sigma}_r^{-1}\right)^{k+1}\right\|_2\\
    \leq &~ \mathcal{O}(1) \cdot \sum_{k=0}^\infty {\cal O}\left(\frac{\mu_{\rm tail}}{\mu_r}\right)^k \cdot \frac{\alpha\sqrt{\log n}}{\mu_r\sqrt{n}}\\
    = &~ {\cal O}\left(\frac{\alpha\sqrt{\log n}}{\mu_r\sqrt{n}}\right)\,,
\end{aligned}
\end{equation}
where the first step follows from triangle inequality, the second step follows from \cref{lem:vander-eigen} and {\it Property I}, the third step follows from $\left\|\mat{E}_{1}\right\|_2=\mathcal{O}\left(n\mu_{\rm tail}\right)$, $\left\|\mat{E}_{2}\right\|_2=\mathcal{O}\left(\alpha\sqrt{n\log(n)}\right)$, and $\|\mat{\Sigma}_r^{-1}\|_2=1/(n\mu_r)$ by  \eqref{eqn:eigenvalue_prop} in \cref{cor:moitra_T}.

Then, we bound the norm of $\widetilde{\mat{E}}^Q_2$:
\begin{align*}
    \big\|\widetilde{\mat{E}}^Q_2\big\|_2\leq {\cal O}\left(\frac{\alpha^2\log n}{\mu^{2}_r\Delta_{\vec{z}} n}\right)\cdot \|\mat{P}^{-1}\|_2 = {\cal O}\left(\frac{\alpha^2\log n}{\mu^{2}_r\Delta_{\vec{z}} n}\right)\,,
\end{align*}
where the first step follows from $\left\|\mat{\Pi}_{\mat{Q}_r^\perp}\hat{\mat{Q}}_{2}\right\|_2=\mathcal{O}(1)$ since $\|\hat{\mat{Q}}_{2}\|_2=\mathcal{O}(1)$, 
and the second step follows from $\|\mat{P}^{-1}\|_2={\cal O}(1)$ by {\it Property I}.

Hence, we get that
\begin{equation}%
\begin{aligned}
\left\|\mat{\hat{Q}}_r\mat{U}_r\mat{P}^{-1}-\mat{Q}_r\right\|_2\leq &~ {\cal O}\left(\frac{\alpha\sqrt{\log n}}{\mu_r\sqrt{n}} + \frac{\alpha^2\log n}{\mu^{2}_r\Delta_{\vec{z}} n}+\frac{\alpha\sqrt{\log n}}{\mu_r\Delta_{\vec{z}}n^{1.5}}\right)\\
=&~ \mathcal{O}\left(\frac{\alpha\sqrt{\log n}}{\mu_r\sqrt{\zgap n}}\right)\,.
\end{aligned}
\end{equation}
\end{proof}

\subsection{Proofs for Step 2: Taylor expansion with respect to the error terms}\label{sec:align_step2_defer}
\stepIIfirst*
\begin{proof}
According to \eqref{eqn:second_Q_perturb} and the assumption that $n$ is large enough, we have $\|\mat{E}^Q\mat{Q}\|_2<1$ and $\|\mat{E}^Q\mat{E}^Q\|_2<1$. Then, we can use the formula of pseudoinverse \cref{eq:pseudo_inv_col} and Neumann series expansion (\cref{def:neumann}) to derive
\begin{equation}\label{eqn:oneside}
\begin{aligned}
&\mat{P}\mat{\hat{Q}}^{+}_\downarrow\mat{\hat{Q}}_\uparrow\mat{P}^{-1}\\
= &~ \left(\hat{\mat{Q}}_{\downarrow}\mat{P}^{-1}\right)^{+} \mat{\hat{Q}}_\uparrow\mat{P}^{-1}\\
= &~ \left(\mat{Q}_\downarrow + \mat{E}^Q_\downarrow\right)^{+}\left(\mat{Q}_\uparrow+\mat{E}^Q_\uparrow\right)\\
=&~\left(\left(\mat{Q}_\downarrow+\mat{E}^Q_\downarrow\right)^\dagger\left(\mat{Q}_\downarrow+\mat{E}^Q_\downarrow\right)\right)^{-1}\left(\mat{Q}_\downarrow+\mat{E}^Q_\downarrow\right)^\dagger\left(\mat{Q}_\uparrow+\mat{E}^Q_\uparrow\right)\\
=&~\left(\mat{Q}_\downarrow^\dagger\mat{Q}_\downarrow+(\mat{E}^Q_\downarrow)^\dagger\mat{Q}_\downarrow+\mat{Q}_\downarrow^\dagger\mat{E}^Q_\downarrow+(\mat{E}^Q_\downarrow)^\dagger\mat{E}^Q_\downarrow\right)^{-1}\left(\mat{Q}_\downarrow+\mat{E}^Q_\downarrow\right)^\dagger\left(\mat{Q}_\uparrow+\mat{E}^Q_\uparrow\right)\\
= &~ \left(\mat{I}_r-\left(\mat{Q}_\downarrow^\dagger\mat{Q}_\downarrow\right)^{-1}\left((\mat{E}^Q_\downarrow)^\dagger\mat{Q}_\downarrow+\mat{Q}_\downarrow^\dagger\mat{E}^Q_\downarrow+(\mat{E}^Q_\downarrow)^\dagger\mat{E}^Q_\downarrow\right)\right)^{-1}\left(\mat{Q}_\downarrow^\dagger\mat{Q}_\downarrow\right)^{-1}\cdot \left(\mat{Q}_\downarrow+\mat{E}^Q_\downarrow\right)^\dagger\left(\mat{Q}_\uparrow+\mat{E}^Q_\uparrow\right)\\
=&~\sum^\infty_{k=0}\left(-\left(\mat{Q}_\downarrow^\dagger\mat{Q}_\downarrow\right)^{-1}\left((\mat{E}^Q_\downarrow)^\dagger\mat{Q}_\downarrow+\mat{Q}_\downarrow^\dagger\mat{E}^Q_\downarrow+(\mat{E}^Q_\downarrow)^\dagger\mat{E}^Q_\downarrow\right)\right)^k\left(\mat{Q}_\downarrow^\dagger\mat{Q}_\downarrow\right)^{-1}\left(\mat{Q}_\downarrow+\mat{E}^Q_\downarrow\right)^\dagger\left(\mat{Q}_\uparrow+\mat{E}^Q_\uparrow\right)\,.
\end{aligned}
\end{equation}
where the first step follows from \cref{fac:rev_order_pinv} with $\hat{\mat{Q}}_\downarrow$ being full column rank and $\mat{P}$ being invertible, the second step follows from \eqref{eq:E_Q_up_down}, the third step follows from the definition of the pseudoinverse, the fourth step is by direct calculation, the fifth step follows from pulling out the term $\mat{Q}^\dagger_\downarrow \mat{Q}_\downarrow$, and the last step follows from Neumann series expansion.

By \cref{clm:bound_Neumann_series}, the Neumann series in \eqref{eqn:oneside} can be truncated up to second order:
\begin{align*}
    &\sum^\infty_{k=0}\left(-\left(\mat{Q}_\downarrow^\dagger\mat{Q}_\downarrow\right)^{-1}\left((\mat{E}^Q_\downarrow)^\dagger\mat{Q}_\downarrow+\mat{Q}_\downarrow^\dagger\mat{E}^Q_\downarrow+(\mat{E}^Q_\downarrow)^\dagger\mat{E}^Q_\downarrow\right)\right)^k\left(\mat{Q}_\downarrow^\dagger\mat{Q}_\downarrow\right)^{-1}\\
    =&~\underbrace{\left(\mat{Q}_\downarrow^\dagger\mat{Q}_\downarrow\right)^{-1}}_{\text{zeroth order term}}-\underbrace{\left((\mat{E}^Q_\downarrow)^\dagger\mat{Q}_\downarrow+\mat{Q}_\downarrow^\dagger\mat{E}^Q_\downarrow+(\mat{E}^Q_\downarrow)^\dagger\mat{E}^Q_\downarrow\right)}_{\text{first order term}}+\underbrace{\left((\mat{E}^Q_\downarrow)^\dagger\mat{Q}_\downarrow+\mat{Q}_\downarrow^\dagger\mat{E}^Q_\downarrow\right)^2}_{\text{second order term}} \\
    ~&+\mathcal{O}\left(\frac{\alpha^3\log^3 n}{\mu^{3}_r(\Delta_{\vec{z}} n)^{1.5}}\right)\,.    
\end{align*}

Then, we can plug it into \eqref{eqn:oneside} to simplify $\mat{P}\mat{\hat{Q}}^{+}_\downarrow\mat{\hat{Q}}_\uparrow\mat{P}^{-1} - \mat{Q}^{+}_\downarrow\mat{Q}_\uparrow$ and drop the terms containing more than two $\mat{E}^Q_\uparrow$ or $\mat{E}^Q_\downarrow$. We first consider the zeroth order term:
\begin{equation}\label{eq:P_Q_Q_P_0th}
\begin{aligned}
    &\left(\mat{Q}_\downarrow^\dagger\mat{Q}_\downarrow\right)^{-1}\left(\mat{Q}_\downarrow+\mat{E}^Q_\downarrow\right)^\dagger\left(\mat{Q}_\uparrow+\mat{E}^Q_\uparrow\right)-\mat{Q}^{+}_\downarrow\mat{Q}_\uparrow\\
    =&~ \left(\mat{Q}_\downarrow^\dagger\mat{Q}_\downarrow\right)^{-1}\left((\mat{E}^Q_\downarrow)^\dagger\mat{Q}_\uparrow+\mat{Q}_\downarrow^\dagger\mat{E}^Q_\uparrow+(\mat{E}^Q_\downarrow)^\dagger\mat{E}^Q_\uparrow\right)\\
    = &~ \left((\mat{E}^Q_\downarrow)^\dagger\mat{Q}_\uparrow+\mat{Q}_\downarrow^\dagger\mat{E}^Q_\uparrow+(\mat{E}^Q_\downarrow)^\dagger\mat{E}^Q_\uparrow\right)+\left(\left(\mat{Q}_\downarrow^\dagger\mat{Q}_\downarrow\right)^{-1} - \mat{I}_r\right)\left((\mat{E}^Q_\downarrow)^\dagger\mat{Q}_\uparrow+\mat{Q}_\downarrow^\dagger\mat{E}^Q_\uparrow+(\mat{E}^Q_\downarrow)^\dagger\mat{E}^Q_\uparrow\right)\\
    \leq &~ \left((\mat{E}^Q_\downarrow)^\dagger\mat{Q}_\uparrow+\mat{Q}_\downarrow^\dagger\mat{E}^Q_\uparrow+(\mat{E}^Q_\downarrow)^\dagger\mat{E}^Q_\uparrow\right)+ \mathcal{O}\left(\frac{1}{\Delta_{\vec{z}} n}\right)\left\|(\mat{E}^Q_\downarrow)^\dagger\mat{Q}_\uparrow+\mat{Q}_\downarrow^\dagger\mat{E}^Q_\uparrow+(\mat{E}^Q_\downarrow)^\dagger\mat{E}^Q_\uparrow\right\|_2\\
    \leq &~ \left((\mat{E}^Q_\downarrow)^\dagger\mat{Q}_\uparrow+\mat{Q}_\downarrow^\dagger\mat{E}^Q_\uparrow+(\mat{E}^Q_\downarrow)^\dagger\mat{E}^Q_\uparrow\right)+ \mathcal{O}\left(\frac{\alpha\sqrt{\log n}}{\mu_r(\Delta_{\vec{z}} n)^{1.5}}\right)\,,
\end{aligned}
\end{equation}
where the first step follows from $(\mat{Q}_\downarrow^\dagger \mat{Q}_\downarrow)^{-1}\mat{Q}_\downarrow^\dagger =\mat{Q}_\downarrow^{+}$ by the definition of pseudoinverse, the second step is straightforward, the third step follows from \eqref{eqn:inverse_close_to_identity}, and the last step follows from \eqref{eqn:second_Q_perturb}. For the first order term, we have
\begin{equation}\label{eq:P_Q_Q_P_1st}
\begin{aligned}
    &\left((\mat{E}^Q_\downarrow)^\dagger\mat{Q}_\downarrow+\mat{Q}_\downarrow^\dagger\mat{E}^Q_\downarrow+(\mat{E}^Q_\downarrow)^\dagger\mat{E}^Q_\downarrow\right)\left(\mat{Q}_\downarrow+\mat{E}^Q_\downarrow\right)^\dagger\left(\mat{Q}_\uparrow+\mat{E}^Q_\uparrow\right)\\
    \leq &~ \left((\mat{E}^Q_\downarrow)^\dagger\mat{Q}_\downarrow+\mat{Q}_\downarrow^\dagger\mat{E}^Q_\downarrow+(\mat{E}^Q_\downarrow)^\dagger\mat{E}^Q_\downarrow\right)(\mat{Q}_\downarrow)^\dagger \mat{Q}_\uparrow+ \left((\mat{E}^Q_\downarrow)^\dagger\mat{Q}_\downarrow+\mat{Q}_\downarrow^\dagger\mat{E}^Q_\downarrow\right)\left((\mat{E}^Q_\downarrow)^\dagger \mat{Q}_\uparrow + (\mat{Q}_\downarrow)^\dagger \mat{E}^Q_\uparrow\right)\\
    ~&+\left\|(\mat{E}^Q_\downarrow)^\dagger\mat{E}^Q_\downarrow \left((\mat{E}^Q_\downarrow)^\dagger \mat{Q}_\uparrow + (\mat{Q}_\downarrow)^\dagger \mat{E}^Q_\uparrow\right)+\left((\mat{E}^Q_\downarrow)^\dagger\mat{Q}_\downarrow+\mat{Q}_\downarrow^\dagger\mat{E}^Q_\downarrow+(\mat{E}^Q_\downarrow)^\dagger\mat{E}^Q_\downarrow\right)(\mat{E}^Q_\downarrow)^\dagger\mat{E}^Q_\uparrow\right\|_2\\
    \leq &~ \left((\mat{E}^Q_\downarrow)^\dagger\mat{Q}_\downarrow+\mat{Q}_\downarrow^\dagger\mat{E}^Q_\downarrow+(\mat{E}^Q_\downarrow)^\dagger\mat{E}^Q_\downarrow\right)(\mat{Q}_\downarrow)^\dagger \mat{Q}_\uparrow+ \left((\mat{E}^Q_\downarrow)^\dagger\mat{Q}_\downarrow+\mat{Q}_\downarrow^\dagger\mat{E}^Q_\downarrow\right)\left((\mat{E}^Q_\downarrow)^\dagger \mat{Q}_\uparrow + (\mat{Q}_\downarrow)^\dagger \mat{E}^Q_\uparrow\right)\\
    ~&+ \mathcal{O}\left(\frac{\alpha^3\log^3 n}{\mu^{3}_r(\Delta_{\vec{z}} n)^{1.5}}\right)\,,
\end{aligned}
\end{equation}
where the second step uses \eqref{eqn:second_Q_perturb} for $\|\mat{E}^Q_\downarrow\|_2$ and \eqref{eqn:inverse_close_to_identity} for $\|\mat{Q}_\downarrow\|_2$. Similarly, for the second order term, we have
\begin{equation}\label{eq:P_Q_Q_P_2nd}
\begin{aligned}
    &\left((\mat{E}^Q_\downarrow)^\dagger\mat{Q}_\downarrow+\mat{Q}_\downarrow^\dagger\mat{E}^Q_\downarrow\right)^2\left(\mat{Q}_\downarrow+\mat{E}^Q_\downarrow\right)^\dagger\left(\mat{Q}_\uparrow+\mat{E}^Q_\uparrow\right)\\
    \leq &~ \left((\mat{E}^Q_\downarrow)^\dagger\mat{Q}_\downarrow+\mat{Q}_\downarrow^\dagger\mat{E}^Q_\downarrow\right)^2 (\mat{Q}_\downarrow)^\dagger \mat{Q}_\uparrow+\left\|\left((\mat{E}^Q_\downarrow)^\dagger\mat{Q}_\downarrow+\mat{Q}_\downarrow^\dagger\mat{E}^Q_\downarrow\right)^2\left((\mat{E}^Q_\downarrow)^\dagger \mat{Q}_\uparrow + (\mat{Q}_\downarrow)^\dagger \mat{E}^Q_\uparrow+(\mat{E}^Q_\downarrow)^\dagger\mat{E}_\uparrow^Q\right)\right\|_2\\
    \leq &~ \left((\mat{E}^Q_\downarrow)^\dagger\mat{Q}_\downarrow+\mat{Q}_\downarrow^\dagger\mat{E}^Q_\downarrow\right)^2 (\mat{Q}_\downarrow)^\dagger \mat{Q}_\uparrow+\mathcal{O}\left(\frac{\alpha^3\log^3 n}{\mu^{3}_r(\Delta_{\vec{z}} n)^{1.5}}\right)\,.
\end{aligned}
\end{equation}
And for the residual term, we have
\begin{equation}\label{eq:P_Q_Q_P_res}
\begin{aligned}
    \mathcal{O}\left(\frac{\alpha^3\log^3 n}{\mu^{3}_r(\Delta_{\vec{z}} n)^{1.5}}\right)\cdot \left\|\left(\mat{Q}_\downarrow+\mat{E}^Q_\downarrow\right)^\dagger\left(\mat{Q}_\uparrow+\mat{E}^Q_\uparrow\right)\right\|_2 = \mathcal{O}\left(\frac{\alpha^3\log^3 n}{\mu^{3}_r(\Delta_{\vec{z}} n)^{1.5}}\right)\,.
\end{aligned}    
\end{equation}
Finally, combining \cref{eq:P_Q_Q_P_0th,eq:P_Q_Q_P_1st,eq:P_Q_Q_P_2nd,eq:P_Q_Q_P_res} together, we obtain that
\begin{align*}
&\mat{P}\mat{\hat{Q}}^{+}_\downarrow\mat{\hat{Q}}_\uparrow\mat{P}^{-1}-\mat{Q}^{+}_\downarrow\mat{Q}_\uparrow\\
= &~ \left((\mat{E}^Q_\downarrow)^\dagger\mat{Q}_\uparrow+\mat{Q}_\downarrow^\dagger\mat{E}^Q_\uparrow+(\mat{E}^Q_\downarrow)^\dagger\mat{E}^Q_\uparrow\right)\\
~&+\left((\mat{E}^Q_\downarrow)^\dagger\mat{Q}_\downarrow+\mat{Q}_\downarrow^\dagger\mat{E}^Q_\downarrow+(\mat{E}^Q_\downarrow)^\dagger\mat{E}^Q_\downarrow\right)(\mat{Q}_\downarrow)^\dagger \mat{Q}_\uparrow+ \left((\mat{E}^Q_\downarrow)^\dagger\mat{Q}_\downarrow+\mat{Q}_\downarrow^\dagger\mat{E}^Q_\downarrow\right)\left((\mat{E}^Q_\downarrow)^\dagger \mat{Q}_\uparrow + (\mat{Q}_\downarrow)^\dagger \mat{E}^Q_\uparrow\right)\\
~&+\left((\mat{E}^Q_\downarrow)^\dagger\mat{Q}_\downarrow+\mat{Q}_\downarrow^\dagger\mat{E}^Q_\downarrow\right)^2 (\mat{Q}_\downarrow)^\dagger \mat{Q}_\uparrow+\mathcal{O}\left(\frac{\alpha^3\log^3 n}{\mu^{3}_r(\Delta_{\vec{z}} n)^{1.5}}\right)\\
=&~\left((\mat{E}^Q_\downarrow)^\dagger\mat{Q}_\uparrow+\mat{Q}_\downarrow^\dagger\mat{E}^Q_\uparrow\right)-\left((\mat{E}^Q_\downarrow)^\dagger\mat{Q}_\downarrow+\mat{Q}_\downarrow^\dagger\mat{E}^Q_\downarrow\right)(\mat{Q}_\downarrow)^\dagger\mat{Q}_\uparrow\\
~&+(\mat{E}^Q_\downarrow)^\dagger(\mat{E}^Q_\uparrow)-\left((\mat{E}^Q_\downarrow)^\dagger\mat{Q}_\downarrow+\mat{Q}_\downarrow^\dagger\mat{E}^Q_\downarrow\right)\left((\mat{E}^Q_\downarrow)^\dagger\mat{Q}_\uparrow+\mat{Q}_\downarrow^\dagger\mat{E}^Q_\uparrow\right)\\
~&+\left(-(\mat{E}^Q_\downarrow)^\dagger(\mat{E}^Q_\downarrow)+\left((\mat{E}^Q_\downarrow)^\dagger\mat{Q}_\downarrow+\mat{Q}_\downarrow^\dagger\mat{E}^Q_\downarrow\right)^2\right)(\mat{Q}_\downarrow)^\dagger\mat{Q}_\uparrow+\mathcal{O}\left(\frac{\alpha^3\log^3 n}{\mu^{3}_r(\Delta_{\vec{z}} n)^{1.5}}\right)\,,
\end{align*}
where we re-group the terms in the second step. 

The proof of the lemma is completed.
\end{proof}

\coarsebound*

\begin{proof}
We first deal with the first order term in \eqref{eqn:error_expansion}:
\begin{equation}\label{eqn:first_order_derivation}
\begin{aligned}
&\left((\mat{E}^Q_\downarrow)^\dagger\mat{Q}_\uparrow+\mat{Q}_\downarrow^\dagger\mat{E}^Q_\uparrow\right)-\left((\mat{E}^Q_\downarrow)^\dagger\mat{Q}_\downarrow+\mat{Q}_\downarrow^\dagger\mat{E}^Q_\downarrow\right)(\mat{Q}_\downarrow)^\dagger\mat{Q}_\uparrow\\
=&\underbrace{(\mat{E}^Q_\downarrow)^\dagger\left(\mat{Q}_\uparrow-\mat{Q}_\downarrow(\mat{Q}_\downarrow)^\dagger\mat{Q}_\uparrow\right)}_{\mathrm{(I)}}+\underbrace{\mat{Q}_\downarrow^\dagger\left(\mat{E}^Q_\uparrow-\mat{E}^Q_\downarrow(\mat{Q}_\downarrow)^\dagger\mat{Q}_\uparrow\right)}_{\mathrm{(II)}}\,.
\end{aligned}
\end{equation}

To bound $\mathrm{(I)}$, we notice
\[
\textbf{R}(\mat{Q}_\uparrow)=\textbf{R}(\left(\mat{V}_n(\zdom)\right)_\uparrow)=\textbf{R}(\left(\mat{V}_n(\zdom)\right)_\downarrow)= \textbf{R}(\mat{Q}_\downarrow)\,,
\]
where the second step follows from $\left(\mat{V}_n(\zdom)\right)_\uparrow=\left(\mat{V}_n(\zdom)\right)_\downarrow\cdot \diag(\vec{z}_{\rm dom})$.
We also notice that
\begin{equation}\label{eq:Q_Q_dagger_Pi}
\begin{aligned}
\left\|\mat{Q}_\downarrow(\mat{Q}_\downarrow)^\dagger-\Pi_{\textbf{R}(\mat{Q}_\uparrow)}\right\|_2=&~\left\|\mat{Q}_\downarrow(\mat{Q}_\downarrow)^\dagger-\Pi_{\textbf{R}(\mat{Q}_\downarrow)}\right\|_2\\
=&~\left\|\mat{Q}_\downarrow(\mat{Q}_\downarrow)^\dagger-\mat{Q}_\downarrow\mat{P}_{Q_\downarrow}\mat{P}^\dagger_{Q_\downarrow}\mat{Q}_\downarrow^\dagger\right\|_2\\
\leq &~ \left\|\mat{Q}_\downarrow\right\|_2\cdot \left\|\mat{I}_r-\mat{P}_{Q_\downarrow}\mat{P}^\dagger_{Q_\downarrow}\right\|_2\cdot \|\mat{Q}_\downarrow^\dagger\|_2\\
=&~\mathcal{O}\left(\frac{1}{\Delta_{\vec{z}} n}\right)\,,
\end{aligned}
\end{equation}
where $\Pi_{\textbf{R}(\mat{Q}_\downarrow)}$ denotes the projector to the column space of $\mat{Q}_\downarrow$, the second step follows the definition of $\mat{P}_{Q_\downarrow}$ \eqref{eq:def_P_Q_downarrow}, and the last step follows from  \eqref{eqn:P_Q_unitary} and $\|\mat{Q}_\downarrow\|_2={\cal O}(1)$.

Combining this and \eqref{eqn:second_Q_perturb}, we obtain
\begin{equation}\label{eq:taylor_series_fo_I}
\begin{aligned}
\Big\|(\mat{E}^Q_\downarrow)^\dagger\left(\mat{Q}_\uparrow-\mat{Q}_\downarrow(\mat{Q}_\downarrow)^\dagger\mat{Q}_\uparrow\right)\Big\|_2= &~ \Big\|(\mat{E}^Q_\downarrow)^\dagger\left(\mat{Q}_\uparrow-\Pi_{\textbf{R}(\mat{Q}_\uparrow)}\mat{Q}_\uparrow\right) +(\mat{E}^Q_\downarrow)^\dagger\left(\Pi_{\textbf{R}(\mat{Q}_\uparrow)} - \mat{Q}_\downarrow(\mat{Q}_\downarrow)^\dagger\right)\mat{Q}_\uparrow\Big\|_2\\
= &~ \Big\|(\mat{E}^Q_\downarrow)^\dagger\left(\Pi_{\textbf{R}(\mat{Q}_\uparrow)} - \mat{Q}_\downarrow(\mat{Q}_\downarrow)^\dagger\right)\mat{Q}_\uparrow\Big\|_2\\
\leq &~ \big\|\mat{E}^Q_\downarrow\big\|_2 \cdot \left\|\Pi_{\textbf{R}(\mat{Q}_\uparrow)} - \mat{Q}_\downarrow(\mat{Q}_\downarrow)^\dagger\right\|_2\cdot \left\|\mat{Q}_\uparrow\right\|_2\\
\leq &~ \mathcal{O}\left(\frac{\alpha\sqrt{\log n}}{\mu_r\sqrt{\Delta_{\vec{z}}n}}\right)\cdot \mathcal{O}\left(\frac{1}{\Delta_{\vec{z}} n}\right)\\
=&~\mathcal{O}\left(\frac{\alpha\sqrt{\log n}}{\mu_r(\Delta_{\vec{z}} n)^{1.5}}\right)\,,
\end{aligned}
\end{equation}
where the first step is straightforward, the second step follows from $\Pi_{\textbf{R}(\mat{Q}_\uparrow)}\mat{Q}_\uparrow=\mat{Q}_\uparrow$ by definition, the third step follows from triangle inequality, the fourth step follows from \eqref{eqn:second_Q_perturb} and \eqref{eq:Q_Q_dagger_Pi}.%

Next, we keep $\mathrm{(II)}$ and first consider the second-order term in \eqref{eqn:error_expansion}, which can be rewritten as follows:
\[
\begin{aligned}
&(\mat{E}^Q_\downarrow)^\dagger(\mat{E}^Q_\uparrow)-\left((\mat{E}^Q_\downarrow)^\dagger\mat{Q}_\downarrow+\mat{Q}_\downarrow^\dagger\mat{E}^Q_\downarrow\right)\left((\mat{E}^Q_\downarrow)^\dagger\mat{Q}_\uparrow+\mat{Q}_\downarrow^\dagger\mat{E}^Q_\uparrow\right)\\
~&+\left(-(\mat{E}^Q_\downarrow)^\dagger(\mat{E}^Q_\downarrow)+\left((\mat{E}^Q_\downarrow)^\dagger\mat{Q}_\downarrow+\mat{Q}_\downarrow^\dagger\mat{E}^Q_\downarrow\right)^2\right)(\mat{Q}_\downarrow)^\dagger\mat{Q}_\uparrow\\
=&~(\mat{E}^Q_\downarrow)^\dagger\left(\mat{E}^Q_\uparrow-\mat{E}^Q_\downarrow(\mat{Q}_\downarrow)^\dagger\mat{Q}_\uparrow\right)\\
~&-\left((\mat{E}^Q_\downarrow)^\dagger\mat{Q}_\downarrow+\mat{Q}_\downarrow^\dagger\mat{E}^Q_\downarrow\right)\left(\left((\mat{E}^Q_\downarrow)^\dagger\mat{Q}_\uparrow+\mat{Q}_\downarrow^\dagger\mat{E}^Q_\uparrow\right)-\left((\mat{E}^Q_\downarrow)^\dagger\mat{Q}_\downarrow+\mat{Q}_\downarrow^\dagger\mat{E}^Q_\downarrow\right)(\mat{Q}_\downarrow)^\dagger\mat{Q}_\uparrow\right)\\
=&~(\mat{E}^Q_\downarrow)^\dagger\left(\mat{E}^Q_\uparrow-\mat{E}^Q_\downarrow(\mat{Q}_\downarrow)^\dagger\mat{Q}_\uparrow\right)\\
~&-\left((\mat{E}^Q_\downarrow)^\dagger\mat{Q}_\downarrow+\mat{Q}_\downarrow^\dagger\mat{E}^Q_\downarrow\right)\Bigg(\underbrace{(\mat{E}^Q_\downarrow)^\dagger\left(\mat{Q}_\uparrow-\mat{Q}_\downarrow(\mat{Q}_\downarrow)^\dagger\mat{Q}_\uparrow\right)}_{\mathrm{(I)}~\text{in~\eqref{eqn:first_order_derivation}}}+\underbrace{\mat{Q}_\downarrow^\dagger\left(\mat{E}^Q_\uparrow-\mat{E}^Q_\downarrow(\mat{Q}_\downarrow)^\dagger\mat{Q}_\uparrow\right)}_{\mathrm{(II)}~\text{in~\eqref{eqn:first_order_derivation}}}\Bigg)\\
=&~(\mat{E}^Q_\downarrow)^\dagger\left(\mat{E}^Q_\uparrow-\mat{E}^Q_\downarrow(\mat{Q}_\downarrow)^\dagger\mat{Q}_\uparrow\right)+\mathcal{O}\left(\frac{\alpha\sqrt{\log n}}{\mu_r\sqrt{\Delta_{\vec{z}}n}}\right)\left\|\mat{Q}_\downarrow^\dagger\left(\mat{E}^Q_\uparrow-\mat{E}^Q_\downarrow(\mat{Q}_\downarrow)^\dagger\mat{Q}_\uparrow\right)\right\|_2\\
~&+\mathcal{O}\left(\frac{\alpha^2\log n}{\mu^{2}_r\Delta_{\vec{z}}^2 n^{2}}\right)
\end{aligned}
\]
where the first two steps follow from re-grouping the terms, and the last step follows from \eqref{eq:mid_bracket_neumann} and \eqref{eq:taylor_series_fo_I}. %

Plugging the bounds for the first order and the second order terms into \eqref{eqn:error_expansion}, we get that
\begin{align*}
&\left\|\mat{P}\mat{\hat{Q}}^{+}_\downarrow\mat{\hat{Q}}_\uparrow\mat{P}^{-1}-\mat{Q}^{+}_\downarrow\mat{Q}_\uparrow\right\|_2\\
\leq &~ \left\|(\mat{E}^Q_\downarrow)^\dagger\left(\mat{E}^Q_\uparrow-\mat{E}^Q_\downarrow(\mat{Q}_\downarrow)^\dagger\mat{Q}_\uparrow\right)\right\|_2+\mathcal{O}\left(1+\frac{\alpha\sqrt{\log n}}{\mu_r\sqrt{\Delta_{\vec{z}} n}}\right)\left\|\mat{Q}_\downarrow^\dagger\left(\mat{E}^Q_\uparrow-\mat{E}^Q_\downarrow(\mat{Q}_\downarrow)^\dagger\mat{Q}_\uparrow\right)\right\|_2\\
~&+ \mathcal{O}\left(\frac{\alpha^3\log^3 n}{\mu^{3}_r(\Delta_{\vec{z}} n)^{1.5}} \right)\,,
\end{align*}
which proves the lemma.
\end{proof}

\subsubsection{Technical claims}
\begin{claim}[Approximate unitarity of $\mat{Q}_\downarrow$]
For $\mat{Q}_\downarrow=\mat{Q}(1:,:r)$, we have
\begin{equation}\label{eqn:inverse_close_to_identity}
\left\|\mat{Q}_\downarrow^\dagger\mat{Q}_\downarrow-\mat{I}_r\right\|_2=\mathcal{O}\left(\frac{1}{\Delta_{\vec{z}} n}\right),\quad \left\|\left(\mat{Q}_\downarrow^\dagger\mat{Q}_\downarrow\right)^{-1}-\mat{I}_r\right\|_2=\mathcal{O}\left(\frac{1}{\Delta_{\vec{z}} n}\right)\,.
\end{equation}

Furthermore, there exists a matrix $\mat{P}_{\mat{Q}_\downarrow}\in \C^{r\times r}$ such that $\mat{P}^\dagger_{Q_\downarrow}\mat{Q}_\downarrow^\dagger\mat{Q}_\downarrow\mat{P}_{Q_\downarrow}=\mat{I}_r$ and
\begin{equation}\label{eqn:P_Q_unitary}
\left\|\mat{P}^\dagger_{Q_\downarrow}\mat{P}_{Q_\downarrow}-\mat{I}_r\right\|_2=\mathcal{O}\left(\frac{1}{\Delta_{\vec{z}} n}\right),\quad \left\|\mat{P}_{Q_\downarrow}\mat{P}^\dagger_{Q_\downarrow}-\mat{I}_r\right\|_2=\mathcal{O}\left(\frac{1}{\Delta_{\vec{z}} n}\right)\,.
\end{equation}
\end{claim}
\begin{proof}
We first simplify $\left(\mat{Q}_\downarrow^\dagger\mat{Q}_\downarrow\right)^{-1}$ in the above formula:%
\[
\begin{aligned}
\left\|\mat{Q}_\downarrow^\dagger\mat{Q}_\downarrow-\mat{I}_r\right\|_2=&~\left\|\mat{P}^\dagger_v\frac{\left(\mat{V}_n(\zdom)\right)_\downarrow^\dagger\left(\mat{V}_n(\zdom)\right)_\downarrow}{n}\mat{P}_v-\mat{I}_r\right\|_2\\
\leq &~\left\|\mat{P}^\dagger_v\mat{P}_v-\mat{I}_r\right\|_2+\left\|\mat{P}_v^\dagger\left(\frac{\left(\mat{V}_n(\zdom)\right)_\downarrow^\dagger\left(\mat{V}_n(\zdom)\right)_\downarrow}{n}-\mat{I}_r\right)\mat{P}_v\right\|_2\\
=&~ \mathcal{O}\left(\frac{1}{\Delta_{\vec{z}} n}\right)\,.
\end{aligned}
\]
where the first step follows from \eqref{eqn:reshape_Q} that
\begin{align*}
    \mat{Q}_\downarrow=\left(\frac{1}{\sqrt{n}}\mat{V}_n(\zdom)\mat{P}_v\right)_\downarrow = \frac{1}{\sqrt{n}}(\mat{V}_n(\zdom))_\downarrow\mat{P}_v\,,
\end{align*}
the second step follows from triangle inequality, and the third step follows from \cref{lem:vander-eigen} and \cref{cor:V_dagger_V_dom}.

This implies that
\begin{align*}%
\left\|\mat{Q}_\downarrow^\dagger\mat{Q}_\downarrow-\mat{I}_r\right\|_2=\mathcal{O}\left(\frac{1}{\Delta_{\vec{z}} n}\right),\quad \left\|\left(\mat{Q}_\downarrow^\dagger\mat{Q}_\downarrow\right)^{-1}-\mat{I}_r\right\|_2=\mathcal{O}\left(\frac{1}{\Delta_{\vec{z}} n}\right)\,.
\end{align*} 

For the furthermore part, let the eigendecomposition of  $\mat{Q}_\downarrow^\dagger\mat{Q}_\downarrow$ be as follows:
\[
\mat{Q}_\downarrow^\dagger\mat{Q}_\downarrow=\mat{U}^\dagger_\downarrow \mat{\Lambda}_\downarrow \mat{U}_\downarrow.
\]
We define an invertible matrix $\mat{P}_{Q_\downarrow}$ as:
\begin{align}\label{eq:def_P_Q_downarrow}
\mat{P}_{Q_\downarrow}:=\mat{U}_\downarrow^\dagger\mat{\Lambda}^{-1/2}_\downarrow \,.
\end{align}
Then, we have 
\begin{align*}
\mat{P}^\dagger_{Q_\downarrow}\mat{Q}_\downarrow^\dagger\mat{Q}_\downarrow\mat{P}_{Q_\downarrow}=\mat{I}_r\,,
\end{align*}
Using \eqref{eqn:inverse_close_to_identity}, we obtain that $\|\mat{\Lambda}^{-1}_\downarrow-\mat{I}_r\|_2=\mathcal{O}\left(\frac{1}{\Delta_{\vec{z}} n}\right)$. This implies that $\mat{P}_{Q_\downarrow}\in\mathbb{C}^{r\times r}$ satisfies
\begin{align*}%
\left\|\mat{P}^\dagger_{Q_\downarrow}\mat{P}_{Q_\downarrow}-\mat{I}_r\right\|_2=\mathcal{O}\left(\frac{1}{\Delta_{\vec{z}} n}\right),\quad \left\|\mat{P}_{Q_\downarrow}\mat{P}^\dagger_{Q_\downarrow}-\mat{I}_r\right\|_2=\mathcal{O}\left(\frac{1}{\Delta_{\vec{z}} n}\right)
\end{align*}
where the first equation follows from $\mat{P}^\dagger_{Q_\downarrow}\mat{P}_{Q_\downarrow} = \mat{\Lambda}^{-1}$, and the second equation follows from $\mat{P}_{Q_\downarrow}\mat{P}^\dagger_{Q_\downarrow}-\mat{I}_r = \mat{U}_\downarrow^\dagger (\mat{\Lambda}^{-1} -\mat{I}_r)\mat{U}_\downarrow$.

The claim is then proved.
\end{proof}

\begin{claim}[Neumann series truncation]\label{clm:bound_Neumann_series}
It holds that
\begin{align*}
    &\sum^\infty_{k=0}\left(-\left(\mat{Q}_\downarrow^\dagger\mat{Q}_\downarrow\right)^{-1}\left((\mat{E}^Q_\downarrow)^\dagger\mat{Q}_\downarrow+\mat{Q}_\downarrow^\dagger\mat{E}^Q_\downarrow+(\mat{E}^Q_\downarrow)^\dagger\mat{E}^Q_\downarrow\right)\right)^k\left(\mat{Q}_\downarrow^\dagger\mat{Q}_\downarrow\right)^{-1}\\
    =&~\underbrace{\left(\mat{Q}_\downarrow^\dagger\mat{Q}_\downarrow\right)^{-1}}_{\text{zeroth order term}}-\underbrace{\left((\mat{E}^Q_\downarrow)^\dagger\mat{Q}_\downarrow+\mat{Q}_\downarrow^\dagger\mat{E}^Q_\downarrow+(\mat{E}^Q_\downarrow)^\dagger\mat{E}^Q_\downarrow\right)}_{\text{first order term}}+\underbrace{\left((\mat{E}^Q_\downarrow)^\dagger\mat{Q}_\downarrow+\mat{Q}_\downarrow^\dagger\mat{E}^Q_\downarrow\right)^2}_{\text{second order term}} \\
    ~&+\mathcal{O}\left(\frac{\alpha^3\log^3 n}{\mu^{3}_r(\Delta_{\vec{z}} n)^{1.5}}\right)\,.
\end{align*}
\end{claim}
\begin{proof}
For the Neumann series term in \eqref{eqn:oneside}:
\begin{align}\label{eq:neumann_series_term}
    \sum^\infty_{k=0}\left(-\left(\mat{Q}_\downarrow^\dagger\mat{Q}_\downarrow\right)^{-1}\left((\mat{E}^Q_\downarrow)^\dagger\mat{Q}_\downarrow+\mat{Q}_\downarrow^\dagger\mat{E}^Q_\downarrow+(\mat{E}^Q_\downarrow)^\dagger\mat{E}^Q_\downarrow\right)\right)^k\left(\mat{Q}_\downarrow^\dagger\mat{Q}_\downarrow\right)^{-1}\,,
\end{align}
we first bound the middle bracket:
\begin{align}\label{eq:mid_bracket_neumann}
&\left\|(\mat{E}^Q_\downarrow)^\dagger\mat{Q}_\downarrow+\mat{Q}_\downarrow^\dagger\mat{E}^Q_\downarrow+(\mat{E}^Q_\downarrow)^\dagger\mat{E}^Q_\downarrow\right\|_2\notag\\
\leq &~\left\|\mat{E}^Q_\downarrow\right\|_2\left(2\left\|\mat{Q}_\downarrow\right\|_2+\left\|\mat{E}^Q_\downarrow\right\|_2\right)\notag\\ \leq&~\left\|\mat{E}^Q\right\|_2\left(2\left\|\mat{Q}_r\right\|_2+\left\|\mat{E}^Q\right\|_2\right)\notag\\
=&~{\cal O}\left(\frac{\alpha\sqrt{\log n}}{\mu_r\sqrt{\Delta_z n}}\right)\,,
\end{align}
where the first step follows from triangle inequality, the second step is straightforward, the third step follows from the definition of $\mat{E}^Q$ and \eqref{eqn:second_Q_perturb}.

The first-order term (i.e., $k=1$) in \eqref{eq:neumann_series_term} can be approximated as follows:
\begin{equation}\label{eq:neumann_fo}
\begin{aligned}
    &\left\|\left(\mat{Q}_\downarrow^\dagger\mat{Q}_\downarrow\right)^{-1}\left((\mat{E}^Q_\downarrow)^\dagger\mat{Q}_\downarrow+\mat{Q}_\downarrow^\dagger\mat{E}^Q_\downarrow+(\mat{E}^Q_\downarrow)^\dagger\mat{E}^Q_\downarrow\right)\left(\mat{Q}_\downarrow^\dagger\mat{Q}_\downarrow\right)^{-1}+\left((\mat{E}^Q_\downarrow)^\dagger\mat{Q}_\downarrow+\mat{Q}_\downarrow^\dagger\mat{E}^Q_\downarrow-(\mat{E}^Q_\downarrow)^\dagger\mat{E}^Q_\downarrow\right)\right\|_2\\
    \leq &~ \left\|\left(\left(\mat{Q}_\downarrow^\dagger\mat{Q}_\downarrow\right)^{-1}-\mat{I}_r\right)\left((\mat{E}^Q_\downarrow)^\dagger\mat{Q}_\downarrow+\mat{Q}_\downarrow^\dagger\mat{E}^Q_\downarrow+(\mat{E}^Q_\downarrow)^\dagger\mat{E}^Q_\downarrow\right)\left(\mat{Q}_\downarrow^\dagger\mat{Q}_\downarrow\right)^{-1}\right\|_2\\
    ~&+\left\|\left((\mat{E}^Q_\downarrow)^\dagger\mat{Q}_\downarrow+\mat{Q}_\downarrow^\dagger\mat{E}^Q_\downarrow+(\mat{E}^Q_\downarrow)^\dagger\mat{E}^Q_\downarrow\right)\left(\left(\mat{Q}_\downarrow^\dagger\mat{Q}_\downarrow\right)^{-1}-\mat{I}_r\right)\right\|_2\\
    \leq &~ \mathcal{O}\left(\frac{1}{\Delta_{\vec{z}} n}\right)\cdot {\cal O}\left(\frac{\alpha\sqrt{\log n}}{\mu_r\sqrt{\Delta_z n}}\right)\\
    = &~ \mathcal{O}\left(\frac{\alpha \sqrt{\log n}}{\mu_r(\Delta_{\vec{z}}n)^{1.5}}\right)\,,
\end{aligned}
\end{equation}
where the first step follows from triangle inequality, the second step follows from \eqref{eqn:inverse_close_to_identity} and \eqref{eq:mid_bracket_neumann}, and the last step is by direct calculation.

By a similar calculation, we can show that the second-order term (i.e., $k=2$) in \eqref{eq:neumann_series_term} can be approximated by
\begin{equation}\label{eq:neumann_so_1}
\begin{aligned}
    &\left\|\left(\left(\mat{Q}_\downarrow^\dagger\mat{Q}_\downarrow\right)^{-1}\left((\mat{E}^Q_\downarrow)^\dagger\mat{Q}_\downarrow+\mat{Q}_\downarrow^\dagger\mat{E}^Q_\downarrow+(\mat{E}^Q_\downarrow)^\dagger\mat{E}^Q_\downarrow\right)\right)^2\left(\mat{Q}_\downarrow^\dagger\mat{Q}_\downarrow\right)^{-1}-\left((\mat{E}^Q_\downarrow)^\dagger\mat{Q}_\downarrow+\mat{Q}_\downarrow^\dagger\mat{E}^Q_\downarrow-(\mat{E}^Q_\downarrow)^\dagger\mat{E}^Q_\downarrow\right)^2\right\|_2\\
    \leq &~ \mathcal{O}\left(\frac{1}{\Delta_{\vec{z}} n}\right)\cdot {\cal O}\left(\frac{\alpha\sqrt{\log n}}{\mu_r\sqrt{\Delta_z n}}\right)^2\\
    = &~ \mathcal{O}\left(\frac{\alpha^2 \log n}{\mu^{2}_r\Delta_{\vec{z}}^2n^2}\right)\,.
\end{aligned}
\end{equation}
Indeed, this term can be further simplified:
\begin{equation}\label{eq:neumann_so_2}
\begin{aligned}
    &\left\|\left((\mat{E}^Q_\downarrow)^\dagger\mat{Q}_\downarrow+\mat{Q}_\downarrow^\dagger\mat{E}^Q_\downarrow-(\mat{E}^Q_\downarrow)^\dagger\mat{E}^Q_\downarrow\right)^2 - \left((\mat{E}^Q_\downarrow)^\dagger\mat{Q}_\downarrow+\mat{Q}_\downarrow^\dagger\mat{E}^Q_\downarrow\right)^2\right\|_2\\
    = &~ \left\|\left((\mat{E}^Q_\downarrow)^\dagger\mat{Q}_\downarrow+\mat{Q}_\downarrow^\dagger\mat{E}^Q_\downarrow\right)(\mat{E}^Q_\downarrow)^\dagger\mat{E}^Q_\downarrow + (\mat{E}^Q_\downarrow)^\dagger\mat{E}^Q_\downarrow\left((\mat{E}^Q_\downarrow)^\dagger\mat{Q}_\downarrow+\mat{Q}_\downarrow^\dagger\mat{E}^Q_\downarrow\right) \right\|_2\\
    \leq &~ {\cal O}\left(\frac{\alpha\sqrt{\log n}}{\mu_r\sqrt{\Delta_z n}}\right)^3\\
    = &~ {\cal O}\left(\frac{\alpha^3\log^3 n}{\mu^{3}_r(\Delta_{\vec{z}}n)^{1.5}}\right)\,,
\end{aligned}
\end{equation}
where the first inequality follows from \eqref{eqn:second_Q_perturb} for $\|\mat{E}^Q_\downarrow\|_2$ and \eqref{eqn:inverse_close_to_identity} for $\|\mat{Q}_\downarrow\|_2$.

The higher-order terms (i.e., $k>2$) in \eqref{eq:neumann_series_term} can be directly bounded as follows:
\begin{equation}\label{eq:neumann_ho}
\begin{aligned}
    &\left\|\sum^\infty_{k=3}\left(-\left(\mat{Q}_\downarrow^\dagger\mat{Q}_\downarrow\right)^{-1}\left((\mat{E}^Q_\downarrow)^\dagger\mat{Q}_\downarrow+\mat{Q}_\downarrow^\dagger\mat{E}^Q_\downarrow+(\mat{E}^Q_\downarrow)^\dagger\mat{E}^Q_\downarrow\right)\right)^k\left(\mat{Q}_\downarrow^\dagger\mat{Q}_\downarrow\right)^{-1}\right\|_2\\
    \leq &~ \sum_{k=3}^\infty \left\|(\mat{Q}_\downarrow^\dagger\mat{Q}_\downarrow)^{-1}\right\|_2^k \left\|(\mat{E}^Q_\downarrow)^\dagger\mat{Q}_\downarrow+\mat{Q}_\downarrow^\dagger\mat{E}^Q_\downarrow+(\mat{E}^Q_\downarrow)^\dagger\mat{E}^Q_\downarrow\right\|_2^k\cdot \left\|(\mat{Q}_\downarrow^\dagger\mat{Q}_\downarrow)^{-1}\right\|_2\\
    = &~ \sum_{k=3}^\infty {\cal O}\left(\frac{\alpha\sqrt{\log n}}{\mu_r\sqrt{\Delta_z n}}\right)^{k}\\
    = &~ {\cal O}\left(\frac{\alpha^3\log^3 n}{\mu^{3}_r(\Delta_{\vec{z}}n)^{1.5}}\right)\,,
\end{aligned}
\end{equation}
where the first step follows from triangle inequality, the second step follows from \eqref{eqn:inverse_close_to_identity} and \eqref{eq:mid_bracket_neumann}, and the last step follows from the geometric summation.

Combining \eqref{eq:neumann_fo} to \eqref{eq:neumann_ho} together, we get that
\[
\begin{aligned}
    &\sum^\infty_{k=0}\left(-\left(\mat{Q}_\downarrow^\dagger\mat{Q}_\downarrow\right)^{-1}\left((\mat{E}^Q_\downarrow)^\dagger\mat{Q}_\downarrow+\mat{Q}_\downarrow^\dagger\mat{E}^Q_\downarrow+(\mat{E}^Q_\downarrow)^\dagger\mat{E}^Q_\downarrow\right)\right)^k\left(\mat{Q}_\downarrow^\dagger\mat{Q}_\downarrow\right)^{-1}\\
    =&~\left(\mat{Q}_\downarrow^\dagger\mat{Q}_\downarrow\right)^{-1}-\left((\mat{E}^Q_\downarrow)^\dagger\mat{Q}_\downarrow+\mat{Q}_\downarrow^\dagger\mat{E}^Q_\downarrow+(\mat{E}^Q_\downarrow)^\dagger\mat{E}^Q_\downarrow\right)+\left((\mat{E}^Q_\downarrow)^\dagger\mat{Q}_\downarrow+\mat{Q}_\downarrow^\dagger\mat{E}^Q_\downarrow\right)^2+\mathcal{O}\left(\frac{\alpha^3\log^3 n}{\mu^{3}_r(\Delta_{\vec{z}} n)^{1.5}}\right)\,.
\end{aligned}
\]
The claim is then proved.
\end{proof}

\subsection{Proofs for Step 3: Error cancellation in the Taylor expansion}\label{sec:align_step3_defer}

\subsubsection{Establishing the first equation in \texorpdfstring{\cref{eqn:key_equality}}{Eq}}\label{sec:defer_key_eq_1st}

\decomposeI*

\begin{proof}
In this proof, we often use the following observations:
\begin{subequations}
\begin{equation}\label{eqn:Q_up_simplify_1st}
\mat{P}_v(\mat{Q}_\downarrow)^\dagger\mat{Q}_\uparrow\mat{P}^\dagger_v=\diag(\zdom^{-1})+\mathcal{O}\left(\frac{1}{\Delta_{\vec{z}} n}\right)\,,
\end{equation}
\begin{equation}\label{eqn:Q_up_simplify_2nd}
\mat{P}_v\mat{\Sigma}^{-k}_r\mat{P}^\dagger_v=\left(n\diag(\mudom)\right)^{-k}+\mathcal{O}\left(\frac{1}{\mu_r n}\right)^{k+1}\frac{1}{\Delta_{\vec{z}} }  \quad \forall k\geq 0\,.
\end{equation}
\end{subequations}
They are proved in \cref{clm:P_v_properties}.%

Plugging the formula of $\mat{E}^Q$ into the first equation of \eqref{eqn:key_equality}, we obtain \cref{eq:EPPV_lem}:%
\begin{equation*}
\begin{aligned}
&\left\|\mat{E}^Q_\uparrow-\mat{E}^Q_\downarrow(\mat{Q}_\downarrow)^\dagger\mat{Q}_\uparrow\right\|_2\\
\leq &~ \left\|\mat{E}^Q_\uparrow\mat{P}^\dagger_v-\mat{E}^Q_\downarrow(\mat{Q}_\downarrow)^\dagger\mat{Q}_\uparrow\mat{P}^\dagger_v\right\|_2\cdot \Big\|(\mat{P}_v^\dagger)^{-1}\Big\|_2\\
=&~\mathcal{O}\left(\left\|\mat{E}^Q_\uparrow\mat{P}^\dagger_v-\mat{E}^Q_\downarrow(\mat{Q}_\downarrow)^\dagger\mat{Q}_\uparrow\mat{P}^\dagger_v\right\|_2\right)\\
\leq &~ \mathcal{O}\left(\left\|\mat{E}^Q_\uparrow\mat{P}\mat{P}^\dagger_v-\mat{E}^Q_\downarrow\mat{P}(\mat{Q}_\downarrow)^\dagger\mat{Q}_\uparrow\mat{P}^\dagger_v\right\|_2+\left\|\mat{E}^Q_\uparrow\left(\mat{P}-\mat{I}_r\right)\mat{P}_v^\dagger\right\|_2+\left\|\mat{E}^Q_\downarrow\left(\mat{P}-\mat{I}_r\right)(\mat{Q}_\downarrow)^\dagger\mat{Q}_\uparrow\mat{P}_v^\dagger\right\|_2\right)\\
\leq &~ \mathcal{O}\left(\left\|\mat{E}^Q_\uparrow\mat{P}\mat{P}^\dagger_v-\mat{E}^Q_\downarrow\mat{P}(\mat{Q}_\downarrow)^\dagger\mat{Q}_\uparrow\mat{P}^\dagger_v\right\|_2\right) + {\cal O}\left(\left\|\mat{E}^Q\right\|_2\left\|\mat{P}-\mat{I}_r\right\|_2\left(1+\big\|(\mat{Q}_\downarrow)^\dagger\mat{Q}_\uparrow\big\|_2\right)\right)\\
= &~ \mathcal{O}\left(\left\|\mat{E}^Q_\uparrow\mat{P}\mat{P}^\dagger_v-\mat{E}^Q_\downarrow\mat{P}(\mat{Q}_\downarrow)^\dagger\mat{Q}_\uparrow\mat{P}^\dagger_v\right\|_2\right)+\mathcal{O}\left(\frac{\alpha\sqrt{\log n}}{\mu_r\sqrt{\Delta_{\vec{z}}n}}\right)\cdot \mathcal{O}\left(\frac{\alpha\sqrt{\log n}}{\mu_r\sqrt{\Delta_{\vec{z}}n}}\right)\cdot {\cal O}(1)\\
= &~ \mathcal{O}\left(\left\|\mat{E}^Q_\uparrow\mat{P}\mat{P}^\dagger_v-\mat{E}^Q_\downarrow\mat{P}(\mat{Q}_\downarrow)^\dagger\mat{Q}_\uparrow\mat{P}^\dagger_v\right\|_2\right)+\mathcal{O}\left(\frac{\alpha^2\log n}{\mu^{2}_r\Delta_{\vec{z}} n}\right)\\
\leq &~ \mathcal{O}\left(\left\|(\widetilde{\mat{E}}^Q_1)_\uparrow\mat{P}\mat{P}^\dagger_v-(\widetilde{\mat{E}}^Q_1)_\downarrow\mat{P}(\mat{Q}_\downarrow)^\dagger\mat{Q}_\uparrow\mat{P}^\dagger_v\right\|_2+\left\|(\widetilde{\mat{E}}^Q_2)_\uparrow\mat{P}\mat{P}^\dagger_v-(\widetilde{\mat{E}}^Q_2)_\downarrow\mat{P}(\mat{Q}_\downarrow)^\dagger\mat{Q}_\uparrow\mat{P}^\dagger_v\right\|_2\right)+\mathcal{O}\left(\frac{\alpha^2\log n}{\mu^{2}_r\Delta_{\vec{z}} n}\right)\\
\leq &~ \mathcal{O}\left(\left\|(\widetilde{\mat{E}}^Q_1)_\uparrow\mat{P}\mat{P}^\dagger_v-(\widetilde{\mat{E}}^Q_1)_\downarrow\mat{P}(\mat{Q}_\downarrow)^\dagger\mat{Q}_\uparrow\mat{P}^\dagger_v\right\|_2+\left\|\widetilde{\mat{E}}^Q_2\right\|_2\right)+\mathcal{O}\left(\frac{\alpha^2\log n}{\mu^{2}_r\Delta_{\vec{z}} n}\right)\\
=&~\mathcal{O}\left(\Big\|\left(\widetilde{\mat{E}}^Q_1\mat{P}\mat{P}^\dagger_v\right)_\uparrow-\left(\widetilde{\mat{E}}^Q_1\mat{P}\right)_\downarrow(\mat{Q}_\downarrow)^\dagger\mat{Q}_\uparrow\mat{P}^\dagger_v\Big\|_2\right)+\mathcal{O}\left(\frac{\alpha^2\log n}{\mu^{2}_r\Delta_{\vec{z}} n}\right)\,,
\end{aligned}
\end{equation*}
where the first step follows from multiplying $\mat{P}_v^\dagger (\mat{P}_v^\dagger)^{-1}$ on the right and pulling out $\|(\mat{P}_v^\dagger)^{-1}\|_2$, the second step follows from $\|\mat{P}_v^{-1}\|_2={\cal O}(1)$ by \eqref{eqn:P_orthogonal}, the third step follows from triangle inequality, the fourth step follows from $\|\mat{P}_v^{-1}\|_2={\cal O}(1)$ again and $\|\mat{E}^Q_\uparrow\|_2,\|\mat{E}^Q_\downarrow\|_2\leq \|\mat{E}^Q\|_2$, the fifth step follows from \eqref{eqn:second_Q_perturb}, \eqref{eqn:P_bound}, and $\|\mat{Q}_r\|_2=1$, the sixth step follows by direct calculation, the seventh step follows from $\mat{E}^Q=\widetilde{\mat{E}}_1^Q+\widetilde{\mat{E}}_2^Q$ by \eqref{eq:E_Q_up_down_2}, where $\widetilde{\mat{E}}^Q_1, \widetilde{\mat{E}}^Q_2$ are defined in \eqref{eqn:E_tilde_Q}, the eighth step follows from $\|\mat{P}\|_2={\cal O}(1)$, $\|\mat{P}_v\|_2={\cal O}(1)$, and $\|\mat{Q}_r\|_2=1$, and the last step uses the bound for $\|\widetilde{\mat{E}}^Q_2\|_2$ derived from its definition \eqref{eqn:E_tilde_Q}.

We first consider $\left(\widetilde{\mat{E}}^Q_1\mat{P}\mat{P}^\dagger_v\right)_\uparrow$. Notice that
\begin{align*}
\widetilde{\mat{E}}^Q_1\mat{P}\mat{P}^\dagger_v=&~
\sum^\infty_{k=0}\left(\mat{I}_n-\frac{1}{n}\mat{V}_n(\zdom)\mat{V}_n(\zdom)^\dagger\right)\mat{E}_1^k\mat{E}_2\mat{Q}_r\left(\mat{\Sigma}_r^{-1}\right)^{k+1}\mat{P}^\dagger_v\\
=&~\sum^\infty_{k=0}\left(\mat{I}_n-\frac{1}{n}\mat{V}_n(\zdom)\mat{V}_n(\zdom)^\dagger\right)\mat{E}_1^k\mat{E}_2\mat{Q}_r\cdot \left(\mat{P}^\dagger_v\mat{P}_v\right)\cdot \left(\mat{\Sigma}_r^{-1}\right)^{k+1}\mat{P}^\dagger_v\\
~&+\sum^\infty_{k=0}\left(\mat{I}_n-\frac{1}{n}\mat{V}_n(\zdom)\mat{V}_n(\zdom)^\dagger\right)\mat{E}_1^k\mat{E}_2\mat{Q}_r\left(\mat{I}_r-\mat{P}^\dagger_v\mat{P}_v\right)\left(\mat{\Sigma}_r^{-1}\right)^{k+1}\mat{P}^\dagger_v,
\end{align*}
where the first step follows from the definition of $\widetilde{\mat{E}}^Q_1$.
To bound the second term, for any $k\geq 0$, we have
\begin{align*}
    &\left\|\left(\mat{I}_n-\frac{1}{n}\mat{V}_n(\zdom)\mat{V}_n(\zdom)^\dagger\right)\mat{E}_1^k\mat{E}_2\mat{Q}_r\cdot \left(\mat{I}_r-\mat{P}^\dagger_v\mat{P}_v\right)\cdot \left(\mat{\Sigma}_r^{-1}\right)^{k+1}\mat{P}^\dagger_v\right\|_2\\
    \leq &~ \left\|\mat{I}_n-\frac{1}{n}\mat{V}_n(\zdom)\mat{V}_n(\zdom)^\dagger\right\|_2\cdot \|\mat{E}_1\|_2^k\cdot \|\mat{E}_2\|_2\cdot \|\mat{Q}_r\|_2\cdot \left\|\mat{I}_r-\mat{P}^\dagger_v\mat{P}_v\right\|_2\cdot \left\|\mat{\Sigma}_r^{-1}\right\|_2^{k+1}\cdot \|\mat{P}_v^\dagger\|_2\\
    \leq &~ {\cal O}(1)\cdot (n\mu_{\rm tail})^k\cdot {\cal O}(\alpha\sqrt{n\log n})\cdot 1\cdot \mathcal{O}\left(\frac{1}{\Delta_{\vec{z}} n}\right)\cdot {\cal O}(n\mu_r)^{-(k+1)}\cdot {\cal O}(1)\\
    = &~ {\cal O}\left(\frac{\mu_{\rm tail}}{\mu_r}\right)\cdot \frac{\alpha\sqrt{\log n}}{\mu_r\Delta_{\vec{z}} n^{1.5}}
\end{align*}
where the third step uses \cref{lem:vander-eigen} for the first term, \eqref{eqn:bound_tail} for the second term, \cref{lem:Mecks} for the third term, \eqref{eqn:P_orthogonal} for the fifth and seventh terms, \eqref{eqn:eigenvalue_prop} for the sixth term. Thus,
\begin{align*}
\widetilde{\mat{E}}^Q_1\mat{P}\mat{P}^\dagger_v\leq &~
\sum^\infty_{k=0}\left(\mat{I}_n-\frac{1}{n}\mat{V}_n(\zdom)\mat{V}_n(\zdom)^\dagger\right)\mat{E}_1^k\mat{E}_2\mat{Q}_r\mat{P}^\dagger_v\mat{P}_v\left(\mat{\Sigma}_r^{-1}\right)^{k+1}\mat{P}^\dagger_v\\
~&+\sum_{k=0}^\infty {\cal O}\left(\frac{\mu_{\rm tail}}{\mu_r}\right)\cdot \frac{\alpha\sqrt{\log n}}{\mu_r\Delta_{\vec{z}} n^{1.5}}\\
= &~ \sum^\infty_{k=0}\left(\mat{I}_n-\frac{1}{n}\mat{V}_n(\zdom)\mat{V}_n(\zdom)^\dagger\right)\mat{E}_1^k\mat{E}_2\mat{Q}_r\mat{P}^\dagger_v\mat{P}_v\left(\mat{\Sigma}_r^{-1}\right)^{k+1}\mat{P}^\dagger_v\\
~&+ {\cal O}\left(\frac{\alpha\sqrt{\log n}}{\mu_r\Delta_{\vec{z}} n^{1.5}}\right)\\
=~&\sum^\infty_{k=0}\left(\mat{I}_n-\frac{1}{n}\mat{V}_n(\zdom)\mat{V}_n(\zdom)^\dagger\right)\mat{E}_1^k\mat{E}_2\left(\frac{1}{\sqrt{n}}\mat{V}_n(\zdom)\right)\left(n\diag(\mudom)\right)^{-(k+1)}\\
~&+\sum^\infty_{k=0}\left(\mat{I}_n-\frac{1}{n}\mat{V}_n(\zdom)\mat{V}_n(\zdom)^\dagger\right)\mat{E}_1^k\mat{E}_2\left(\frac{1}{\sqrt{n}}\mat{V}_n(\zdom)\right)\\
~&\cdot\left(\mat{P}_v\left(\mat{\Sigma}_r^{-1}\right)^{k+1}\mat{P}^\dagger_v-\left(n\diag(\mudom)\right)^{-(k+1)}\right)\\
~&+\mathcal{O}\left(\frac{\alpha\sqrt{\log n}}{\mu_r\Delta_{\vec{z}} n^{1.5}}\right)\,.
\end{align*}
To bound the second term in the above equation, for any $k\geq 0$, we have
\begin{align*}
    &\left\|\left(\mat{I}_n-\frac{1}{n}\mat{V}_n(\zdom)\mat{V}_n(\zdom)^\dagger\right)\mat{E}_1^k\mat{E}_2\left(\frac{1}{\sqrt{n}}\mat{V}_n(\zdom)\right)\left(\mat{P}_v\left(\mat{\Sigma}_r^{-1}\right)^{k+1}\mat{P}^\dagger_v-\left(n\diag(\mudom)\right)^{-(k+1)}\right)\right\|_2\\
    \leq &~ {\cal O}(1)\cdot \|\mat{E}_1\|_2^k\cdot \|\mat{E}_2\|_2 \cdot \left\|\mat{P}_v\left(\mat{\Sigma}_r^{-1}\right)^{k+1}\mat{P}^\dagger_v-\left(n\diag(\mudom)\right)^{-(k+1)}\right\|_2\\
    \leq &~ (n\mu_{\rm tail})^k\cdot {\cal O}(\alpha\sqrt{n\log n})\cdot \mathcal{O}\left(\frac{1}{\mu_r n}\right)^{k+2}\frac{1}{\Delta_{\vec{z}} }\\
    = &~ {\cal O}\left(\frac{\mu_{\rm tail}}{\mu_r}\right)^k \frac{\alpha\sqrt{\log n}}{\mu_r^2\Delta_{\vec{z}} n^{1.5}}
\end{align*}
where the first step uses \cref{lem:vander-eigen} to bound $\|\mat{I}_n-\frac{1}{n}\mat{V}_n(\zdom)\mat{V}_n(\zdom)^\dagger\|_2$ and $\|\frac{1}{\sqrt{n}}\mat{V}_n(\zdom)\|_2$, the second step uses \eqref{eqn:bound_tail}, \cref{lem:Mecks}, \eqref{eqn:Q_up_simplify_2nd} to bound the three terms, respectively. Therefore, we obtain that
\begin{equation}\label{eqn:EPPV}
\begin{aligned}
\widetilde{\mat{E}}^Q_1\mat{P}\mat{P}^\dagger_v\leq &~ \sum^\infty_{k=0}\left(\mat{I}_n-\frac{1}{n}\mat{V}_n(\zdom)\mat{V}_n(\zdom)^\dagger\right)\mat{E}_1^k\mat{E}_2\left(\frac{1}{\sqrt{n}}\mat{V}_n(\zdom)\right)\left(n\diag(\mudom)\right)^{-(k+1)}\\
~&+\sum_{k=0}^\infty {\cal O}\left(\frac{\mu_{\rm tail}}{\mu_r}\right)^k \frac{\alpha\sqrt{\log n}}{\mu_r^2\Delta_{\vec{z}} n^{1.5}} + \mathcal{O}\left(\frac{\alpha\sqrt{\log n}}{\mu_r\Delta_{\vec{z}} n^{1.5}}\right)\\
= &~ \sum^\infty_{k=0}\left(\mat{I}_n-\frac{1}{n}\mat{V}_n(\zdom)\mat{V}_n(\zdom)^\dagger\right)\mat{E}_1^k\mat{E}_2\left(\frac{1}{\sqrt{n}}\mat{V}_n(\zdom)\right)\left(n\diag(\mudom)\right)^{-(k+1)}\\
~&+ \mathcal{O}\left(\frac{\alpha\sqrt{\log n}}{\mu_r^2\Delta_{\vec{z}} n^{1.5}}\right)\,.
\end{aligned}
\end{equation}

Similarly, we also have
\begin{equation}\label{eqn:EPPV_2}
\begin{aligned}
&\left(\widetilde{\mat{E}}^Q_1\mat{P}\right)_\downarrow(\mat{Q}_\downarrow)^\dagger\mat{Q}_\uparrow\mat{P}^\dagger_v\\
= &~ \left(\widetilde{\mat{E}}^Q_1\mat{P}\right)_\downarrow\cdot \mat{P}^\dagger_v\mat{P}_v\cdot (\mat{Q}_\downarrow)^\dagger\mat{Q}_\uparrow\mat{P}^\dagger_v + \left(\widetilde{\mat{E}}^Q_1\mat{P}\right)_\downarrow\left(\mat{I}_r-\mat{P}_v^\dagger \mat{P}_v\right)(\mat{Q}_\downarrow)^\dagger\mat{Q}_\uparrow\mat{P}^\dagger_v\\
\leq &~ \left(\widetilde{\mat{E}}^Q_1\mat{P}\right)_\downarrow\cdot \mat{P}^\dagger_v\mat{P}_v\cdot (\mat{Q}_\downarrow)^\dagger\mat{Q}_\uparrow\mat{P}^\dagger_v + {\cal O}\left(\frac{\alpha\sqrt{\log n}}{\mu_r\sqrt{n}}\right)\cdot \mathcal{O}\left(\frac{1}{\Delta_{\vec{z}} n}\right)\\
= &~ \left(\widetilde{\mat{E}}^Q_1\mat{P}\mat{P}^\dagger_v\right)_\downarrow\mat{P}_v(\mat{Q}_\downarrow)^\dagger\mat{Q}_\uparrow\mat{P}^\dagger_v+\mathcal{O}\left(\frac{\alpha\sqrt{\log n}}{\mu_r\Delta_{\vec{z}} n^{1.5}}\right)\\
\leq &~ \left(\widetilde{\mat{E}}^Q_1\mat{P}\mat{P}^\dagger_v\right)_\downarrow \diag(\zdom^{-1}) +\left\|\widetilde{\mat{E}}^Q_1\mat{P}\mat{P}^\dagger_v\right\|_2\cdot {\cal O}\left(\frac{1}{\Delta_{\vec{z}}n}\right) + \mathcal{O}\left(\frac{\alpha\sqrt{\log n}}{\mu_r\Delta_{\vec{z}} n^{1.5}}\right)\\
= &~ \left(\widetilde{\mat{E}}^Q_1\mat{P}\mat{P}^\dagger_v\right)_\downarrow \diag(\zdom^{-1}) + \mathcal{O}\left(\frac{\alpha\sqrt{\log n}}{\mu_r\Delta_{\vec{z}} n^{1.5}}\right)\\
\leq &~\left(\sum^\infty_{k=0}\left(\mat{I}_n-\frac{1}{n}\mat{V}_n(\zdom)\mat{V}_n(\zdom)^\dagger\right)\mat{E}_1^k\mat{E}_2\left(\frac{1}{\sqrt{n}}\mat{V}_n(\zdom)\right)\left(n\diag(\mudom)\right)^{-(k+1)}\right)_\downarrow\diag(\zdom^{-1})\\
~&+\left\|\diag(\zdom^{-1})\right\|_2\cdot \mathcal{O}\left(\frac{\alpha\sqrt{\log n}}{\mu_r^2\Delta_{\vec{z}} n^{1.5}}\right)+\mathcal{O}\left(\frac{\alpha\sqrt{\log n}}{\mu_r\Delta_{\vec{z}} n^{1.5}}\right)\\
= &~ \left(\sum^\infty_{k=0}\left(\mat{I}_n-\frac{1}{n}\mat{V}_n(\zdom)\mat{V}_n(\zdom)^\dagger\right)\mat{E}_1^k\mat{E}_2\left(\frac{1}{\sqrt{n}}\mat{V}_n(\zdom)\right)\left(n\diag(\mudom)\right)^{-(k+1)}\right)_\downarrow\diag(\zdom^{-1})\\
~&+ \mathcal{O}\left(\frac{\alpha\sqrt{\log n}}{\mu^{2}_r\Delta_{\vec{z}} n^{1.5}}\right)\,,
\end{aligned}
\end{equation}
where the first step follows from inserting $\mat{P}_v^\dagger \mat{P}_v$ in the middle, the second step follows from \eqref{eq:bound_E1_Q_wt} for $\|\widetilde{\mat{E}}^Q_1\|_2$ and \eqref{eqn:P_orthogonal} for $\|\mat{I}_r-\mat{P}_v^\dagger \mat{P}_v\|_2$, the third step follows from re-grouping terms, the fourth step follows from \eqref{eqn:Q_up_simplify_1st}, the fifth step follows from \eqref{eq:bound_E1_Q_wt}, the sixth step follows from \eqref{eqn:EPPV}, and the last step follows from $\|\diag(\zdom^{-1})\|_2=1$.%

We also notice the following bound to get rid of some terms involving $\frac{1}{n}\mat{V}_n(\zdom)\mat{V}_n(\zdom)^\dagger$ and $\mat{E}^k_1$ for $k\geq 1$:
\begin{equation}\label{eqn:V_E_vanish}
\begin{aligned}
&\left\|\sum^\infty_{k=1}\frac{1}{n}\mat{V}_n(\zdom)\mat{V}_n(\zdom)^\dagger\mat{E}_1^k\mat{E}_2\left(\frac{1}{\sqrt{n}}\mat{V}_n(\zdom)\right)\left(n\diag(\mudom)\right)^{-(k+1)}\right\|_2\\
\leq &~\sum^\infty_{k=1}\left\|\frac{1}{\sqrt{n}}\mat{V}_n(\zdom)\right\|_2\left\|\frac{1}{\sqrt{n}}\mat{V}_n(\zdom)^\dagger\mat{E}_1\right\|_2\left\|\mat{E}_1\right\|_2^{k-1}\left\|\mat{E}_2\right\|_2\left\|\frac{1}{\sqrt{n}}\mat{V}_n(\zdom)\right\|_2\left\|\left(n\diag(\mudom)\right)^{-1}\right\|_2^{k+1}\\
\leq &~\sum^\infty_{k=1}\mathcal{O}\left(\frac{\mu_{\rm tail}}{\Delta_{\vec{z}}}\right)\cdot (n\mu_{\rm tail})^{k-1}\cdot \mathcal{O}\left(\alpha\sqrt{n\log n}\right)\cdot \mathcal{O}\left(\frac{1}{(n\mu_r)^{k+1}}\right)\\
= &~ \sum_{k=0}^\infty \left(\frac{\mu_{\rm tail}}{\mu_r}\right)^k \cdot {\cal O}\left(\frac{\mu_{\rm tail}\cdot \alpha\sqrt{\log n}}{\mu_r^2 \Delta_{\vec{z}}n^{1.5}}\right)\\
=&~ \mathcal{O}\left(\frac{\alpha\sqrt{\log n}}{\mu_r\Delta_{\vec{z}} n^{1.5}}\right)\,,
\end{aligned}
\end{equation}
where in the third step, we use  \cref{cor:V_dagger_V_dom} to bound $\left\|\frac{1}{\sqrt{n}}\mat{V}_n(\zdom)\right\|_2$, use \cref{clm:bound_V_dom_E1} to bound  $\left\|\frac{1}{\sqrt{n}}\mat{V}_n(\zdom)^\dagger\mat{E}_1\right\|_2$, use \eqref{eqn:bound_tail} to bound $\left\|\mat{E}^{k-1}_1\right\|_2$, use \cref{lem:Mecks} to bound $\|\mat{E}_2\|_2$, and use the fact that $\|\diag(\mu_{\rm dom})^{-1}\|_2=\mu_r^{-1}$.

Plugging \eqref{eqn:V_E_vanish} into \eqref{eqn:EPPV} and \eqref{eqn:EPPV_2}, we have
\begin{align*}
\widetilde{\mat{E}}^Q_1\mat{P}\mat{P}^\dagger_v=
&~\sum^\infty_{k=0}\mat{E}_1^k\mat{E}_2\left(\frac{1}{\sqrt{n}}\mat{V}_n(\zdom)\right)\left(n\diag(\mudom)\right)^{-(k+1)}\\
~&- \frac{1}{n}\mat{V}_n(\zdom)\mat{V}_n(\zdom)^\dagger \mat{E}_2\left(\frac{1}{\sqrt{n}}\mat{V}_n(\zdom)\right) \left(n\diag(\mudom)\right)^{-1}\\
~&+\mathcal{O}\left(\frac{\alpha\sqrt{\log n}}{\mu^{2}_r\Delta_{\vec{z}} n^{1.5}}\right)\,,
\end{align*}    
and
\begin{align*}
&\left(\widetilde{\mat{E}}^Q_1\mat{P}\right)_\downarrow(\mat{Q}_\downarrow)^\dagger\mat{Q}_\uparrow\mat{P}^\dagger_v\\
=&~\left(\sum^\infty_{k=0}\mat{E}_1^k\mat{E}_2\left(\frac{1}{\sqrt{n}}\mat{V}_n(\zdom)\right)\left(n\diag(\mudom)\right)^{-(k+1)}\right)_\downarrow\diag(\zdom^{-1})\\
~&-\left(\frac{1}{n}\mat{V}_n(\zdom)\mat{V}_n(\zdom)^\dagger\mat{E}_2\left(\frac{1}{\sqrt{n}}\mat{V}_n(\zdom)\right)\left(n\diag(\mudom)\right)^{-1}\right)_\downarrow\diag(\zdom^{-1})\\
~&+\mathcal{O}\left(\frac{\alpha\sqrt{\log n}}{\mu^{2}_r\Delta_{\vec{z}} n^{1.5}}\right)\,,
\end{align*}
where the $\frac{1}{n}\mat{V}_n(\zdom)\mat{V}_n(\zdom)^\dagger$ terms come from the $k=0$ case in \eqref{eqn:EPPV} and \eqref{eqn:EPPV_2}.

The proof of the lemma is completed.
\end{proof}

\finegrainedI*
\begin{proof}
We analyze \eqref{eqn:EPPV_1_all} and \eqref{eqn:EPPV_2_all} column-by-column. %

For convenience, we define $\lambda_i\in \C$ and $\mat{F}_i\in \C^{n\times n}$ for each $1\leq i \leq r$ as follows:
\begin{align}
\lambda_i:=&~ n\mu_i,\\
\mat{F}_i:=&~\frac{1}{\lambda_i}\left(\sum^\infty_{k=0}\left(\frac{\mat{E}_1}{\lambda_i}\right)^k-\frac{1}{n}\mat{V}_n(\zdom)\mat{V}_n(\zdom)^\dagger\right)\mat{E}_2\,.
\end{align}
Observe that 
\begin{align*}
    &\mat{E}_1^k\mat{E}_2\left(\frac{1}{\sqrt{n}}\mat{V}_n(\zdom)\right)\left(n\diag(\mudom)\right)^{-(k+1)}\\
    = &~ \mat{E}_1^k\mat{E}_2 \cdot \left(\frac{1}{\sqrt{n}}\begin{bmatrix}
        \vec{v}_n(z_1) & \cdots & \vec{v}_n(z_r)
    \end{bmatrix}
    \cdot \diag(\vec{\lambda}^{-(k+1)})\right)\\
    = &~ \mat{E}_1^k\mat{E}_2 \cdot \frac{1}{\sqrt{n}}\left[
        \frac{\vec{v}_n(z_1)}{\lambda_1^{k+1}} \quad \cdots \quad \frac{\vec{v}_n(z_r)}{\lambda_r^{k+1}}
    \right]\\
    = &~ \left[\frac{1}{\lambda_1}\left(\frac{\mat{E}_1}{\lambda_1}\right)^k\mat{E}_2\frac{\vec{v}_n(z_1)}{\sqrt{n}} \quad \cdots \quad \frac{1}{\lambda_r}\left(\frac{\mat{E}_1}{\lambda_r}\right)^k\mat{E}_2\frac{\vec{v}_n(z_r)}{\sqrt{n}}\right]\,. 
\end{align*}
where
\begin{align*}
    \vec{v}_n(z_i):=\begin{bmatrix}1 & z_i & z_i^2 & \cdots & z_i^{n-1}\end{bmatrix}^\top\in \C^n\,.
\end{align*}
Similarly,
\begin{align*}
    &\frac{1}{n}\mat{V}_n(\zdom)\mat{V}_n(\zdom)^\dagger \mat{E}_2\left(\frac{1}{\sqrt{n}}\mat{V}_n(\zdom)\right) \left(n\diag(\mudom)\right)^{-1} \\
    = &~\frac{1}{n}\mat{V}_n(\zdom)\mat{V}_n(\zdom)^\dagger \mat{E}_2 \left[\frac{1}{\lambda_1}\frac{\vec{v}_n(z_1)}{\sqrt{n}}\quad \cdots \quad \frac{1}{\lambda_r}\frac{\vec{v}_n(z_r)}{\sqrt{n}}\right]\,.
\end{align*}
Combining the above two equations together and summing over $k$, we obtain that
\begin{align*}
    &\sum^\infty_{k=0}\mat{E}_1^k\mat{E}_2\left(\frac{1}{\sqrt{n}}\mat{V}_n(\zdom)\right)\left(n\diag(\mudom)\right)^{-(k+1)}\\
    ~&- \frac{1}{n}\mat{V}_n(\zdom)\mat{V}_n(\zdom)^\dagger \mat{E}_2\left(\frac{1}{\sqrt{n}}\mat{V}_n(\zdom)\right) \left(n\diag(\mudom)\right)^{-1}\\
    = &~ \left[\frac{\mat{F}_1\vec{v}_n(z_1)}{\sqrt{n}}\quad \frac{\mat{F}_2\vec{v}_n(z_2)}{\sqrt{n}}\quad \cdots\quad \frac{\mat{F}_r\vec{v}_n(z_r)}{\sqrt{n}}\right]\,.
\end{align*}
By \eqref{eqn:EPPV_1_all} and \eqref{eqn:EPPV_2_all}, we have
\begin{equation}\label{eqn:E_decomposition}
\begin{aligned}
&\left(\widetilde{\mat{E}}^Q_1\mat{P}\mat{P}^\dagger_v\right)_\uparrow-\left(\widetilde{\mat{E}}^Q_1\mat{P}\right)_\downarrow(\mat{Q}_\downarrow)^\dagger\mat{Q}_\uparrow\mat{P}^\dagger_v\\
=&~\left[\frac{\mat{F}_1\vec{v}_n(z_1)}{\sqrt{n}}\quad \frac{\mat{F}_2\vec{v}_n(z_2)}{\sqrt{n}}\quad \cdots\quad \frac{\mat{F}_r\vec{v}_n(z_r)}{\sqrt{n}}\right]_{\uparrow}\\
~&-\left[\frac{\mat{F}_1\vec{v}_n(z_1)}{\sqrt{n}}\quad \frac{\mat{F}_2\vec{v}_n(z_2)}{\sqrt{n}}\quad \cdots\quad \frac{\mat{F}_r\vec{v}_n(z_r)}{\sqrt{n}}\right]_\downarrow\diag(\zdom^{-1})\\
~&+\mathcal{O}\left(\frac{\alpha\sqrt{\log n}}{\mu^{2}_r\Delta_{\vec{z}} n^{1.5}}\right)  
\end{aligned}
\end{equation}

Without loss of generality, we only consider the first column of \eqref{eqn:E_decomposition}:
\begin{equation}\label{eq:E_decomp_1st_col}
\begin{aligned}
    &\left(\frac{\mat{F}_1\vec{v}_n(z_1)}{\sqrt{n}}\right)_\uparrow - \left(\frac{\mat{F}_1\vec{v}_n(z_1)}{\sqrt{n}}\right)_\downarrow z_1^{-1}\\
    =&~ \left(\sum^\infty_{k=0}\left(\frac{\mat{E}_1}{\lambda_1}\right)^k\frac{\mat{E}_2\vec{v}_n(z_1)}{\lambda_1\sqrt{n}}\right)_\uparrow - \left(\sum^\infty_{k=0}\left(\frac{\mat{E}_1}{\lambda_1}\right)^k\frac{\mat{E}_2\vec{v}_n(z_1)}{\lambda_1\sqrt{n}}z_i^{-1}\right)_\downarrow\\
    ~&- \left(\frac{1}{n}\mat{V}_n(\zdom)\mat{V}_n(\zdom)^\dagger\frac{\mat{E}_2\vec{v}_n(z_1)}{\lambda_1\sqrt{n}}\right)_\uparrow  + \left(\frac{1}{n}\mat{V}_n(\zdom)\mat{V}_n(\zdom)^\dagger\frac{\mat{E}_2\vec{v}_n(z_1)}{\lambda_1\sqrt{n}}z_1^{-1}\right)_\downarrow 
    \,.
\end{aligned}
\end{equation}
Recall
\[
\mat{E}_{2}=\begin{bmatrix}
        E_0 & \overline{E_1} & \overline{E_2} & \cdots & \overline{E_{n-1}} \\
        E_1 & E_0 & \overline{E_1} & \cdots & \overline{E_{n-2}} \\
        E_2 & E_1 & E_0 & \cdots & \overline{E_{n-3}} \\ 
        \vdots & \vdots & \vdots & \ddots & \vdots \\
        E_{n-1} & E_{n-2} & E_{n-3} & \cdots & E_0
    \end{bmatrix}.
\]
Define vectors
\[
\vec{p}_0:=\frac{\mat{E}_2\vec{v}_n(z_1)}{\sqrt{n}\lambda_1},\quad \vec{q}_0:=\frac{\mat{E}_2\vec{v}_n(z_1)}{\sqrt{n}\lambda_1}z^{-1}_1\,.
\]
Because $\lambda_1=n\mu_{1}\geq n\mu_r$ and $\|\vec{v}_n(z_1)\|_\infty=1$, we have
\begin{align}\label{eq:p0_inf_norm}
    \|\vec{p}_0\|_\infty\leq  \frac{\|\mat{E}_2\|_1}{\sqrt{n}\lambda_1}\leq \frac{n\cdot {\cal O}(\alpha\sqrt{\log n})}{\sqrt{n}\cdot \mu_r n}= {\cal O}\left(\frac{\alpha\sqrt{\log n}}{\mu_r \sqrt{n}}\right)\,,
\end{align}
where the second step follows from $|E_i|={\cal O}(\alpha\sqrt{\log n})$. We also have
$\|\vec{q}_0\|_\infty=\mathcal{O}(\alpha\sqrt{\log n}/(\mu_r\sqrt{n}))$ since $|z_1|=1$. Using the structure of $\mat{E}_2$, we obtain
\begin{equation}\label{eqn:p_q_0}
\left\|\left(\vec{p}_0\right)_\uparrow-\left(\vec{q}_0\right)_\downarrow \right\|_\infty=\left\|\frac{1}{\sqrt{n}\lambda_1}\begin{bmatrix}
        -E_1z^{-1}_1+\overline{E_{n-1}}z^{n-1}_1\\
        -E_2z^{-1}_1+\overline{E_{n-2}}z^{n-1}_1\\
        \vdots\\
        -E_{n-1}z^{-1}_1+\overline{E_{1}}z^{n-1}_1\\
    \end{bmatrix}\right\|_\infty=\mathcal{O}\left(\frac{\alpha}{\mu_r n^{1.5}}\right)\,,
\end{equation}
where we use $|E_i|=\mathcal{O}(\alpha\sqrt{\log n})$ and $\lambda_1\geq \mu_r n$ in the last equality.

We first consider the terms in \eqref{eq:E_decomp_1st_col} that involves the second term of $\mat{F}_1$:
\[
\left(\frac{1}{n}\mat{V}_n(\zdom)\mat{V}_n(\zdom)^\dagger\vec{p}_0\right)_\uparrow-\left(\frac{1}{n}\mat{V}_n(\zdom)\mat{V}_n(\zdom)^\dagger\vec{q}_0\right)_\downarrow\,.
\]
Note that
\begin{align*}
    \left(\frac{1}{n} \mat{V}_n(\zdom)\mat{V}_n(\zdom)^\dagger\right)_{i,j} = \frac{1}{n} \sum_{l=1}^{r} z_l^{i-j}\quad\forall ~1\leq i,j\leq n\,,
\end{align*}
which implies that $\frac{1}{n}\mat{V}_n(\zdom)\mat{V}_n(\zdom)^\dagger$ is also a Toeplize matrix (since the $(i,j)$-th entry only depends on $i-j$), and 
\begin{align}\label{eq:V_V_dagger_inf_norm}
    \left\|\frac{1}{n}\mat{V}_n(\zdom)\mat{V}_n(\zdom)^\dagger\right\|_{\text{entry},\infty}=\frac{r}{n}\,,
\end{align}
where $\left\|A\right\|_{\text{entry},\infty}=\max_{i,j}|A_{i,j}|$.
Thus, for any $1\leq i\leq n-1$, 
\begin{align*}
    &\left(\Big(\frac{1}{n}\mat{V}_n(\zdom)\mat{V}_n(\zdom)^\dagger\vec{p}_0\Big)_\uparrow-\Big(\frac{1}{n}\mat{V}_n(\zdom)\mat{V}_n(\zdom)^\dagger\vec{q}_0\Big)_\downarrow\right)_i\\
    = &~ \sum_{j=1}^n \Big(\frac{1}{n}\mat{V}_n(\zdom)\mat{V}_n(\zdom)^\dagger\Big)_{i,j}(\vec{p}_0)_j- \Big(\frac{1}{n}\mat{V}_n(\zdom)\mat{V}_n(\zdom)^\dagger\Big)_{i+1,j}(\vec{q}_0)_j\\
    = &~ \sum_{j=2}^{n-1} \Big(\frac{1}{n}\mat{V}_n(\zdom)\mat{V}_n(\zdom)^\dagger\Big)_{i,j} \left((\vec{p}_0)_\uparrow - (\vec{q}_0)_\downarrow\right)_j + \Big(\frac{1}{n}\mat{V}_n(\zdom)\mat{V}_n(\zdom)^\dagger\Big)_{i,1}((\vec{p}_0)_\uparrow)_1\\
    ~&- \Big(\frac{1}{n}\mat{V}_n(\zdom)\mat{V}_n(\zdom)^\dagger\Big)_{i+1,n}((\vec{q}_0)_\downarrow)_{n-1}\,,
\end{align*}
where the second step follows from the Toeplize structure of the matrix. Therefore, by triangle inequality,
\begin{equation}\label{eqn:VV}
\begin{aligned}
&\left\|\left(\frac{1}{n}\mat{V}_n(\zdom)\mat{V}_n(\zdom)^\dagger\vec{p}_0\right)_\uparrow-\left(\frac{1}{n}\mat{V}_n(\zdom)\mat{V}_n(\zdom)^\dagger\vec{q}_0\right)_\downarrow\right\|_\infty\\
\leq &~\left\|\frac{1}{n}\mat{V}_n(\zdom)\mat{V}_n(\zdom)^\dagger\right\|_{\text{entry},\infty}\cdot \left((n-2)\left\|\left(\vec{p}_0\right)_\uparrow-\left(\vec{q}_0\right)_\downarrow\right\|_\infty+\left\|\vec{p}_0\right\|_\infty+\left\|\vec{q}_0\right\|_\infty\right)\\
=&~\frac{r}{n}\cdot \mathcal{O}\left(n\cdot \frac{\alpha}{\mu_r n^{1.5}} + \frac{\alpha}{\mu_r\sqrt{n}}\right) = \mathcal{O}\left(\frac{r\alpha}{\mu_r  n^{1.5}}\right)\,.
\end{aligned}
\end{equation}
where the second step follows from \eqref{eq:p0_inf_norm}-\eqref{eq:V_V_dagger_inf_norm}.

Now, we consider the first term of $\mat{F}_1$. Define
\[
\vec{p}_k:=\left(\frac{\mat{E}_1}{\lambda_1}\right)\vec{p}_{k-1},\quad \vec{q}_k:=\left(\frac{\mat{E}_1}{\lambda_1}\right)\vec{q}_{k-1}
\]
for $k\geq 1$. We first notice 
\begin{align}\label{eq:E1_entry_inf_norm}
    \left\|\frac{\mat{E}_1}{\lambda_1}\right\|_{\text{entry},\infty} = \frac{1}{\lambda_1}\max_{i,j} \sum_{l=r+1}^d |\mu_lz_l^{i-j}|\leq \frac{\mu_{\rm tail}}{\lambda_1}={\cal O}\left(\frac{\mu_{\rm tail}}{\mu_r n}\right)\,,
\end{align}
where the second step follows from the definition of $\mu_{\rm tail}$.
Then, we have
\[
\left\|\vec{p}_{k+1}\right\|_\infty= \mathcal{O}\left(\frac{\mu_{\rm tail}}{\mu_r}\right) \left\|\vec{p}_{k}\right\|_\infty,\quad \left\|\vec{q}_{k+1}\right\|_\infty= \mathcal{O}\left(\frac{\mu_{\rm tail}}{\mu_r}\right) \left\|\vec{p}_{k}\right\|_\infty\,.
\]
Furthermore, because $\mat{E}_1$ is also a Toeplize matrix, similar to \eqref{eqn:VV}, we have
\[
\begin{aligned}
&\left\|\left(\mat{p}_{k+1}\right)_\uparrow-\left(\mat{q}_{k+1}\right)_\downarrow\right\|_\infty\\
= &~ \left\|\left(\frac{\mat{E}_1}{\lambda_1}\mat{p}_{k}\right)_\uparrow-\left(\frac{\mat{E}_1}{\lambda_1}\mat{q}_{k}\right)_\downarrow\right\|_\infty\\
\leq &~\|\mat{E}_1/\lambda_1\|_{\text{entry},\infty}\cdot \left((n-2)\left\|\left(\mat{p}_{k}\right)_\uparrow-\left(\mat{q}_{k}\right)_\downarrow\right\|_\infty+\left\|\vec{p}_k\right\|_\infty+\left\|\vec{q}_k\right\|_\infty\right)\\
=&~\mathcal{O}\left(\frac{\mu_{\rm tail}}{\mu_r}\left\|\left(\mat{p}_{k}\right)_\uparrow-\left(\mat{q}_{k}\right)_\downarrow\right\|_\infty+\frac{\mu_{\rm tail}}{\mu_r n}(\left\|\vec{p}_k\right\|_\infty+\left\|\vec{q}_k\right\|_\infty)\right)\,,
\end{aligned}
\]
where the second step follows from \eqref{eq:E1_entry_inf_norm}.
The above two equations imply that
\[
\begin{aligned}
&n\left\|\left(\mat{p}_{k+1}\right)_\uparrow-\left(\mat{q}_{k+1}\right)_\downarrow\right\|_\infty+\left\|\vec{p}_{k+1}\right\|_\infty+\left\|\vec{q}_{k+1}\right\|_\infty\\
=&~\mathcal{O}\left(\frac{\mu_{\rm tail}}{\mu_r}\right)\cdot \left(n\left\|\left(\mat{p}_{k}\right)_\uparrow-\left(\mat{q}_{k}\right)_\downarrow\right\|_\infty+(\left\|\vec{p}_k\right\|_\infty+\left\|\vec{q}_k\right\|_\infty)\right)\,.
\end{aligned}
\]
Plugging in the values of $k=0$ in \eqref{eq:p0_inf_norm} and \eqref{eqn:p_q_0}, we obtain for any $k>0$,
\begin{align*}
n\left\|\left(\mat{p}_{k}\right)_\uparrow-\left(\mat{q}_{k}\right)_\downarrow\right\|_\infty + \|\vec{p}_k\|_\infty + \|\vec{q}_k\|_\infty\leq \mathcal{O}\left(\frac{\mu_{\rm tail}}{\mu_r}\right)^k\cdot \frac{\alpha}{\mu_r \sqrt{n}}\,,
\end{align*}
which implies that
\begin{equation}\label{eqn:p_k_bound}
    \left\|\left(\mat{p}_{k}\right)_\uparrow-\left(\mat{q}_{k}\right)_\downarrow\right\|_\infty \leq \mathcal{O}\left(\frac{\mu_{\rm tail}}{\mu_r}\right)^k\cdot \frac{\alpha}{\mu_r n^{1.5}}~~~\forall k\geq 0\,.
\end{equation}
This gives a bound for each term in the summation of the first term of $\mat{F}_i$.

Therefore, we get that
\begin{equation}\label{eq:bound_F1_1st_column}
\begin{aligned}
&\left\|\left(\frac{\mat{F}_1\vec{v}_n(z_1)}{\sqrt{n}\lambda_1}\right)_\uparrow-\left(\frac{\mat{F}_1\vec{v}_n(z_1)}{\sqrt{n}\lambda_1}\right)_\downarrow z^{-1}_1\right\|_2\\
\leq &~\sqrt{n}\left\|\left(\frac{\mat{F}_1\vec{v}_n(z_1)}{\sqrt{n}\lambda_1}\right)_\uparrow-\left(\frac{\mat{F}_1\vec{v}_n(z_1)}{\sqrt{n}\lambda_1}\right)_\downarrow z^{-1}_1\right\|_\infty\\
\leq &~\sqrt{n}\sum_{k=0}^\infty \left\|(\vec{p}_k)_\uparrow - (\vec{q}_k)_\downarrow\right\|_\infty + \sqrt{n}\left\|\left(\frac{1}{n}\mat{V}_n(\zdom)\mat{V}_n(\zdom)^\dagger\vec{p}_0\right)_\uparrow-\left(\frac{1}{n}\mat{V}_n(\zdom)\mat{V}_n(\zdom)^\dagger\vec{q}_0\right)_\downarrow\right\|_\infty\\
\leq &~ {\cal O}\left( \frac{\alpha}{\mu_r n}\right) + \mathcal{O}\left(\frac{r\alpha}{\mu_r  n}\right)= ~\mathcal{O}\left(\frac{r\alpha}{\mu_r n}\right)\,,
\end{aligned}
\end{equation}
where the first step follows from $\|x\|_2\leq \sqrt{n}\|x\|_\infty$ for any $x\in \C^n$, the second step follows from \eqref{eq:E_decomp_1st_col}, and the third step follows from \eqref{eqn:VV} and \eqref{eqn:p_k_bound}.

Performing the same computation as \eqref{eq:E_decomp_1st_col} to \eqref{eq:bound_F1_1st_column} for every $1\leq i\leq r$, we derive
\[
\max_{1\leq i\leq r}\left\|\left(\frac{\mat{F}_i\vec{v}_n(z_i)}{\sqrt{n}\lambda_i}\right)_\uparrow-\left(\frac{\mat{F}_i\vec{v}_n(z_i)}{\sqrt{n}\lambda_i}\right)_\downarrow z^{-1}_i\right\|_2=\mathcal{O}\left(\frac{r\alpha}{\mu_r  n}\right)\,.
\]
It implies that
\begin{align*}
&\Big\|\left(\widetilde{\mat{E}}^Q_1\mat{P}\mat{P}^\dagger_v\right)_\uparrow-\left(\widetilde{\mat{E}}^Q_1\mat{P}\right)_\downarrow(\mat{Q}_\downarrow)^\dagger\mat{Q}_\uparrow\mat{P}^\dagger_v\Big\|_2\\
\leq &~ \left\|\left[\left(\frac{\mat{F}_1\vec{v}_n(z_1)}{\sqrt{n}}\right)_\uparrow -\left(\frac{\mat{F}_1\vec{v}_n(z_1)}{\sqrt{n}}\right)_\downarrow z_1^{-1} \quad \cdots\quad \left(\frac{\mat{F}_r\vec{v}_n(z_r)}{\sqrt{n}}\right)_\uparrow -\left(\frac{\mat{F}_r\vec{v}_n(z_r)}{\sqrt{n}}\right)_\downarrow z_r^{-1} \right]\right\|_2\\
~&+\mathcal{O}\left(\frac{\alpha\sqrt{\log n}}{\mu^{2}_r\Delta_{\vec{z}} n^{1.5}}\right)\\
\leq &~\sqrt{r}\max_i\left\|\left(\frac{\mat{F}_i\vec{v}_n(z_i)}{\sqrt{n}\lambda_i}\right)_\uparrow-\left(\frac{\mat{F}_i\vec{v}_n(z_i)}{\sqrt{n}\lambda_i}\right)_\downarrow z^{-1}_i\right\|_2 + \mathcal{O}\left(\frac{\alpha\sqrt{\log n}}{\mu^{2}_r\Delta_{\vec{z}} n^{1.5}}\right)\\
= &~ \mathcal{O}\left(\frac{r^{1.5}\alpha}{\mu_r n}+\frac{\alpha\sqrt{\log n}}{\mu^{2}_r\Delta_{\vec{z}} n^{1.5}}\right)\,,%
\end{align*}
where the first step follows from \eqref{eqn:E_decomposition}, the second step follows from $\|A\|_2\leq \sqrt{r}\max_{i\in [n]}\|A_i\|_2$ for any $A\in \C^{n\times r}$, the last step follows from \eqref{eq:bound_F1_1st_column}.

The lemma is then proved.
\end{proof}

\subsubsection{Establishing the second equation in \texorpdfstring{\cref{eqn:key_equality}}{Eq}}\label{sec:defer_key_eq_2nd}

\decomposeII*

\begin{proof}
For simplicity, we mainly consider the first term $\mat{Q}_\downarrow^\dagger\mat{E}^Q_\uparrow$. Another term $\mat{Q}_\downarrow^\dagger\mat{E}^Q_\downarrow \mat{Q}_\downarrow^\dagger \mat{Q}_\uparrow$ can be bounded similarly. 

First, we notice that %
\[
\begin{aligned}
\mat{Q}_\downarrow^\dagger \mat{E}_\uparrow^Q = &~\mat{P}^\dagger_v\left(\left(\frac{1}{\sqrt{n}}\mat{V}_n(\zdom)\right)_{\downarrow}\right)^\dagger\mat{E}^Q_{\uparrow}\\
=&~\mat{P}^\dagger_v\diag(\zdom)^{\dagger}\left(\left(\frac{1}{\sqrt{n}}\mat{V}_n(\zdom)\right)_{\uparrow}\right)^\dagger\mat{E}^Q_{\uparrow}\\
\leq &~\mat{P}^\dagger_v\diag(\zdom)^{\dagger}\mat{P}_v\mat{P}^\dagger_v\left(\left(\frac{1}{\sqrt{n}}\mat{V}_n(\zdom)\right)_{\uparrow}\right)^\dagger\mat{E}^Q_{\uparrow}+\mathcal{O}\left(\left\|\mat{I}_r-\mat{P}_v\mat{P}_v^\dagger\right\|_2\cdot \left\|\mat{E}^Q_\uparrow\right\|_2\right)\\
= &~ \mat{P}^\dagger_v\diag(\zdom)^{\dagger}\mat{P}_v\mat{P}^\dagger_v\left(\left(\frac{1}{\sqrt{n}}\mat{V}_n(\zdom)\right)_{\uparrow}\right)^\dagger\mat{E}^Q_{\uparrow} + {\cal O}\left(\frac{1}{\Delta_{\vec{z}}n}\right)\cdot {\cal O}\left(\frac{\alpha\sqrt{\log n}}{\mu_r\sqrt{\Delta_{\vec{z}}n}}\right)\\
=&~\mat{P}^\dagger_v\diag^{\dagger}(\zdom)\mat{P}_v\mat{Q}_{\uparrow}^\dagger\mat{E}^Q_{\uparrow}+\mathcal{O}\left(\frac{\alpha\sqrt{\log n}}{\mu_r(\Delta_{\vec{z}} n)^{1.5}}\right)\,,
\end{aligned}
\]
where the first step follows from $\mat{Q}_r=\frac{1}{\sqrt{n}}\mat{V}_n(\zdom)\mat{P}_v$ by \cref{lem:vander-eigen}, the second step follows from $\mat{V}_n(\zdom)_\downarrow = \mat{V}_n(\zdom)_\downarrow\cdot \diag(\zdom)$ by the structure of the Vandermonde matrix, the third step follows from $\|\mat{P}_v\|_2={\cal O}(1)$ by \eqref{eqn:P_orthogonal}, $\|\diag(\zdom)\|_2=1$ by definition, and $\big\|\frac{1}{\sqrt{n}}\mat{V}_n(\zdom)\big\|_2={\cal O}(1)$ by \cref{cor:V_dagger_V_dom},
the fourth step follows from \eqref{eqn:P_orthogonal} and \eqref{eqn:second_Q_perturb}, and the last step follows from \cref{lem:vander-eigen} again.

Thus, we obtain that
\begin{equation}\label{eq:Q_down_E_up_1}
\begin{aligned}
    \left\|\mat{Q}_\downarrow^\dagger \mat{E}_\uparrow^Q\right\|_2\leq &~ \left\|\mat{P}^\dagger_v\diag^{\dagger}(\zdom)\mat{P}_v\mat{Q}_{\uparrow}^\dagger\mat{E}^Q_{\uparrow}\right\|_2+\mathcal{O}\left(\frac{\alpha\sqrt{\log n}}{\mu_r(\Delta_{\vec{z}} n)^{1.5}}\right)\\
    = &~ {\cal O}\left(\left\|\mat{P}_v\mat{Q}_{\uparrow}^\dagger\mat{E}^Q_{\uparrow}\right\|_2+\frac{\alpha\sqrt{\log n}}{\mu_r(\Delta_{\vec{z}} n)^{1.5}}\right)\,,
\end{aligned}
\end{equation}
where the second step follows from $\|\mat{P}^\dagger_v\diag^{\dagger}(\zdom)\|_2=\mathcal{O}(1)$. Moreover, $\left\|\mat{P}_v\mat{Q}_{\uparrow}^\dagger\mat{E}^Q_{\uparrow}\right\|_2$ can be further bounded as follows:
\begin{equation}\label{eq:Pv_Q_up_E_up}
\begin{aligned}
\left\|\mat{P}_v\mat{Q}^\dagger_\uparrow \mat{E}_\uparrow^Q\right\|_2 = &~ \left\|\mat{P}_v\mat{Q}_{\uparrow}^\dagger\left(\widetilde{\mat{E}}^Q_1+\widetilde{\mat{E}}^Q_2\right)_\uparrow\right\|_2\\
= &~ \left\|\mat{P}_v\mat{P}_v^\dagger \left(\frac{1}{\sqrt{n}}\mat{V}_n(\zdom)\right)_\uparrow^\dagger \left(\widetilde{\mat{E}}^Q_1+\widetilde{\mat{E}}^Q_2\right)_\uparrow\right\|_2\\
\leq &~ \left\|\left(\frac{1}{\sqrt{n}}\mat{V}_n(\zdom)\right)_\uparrow^\dagger \left(\widetilde{\mat{E}}^Q_1+\widetilde{\mat{E}}^Q_2\right)_\uparrow\right\|_2+{\cal O}\left(\left\|\mat{P}_v\mat{P}_v^\dagger -\mat{I}_r\right\|_2 \left\|\mat{E}_\uparrow^Q\right\|_2\right)\\
= &~\left\|\left(\frac{1}{\sqrt{n}}\mat{V}_n(\zdom)\right)_\uparrow^\dagger \left(\widetilde{\mat{E}}^Q_1+\widetilde{\mat{E}}^Q_2\right)_\uparrow\mat{P}\mat{P}^{-1}\right\|_2+\mathcal{O}\left(\frac{\alpha\sqrt{\log n}}{\mu_r(\Delta_{\vec{z}} n)^{1.5}}\right)\\
\leq &~ \left\|\left(\frac{1}{\sqrt{n}}\mat{V}_n(\zdom)\right)_\uparrow^\dagger \left(\widetilde{\mat{E}}^Q_1+\widetilde{\mat{E}}^Q_2\right)_\uparrow\mat{P}\right\|_2\left\|\mat{P}^{-1}\right\|_2+\mathcal{O}\left(\frac{\alpha\sqrt{\log n}}{\mu_r(\Delta_{\vec{z}} n)^{1.5}}\right)\\
= &~ \mathcal{O}\left(\left\|\left(\frac{1}{\sqrt{n}}\mat{V}_n(\zdom)\right)_\uparrow^\dagger \left(\left(\widetilde{\mat{E}}^Q_1+\widetilde{\mat{E}}^Q_2\right)\mat{P}\right)_\uparrow\right\|_2\right)+\mathcal{O}\left(\frac{\alpha\sqrt{\log n}}{\mu_r(\Delta_{\vec{z}} n)^{1.5}}\right)\,,
\end{aligned}
\end{equation}
where the first step follows from \eqref{eq:E_Q_up_down_2}, the second step follows from \cref{lem:vander-eigen}, the third step follows from triangle inequality and $\big\|\frac{1}{\sqrt{n}}\mat{V}_n(\zdom)\big\|_2={\cal O}(1)$ by \cref{cor:V_dagger_V_dom}, the fourth step follows from \cref{eqn:P_orthogonal,eqn:second_Q_perturb}, the fifth step follows from the sub-multiplicitivity of the spectral norm, the sixth step follows from $\left\|\mat{P}^{-1}\right\|_2=\mathcal{O}(1)$ by \eqref{eqn:P_bound}.%

Let $\mat{V}_r:=\frac{1}{\sqrt{n}}\mat{V}_n(\zdom)$. Then, by \cref{eq:Q_down_E_up_1,eq:Pv_Q_up_E_up}, we have
\begin{align}\label{eq:Q_down_E_up_bound_1}
    \left\|\mat{Q}_\downarrow^\dagger \mat{E}_\uparrow^Q\right\|_2\leq \mathcal{O}\left(\left\|\left((\mat{V}_r)_\uparrow\right)^\dagger\left(\widetilde{\mat{E}}^Q_1\right)_\uparrow\mat{P}\right\|_2+\left\|\left((\mat{V}_r)_\uparrow\right)^\dagger\left(\widetilde{\mat{E}}^Q_2\mat{P}\right)_\uparrow\right\|_2 + \frac{\alpha\sqrt{\log n}}{\mu_r(\Delta_{\vec{z}} n)^{1.5}}\right)\,.
\end{align}

For the second term $\mat{Q}_\downarrow^\dagger\mat{E}^Q_\downarrow \mat{Q}_\downarrow^\dagger \mat{Q}_\uparrow$ in the second equation of \eqref{eqn:key_equality}, we notice that
\[
\begin{aligned}
\left\|\mat{Q}_\downarrow^\dagger\mat{E}^Q_\downarrow \mat{Q}_\downarrow^\dagger \mat{Q}_\uparrow\right\|_2\leq&~ \left\|\mat{Q}_\downarrow^\dagger\mat{E}^Q_\downarrow\right\|_2\left\|\mat{Q}_\downarrow\right\|_2\left\|\mat{Q}_\uparrow\right\|_2\\
\leq&~ {\cal O}\left(\left\|\mat{Q}_\downarrow^\dagger\mat{E}^Q_\downarrow\right\|_2\right)\\
\leq &~ {\cal O}\left(\left\|\mat{P}_v\mat{Q}_\downarrow^\dagger\mat{E}^Q_\downarrow\right\|_2\left\|(\mat{P}_v)^{-1}\right\|_2\right)\\
=&~ \mathcal{O}\left(\left\|\mat{P}_v\mat{Q}_\downarrow^\dagger\mat{E}^Q_\downarrow\right\|_2\right)\,,    
\end{aligned}
\]
where we use $\|\mat{P}_v^{-1}\|_2={\cal O}(1)$ by \eqref{eqn:P_orthogonal} in the last step. Similar to \eqref{eq:Pv_Q_up_E_up}, we can show that
\begin{equation}\label{eq:Q_down_E_up_bound_2}
\begin{aligned}
\left\|\mat{P}_v\mat{Q}_\downarrow^\dagger\mat{E}^Q_\downarrow\right\|_2
\leq &~\mathcal{O}\left(\left\|\left((\mat{V}_r)_\downarrow\right)^\dagger\left(\widetilde{\mat{E}}^Q_1\right)_\downarrow\mat{P}\right\|_2+\left\|\left((\mat{V}_r)_\downarrow\right)^\dagger\left(\widetilde{\mat{E}}^Q_2\mat{P}\right)_\downarrow\right\|_2 + \frac{\alpha\sqrt{\log n}}{\mu_r(\Delta_{\vec{z}} n)^{1.5}}\right)\,.
\end{aligned}
\end{equation}

Combining \eqref{eq:Q_down_E_up_bound_1} and \eqref{eq:Q_down_E_up_bound_2} together, the lemma is proved.
\end{proof}

\finegrainedII*

\begin{proof}
First, we consider $\left((\mat{V}_r)_\uparrow\right)^\dagger\left(\widetilde{\mat{E}}^Q_2\mat{P}\right)_\uparrow$. Recall from \eqref{eqn:E_tilde_Q} that
\begin{align}\label{eq:E2_P_recall}
\widetilde{\mat{E}}^Q_2\mat{P}={\cal O}\left(\frac{\alpha^2\log n}{\mu^{2}_r\Delta_{\vec{z}} n}\right)\mat{\Pi}_{\mat{Q}^\perp_r}\tilde{\mat{Q}}_{2}=:{\cal O}\left(\frac{\alpha^2\log n}{\mu^{2}_r\Delta_{\vec{z}} n}\right)\tilde{\mat{Q}}_{2}^{\perp}
\end{align}
Since the column space of $\mat{V}_r$ equals to the column space of $\mat{Q}_r$, we note that $\left(\mat{V}_r\right)^\dagger\left(\widetilde{\mat{E}}^Q_2\right)=\mat{0}_r$, where $\mat{0}_r$ is a $r$ by $r$ matrix with all zero entries. Although $\left((\mat{V}_r)_\uparrow\right)^\dagger\left(\widetilde{\mat{E}}^Q_2\mat{P}\right)_\uparrow$ is the multiplication with $\uparrow$, the remaining term for $\widetilde{\mat{E}}^Q_2$ is small. In particular, for any $(i,j)\in [r]\times [r]$, we have the following
\begin{equation}
\begin{aligned}
&\left|\left(\left((\mat{V}_r)_\uparrow\right)^\dagger\left(\widetilde{\mat{E}}^Q_2\mat{P}\right)_\uparrow\right)_{i,j}\right|\\
= &~ \left|\left((\mat{V}_r)^\dagger\left(\widetilde{\mat{E}}^Q_2\mat{P}\right)\right)_{i,j}-(\overline{\mat{V}_r})_{n,i}(\widetilde{\mat{E}}^Q_2\mat{P})_{n,j}\right|\\
= &~ \left|(\overline{\mat{V}_r})_{n,i}(\widetilde{\mat{E}}^Q_2\mat{P})_{n,j}\right|\\
=&~{\cal O}\left(\frac{\alpha^2\log n}{\mu^{2}_r\Delta_{\vec{z}} n}\right)\cdot \frac{|\overline{z^{n-1}_i}|
}{\sqrt{n}}\cdot \left|\left(\tilde{\mat{Q}}_{2}^{\perp}\right)_{n,j}\right|\\
\leq &~ {\cal O}\left(\frac{\alpha^2\log n}{\mu^{2}_r\Delta_{\vec{z}} n^{1.5}}\right)
\end{aligned}
\end{equation}
where the first step follows from the definition of $(\cdot)_\uparrow$, the second step follows from $(\mat{V}_r)^\dagger\left(\widetilde{\mat{E}}^Q_2\mat{P}\right)=\mat{0}_r\cdot \mat{P}=\mat{0}_r$, the third step follows from \eqref{eq:E2_P_recall}, and the fourth step follows from $\big|\big(\tilde{\mat{Q}}_{2}^{\perp}\big)_{n,j}\big|\leq \big\|\tilde{\mat{Q}}_{2}^{\perp}\big\|_{\text{entry},\infty}\leq \big\|\tilde{\mat{Q}}_{2}^{\perp}\big\|_2=\mathcal{O}(1)$. %
It implies that
\begin{align}\label{eqn:Q_E_2_up}
\left\|\left((\mat{V}_r)_\uparrow\right)^\dagger\left(\widetilde{\mat{E}}^Q_2\mat{P}\right)_\uparrow\right\|_2 \leq \left\|\left((\mat{V}_r)_\uparrow\right)^\dagger\left(\widetilde{\mat{E}}^Q_2\mat{P}\right)_\uparrow\right\|_F\leq r\cdot {\cal O}\left(\frac{\alpha^2\log n}{\mu^{2}_r\Delta_{\vec{z}} n^{1.5}}\right)\,.    
\end{align}

Next, we deal with $\left((\mat{V}_r)_\uparrow\right)^\dagger\left(\widetilde{\mat{E}}^Q_1\right)_\uparrow\mat{P}$. First, because $\mat{P}_v$ is almost unitary according to \cref{lem:vander-eigen}, we have
\begin{equation}
\begin{aligned}
\left\|\left((\mat{V}_r)_\uparrow\right)^\dagger\left(\widetilde{\mat{E}}^Q_1\right)_\uparrow\mat{P}\right\|_2\leq &~ \left\|\left((\mat{V}_r)_\uparrow\right)^\dagger\left(\widetilde{\mat{E}}^Q_1\right)_\uparrow\mat{P}\mat{P}^\dagger_v\right\|_2\left\|(\mat{P}_v^\dagger)^{-1}\right\|_2\\
=&~ \mathcal{O}\left(\left\|\left((\mat{V}_r)_\uparrow\right)^\dagger\left(\widetilde{\mat{E}}^Q_1\right)_\uparrow\mat{P}\mat{P}^\dagger_v\right\|_2\right)\,.
\end{aligned}
\end{equation}
Without loss of generality, we bound the norm of the first column of $\left((\mat{V}_r)_\uparrow\right)^\dagger\left(\widetilde{\mat{E}}^Q_1\right)_\uparrow\mat{P}\mat{P}^\dagger_v$.
By using \eqref{eqn:EPPV_1_all} to approximate $\left(\widetilde{\mat{E}}^Q_1\right)_\uparrow\mat{P}\mat{P}^\dagger_v$, we have %
\begin{align}
&\left\|\left(\left((\mat{V}_r)_\uparrow\right)^\dagger\left(\widetilde{\mat{E}}^Q_1\right)_\uparrow\mat{P}\mat{P}^\dagger_v\right)[:,1]\right\|_2\notag\\
= &~ \left\|\left(\left(\mat{V}_r\right)_{\uparrow}\right)^\dagger\frac{1}{\lambda_1}\left(\sum^\infty_{k=1}\left(\frac{\mat{E}_1}{\lambda_i}\right)^k\right)_\uparrow\frac{\mat{E}_2\vec{v}_n(z_1)}{\sqrt{n}}\right\|_2\label{eq:second_Q_E_t1}\\
~&+\left\|\left(\left(\mat{V}_r\right)_{\uparrow}\right)^\dagger\frac{1}{\lambda_1}\left(\mat{I}_n-\mat{V}_r\mat{V}_r^\dagger\right)_\uparrow\frac{\mat{E}_2\vec{v}_n(z_1)}{\sqrt{n}}\right\|_2 \label{eq:second_Q_E_t2}\\
~&+\mathcal{O}\left(\frac{\alpha\sqrt{\log n}}{\mu_r^2\Delta_{\vec{z}} n^{1.5}}\right)\,,\notag
\end{align}
where $\lambda_1=\mu_1 n\geq \mu_r n$.

Then, we bound \eqref{eq:second_Q_E_t1} and \eqref{eq:second_Q_E_t2} separately.

For \eqref{eq:second_Q_E_t1}, we have
    \begin{equation}\label{eq:second_Q_E_t1_bound}
    \begin{aligned}
    &\left\|\frac{1}{\lambda_1}\left(\left(\mat{V}_r\right)_{\uparrow}\right)^\dagger\left(\sum^\infty_{k=1}\left(\frac{\mat{E}_1}{\lambda_1}\right)^k\right)_\uparrow\frac{\mat{E}_2\vec{v}_n(z_1)}{\sqrt{n}}\right\|_2\\
    \leq &~ \left\|\left(\left(\mat{V}_r\right)_{\uparrow}\right)^\dagger\left(\frac{\mat{E}_1}{\lambda_1}\right)_\uparrow\right\|_2\cdot \left(\sum_{k=0}^\infty \frac{\left\|\mat{E}_1\right\|_2^k}{|\lambda_1|^k}\right)\cdot \frac{\|\mat{E}_2\vec{v}_n(z_1)\|_2}{\lambda_1\sqrt{n}}\\
    \leq &~ \left\|\left(\left(\mat{V}_r\right)_{\uparrow}\right)^\dagger\left(\frac{\mat{E}_1}{\lambda_1}\right)_\uparrow\right\|_2\cdot {\cal O}\left(\sum_{k=0}^\infty \left(\frac{n\mu_{\rm tail}}{\mu_r n}\right)^k\cdot \frac{\alpha\sqrt{n\log n}\cdot \sqrt{n}}{\mu_r n^{1.5}}\right)\\
    = &~ \left\|\left(\left(\mat{V}_r\right)_{\uparrow}\right)^\dagger\left(\frac{\mat{E}_1}{\lambda_1}\right)_\uparrow\right\|_2\cdot {\cal O}\left(\frac{\alpha\sqrt{\log n}}{\mu_r\sqrt{n}}\right)\\
    = &~ {\cal O}\left(\frac{\mu_{\rm tail}}{\Delta_{\vec{z}}\lambda_1}\right)\cdot {\cal O}\left(\frac{\alpha\sqrt{\log n}}{\mu_r\sqrt{n}}\right)\\
    = &~ {\cal O}\left(\frac{\alpha \sqrt{\log n}}{\mu_r \Delta_{\vec{z}}n^{1.5}}\right)\,,
    \end{aligned}
    \end{equation}
    where the first step follows from the triangle inequality, the second step follows from \eqref{eqn:bound_tail}, \cref{lem:Mecks}, the definition of $\lambda_1$, and $\|\vec{v}_n(z_1)\|_2=\sqrt{n}$, the third step follows from $\mu_{\rm tail}<\mu_r$, the fourth step follows from \cref{clm:bound_V_dom_E1}, and the last step follows from $\lambda_1\geq \mu_r n$ and $\mu_{\rm tail}< \mu_r$.
    
For \eqref{eq:second_Q_E_t2}, we notice that
\begin{align*}
&\left(\left(\mat{V}_r\right)_{\uparrow}\right)^\dagger\left(\mat{I}_n-\mat{V}_r\mat{V}_r^\dagger\right)_\uparrow\\    
= &~ \left[\left(\left(\mat{V}_r\right)_{\uparrow}\right)^\dagger\quad \vec{0}_1\right] - \left(\left(\mat{V}_r\right)_{\uparrow}\right)^\dagger \left(\mat{V}_r\right)_\uparrow \mat{V}_r^\dagger\\
\leq &~ \left[\left(\left(\mat{V}_r\right)_{\uparrow}\right)^\dagger\quad \vec{0}_1\right] - \frac{n-1}{n}\mat{V}_r^\dagger + \frac{n-1}{n}\left\|\mat{I}_r - \frac{1}{n-1}\mat{V}_{n-1}(\zdom)^\dagger \mat{V}_{n-1}(\zdom)\right\|_2\cdot \left\|\mat{V}_r^\dagger\right\|_2\\
= &~ \left[\left(\left(\mat{V}_r\right)_{\uparrow}\right)^\dagger\quad \vec{0}_1\right] - \frac{n-1}{n}\mat{V}_r^\dagger + {\cal O}\left(\frac{1}{\Delta_{\vec{z}}n}\right)\\
\leq &~ \left[\left(\left(\mat{V}_r\right)_{\uparrow}\right)^\dagger\quad \vec{0}_1\right] - \mat{V}_r^\dagger +  {\cal O}\left(\frac{1}{\Delta_{\vec{z}}n}\right)\\
= &~ -\frac{1}{\sqrt{n}} \left[\mat{0}_{r\times n-1}\quad \begin{matrix} z_1^{1-n} \\
 \vdots \\
    z_r^{1-n}  \end{matrix}\right]+\mathcal{O}\left(\frac{1}{\Delta_{\vec{z}}n}\right)\,,
\end{align*}
where $\vec{0}_1\in \C^r$ is an all-zero vector, the second step follows from $(\mat{V}_n(\zdom))_\uparrow = \mat{V}_{n-1}(\zdom)$ by definition, the third step follows from \cref{cor:V_dagger_V_dom}, the fourth step follows from $\|\mat{V}_r\|_2={\cal O}(1)$, and the last step is by direct calculation.
Plugging it into \eqref{eq:second_Q_E_t2}, we have
\begin{equation}\label{eq:second_Q_E_t2_bound}
\begin{aligned}
&\left\|\left(\left(\mat{V}_r\right)_{\uparrow}\right)^\dagger\frac{1}{\lambda_1}\left(\mat{I}_n-\mat{V}_r\mat{V}_r^\dagger\right)_\uparrow\frac{\mat{E}_2\vec{v}_n(z_1)}{\sqrt{n}}\right\|_2\\
=&~\left\|-\frac{1}{\sqrt{n}} \left[\mat{0}_{r\times n-1}\quad \begin{matrix} z_1^{1-n} \\
 \vdots \\
    z_r^{1-n}  \end{matrix}\right]\frac{\mat{E}_2\vec{v}_n(z_1)}{\lambda_1\sqrt{n}}\right\|_2+\left\|\frac{\mat{E}_2\vec{v}_n(z_1)}{\lambda_1\sqrt{n}}\right\|_2\cdot \mathcal{O}\left(\frac{1}{\Delta_{\vec{z}}n}\right)\\
=&~\left\|\begin{bmatrix} z_1^{1-n} \\
 \vdots \\
z_r^{1-n}  \end{bmatrix}\right\|_2\cdot \frac{|(\mat{E}_2\vec{v}_n(z_1))_n|}{\lambda_1 n} +\mathcal{O}\left(\frac{\alpha\sqrt{\log n}}{\mu_r\Delta_{\vec{z}}n^{1.5}}\right)\\
\leq &~\frac{\sqrt{r}\cdot {\cal O}(\sqrt{n\log n})}{\mu_r n^2}+\mathcal{O}\left(\frac{\alpha\sqrt{\log n}}{\mu_r\Delta_{\vec{z}}n^{1.5}}\right) = \mathcal{O}\left(\frac{\alpha\sqrt{r\log n}}{\mu_rn^{1.5}}+\frac{\alpha\sqrt{\log n}}{\mu_r\Delta_{\vec{z}}n^{1.5}}\right)\,,
\end{aligned}
\end{equation}
where the second step follows from \cref{lem:Mecks} and $\|\vec{v}_n(z_1)\|_2=\sqrt{n}$, and the third step follows from the fact that $(\mat{E}_2\vec{v}_n(z_1))_n$ is a sub-Gaussian random variable with parameter ${\cal O}(\alpha\|\vec{v}_n(z_1)\|_2)={\cal O}(\alpha\sqrt{n})$. 

Hence, combining \eqref{eq:second_Q_E_t1_bound} and \eqref{eq:second_Q_E_t2_bound} together, we obtain
\begin{align*}
\left\|\left(\left((\mat{V}_r)_\uparrow\right)^\dagger\left(\widetilde{\mat{E}}^Q_1\right)_\uparrow\mat{P}\mat{P}^\dagger_v\right)[:,1]\right\|_2\leq &~ \mathcal{O}\left(\frac{\alpha\sqrt{\log n}}{\mu_r\Delta_{\vec{z}} n^{1.5}}+\frac{\alpha\sqrt{r\log n}}{\mu_rn^{1.5}}+\frac{\alpha\sqrt{\log n}}{\mu_r\Delta_{\vec{z}}n^{1.5}}\right)\\
= &~ \mathcal{O}\left(\frac{\alpha\sqrt{r\log n}}{\mu_rn^{1.5}}+\frac{\alpha\sqrt{\log n}}{\mu_r\Delta_{\vec{z}}n^{1.5}}\right)\,.
\end{align*}
This implies
\begin{align}\label{eq:V_up_E1_up_P}
\left\|\left((\mat{V}_r)_\uparrow\right)^\dagger\left(\widetilde{\mat{E}}^Q_1\right)_\uparrow\mat{P}\right\|_2\leq \left\|\left((\mat{V}_r)_\uparrow\right)^\dagger\left(\widetilde{\mat{E}}^Q_1\right)_\uparrow\mat{P}\right\|_F\leq \sqrt{r}\cdot \mathcal{O}\left(\frac{\alpha\sqrt{r\log n}}{\mu_rn^{1.5}}+\frac{\alpha\sqrt{\log n}}{\mu_r\Delta_{\vec{z}}n^{1.5}}\right)\,,
\end{align}

Thus, by \eqref{eqn:Q_E_2_up} and \eqref{eq:V_up_E1_up_P}, we prove the first part of the lemma \eqref{eq:V_EQ_P_up}.
For \eqref{eq:V_EQ_P_down}, we can verify that the above proof still holds when we change $\uparrow$ to $\downarrow$. 

Therefore, we complete the proof of the lemma.
\end{proof}

\subsubsection{Techincal claims}
\begin{claim}\label{clm:P_v_properties}
Let $\mat{P}_v$ be defined as \eqref{eq:def_P_v}. Then, 
\begin{subequations}\label{eqn:P_v_addition_properties_1}
\begin{equation}
\mat{P}_v(\mat{Q}_\downarrow)^\dagger\mat{Q}_\uparrow\mat{P}^\dagger_v=\diag(\zdom^{-1})+\mathcal{O}\left(\frac{1}{\Delta_{\vec{z}} n}\right)\,,
\end{equation}
\begin{equation}\label{eqn:P_v_addition_properties_2}
\mat{P}_v\mat{\Sigma}^{-k}_r\mat{P}^\dagger_v=\left(n\diag(\mudom)\right)^{-k}+\mathcal{O}\left(\frac{1}{\mu_r n}\right)^{k+1}\frac{1}{\Delta_{\vec{z}} }\quad\forall k\geq 0\,.
\end{equation}
\end{subequations}
\end{claim}
\begin{proof}
We first prove \eqref{eqn:P_v_addition_properties_1}.

By \eqref{eqn:reshape_Q}, we have
\begin{align*}
    \mat{Q}_\downarrow = \frac{1}{\sqrt{n}}{\mat{V}_n(\zdom)}_\downarrow\mat{P}_v,\quad \mat{Q}_\uparrow = \frac{1}{\sqrt{n}}{\mat{V}_n(\zdom)}_\uparrow\mat{P}_v\,.
\end{align*}
Since ${\mat{V}_n(\zdom)}_\downarrow = {\mat{V}_n(\zdom)}_\uparrow\diag(\zdom)$ and $\mat{V}_n(\zdom)_\uparrow = \mat{V}_{n-1}(\zdom)$, we also have
\begin{align*}
    \mat{Q}_\downarrow = \frac{1}{\sqrt{n}}\mat{V}_{n-1}(\zdom)\diag(\zdom)\mat{P}_v,\quad \mat{Q}_\uparrow = \frac{1}{\sqrt{n}}\mat{V}_{n-1}(\zdom)\mat{P}_v\,.
\end{align*}
Thus, LHS of \eqref{eqn:P_v_addition_properties_1} can be expressed as:
\begin{align*}
&\mat{P}_v(\mat{Q}_\downarrow)^\dagger\mat{Q}_\uparrow\mat{P}^\dagger_v\\
= &~ \mat{P}_v \cdot \left(\mat{P}_v^\dagger \diag(\zdom)^\dagger\frac{1}{\sqrt{n}}\mat{V}_{n-1}(\zdom)^\dagger\right) \cdot \left(\frac{1}{\sqrt{n}}\mat{V}_{n-1}(\zdom)\mat{P}_v\right)\cdot \mat{P}^\dagger_v\\
= &~ \left(\mat{P}_v \mat{P}_v^\dagger\right)\cdot \diag(\zdom^{-1}) \cdot \left(\frac{1}{n}\mat{V}_{n-1}(\zdom)^\dagger \mat{V}_{n-1}(\zdom)\right) \cdot \left(\mat{P}_v\mat{P}^\dagger_v\right)\,,
\end{align*}
where the second step follows from $\overline{z_i}=z_i^{-1}$ for all $i\in [m]$. Then, we have
\begin{align*}
    &\left\|\mat{P}_v(\mat{Q}_\downarrow)^\dagger\mat{Q}_\uparrow\mat{P}^\dagger_v - \diag(\zdom^{-1})\right\|_2\\
    =&~ \left\|\left(\mat{P}_v \mat{P}_v^\dagger\right)\cdot \diag(\zdom^{-1}) \cdot \left(\frac{1}{n}\mat{V}_{n-1}(\zdom)^\dagger \mat{V}_{n-1}(\zdom)\right) \cdot \left(\mat{P}_v\mat{P}^\dagger_v\right)-\diag(\zdom^{-1})\right\|_2\\
    \leq &~ \left\|\mat{P}_v \mat{P}_v^\dagger-\mat{I}_r\right\|_2\cdot \left\|\left(\frac{1}{n}\mat{V}_{n-1}(\zdom)^\dagger \mat{V}_{n-1}(\zdom)\right) \cdot \left(\mat{P}_v\mat{P}^\dagger_v\right)\right\|_2\\
    ~&+ \left\|\frac{1}{n}\mat{V}_{n-1}(\zdom)^\dagger \mat{V}_{n-1}(\zdom)-\mat{I}_r\right\|_2\cdot \left\|\mat{P}_v\mat{P}^\dagger_v\right\|_2\\
    ~&+ \left\|\mat{P}_v\mat{P}^\dagger_v - \mat{I}_r\right\|_2\,,
\end{align*}
where the second step follows from the triangle inequality and $\|\diag(\zdom^{-1})\|_2=1$.

By a slightly modified proof of \cref{cor:V_dagger_V_dom}, it is easy to show that
\begin{align*}
     \left\|\frac{1}{n}\mat{V}_{n-1}(\zdom)^\dagger \mat{V}_{n-1}(\zdom)-\mat{I}_r\right\|_2 = {\cal O}\left(\frac{1}{\Delta_{\vec{z}} n}\right)\,.
\end{align*}
Together with \cref{lem:vander-eigen}, we obtain that
\begin{align*}
    \left\|\mat{P}_v(\mat{Q}_\downarrow)^\dagger\mat{Q}_\uparrow\mat{P}^\dagger_v - \diag(\zdom^{-1})\right\|_2\leq {\cal O}\left(\frac{1}{\Delta_{\vec{z}} n}\right)\,
\end{align*}
which proves the first part of the claim.

Next, we prove \eqref{eqn:P_v_addition_properties_2}.

When $k=0$, \eqref{eqn:P_v_addition_properties_2} becomes
\begin{align*}
    \mat{P}_v\mat{P}^\dagger_v&=\mat{I}_r+\mathcal{O}\left(\frac{1}{\mu_r\Delta_{\vec{z}} n}\right)\,,
\end{align*}
which follows from \cref{lem:vander-eigen}.

When $k\geq 1$, we have
\begin{equation}\label{eq:induction_0}
\begin{aligned}
&\left\|\mat{P}_v\mat{\Sigma}_r^{-{k}}\mat{P}_v^\dagger-\left(n\diag(\mudom)\right)^{-k}\right\|_2\\
\leq &~ \left\|\mat{P}_v\mat{\Sigma}_r^{-{(k-1)}}\left(\mat{I_r} - \mat{P}^\dagger_v\mat{P}_v\right)\mat{\Sigma}_r^{-1}\mat{P}_v^\dagger\right\|_2 \\
~&+ \left\|\left(\mat{P}_v\mat{\Sigma}_r^{-{(k-1)}}\mat{P}^\dagger_v- \left(n\diag(\mudom)\right)^{-(k-1)}\right)\mat{P}_v\mat{\Sigma}_r^{-1}\mat{P}_v^\dagger\right\|_2\\
~&+ \left\|\left(n\diag(\mudom)\right)^{-(k-1)}\left(\mat{P}_v\mat{\Sigma}_r^{-1}\mat{P}_v^\dagger-\left(n\diag(\mudom)\right)^{-1}\right)\right\|_2\,.
\end{aligned}
\end{equation}
Since $\|\mat{\Sigma}_r^{-1}\|_2={\cal O}\left(\frac{1}{\mu_r n}\right)$ by \eqref{eqn:eigenvalue_prop} and $\|\mat{I_r} - \mat{P}^\dagger_v\mat{P}_v\|_2={\cal O}\left(\frac{1}{\Delta_{\vec{z}}n}\right)$ by \cref{lem:vander-eigen}, the first term can be bounded by:
\begin{equation}\label{eq:induction_1}
    \left\|\mat{P}_v\mat{\Sigma}_r^{-{(k-1)}}\left(\mat{I_r} - \mat{P}^\dagger_v\mat{P}_v\right)\mat{\Sigma}_r^{-1}\mat{P}_v^\dagger\right\|_2 \leq {\cal O}\left(\left(\frac{1}{\mu_r n}\right)^k\cdot \frac{1}{\Delta_{\vec{z}}n}\right)\,.
\end{equation}
Similarly, the second term can be bounded by:
\begin{equation}\label{eq:induction_2}
\begin{aligned}
    &\left\|\left(\mat{P}_v\mat{\Sigma}_r^{-{(k-1)}}\mat{P}^\dagger_v- \left(n\diag(\mudom)\right)^{-(k-1)}\right)\mat{P}_v\mat{\Sigma}_r^{-1}\mat{P}_v^\dagger\right\|_2\\ \leq &~ {\cal O}\left(\frac{1}{\mu_r n}\right)\cdot \left\|\mat{P}_v\mat{\Sigma}_r^{-{(k-1)}}\mat{P}^\dagger_v- \left(n\diag(\mudom)\right)^{-(k-1)}\right\|_2
\end{aligned}
\end{equation}
The third term can be bounded by:
\begin{equation}\label{eq:induction_3}
\begin{aligned}
    &\left\|\left(n\diag(\mudom)\right)^{-(k-1)}\left(\mat{P}_v\mat{\Sigma}_r^{-1}\mat{P}_v^\dagger-\left(n\diag(\mudom)\right)^{-1}\right)\right\|_2\\
    \leq &~ {\cal O}\left(\frac{1}{\mu_r n}\right)^{k-1}\cdot \left\|\mat{P}_v\mat{\Sigma}_r^{-1}\mat{P}_v^\dagger-\left(n\diag(\mudom)\right)^{-1}\right\|_2\,.
\end{aligned}
\end{equation}

Hence, to prove \eqref{eqn:P_v_addition_properties_2}, it suffices to bound $$\left\|\mat{P}_v\mat{\Sigma}_r^{-1}\mat{P}_v^\dagger-\left(n\diag(\mudom)\right)^{-1}\right\|_2.$$ Using the formula of $\mat{P}_v$ in \eqref{eq:def_P_v}, we have
\begin{equation}
\begin{aligned}
&\left\|\mat{P}_v\mat{\Sigma}_r^{-1}\mat{P}_v^\dagger-\left(n\diag(\mudom)\right)^{-1}\right\|_2\\
=&~\left\|\left(\frac{1}{\sqrt{n}}\mat{V}_n(\zdom)\right)^{+}\cdot \left(\mat{Q}_r\mat{\Sigma}_r^{-1}\mat{Q}_r^\dagger-\left(\frac{1}{\sqrt{n}}\mat{V}_n(\zdom)\right)\left(n\diag(\mudom)\right)^{-1}\left(\frac{1}{\sqrt{n}}\mat{V}_n(\zdom)^\dagger\right)\right)\right.\\
~&\cdot \left.{\left(\frac{1}{\sqrt{n}}\mat{V}_n(\zdom)\right)^{+}}^\dagger\right\|_2\\
=&~\mathcal{O}\left(\left\|\mat{Q}_r\mat{\Sigma}_r^{-1}\mat{Q}_r^\dagger-\frac{1}{\sqrt{n}}\mat{V}_n(\zdom)\left(n\diag(\mudom)\right)^{-1}\frac{1}{\sqrt{n}}\mat{V}_n(\zdom)^\dagger\right\|_2\right)\,,
\end{aligned}
\end{equation}
where $\left(\frac{1}{\sqrt{n}}\mat{V}_n(\zdom)\right)^{+} = \left(\frac{1}{n}\mat{V}_n(\zdom)^\dagger\mat{V}_n(\zdom)\right)^{-1}\cdot \frac{1}{\sqrt{n}}\mat{V}_n(\zdom)^\dagger$ is the pseudo-inverse of $\frac{1}{\sqrt{n}}\mat{V}_n(\zdom)$, and the second step follows from $\Big\|\left(\frac{1}{\sqrt{n}}\mat{V}_n(\zdom)\right)^{+}\Big\|_2={\cal O}(1)$ by \cref{cor:V_dagger_V_dom}.

Since $\mat{T}=\mat{Q}\mat{\Sigma}\mat{Q}^\dagger=\mat{Q}_r\mat{\Sigma}_r\mat{Q}^\dagger_r$ is the eigendecomposition, we notice that
\[
\mat{T}\mat{Q}_r\mat{\Sigma}_r^{-1}\mat{Q}_r^\dagger=\mat{Q}_r\mat{Q}_r^\dagger\,.
\]
On the other hand, since $\mat{T}=\mat{V}_n(\zdom) \diag(\vec{\mu}_{\rm dom})\mat{V}_n(\zdom)^\dagger$, we have
\[
\begin{aligned}
&\left\|\mat{T}\frac{1}{\sqrt{n}}\mat{V}_n(\zdom)\left(n\diag(\mudom)\right)^{-1}\frac{1}{\sqrt{n}}\mat{V}_n(\zdom)^\dagger-\frac{1}{n}\mat{V}_n(\zdom)\mat{V}_n(\zdom)^\dagger\right\|_2\\
= &~ \frac{1}{n}\left\|\mat{V}_n(\zdom)\diag(\vec{\mu}_{\rm dom}) \left(\frac{1}{n}\mat{V}_n(\zdom)^\dagger \mat{V}_n(\zdom)- \mat{I}_r\right)\diag(\vec{\mu}_{\rm dom})^{-1} \mat{V}_n(\zdom)^\dagger\right\|_2\\
\leq &~ \frac{\mu_1}{\mu_r}\cdot \left\|\frac{1}{n}\mat{V}_n(\zdom)^\dagger \mat{V}_n(\zdom)- \mat{I}_r\right\|_2\\
\leq&~{\cal O}\left(\frac{\mu_1}{\mu_r}\cdot \frac{1}{\Delta_{\vec{z}}n}\right)\\
=&~\mathcal{O}\left(\frac{1}{\mu_r\Delta_{\vec{z}}n}\right)\,,
\end{aligned}
\]
where the second step follows from $\Big\|\frac{1}{\sqrt{n}}\mat{V}_n(\zdom)\Big\|_2={\cal O}(1)$, the third step follows from \cref{cor:V_dagger_V_dom}, and the last step follows from $\mu_1={\Theta}(1)$. We further have
\begin{align*}
&\left\|\mat{Q}_r\mat{Q}_r^\dagger-\frac{1}{n}\mat{V}_n(\zdom)\mat{V}_n(\zdom)^\dagger\right\|_2\\
= &~ \left\|\left(\frac{1}{\sqrt{n}}\mat{V}_n(\zdom)\right)\left(\mat{P}_v\mat{P}_v^\dagger-\mat{I}_r\right)\left(\frac{1}{\sqrt{n}}\mat{V}_n(\zdom)\right)^\dagger \right\|_2\\
\leq &~ \mathcal{O}\left(\frac{1}{\Delta_{\vec{z}} n}\right)\,,
\end{align*}
where the first step follows from \cref{lem:vander-eigen}, and the second step follows from \eqref{eqn:P_orthogonal}. Combining the above three equations together, we obtain
\begin{align*}
&\left\|\mat{T}\left(\mat{Q}_r\mat{\Sigma}_r^{-1}\mat{Q}_r^\dagger-\left(\frac{1}{\sqrt{n}}\mat{V}_n(\zdom)\left(n\diag(\mudom)\right)^{-1}\frac{1}{\sqrt{n}}\mat{V}_n(\zdom)^\dagger\right)\right)\right\|_2\\
=&~ \left\|\mat{Q}_r\mat{Q}_r^\dagger - \mat{T}\frac{1}{\sqrt{n}}\mat{V}_n(\zdom)\left(n\diag(\mudom)\right)^{-1}\frac{1}{\sqrt{n}}\mat{V}_n(\zdom)^\dagger\right\|_2\\
\leq &~ \left\|\mat{Q}_r\mat{Q}_r^\dagger-\frac{1}{n}\mat{V}_n(\zdom)\mat{V}_n(\zdom)^\dagger\right\|_2 + \mathcal{O}\left(\frac{1}{\mu_r\Delta_{\vec{z}}n}\right)\\
\leq &~ 
\mathcal{O}\left(\frac{1}{\Delta_{\vec{z}} n}\right)+ \mathcal{O}\left(\frac{1}{\mu_r\Delta_{\vec{z}}n}\right)=\mathcal{O}\left(\frac{1}{\mu_r\Delta_{\vec{z}}n}\right)\,.
\end{align*}
Because 
\[
\mathrm{Ker}(\mat{T})=\mathrm{Ker}\left(\mat{Q}_r\mat{\Sigma}_r^{-1}\mat{Q}_r^\dagger\right)=\mathrm{Ker}\left(\frac{1}{\sqrt{n}}\mat{V}_n(\zdom)\left(n\diag(\mudom)\right)^{-1}\frac{1}{\sqrt{n}}\mat{V}_n(\zdom)^\dagger\right)
\]
and
\[
\mathrm{Range}(\mat{T})=\mathrm{Range}\left(\mat{Q}_r\mat{\Sigma}_r^{-1}\mat{Q}_r^\dagger\right)=\mathrm{Range}\left(\frac{1}{\sqrt{n}}\mat{V}_n(\zdom)\left(n\diag(\mudom)\right)^{-1}\frac{1}{\sqrt{n}}\mat{V}_n(\zdom)^\dagger\right)\,,
\]
we obtain
\begin{align*}
&\left\|\mat{Q}_r\mat{\Sigma}_r^{-1}\mat{Q}_r^\dagger-\frac{1}{\sqrt{n}}\mat{V}_n(\zdom)\left(n\diag(\mudom)\right)^{-1}\frac{1}{\sqrt{n}}\mat{V}_n(\zdom)^\dagger\right\|_2\\
\leq &~ \sigma_{\min}(\mat{T})^{-1}\cdot \left\|\mat{T}\left(\mat{Q}_r\mat{\Sigma}_r^{-1}\mat{Q}_r^\dagger-\left(\frac{1}{\sqrt{n}}\mat{V}_n(\zdom)\left(n\diag(\mudom)\right)^{-1}\frac{1}{\sqrt{n}}\mat{V}_n(\zdom)^\dagger\right)\right)\right\|_2\\
=&~ \mathcal{O}\left(\frac{1}{\mu_r^2\Delta_{\vec{z}} n^2}\right)\,,
\end{align*}
where the second step follows from using \cref{cor:moitra_T}. This proves the $k=1$ case of \eqref{eqn:P_v_addition_properties_2}. 

Finally, by induction \cref{eq:induction_0,eq:induction_1,eq:induction_2,eq:induction_3}, we can prove that for any $k\geq 2$, 
\begin{align*}
    &\left\|\mat{P}_v\mat{\Sigma}_r^{-{k}}\mat{P}_v^\dagger-\left(n\diag(\mudom)\right)^{-k}\right\|_2\\
    \leq &~ {\cal O}\left(\left(\frac{1}{\mu_r n}\right)^k\cdot \frac{1}{\Delta_{\vec{z}}n}\right) + \frac{1}{\mu_r n}\cdot \mathcal{O}\left(\frac{1}{\mu_r n}\right)^{k}\frac{1}{\Delta_{\vec{z}}} + {\cal O}\left(\frac{1}{\mu_r n}\right)^{k-1}\cdot \frac{1}{\mu_r^2\Delta_{\vec{z}} n^2}\\
    = &~ \mathcal{O}\left(\frac{1}{\mu_r n}\right)^{k+1}\frac{1}{\Delta_{\vec{z}}}\,.
\end{align*}
This concludes the proof of \eqref{eqn:P_v_addition_properties_2}.
\end{proof}

\begin{claim}\label{clm:bound_V_dom_E1}
\begin{align*}
    \left\|\frac{1}{\sqrt{n}}\mat{V}(\zdom)^\dagger \mat{E}_1\right\|_2=\mathcal{O}\left(\frac{\mu_{\rm tail}}{\Delta_{\vec{z}}}\right)\,.
\end{align*}    
\end{claim}
\begin{proof} Recall that
\[
\mat{E}_1=\mat{V}_n(\ztail) \cdot \diag(\mutail) \cdot \mat{V}_n(\ztail)^\dagger\,.
\]
By \cref{prop:tail-error}, we have
\begin{align*}
    \|\mat{E}_1\mat{Q}_r\|_2={\cal O}\left(\frac{\mu_{\rm tail}}{\Delta_{\vec{z}}}\right)\,.
\end{align*}
Since $\mat{Q}_r=\frac{1}{\sqrt{n}}\mat{V}_n(\zdom)\mat{P}_v$ and $\|\mat{P}_v^{-1}\|_2={\cal O}(1)$ by \cref{lem:vander-eigen}, we have
\begin{align*}
     \left\|\frac{1}{\sqrt{n}}\mat{V}(\zdom)^\dagger \mat{E}_1\right\|_2 \leq \|\mat{E}_1\mat{Q}_r\|_2 \cdot \|\mat{P}_v^{-1}\|_2 = {\cal O}\left(\frac{\mu_{\rm tail}}{\Delta_{\vec{z}}}\right)\,.
\end{align*}
The claim then follows.
\end{proof}

\section{Lower bound for spectral estimation}\label{sec:lower_bound}
In this section, we prove a sample complexity lower bound for estimating the locations of the point sources in the spectral estimation problem. Information-theoretic lower bounds for spectral estimation have been extensively studied in prior works \cite{sn89,sn90,don92,lee92,moitra_super-resolution_2015,dn15,ps15,kp16,9000636,bdg20,bgy21,ll21,li2022stability}. In particular, \cite{sn89} proved a Cram\'er-Rao lower bound for estimation the location of a single point-source with i.i.d.\ Gaussian noise, showing that for any \emph{unbiased estimator} using $n$ samples, the variance is at least $1/n^3$. And \cite{ll21} proved a min-max lower bound in terms of the ``super-resolution factor'' $1/(\Delta_{\vec{z}}n)$ with \emph{arbitrary} noise.
The settings of these lower bounds are somewhat different from ours and the proofs are complicated.
Below, we give a simple construction showing, that in our noisy super-resolution setting, the location estimation error is $\Omega(n^{-1.5})$ and the intensity estimation error is $\Omega(n^{-0.5})$.

\begin{theorem}[Spectral estimation lower bound, {\normalfont formal version of \cref{thm:lower_bound}}]\label{thm:lower_bound_formal}
Let $\{g_j\}_{j=1}^n$ be defined in \eqref{eq:noisy-signal}, and $\{E_j\}_{j=1}^n\sim {\cal N}(0,1)$ be i.i.d.\ standard Gaussians. Then, 
\begin{itemize}
    \item \textbf{Location estimation.} For any $\eta>0$, one cannot estimate the locations $\{z_i\}_{i=1}^r$ up to $\order(n^{-1.5-\eta})$ matching distance error for all $i=1,\ldots,r$ with constant success probability.
    \item \textbf{Intensity estimation.} For any $\eta>0$, one cannot estimate the intensities $\{\mu_i\}_{i=1}^r$ within $\order(n^{-0.5-\eta})$ matching distance error for all $i=1,\ldots,r$ with constant success probability.
\end{itemize}
\end{theorem}

To prove this theorem, we will employ the following lemma \cite[Thm.~1.8]{aal23}:

\begin{lemma}[TV distance of standard Gaussian distributions]\label{lem:gaussian_tv}
Let $\vec{\mu}_1,\vec{\mu}_2\in \C^n$. Then, 
\begin{align*}
    d_{\rm TV}\left({\cal N}(\vec{\mu}_1, \sigma^2 \Id_n), {\cal N}(\vec{\mu}_2, \sigma^2 \Id_n)\right)\leq \|\vec{\mu}_1-\vec{\mu}_2\|_2\,.
\end{align*}
\end{lemma}

\begin{proof}[Proof of \cref{thm:lower_bound_formal}]
We first show lower bound for location estimation.
Let $\epsilon, \eta \in (0,1)$. Suppose there exists an algorithm that can estimate all the locations within $\epsilon={\cal O}(n^{-1.5-\eta})$ error using $n$ noisy measurements with constant success probability.
Consider two spectral densities with parameters $(\mu_1=1, z_1 = e^{2i\epsilon})$ and $(\mu_1'=1, z_1'=1)$, respectively. Since
\begin{align*}
    |z_1 - z'_1| = |e^{2i\epsilon} - 1|=2\sin(\epsilon)>\epsilon\,,
\end{align*}
the algorithm can distinguish these two spectral densities with constant probability. On the other hand, it is easy to see that the $n$ measurement outcomes for each spectral density are sampled from an $n$-dimensional Gaussian distribution ${\cal N}(\vec{g}, \Id_n)$ (or ${\cal N}(\vec{g}', \Id_n)$), where
\begin{align*}
    \vec{g}_j = e^{2i\epsilon j}\quad \text{and}\quad \vec{g}'_j=1 \for j = 0,\ldots,n-1.
\end{align*}
Then, we have
\begin{align*}
    \|\vec{g}-\vec{g}'\|_2^2=\sum_{j=1}^n |e^{2ij\epsilon}-1|^2=\sum_{j=1}^n 4|\sin(j\epsilon)|^2\leq \sum_{j=1}^n 4(j\epsilon)^2 = {\cal O}(n^3 \epsilon^2)= {\cal O}(n^{-2\eta})\,.
\end{align*}
Thus, \cref{lem:gaussian_tv}, we conclude
\begin{align*}
    d_{\rm TV}\left({\cal N}(\vec{g}, \Id_n), {\cal N}(\vec{g}', \Id_n)\right)\leq \|\vec{g}-\vec{g}'\|_2 = {\cal O}(n^{-2\eta})\,.
\end{align*}
Therefore, we can only distinguish the two spectral densities with probability ${\cal O}(n^{-2\eta}) = o(1)$, contradicting the assumption that the algorithm succeeds with constant probability.

Next, we prove the lower bound for intensity estimation.
Let $\epsilon, \eta \in (0,1)$.
Suppose there exists an algorithm that can estimate all the locations within $\epsilon={\cal O}(n^{-0.5-\eta})$ error using $n$ noisy measurements with constant success probability.
Consider the spectral densities with parameters $(\mu_1=1, z_1 = 1)$ and $(\mu_1'=1+2\epsilon, z_1'=1)$, respectively.
Then, the algorithm can distinguish them with constant probability since $|\mu_1 - \mu_1'|>\epsilon$.
As above, the $n$ measurement outcomes are samples from ${\cal N}(\vec{1}, \Id_n)$ (or ${\cal N}((1+2\epsilon)\vec{1}, \Id_n)$). By \cref{lem:gaussian_tv}, their TV distance is at most 
\begin{align*}
    d_{\rm TV}\left({\cal N}(\vec{1}, \Id_n), {\cal N}((1+2\epsilon)\vec{1}, \Id_n)\right)\leq \|\vec{1}-(1+2\epsilon)\vec{1}\|_2={\cal O}(n^{0.5}\epsilon) = {\cal O}(n^{-\eta}),
\end{align*}
again giving a contradiction.
\end{proof}

\clearpage
\phantomsection
\addcontentsline{toc}{section}{References}
\bibliographystyle{alpha}
\bibliography{ref,more_refs}
\end{document}